\documentclass[fleqn,11pt]{article}
\usepackage{amsfonts}
\usepackage{eurosym}
\usepackage{amssymb}
\usepackage{amsmath}
\usepackage{graphicx}
\usepackage{float}
\usepackage{color}
\usepackage{natbib}
\usepackage{hyperref}
\usepackage[a4paper, margin=1in]{geometry}
\usepackage[nodisplayskipstretch]{setspace}

\hypersetup{
colorlinks=true,
linkcolor=black,
citecolor=black,
filecolor=magenta,
urlcolor=black}
\newtheorem{theorem}{Theorem}
\newtheorem{example}{Example}

\newtheorem{assumption}{Assumption}
\newtheorem{proofproposition}{Proof of Proposition}
\newtheorem{proofcorollary}{Proof of Corollary}
\newtheorem{proposition}{Proposition}
\newtheorem{corollary}{Corollary}
\newtheorem{lemma}{Lemma}
\numberwithin{lemma}{section}
\newtheorem{remark}{Remark}

\newenvironment{proof}[1][Proof]{\noindent \textbf{#1.} }{\rule{0.5em}{0.5em}}

\begin{document}
{\setstretch{1.1} 
\title{Variable Selection in High Dimensional Linear Regressions with
Parameter Instability\thanks{{\scriptsize We are grateful to Elie Tamer
(Editor), two anonymous reviewers and an associate editor, for their
constructive comments and helpful suggestions. We have also benefited from
discussions and comments by George Kapetanios, Oliver Linton, Ron Smith, and
seminar participants at Cambridge University. The views expressed in this
paper are those of the authors and do not necessarily reflect those of the
Federal Reserve Bank of Dallas or the Federal Reserve System. This research
was supported in part through computational resources provided by the
Big-Tex High Performance Computing Group at the Federal Reserve Bank of
Dallas. This paper in part was written when Sharifvaghefi was a doctoral
student at the University of Southern California (USC). Sharifvaghefi
gratefully acknowledges financial support from the Center for Applied
Financial Economics at USC.}}}
\date{\today}
\author{Alexander Chudik \\
{\footnotesize Federal Reserve Bank of Dallas, Dallas, USA} \and M. Hashem
Pesaran \\
{\footnotesize University of Southern California, Los Angeles, USA and
Trinity College, Cambridge, UK} \and Mahrad Sharifvaghefi\thanks{%
{\scriptsize Corresponding author. Postal address: 230 S Bouquet St.,
Pittsburgh, PA, USA, 15260. Email: \href{mailto:sharifvaghefi@pitt.edu}{%
sharifvaghefi@pitt.edu}. }} \\
{\footnotesize University of Pittsburgh, Pittsburgh, USA} }
\maketitle

\begin{abstract}
This paper considers the problem of variable selection allowing for
parameter instability. It distinguishes between signal and pseudo-signal
variables that are correlated with the target variable, and noise variables
that are not, and investigate the asymptotic properties of the One Covariate
at a Time Multiple Testing (OCMT) method proposed by \cite{chudik2018one}
under parameter insatiability. It is established that OCMT continues to
asymptotically select an approximating model that includes all the signals
and none of the noise variables. Properties of post selection regressions
are also investigated, and in-sample fit of the selected regression is shown
to have the oracle property. The theoretical results support the use of
unweighted observations at the selection stage of OCMT, whilst applying
down-weighting of observations only at the forecasting stage. Monte Carlo
and empirical applications show that OCMT without down-weighting at the
selection stage yields smaller mean squared forecast errors compared to
Lasso, Adaptive Lasso, and boosting.
\end{abstract}


}
\thispagestyle{empty} \pagenumbering{arabic} \newpage

\section{Introduction}

\label{introduction}

Models fitted to statistical relationships could be subject to parameter
instabilities. In an extensive early study, \cite{StockWatson1996} find that
a large number of time series regressions in economics are subject to
breaks. \cite{ClementsHendry1998} consider parameter instability to be one
of the main sources of forecast failure. This problem has been addressed at
the estimation/forecasting stage for a given set of selected regressors.
Typical solutions are either to use rolling windows or exponential
down-weighting. For instance, \cite{pesaran2007selection}, \cite%
{pesaran2011forecast} and \cite{inoue2017rolling} consider the choice of an
observation window, and \cite{hyndman2008forecasting} and \cite%
{pesaran2013optimal}, respectively consider exponential and non-exponential
down-weighting of the observations. There are also Bayesian approaches to
prediction that allow for the possibility of breaks over the forecast
horizon, such as \cite{chib1998estimation}, \cite{koop2004forecasting}, and 
\cite{pesaran2006forecasting}. \cite{rossi2013advances} provides a review of
the literature on forecasting under instability. There are also related time
varying parameter and regime switching models that are used for forecasting.
See, for example, \cite{hamilton1988rational} and \cite{dangl2012predictive}%
. This literature does not address the problem of variable selection and
takes the model specification as given.

The theory of variable selection in the presence of parameter instability is
still largely underdeveloped. The application of penalized regression
methods to variable selection is often theoretically justified under two key
parameter stability assumptions: the stability of the coefficients in the
data generating process and the stability of the correlation matrix of the
covariates in the active set. Under these assumptions, the penalized
regression methods can proceed using the full sample without down-weighting
or separating the variable selection from the estimation stage. However, in
the presence of parameter instability penalized regression methods must be
adapted to simultaneously deal with selection and parameter change. There
are a number of recent studies that use machine learning techniques to allow
for parameter instability, in particular penalized regression, especially
the Least Absolute Shrinkage and Selection Operator (Lasso) initially
proposed by \cite{tibshirani1996regression}. For example, \cite%
{qian2016shrinkage} consider a linear regression model with a finite number
of covariates but allow for an unknown number of breaks and use group fused
Lasso by \cite{10.1007/978-3-642-40728-4_9} to consistently estimate the
number of breaks and their locations. \cite{lee2016lasso} have proposed a
Lasso procedure that allows for threshold effects. \cite{kapetanios2018time}
have proposed a time-varying Lasso procedure, where all the parameters of
the model vary locally. \cite{fan2014nonparametric} suggest an extension of
the screening procedure initially proposed by \cite{fan2008sure} to the case
where the regression coefficients vary smoothly with an observable exposure
variable. Also recently, \cite{yousuf2019boosting} propose an interesting
boosting procedure for the estimation of high-dimensional models with
locally time varying parameters. These studies focus on specific forms of
discrete or continuous time varying parameter models, and often carry out
variable selection and estimation simultaneously using the penalized
regression or boosting procedures.

This paper proposes the use of One Covariate at a Time Multiple Testing
(OCMT) procedure proposed by \cite{chudik2018one} which is readily adapted
to the task of variable selection under parameter instability. The key
insight comes from the fact that coefficients of the noise variables that do
not enter the data generating process are zero at all times. Consequently,
using unweighted observations at the variable selection stage will be most
effective in removing noise variables, while using weighted observations at
the estimation stage can provide gains in terms of mean squared forecast
errors. In this study, we allow the marginal effects of signals on the
target variable, as well as the correlation of the covariates under
consideration, to vary over time, assuming time variations in the marginal
effects are not correlated with the signals. We establish the conditions
required for OCMT with unweighted observations to select a model that
contains all the signal variables and none of the noise variables with
probability approaching one as the sample size, $T$, and the number of
covariates under consideration, $N$, tend to infinity.

Clearly, it is also possible to use penalized regression methods with
unweighted observations for the purpose of variable selection, and then
estimate the selected model by the least squares method using weighted
observations. However, as far as we know, there are no studies that consider
the choice of the penalty term to achieve variable selection consistency
under parameter instability. It is hoped that the present paper provides an
impetus for further theoretical analysis of penalized regression techniques
under parameter instability. Although at this stage a comparison of the assumptions required for variable selection consistency of OCMT and Lasso under parameter instability is not possible, in Section \ref{sec:lasso_vs_ocmt} we provide a discussion of the assumptions required for the variable selection consistency of Lasso under parameter stability that are comparable with the those required for the validity of the OCMT procedure.

The OCMT procedure selects variables based on the statistical significance
of the net effect of the covariates in the active set on the target
variable, one-at-a-time subject to the multiple testing nature of the
inferential problem involved. The idea of using one-at-a-time regressions is
not unique to OCMT and has been used in boosting as well as in screening
approaches. See, for example, \cite{buhlmann2006boosting} and \cite%
{fan2018sure} as prominent examples of these approaches. What is unique
about the OCMT procedure is its inferentially motivated stopping rule
without resorting to the use of information criteria, or penalized
regression after the initial stage.\ In the case of models with stable
parameters, \cite{chudik2018one} establish that OCMT asymptotically selects
an approximating model that includes all the signals and none of the noise
variables. This model can contain covariates that do not enter the data
generating process for the target variable but exhibit non-zero correlation
with at least one signal, known as pseudo-signals.

Lasso and OCMT exploit different aspects of the low-dimensional structure
assumed for the underlying data generating process. Lasso restricts the
magnitude of the correlations within signals as well as the correlations
between signals and the remaining covariates in the active set. OCMT limits
the rate at which the number of pseudo-signals, $k_{T}^{\ast }$, rises with
the sample size, $T$. Under parameter stability, the variable selection
consistency of Lasso has been investigated by \cite{zhao2006model}, \cite%
{meinshausen2006high} and more recently by \cite{lahiri2021necessary}. These
conditions, and how they compare with the conditions that underlie OCMT, are
discussed in Section \ref{sec:lasso_vs_ocmt} of the paper. Although Lasso
does not directly impose any restrictions on $k_{T}^{\ast }$, its
Irrepresentable Condition (IRC), by restricting the magnitude of
correlations within and between the signals and pseudo-signals, does have
implications for the number of pseudo-signals that Lasso selects. OCMT
requires $k_{T}^{\ast }$\ not to rise faster than $\sqrt{T}$. When this
condition is violated, then the true signals must end up as common factors
for the pseudo-signals, and what matters is the number of residuals (from
the regressions of pseudo-signals on the common factors) that are correlated
with the residuals of the true signals from the same set of common factors. 
\cite{sharifvaghefi2023variable} shows that such common factors can be
estimated from the principal components of the covariates in the active set
and the OCMT\ condition on the number pseudo-signals, now defined in terms
of the correlation of the residuals, is no longer restrictive.\footnote{%
Another extension of OCMT\ is provided by \cite{su2023one} who allow for
unknown potentially non-linear relationship between the signals and the
target variable.} Once the model is selected, Theorem \ref{estimation
consistency} establishes how the convergence rate of estimated coefficients
of the selected variables depends on $k_{T}^{\ast }$. The regular
convergence rate of $\sqrt{T}$ is achieved only if $k_{T}^{\ast }$ is fixed
in $T$. A similar issue also arises for Lasso, as shown by \cite%
{lahiri2021necessary} who establishes that the Lasso procedure cannot
achieve both variable selection consistency and $\sqrt{T}$-consistency in
coefficient estimation. As noted above, the focus of the present paper is on
the application of OCMT to variable selection in the presence of parameter
instability, broadly defined. To the best of our knowledge, there are no
studies that investigate the variable selection properties of Lasso under
parameter instability.

To take account of the time variations in the coefficients of the signals,
we consider their time averages and distinguish between strong signals whose
average marginal effects go to a non-zero value, semi-strong signals whose
average marginal effects converge to zero, but sufficiently slow, and weak
signals whose average marginal effects approach to zero quite fast. In this
way we allow for variety of time variations that could arise in practice.
Strong signals tend to have non-zero effects at all times, semi-strong
signals could have zero effects during some periods, with weak signals enter
the model relatively rarely. Weak signals are often indistinguishable from
noise variables. In our theoretical analysis we will focus on selection of
strong and semi-strong signals.

We provide three main theorems in support of our proposed variable selection
method. Under certain fairly general regularity conditions we show that the
probability of OCMT selecting the approximating model that contains all the
signals (strong and semi-strong) and none of the noise variables approaches
to one as $T $ goes to infinity. Our results apply both when $N$ is fixed as
well as when $N$ goes to infinity jointly with $T$, covering the case where $%
N \gg T$. We also establish conditions under which (a) least squares
estimates of the coefficients of selected covariates converge to zero unless
they are signals, and (b) the average squared residuals of the selected
model achieves the oracle rate for regression models with time-varying
coefficients. These theoretical findings provide a formal justification for
application of statistical techniques from the time-varying parameters
literature to the post OCMT selected model. Our Monte Carlo experiments show
that the OCMT procedure with weighted observations only at the estimation
stage outperforms, in terms of mean squared forecast errors, Lasso and
Adaptive Lasso (A-Lasso by \cite{zou2006adaptive}), as well as boosting by 
\cite{buhlmann2006boosting}, under many different settings.

Finally, we provide three empirical applications, forecasting monthly rates
of price changes of 28 stocks in Dow Jones using large number of financial,
economic and technical indicators, forecasting output growths across 33
countries using a large number of macroeconomic indicators, and forecasting
euro area output growth using ECB surveys of 25 professional forecasters. To
save space the third application is included in the online supplement. We
generate a large number of forecasts using OCMT with and without
down-weighting of the observations at the selection stage and compare the
results with the forecasts obtained using Lasso, A-Lasso and boosting. The
empirical results are in line with our theoretical and MC findings and
suggest that using down-weighted observations at the selection stage of the
OCMT procedure worsens forecast performance in terms of mean squared
forecast errors and mean directional forecast accuracy. The empirical
results also provide that OCMT with no down-weighting at the selection stage
outperforms, in terms of mean squared forecast errors, boosting, Lasso and
A-Lasso.

The rest of the paper is organized as follows: Section \ref{model setting}
sets out the model specification. Section \ref{ocmt method} explains the
basic idea behind the OCMT procedure for variable selection without
down-weighting in the presence of parameter instability. {Section \ref%
{sec:lasso_vs_ocmt} provides a discussion of key assumptions of Lasso and
OCMT under parameter stability.} Section \ref{asymptotic properties}
discusses the technical assumptions and the asymptotic properties of the
OCMT procedure under parameter instability. Section \ref{sec:MC-studies}
provides the details of the Monte Carlo experiments and a summary of the
main results. Section \ref{empirical section} presents the empirical
applications, and Section \ref{conclusion} concludes. The paper is also
accompanied with three online supplements. A theory supplement contains the
mathematical proofs of the theorems and related lemmas. A Monte Carlo
supplement provides additional summary tables, the full set of Monte Carlo
results, as well as the description of the algorithms used for Lasso,
A-Lasso and boosting. Further details of the empirical applications are
given in an empirical supplement.

\textbf{Notations:} Generic finite positive constants are denoted by $C_{i}$
for $i=1,2,\cdots $. $\lVert \mathbf{A}\rVert _{2}$ and $\lVert \mathbf{A}%
\rVert _{F}$ denote the spectral and Frobenius norms of matrix $\mathbf{A}$,
respectively. $\text{tr}(\mathbf{A})$ and $\lambda _{i}(\mathbf{A})$ denote
the trace and the $i^{th}$ eigenvalue of a square matrix $\mathbf{A}$,
respectively. $\left\Vert \mathbf{x}\right\Vert $ denotes the $\ell _{2}$
norm of vector $\mathbf{x}$. {\ If $\{f_{n}\}_{n=1}^{\infty }$ and $%
\{g_{n}\}_{n=1}^{\infty }$ are both positive sequences of real numbers, then
we say $f_{n}=\ominus (g_{n})$ if there exist $n_{0}\geq 1$ and positive
constants $C_{0}$ and $C_{1}$, such that $\inf_{n\geq n_{0}}\left(
f_{n}/g_{n}\right) \geq C_{0}$ and $\sup_{n\geq n_{0}}\left(
f_{n}/g_{n}\right) \leq C_{1}$. }Similarly, if $f_{iT}$\ and $g_{iT}$\ are
positive double sequences of real numbers for $i=1,2,3,\cdots $; and $%
T=1,2,3,\cdots $, then $f_{iT}=\ominus (g_{iT})$\ if there exist $T_{0}\geq
1 $\ and positive constants $C_{0}$\ and $C_{1}$, such that $\inf_{T\geq
T_{0}}\left( f_{iT}/g_{iT}\right) \geq C_{0}$\ and $\sup_{T\geq T_{0}}\left(
f_{iT}/g_{iT}\right) \leq C_{1}$.

\section{Model specification under parameter instability\label{model setting}%
}

We consider the following data generating process (DGP) for the target
variable, $y_{t}$, in terms of the signal variables $(x_{it}$, for $%
i=1,2,...,k$) 
\begin{equation}
\textstyle y_{t}=\sum_{i=1}^{k}\beta _{it}x_{it}+u_{t},\text{ for }%
t=1,2,\cdots ,T  \label{dgp y_t}
\end{equation}%
with time-varying parameters, $\left\{ \beta _{it}\text{, }
i=1,2,...,k\right\} $, and an error term, $u_{t}$. Intercepts and other
pre-selected variables can also be included.\footnote{%
See the working paper version of the paper available at %
\url{https://doi.org/10.24149/gwp394r2}.} Since the parameters are
time-varying we refer to the covariate $i$ as \textquotedblleft \textit{%
signal}" if its average marginal effect, $\bar{\beta}_{i,T}=T^{-1}%
\sum_{t=1}^{T}\mathbb{E}(\beta _{it})$, is not equal to zero. The strength
of the signal can be captured by the exponent coefficient $\vartheta _{i}$
in $\bar{\beta}_{i,T}=\ominus (T^{-\vartheta _{i}})$. For $\vartheta _{i}=0$%
, the signal is strong and the average marginal effect, $\bar{\beta}_{i,T}$,
does not converge to zero. For $0<\vartheta _{i}<1/2$, the signal is
semi-strong and the average marginal effect converges to zero, but not too
fast. For $\vartheta _{i}\geq 1/2$, the average marginal effect tends to
zero very fast, making it infeasible for the OCMT procedure to distinguish
such weak signals from noise, unless weak signals are sufficiently
correlated with at least one strong or semi-strong signal. In this paper, we
do not impose any restrictions on the correlations among signals, and we
focus only on the covariates with strong and semi-strong signals, where $%
0\leq \vartheta _{i}<1/2$. For simplicity of exposition, unless specified
otherwise, we will refer to both strong and semi-strong signals simply as
signals.

The identity of the $k$ signals are unknown, and the task facing the
investigator is to select the signals from a set of covariates under
consideration, $\mathcal{S}_{Nt}=\{x_{1t},x_{2t},\cdots ,x_{Nt}\}$, known as
the active set, with $N$, the number of covariates in the active set,
possibly much larger than $T$, the number of data points available for
estimation prior to forecasting. The time variations in $\beta _{it}$, for $%
i=1,2,...,k$, are assumed to be exogenous, in the sense that $\beta _{it}$
are distributed independently of the covariates in the active set $\mathcal{S%
}_{Nt}$. This assumption rules out correlated time variations that can arise
in non-linear regressions where $y_{t}$ is a non-linear function of the
signals. One important example is given by the bilinear model 
\begin{equation*}
y_{t}=\sum_{i=1}^{k}\beta _{i}(x_{it})x_{it}+u_{t},  \label{nonL}
\end{equation*}%
where it is assumed that $\beta _{it}$ systematically varies with $x_{it}$.
Nevertheless, in the context of linear regressions, our assumptions about
parameter instability includes many models of parameter instability studied
in the literature. Specifically, our analysis accommodates cases where the
coefficients vary continuously following a stochastic process as in the
standard random coefficient model, 
\begin{equation*}
\beta _{it}=\beta _{i}+\sigma _{it}\xi _{it},  \label{rcm}
\end{equation*}%
or could change at discrete time intervals, as 
\begin{equation*}
\beta _{it}=\beta _{i}^{\left( s\right) },\text{ if }t\in \lbrack
T_{s-1},T_{s})\text{ for }s=1,2,\cdots ,S,  \label{Disc}
\end{equation*}%
where $T_{0}=1$ and $T_{S}=T$.

In this paper we follow \cite{chudik2018one} and consider the application of
the OCMT procedure for variable selection even when the parameters are
time-varying, and provide theoretical arguments in favour of using the full
sample of data available without down-weighting. We first recall that OCMT's
variable selection is based on the net effect of $x_{it}$ on $y_{t}$.
However, when the regression coefficients and/or the correlations across the
covariates in the active set are time-varying, the net effects will also be
time-varying and we need to base our selection on average net effects. The
average net effect of the covariate $x_{it}$ on $y_{t}$ can be defined as 
\begin{equation*}
\textstyle\bar{\theta}_{i,T}=T^{-1}\sum_{t=1}^{T}\mathbb{E}(x_{it}y_{t}).
\end{equation*}%
By substituting $y_{t}$ from (\ref{dgp y_t}) we can further write $\bar{%
\theta}_{i,T}$ as (noting that $\beta _{jt}$ and $x_{it}$ are assumed to be
independently distributed) 
\begin{equation*}
\bar{\theta}_{i,T}=\sum_{j=1}^{k}\left( T^{-1}\sum_{t=1}^{T}\mathbb{E}(\beta
_{jt})\sigma _{ij,t}\right) +\bar{\sigma}_{iu,T},
\end{equation*}%
where $\sigma _{ij,t}=\mathbb{E}(x_{it}x_{jt})$, and $\bar{\sigma}%
_{iu,T}=T^{-1}\sum_{t=1}^{T}\mathbb{E}(x_{it}u_{t})$. In what follows we
allow for a mild degree of correlation between $x_{it},$ and $u_{t},$ by
assuming that $\bar{\sigma}_{iu,T}=O(T^{-\epsilon _{i}}),$ for some $%
\epsilon _{i} \geq 1/2$. In this case the average net effect of the $i^{th}$
covariate simplifies to 
\begin{equation*}
\bar{\theta}_{i,T}=\sum_{j=1}^{k}\left( T^{-1}\sum_{t=1}^{T}\mathbb{E}(\beta
_{jt})\sigma _{ij,t}\right) +O(T^{-\epsilon _{i}}).
\end{equation*}

In line with our assumption about the average marginal effects, namely that $%
\bar{\beta}_{i,T}=\ominus (T^{-\vartheta _{i}}),$ for some $0\leq \vartheta
_{i}<1/2$, we distinguish between covariates with strong and semi-strong net
effects, and the noise variables whose net effects, averaged over time, tend
to zero sufficiently fast. Specifically, for covariates with strong or
semi-strong net effects we set $\bar{\theta}_{i,T}=\ominus (T^{-\vartheta
_{i}}),$ for some $0\leq \vartheta _{i}<1/2$, and for the noise variables we
shall assume that $\bar{\theta}_{i,T}=\ominus (T^{-\epsilon _{i}}),$ for
some $\epsilon _{i}\geq 1/2$.

In what follows, we first describe the OCMT procedure and then discuss the
conditions under which the approximating model (that includes all the
signals and none of the noise variables) is selected with probability
approaching one by OCMT.

\section{Parameter instability and OCMT\label{ocmt method}}

The OCMT procedure begins with $N$ separate regressions, for each of the $N$
covariates in the active set $\mathcal{S}_{Nt}$. Specifically, the focus is
on the statistical significance of $\phi _{i,T}$ in the following simple
regressions: 
\begin{equation}
y_{t}=\phi _{i,T}x_{it}+\eta _{it},\text{ for }t=1,2,\cdots ,T;\text{ }%
i=1,2,...,N,  \label{eq:ocmt_reg}
\end{equation}%
where%
\begin{equation}
\phi _{i,T}\equiv \left( T^{-1}\sum_{t=1}^{T}\mathbb{E}(x_{it}^{2})\right)
^{-1}\left( T^{-1}\sum_{t=1}^{T}\mathbb{E}(x_{it}y_{t})\right) =\left[ \bar{%
\sigma}_{ii,T}\right] ^{-1}\bar{\theta}_{i,T}\text{,}  \label{phiiTdef}
\end{equation}%
with $\bar{\sigma}_{ii,T}=$ $T^{-1}\sum_{t=1}^{T}\sigma _{ii,t}$. Due to
non-zero cross-covariate correlations, knowing whether $\phi _{i,T}$ $\ ($or
equivalently $\bar{\theta}_{i,T}$) is zero does not necessarily allow us to
establish whether $\bar{\beta}_{i,T}$ is sufficiently close to zero or not.
There are four possibilities:

\vspace{0.3cm}

{\small \renewcommand{\arraystretch}{0.8}%
\begin{tabular}{l|l}
\hline
(I) \textit{Signals} & $\bar{\beta}_{i,T}=\ominus (T^{-\vartheta
_{i}})^{\dagger }$ and $\bar{\theta}_{i,T}=\ominus (T^{-\vartheta _{i}})$ \\ 
(II) \textit{Hidden Signals} & $\bar{\beta}_{i,T}=\ominus (T^{-\vartheta
_{i}})$ and $\bar{\theta}_{i,T}=\ominus (T^{-\epsilon _{i}})$ \\ 
(III) \textit{Pseudo-signals} & $\beta _{it}=0$ for all $t$ and $\bar{\theta}%
_{i,T}=\ominus (T^{-\vartheta _{i}})$ \\ 
(IV) \textit{Noise variables} & $\beta _{it}=0$ for all $t$ and $\bar{\theta}%
_{i,T}=\ominus (T^{-\epsilon _{i}})$ \\ \hline
\end{tabular}%
}

{\footnotesize {$\dagger $ The signals are assumed to be (semi) strong such
that $0\leq \vartheta _{i}<1/2$.}}

\noindent for some $0\leq \vartheta _{i}<1/2$, and $\epsilon _{i}\geq 1/2$. {%
To simplify the exposition, we consider the covariates $x_{it}$, for $%
i=1,2,\cdots ,k$, as signals, and for $i=k+1,k+2,\cdots ,k+k_{T}^{\ast }$,
as pseudo-signals. The remaining covariates in the active set, }$\left\{ {%
x_{it},}\text{ for }{i=k+k_{T}^{\ast }+1,k+k_{T}^{\ast }+2,\cdots ,N}%
\right\} ${, are classified as (pure) noise variables. We assume that the
number of signals, $k$, is a finite fixed integer but we allow the number of
pseudo-signals, denoted by $k_{T}^{\ast }$, to grow with $N$ and $T$.}
Notice, if the covariate $x_{it}$ is a noise variable, then $\bar{\theta}%
_{i,T}$ converges to zero very fast. Therefore, down-weighting of
observations at the variable selection stage is likely to be inefficient for
eliminating the noise variables. Moreover, for a signal to remain hidden, we
need the terms of higher order, $\ominus (T^{-\vartheta _{j}})$ with $0\leq
\vartheta _{i}<1/2$, to \textit{exactly} cancel out such that $\theta _{i,T}$
becomes a lower order, i.e. $\ominus (T^{-\varepsilon _{i}})$, that tends to
zero at a sufficiently fast rate (with $\epsilon _{i}\geq 1/2)$. This
combination of events seem quite unlikely, and to simplify the theoretical
derivations in what follows we abstract from such a possibility and assume
that there are no hidden signals and consider a single stage version of the
OCMT procedure for variable selection. To allow for hidden signals, \cite%
{chudik2018one} extend the OCMT method to have multiple stages. \vspace{%
-0.25cm}

\begin{flushleft}
\textbf{The OCMT procedure}\vspace{-0.35cm}
\end{flushleft}

{\itshape}

\begin{enumerate}
\item For $i=1,2,\cdots ,N$, regress $y_{t}$ on $x_{it}$; $y_{t}=\phi
_{i,T}x_{it}+\eta _{it}$; and compute the $t$-ratio of $\phi _{i,T}$, given
by 
\begin{equation}  \label{eq:t-stat}
t_{i,T}=\frac{\hat{\phi}_{i,T}}{s.e.\left( \hat{\phi}_{i,T}\right) }=\frac{%
\sum_{t=1}^{T}x_{it}y_{t}}{\hat{\sigma}_{i}\sqrt{\sum_{t=1}^{T}x_{it}^{2}}},
\end{equation}%
where $\hat{\phi}_{i,T}=\left( \sum_{t=1}^{T}x_{it}^{2}\right) ^{-1}\left(
\sum_{t=1}^{T}x_{it}y_{t}\right) $ is the least squares estimator of $\phi
_{i,T}$, $\hat{\sigma}_{i}^{2}=T^{-1}\sum_{t=1}^{T}\hat{\eta}_{it}^{2}$, and 
$\hat{\eta}_{it}=y_{t}-\hat{\phi}_{i,T}x_{it}$, is the regression residual.

\vspace{-0.25cm}

\item Consider the critical value function, $c_{p}(N,\delta )$, defined by 
\begin{equation}
c_{p}(N,\delta )=\Phi ^{-1}\left( 1-\frac{p}{2N^{\delta }}\right) ,
\label{cv_function}
\end{equation}%
where $\Phi ^{-1}(.)$ is the inverse of a standard normal distribution
function, $\delta $ is a finite positive constant, and $p$ is the nominal
size of the tests to be set by the investigator.

\vspace{-0.25cm}

\item Given $c_{p}(N,\delta )$, the selection indicator is given by 
\begin{equation}
\hat{\mathcal{J}}_{i}=\mathbf{\mathit{I}}\left[ \lvert t_{i,T}\rvert
>c_{p}(N,\delta )\right] ,\ \text{for}\ i=1,2,\cdots ,N.
\label{selection indicator}
\end{equation}%
The covariate $x_{it}$ is selected if $\hat{\mathcal{J}}_{i}=1$.
\end{enumerate}

{OCMT uses the t-ratio of $\phi _{i,T}$, defined by (\ref{eq:t-stat}), to
select the signals (strong as well as semi-strong), $\{x_{it}:i=1,2,\cdots
,k\}$, and none of the noise variables, $\{x_{it}:k+k_{T}^{\ast
}+1,k+k_{T}^{\ast }+2,\cdots ,N\}$. The selected model is referred to as an
approximating model since it can include pseudo-signals, $%
\{x_{it}:k+1,k+2,\cdots ,k+k_{T}^{\ast }\}$, that proxy for the true
signals. To deal with the multiple testing nature of the problem, the
critical value $c_{p}(N,\delta )$ used for the separate-induced tests is
chosen to be an appropriately increasing function of $N$, by setting $\delta
>0$. The choice of $\delta $ is guided by our theoretical derivations, to be
discussed below in Section \ref{asymptotic properties}. }

Before presenting our technical assumptions and theoretical results under
parameter instability, it is instructive to discuss and compare the key
conditions under which Lasso and OCMT lead to consistent model selection
under parameter stability.

\section{Lasso and OCMT under parameter stability\label{sec:lasso_vs_ocmt}}

{As formally established by \cite{zhao2006model} and \cite%
{meinshausen2006high}, three main conditions are required for the Lasso
variable selection to be consistent. Here we follow \cite%
{lahiri2021necessary} who also considers the convergence of Lasso estimated
coefficients to their true values. The key condition is the
\textquotedblleft Irrepresentable Condition" (IRC) that places restrictions
on the magnitudes of the sample correlations across the signals, }${\mathbf{x%
}_{1t}=(x_{1t},x_{2t},\cdots ,x_{kt})^{\prime }}$, and the rest of the {%
covariates in the active set, namely $\mathbf{x}_{2t}=(x_{k+1,t},x_{k+2,t},%
\cdots ,x_{Nt})^{\prime }$. Let 
\begin{equation*}
\mathbf{R}=%
\begin{pmatrix}
\mathbf{R}_{11} & \mathbf{R}_{12} \\ 
\mathbf{R}_{21} & \mathbf{R}_{22}%
\end{pmatrix}%
\end{equation*}%
be the $N\times N$ matrix of sample correlations of the covariates in the
active set, partitioned conformably to $\mathbf{x}_{t}=(\mathbf{x}%
_{1t}^{\prime },\mathbf{x}_{2t}^{\prime })^{\prime }$. The IRC can be
written as 
\begin{equation}
\left\Vert \mathbf{R}_{21}\mathbf{R}_{11}^{-1}\text{sign}\left( \boldsymbol{%
\beta }_{0}\right) \right\Vert _{\infty }\leq 1,  \label{eq:ir_con}
\end{equation}%
where $\left\Vert .\right\Vert _{\infty }$ is the $\ell _{\infty }$ norm of
a vector, $\text{sign}(.)$ is the sign function, and $\boldsymbol{\beta }%
_{0}=\left( \beta _{01},\beta _{02},\cdots ,\beta _{0k}\right) ^{\prime }$
is the $k\times 1$ vector of the coefficients of the signals. The following
example provides more intuition on how IRC imposes restrictions on the
magnitudes of the sample correlations between the covariates in the active
set. }

\begin{example}
Suppose the DGP for $y_{t}$ contains only two signals, $x_{1t}$ and $x_{2t}$%
. Denote the sample correlation coefficient between $x_{1t}$ and $x_{2t}$ by 
$\hat{\rho}$, and the sample correlation coefficients of $x_{1t}$ and $%
x_{2t} $ with the rest of the covariates in the active set, $%
x_{3,t},x_{4,t},\cdots ,x_{Nt}$, by $\hat{\rho}_{i1}$ and $\hat{\rho}_{i2}$,
for $i=3,4,\cdots ,N$, respectively. Then, after some algebra, the IRC given
by (\ref{eq:ir_con}) simplifies to 
\begin{equation*}
\max_{i\in \{3,4,\cdots ,N\}}\left\vert (\hat{\rho}_{i1}-\hat{\rho}\hat{\rho}%
_{i2})\text{sign}(\beta _{01})+(\hat{\rho}_{i2}-\hat{\rho}\hat{\rho}_{i1})%
\text{sign}(\beta _{02})\right\vert \leq 1-\hat{\rho}^{2}.
\end{equation*}%
There are two cases: (A) $\text{sign}\left( \beta _{01}\right) =\text{sign}%
\left( \beta _{02}\right) $ and (B) $\text{sign}\left( \beta _{01}\right)
\neq \text{sign}\left( \beta _{02}\right) $. Under case (A) it follows that
the IRC condition is met if 
\begin{equation*}
\max_{i\in \{3,4,\cdots ,N\}}\left\vert \hat{\rho}_{i1}+\hat{\rho}%
_{i2}\right\vert \leq 1+\hat{\rho}.
\end{equation*}%
Similarly under case (B) it is required that 
\begin{equation*}
\max_{i\in \{3,4,\cdots ,N\}}\left\vert \hat{\rho}_{i1}-\hat{\rho}%
_{i2}\right\vert \leq 1-\hat{\rho}.
\end{equation*}

\end{example}

From the above example, it is clear that IRC places restrictions on the
magnitude of sample correlation among signals ($\hat{\rho}$ in the above
example), as well as the magnitude of sample correlation between signals and
pseudo-signals ($\hat{\rho}_{i1}$ and $\hat{\rho}_{i2}$). Notably, the IRC
is met for noise variables but need not hold for pseudo-signals. OCMT also
has no difficulty in dealing with noise variables, and is very effective at
eliminating them. However, for consistent estimation of the approximate
model, post OCMT selection, it is necessary to restrict the number of
selected covariates relative to the sample size, $T$. To this end, OCMT
assumes that the number of pseudo-signals, $k_{T}^{\ast }$, could grow at an
order less than the square root of the number of observations, namely 
\begin{equation*}
k_{T}^{\ast }=\ominus (T^{d})\text{ for some }0\leq d<\frac{1}{2}.
\end{equation*}%
It is important to note that OCMT does not place any restrictions on the
magnitude of correlations of signals and pseudo-signals. Instead, it limits
the number of covariates that are correlated with the signals ($k_{T}^{\ast
} $). Clearly, the IRC could be violated even when the number of
pseudo-signals grows at an order less than $\sqrt{T}$. Hence the OCMT's
requirement on the number of pseudo-signals allows for cases where the IRC
does not hold, and \textit{vice versa}.

The condition on the number of pseudo-signals ($k_{T}^{\ast }$) in the OCMT
framework has been recently relaxed by \cite{sharifvaghefi2023variable}. To
illustrate how this is done, suppose there are no noise variables and hence
the signals, $\mathbf{x}_{1t}=\left( x_{1t},x_{2t},\cdots ,x_{kt}\right)
^{\prime }$, are correlated with all the remaining covariates in the active
set. In this case if $N \gg \sqrt{T} $, a straightforward application of
OCMT will not be valid. But, we can model the correlation between the
signals, $\mathbf{x}_{1t}$, and the remaining covariates, {$\mathbf{x}_{2t}$,%
} as 
\begin{equation*}
x_{it}=\sum_{j=1}^{k}\psi _{ij}x_{jt}+\xi _{it}=\boldsymbol{\psi }%
_{i}^{\prime }\mathbf{x}_{1t}+\xi _{it},\text{ for }i=k+1,k+2,\cdots ,N.
\end{equation*}%
The signals thus act as strong factors for the pseudo-signals. Given that
the identity of signals and pseudo-signals are unknown and the number of
pseudo-signals is large, it is reasonable to propose the existence of latent
factors, $\mathbf{f}_{t}$, that are common across the covariates in the
active set. This idea can be formally expressed as: 
\begin{equation*}
x_{it}=\boldsymbol{\psi }_{i}^{\prime }\mathbf{f}_{t}+\varepsilon _{it}\quad 
\text{for }i=1,2,\cdots ,N,
\end{equation*}%
where $\boldsymbol{\psi }_{i}$ is vector of factor loadings, and $%
\varepsilon _{it}$ refers to the idiosyncratic components that are weakly
cross-correlated such that%
\begin{equation}
\sup_{j}\sum_{i=1}^{N}\left\vert cov(\varepsilon _{it},\varepsilon
_{jt})\right\vert <C<\infty .  \label{Wcd}
\end{equation}%
Substituting $x_{it}$ into the DGP for $y_{t}$, given by (\ref{dgp y_t}), we
obtain: 
\begin{equation*}
y_{t}=\boldsymbol{\delta }_{0}^{\prime }\mathbf{f}_{t}+\sum_{i=1}^{k}\beta
_{i0}\varepsilon _{it}+u_{t},
\end{equation*}%
with $\boldsymbol{\delta }_{0}=\sum_{i=1}^{k}\beta _{i0}\boldsymbol{\psi }%
_{i}$. When the common factors, $\mathbf{f}_{t}$, and idiosyncratic
components, $\varepsilon _{it}$ , are known, this model would correspond to
that presented in working paper version of our work, where common factors $%
\mathbf{f}_{t}$ can be used as preselected variables. Since $\mathbf{f}_{t}$
and $\varepsilon _{it}$ are not known, \cite{sharifvaghefi2023variable}
shows that when both $N$ and $T$ are large the OCMT selection can be carried
out using the principal component estimators of $\mathbf{f}_{t}$ and $%
\varepsilon _{it}$, denoted by $\hat{\mathbf{f}}_{t}$ and $\hat{\varepsilon}%
_{it}$, using all the covariates in the active set. The large $N$ is
required for consistent estimation of the common factors. As a result, the
OCMT\ condition on the number of pseudo-signals now relates to the number of 
$\varepsilon _{it}$ for $i=k+1,k+2,...,N$ that are correlated with $%
\varepsilon _{it}$ for $i=1,2,...,k$, which is bounded under condition (\ref%
{Wcd}).

For variable selection consistency of Lasso under parameter stability, the
literature further requires the penalty term, $\lambda _{T}$, to grow at an
order greater than $\sqrt{T}$ such that: 
\begin{equation*}
\lim_{T\rightarrow \infty }\Pr \left( \left\Vert \frac{1}{\sqrt{T}}%
\sum_{t=1}^{T}\mathbf{x}_{2t}^{\perp} u_{t}\right\Vert _{\infty }>\frac{%
\lambda _{T}}{\sqrt{T}}\right) =0,
\end{equation*}%
where $\mathbf{x}_{2t}^{\perp} $ is the part of variation in $\mathbf{x}%
_{2t} $ that is orthogonal to $\mathbf{x}_{1t} $ and $u_{t}$ is the error
term in the data generating process. The exact choice of $\lambda _{T}$ in
practice is often unclear, with practitioners typically relying on
cross-validation methods.

A third condition required by Lasso for variable selection consistency is
the beta-min condition: 
\begin{equation*}
\min_{j=1,2,\cdots ,k}\left\vert \beta _{j0}\right\vert >(2T)^{-1}\lambda
_{T}\left\vert \mathbf{R}_{11}^{-1}\text{sign}\left( \boldsymbol{\beta }%
_{0}\right) \right\vert _{j}
\end{equation*}%
where $\left\vert .\right\vert _{j}$ denotes the absolute value of the $%
j^{th}$ element of a vector. Given that $\lambda _{T}$ must grow at an order
greater than $\sqrt{T}$, we can conclude from the beta-min condition that $%
\beta _{i0}\gg \frac{1}{\sqrt{T}}$ for $i=1,2,\cdots ,k$. For example, \cite%
{lahiri2021necessary} assumes that $\beta _{i0}\gg \sqrt{\frac{k\log (T)}{T}}
$. The OCMT's requirement on the strength of signals (under parameter
stability) is given by $\beta _{i0}=\ominus (T^{-\vartheta _{i}})$, for some 
$0\leq \vartheta _{i}<1/2$. This condition is essentially very similar to
the Lasso's beta-min condition.

\section{ Asymptotic properties of OCMT under parameter instability\label%
{asymptotic properties}}

We establish the asymptotic properties of the OCMT procedure for variable
selection assuming the time variations in $\beta _{it}$ for $i=1,2,...,k$
are distributed independently of the regressors in the active set. We also
make additional assumptions that bound the degree of time variations in $%
\beta _{it}$ and $x_{it}$, in addition to assuming the exponentially
decaying tail probabilities for $\beta _{it}$ and $x_{it}$. Our assumptions
on $x_{it}$, $i=1,2,...,k$ and their correlations with the other variables
in the active set are in line with those assumed in the literature. A formal
statement of these assumptions are set out in Section \ref{technical
assumptions}. Theorem \ref{sel_consistency_theorem} establishes that OCMT
continues to asymptotically select an approximating model that includes all
the signals and none of the noise variables. Additional assumptions are
required for investigating the asymptotic properties of the least squares
estimates of the post OCMT selected model. These assumptions and the related
theorems are provided in Section \ref{sec: post OCMT}. Theorem \ref%
{estimation consistency} establishes the rate at which the least squares
estimates of the coefficients of the selected model converge to their true
time averages. It is shown that the regular convergence rate of $\sqrt{T}$
is achieved only if $k_{T}^{\ast }$ (the number of selected covariates) is
fixed in $T$. Irregular convergence rates result when $k_{T}^{\ast }$ rises
in $T$. Theorem \ref{mean square error} shows that the sum of squared
residuals of the estimated model converges in probability to its limiting
value at the oracle rate of $\sqrt{T}$. The limiting value consists of two
components: the first is the unavoidable uncertainty due to the unobserved
error term, $u_{t}$, and the second is the cost (in terms of fit) of
ignoring the time variations in the coefficients of the signals.

Suppose the target variable, $y_{t}$, is generated by (\ref{dgp y_t}) in
terms of $x_{it}$ for $i=1,2,...,k$, and $\mathbf{x}_{t}=(x_{1t},x_{2t},%
\cdots ,x_{kt},x_{k+1,t},....,x_{Nt})^{\prime }$ is the $N\times 1$ vector
of covariates in the active set ($N \gg k$). Let $\bar{\beta}_{i,T}\equiv
T^{-1}\sum_{t=1}^{T}\mathbb{E}(\beta _{it})$, for $i=1,2,...,k$, and $\bar{%
\theta}_{i,T}=\sum_{j=1}^{k}\left( T^{-1}\sum_{t=1}^{T}\mathbb{E}(\beta
_{jt})\sigma _{ij,t}\right) +\bar{\sigma}_{iu,T},$ for $i=1,2,...,N$, where $%
\sigma _{ij,t}=\mathbb{E}(x_{it}x_{jt})$, and $\bar{\sigma}%
_{iu,T}=T^{-1}\sum_{t=1}^{T}\mathbb{E}(x_{it}u_{t})$. Define the filtrations 
$\mathcal{F}_{t}^{u}=\sigma (u_{t},u_{t-1},\cdots )$, $\mathcal{F}%
_{t}^{x}=\sigma (\mathbf{x}_{t},\mathbf{x}_{t-1},\cdots )$, and $\mathcal{F}%
_{jt}^{\beta }=\sigma (\beta _{jt},\beta _{j,t-1},\cdots )$, for $%
j=1,2,\cdots ,k$. Set $\mathcal{F}_{t}^{\beta }=\cup _{j=1}^{k}\mathcal{F}%
_{jt}^{\beta }$ and $\mathcal{F}_{t}=\mathcal{F}_{t}^{q}\cup \mathcal{F}%
_{t}^{\mathrm{a}}\cup \mathcal{F}_{t}^{\beta }\cup \mathcal{F}_{t}^{u}$, and
consider the following assumptions:

\subsection{Assumptions\label{technical assumptions}}

\begin{assumption}[\textbf{Coefficients of signals}]
\label{signal} \textcolor{white}{enter} \newline
(a) The number of signals, $k$, is a finite fixed integer. (b) $\beta _{jt}$%
, $j=1,2,\cdots ,k$, are distributed independently of $x_{it^{\prime }}$, $%
i=1,2,\cdots ,N$, and $u_{t^{\prime }}$ for all $t$ and $t^{\prime }$. (c)
The signals are (semi) strong in the sense that $\bar{\beta}_{j,T}=\ominus
(T^{-\vartheta _{j}})$ for $0\leq \vartheta _{j}<1/2$, $j=1,2,...,k.$ (d)
There are no hidden signals in the sense that $\bar{\theta}_{j,T}=\ominus
(T^{-\vartheta _{j}})$, for $0\leq \vartheta _{j}<1/2$, $j=1,2,...,k$.
\end{assumption}

\begin{assumption}[\textbf{Martingale difference processes}]
\label{md} \textcolor{white}{enter} \newline
For $i, i^{\prime} =1,2,\cdots ,N$, $j=1,2,\cdots ,k$, and $t=1,2,\cdots ,T$%
, (a) $\mathbb{E}\left[ x_{it}x_{ i^{\prime} t}-\mathbb{E}%
(x_{it}x_{i^{\prime} t})|\mathcal{F}_{t-1}\right] =0$, (b) $\mathbb{E}\left[
u_{t}^{2}-\mathbb{E}\left( u_{t}^{2}\right) |\mathcal{F}_{t-1}\right] =0$,
(c) $\mathbb{E}\left[ x_{it}u_{t}-\mathbb{E}(x_{it}u_{t})|\mathcal{F}_{t-1}%
\right] =0$, where $T^{-1}\sum_{t=1}^{T}\mathbb{E}(x_{it}u_{t})=O(T^{-%
\epsilon _{i}}),$ with $\epsilon _{i} \geq 1/2$, and (d) $\mathbb{E}\left[
\beta _{j t}-\mathbb{E}(\beta _{j t})|\mathcal{F}_{t-1}\right] =0$.
\end{assumption}

\begin{assumption}[\textbf{Exponential decaying probability tails}]
\label{subg} \textcolor{white}{enter} \newline
There exist sufficiently large positive constants $C_{0}$ and $C_{1}$, and $%
s>0$ such that for all $\alpha >0$, (a) $\sup_{i,t}\Pr (|x_{it}|>\alpha
)\leq C_{0}\exp (-C_{1}\alpha ^{s})$, (b) $\sup_{i,t}\Pr (|\beta
_{it}|>\alpha )\leq C_{0}\exp (-C_{1}\alpha ^{s})$, and (c) $\sup_{t}\Pr
(|u_{t}|>\alpha )\leq C_{0}\exp (-C_{1}\alpha ^{s})$.
\end{assumption}

{Before presenting the theoretical results, we briefly discuss the rationale
behind our assumptions and compare them with the assumptions typically made
in the high-dimensional linear regressions and the parameter instability
literature. }

Assumption \ref{signal}(a) posits that the number of signals is a fixed
integer. This is crucial to ensure that the random variable $y_{t}$ has a
distribution with an exponentially decaying probability tail. Under the
premise that the covariates $x_{it}$ for all $i$ and $t$ are non-random and
fixed, which is a common assumption in the penalized regression setting, it
becomes permissible for the number of signals to grow with the sample size
at an order slower that $\sqrt{T}$. Assumption \ref{signal}(b) is common in
the literature under parameter instability and restrict the distribution of
time-varying parameters to be independent of the covariates. Assumption \ref%
{signal}(c) is an identification assumption needed to distinguish signals
from noise variables and is similar to the beta-min condition already
discussed in Section \ref{sec:lasso_vs_ocmt}. Finally, Assumption \ref%
{signal}(d) ensures that there are no hidden signals. As discussed in
Section \ref{ocmt method}, we make this assumption to simplify the
theoretical derivations, and one can use the multi-stage OCMT procedure
suggested by \cite{chudik2018one} to allow for hidden signals.

To establish that the OCMT procedure with the critical value function $%
c_{p}(N,\delta )=\Phi ^{-1}\left( 1-\frac{p}{2N^{\delta }}\right) $ does not
select any of the noise variables with a probability approaching one as $N$
and $T$ go to infinity, we need to show that the t-statistic given by (\ref%
{eq:t-stat}) follows a distribution with exponentially decaying tails. We
utilize the concentration inequality of an exponential decaying rate to
accomplish this goal. Assumptions \ref{md} and \ref{subg} place constraints
on the sequence of random variables, $x_{it}$ for $i=1,2,\ldots ,N$, $\beta
_{jt}$ for $j=1,2,\ldots ,k$, and $u_{t}$ such that they adhere to a
martingale difference process and exhibit exponential decaying probability
tails. These assumptions are sufficient to establish the exponential
decaying concentration inequality, as provided in Lemma \ref%
{mart_diff_proc_exp_tail} in the online theory supplement. Notably, these
assumptions could be relaxed provided that the exponential decaying
concentration inequality holds. For example, Theorem 1 of \cite%
{merlevede2011bernstein} and Lemma D1 of the online theory supplement for 
\cite{chudik2018one} establishes that this inequality can be achieved while
allowing for weak time-series dependence. In penalized regression
literature, a commonly held assumption is that the covariates are non-random
and fixed. Moreover, error terms $\{u_{t}\}_{t=1}^{T}$ are typically assumed
to be serially independent. See, for example, see \cite{zhao2006model}, \cite%
{javanmard2013model}, \cite{lee2015model}, \cite{belloni2014inference}, \cite%
{javanmard2020flexible}, and \cite{lahiri2021necessary}. Additionally, in
the Lasso literature it is often assumed that $u_{t}$ possesses an
exponentially decaying probability tail. See, for example, \cite%
{javanmard2018debiasing}, \cite{hansen2019factor}, \cite{fan2020factor}, and 
\cite{javanmard2020flexible}.

\subsection{Variable selection consistency}

\label{sec: selection consistency}

As mentioned in Section \ref{introduction}, the purpose of this paper is to
provide the theoretical argument for applying the OCMT procedure with no
down-weighting at the variable selection stage in linear high-dimensional
settings subject to parameter instability. We now show that under the
assumptions set out in Section \ref{technical assumptions}, the OCMT
procedure selects the approximating model that contains all the signals; $%
\{x_{it}:i=1,2,\cdots ,k\}$; and none of the noise variables; $%
\{x_{it}:k+k_{T}^{\ast }+1,k+k_{T}^{\ast }+2,\cdots ,N\}$. The event of
choosing the approximating model is defined by \vspace{-0.25cm} 
\begin{equation}
\textstyle\mathcal{A}_{0}=\left\{ \sum_{i=1}^{k}\hat{\mathcal{J}}%
_{i}=k\right\} \cap \left\{ \sum_{i=k+k_{T}^{\ast }+1}^{N}\hat{\mathcal{J}}%
_{i}=0\right\} .  \label{approx_model_sel_def}
\end{equation}%
Note that the approximating model can contain pseudo-signals. In what
follows, we show that $\Pr (\mathcal{A}_{0})\rightarrow 1$, as $%
N,T\rightarrow \infty $.

\begin{theorem}
\label{sel_consistency_theorem} Consider the DGP for $y_{t}$, $t=1,2,\cdots
,T$, given by (\ref{dgp y_t}), and the set $\mathcal{S}_{Nt}=%
\{x_{1t},x_{2t},\cdots ,x_{Nt}\}$ that contains $k$ signals, $k_{T}^{\ast }$
pseudo-signals, and $N-k-k_{T}^{\ast }$ noise variables. Suppose that
Assumptions \ref{signal}-\ref{subg} hold and $N=\ominus (T^{\kappa })$ with $%
\kappa >0$. Then, there exist finite positive constants $C_{0}$ and $C_{1}$
such that, for any $0<\pi <1$ and any null sequence $d_{T}>0$, the
probability of selecting the approximating model $\mathcal{A}_{0}$, as
defined by (\ref{approx_model_sel_def}), by the OCMT procedure with the
critical value function $c_{p}(N,\delta )$ given by (\ref{cv_function}), for
some $\delta >0$, is 
\begin{equation}
\Pr (\mathcal{A}_{0})=1-O\left[ T^{\kappa \left( 1-\mathcal{X}_{NT}\left( 
\frac{1-\pi }{1+d_{T}}\right) ^{2}\delta \right) }\right] -O\left[ T^{\kappa
}\exp \left( -C_{0}T^{C_{1}}\right) \right] ,  \label{approx_model_selection}
\end{equation}%
where, 
\begin{equation*}
\textstyle\mathcal{X}_{NT}=\inf_{i\in \{k+k^{\ast }+1,\cdots ,N\}}\frac{\bar{%
\sigma}_{\eta _{i},T}^{2}\bar{\sigma}_{x_{i},T}^{2}}{\bar{\omega}_{iy,T}^{2}}%
,
\end{equation*}
$\bar{\sigma}_{x_{i},T}^{2}=T^{-1}\sum_{t=1}^{T}\mathbb{E}(x_{it}^{2})$, $%
\bar{\omega}_{iy,T}^{2}=T^{-1}\sum_{t=1}^{T}\mathbb{E}(x_{it}^{2}y_{t}^{2}|%
\mathcal{F}_{t-1})$, $\bar{\sigma}_{\eta _{i},T}^{2}=T^{-1}\sum_{t=1}^{T}%
\mathbb{E}(\eta _{it}^{2})$, $\eta _{it}=y_{t}-\phi _{i,T}x_{it}$, and $\phi
_{i,T}$ is defined by ({\ref{phiiTdef}}).
\end{theorem}

This theorem shows that the probability of selecting the approximating model
is unaffected by parameter instability, so long as the average net effects
of the signals are non-zero or converge to zero sufficiently slowly in $T$,
as defined formally by Assumption \ref{signal}. The theorem also highlights
the importance of an appropriate choice of $\delta $ for model selection
consistency. Corollary \ref{cor:delta_choice} in the online theory
supplement shows that if the covariates in the active set are generated by a
stationary process and the noise variables are independent of $y_{t}$ then $%
\mathcal{X}_{NT}=1$. As a result, for any $\delta >1$, OCMT consistently
selects the approximating model, $\mathcal{A}_{0}$. Notably, $c_{p}(N,\delta
)$ is reasonably stable with respect to small increases in $\delta $ in the
neighborhood of $\delta =1$ and the extensive Monte Carlo studies in \cite%
{chudik2018one} also suggest that setting $\delta =1$ performs well in
practice.\footnote{%
One could also use the heteroscedasticity and/or autocorrelation robust
standard errors in computation of t-statistics given by (\ref{eq:t-stat}) to
ensure the consistent selection of the approximating model for any $\delta
>1 $ in a more general setup.}

\subsection{Properties of the post OCMT selected model\label{sec: post OCMT}}

To investigate the asymptotic properties of the least squares estimates of
the selected model (post OCMT) we require the following additional
assumption:

\begin{assumption}[\textbf{Eigenvalues}]
\label{eigenvalues signals and pseudo signal} The eigenvalue condition 
\begin{equation*}
\lambda _{\min }\left[ T^{-1}\sum\limits_{t=1}^{T}\mathbb{E}(\mathbf{x}_{%
\tilde{k}_{T},t}\mathbf{x}_{\tilde{k}_{T},t}^{\prime })\right] >c>0,
\end{equation*}%
holds, where $\mathbf{x}_{\tilde{k}_{T},t}$, for $t=1,2,...,T$ are the $%
\tilde{k}_{T}\times 1$ vector of observations on signals ($k$) and
pseudo-signals ($k_{T}^{\ast }$) with $\tilde{k}_{T}=k+k_{T}^{\ast }$.
\end{assumption}

This assumption ensures that the post OCMT selected model can be
consistently estimated subject to certain regularity conditions to be
discussed below. The post OCMT selected model can be written as 
\begin{equation*}
\textstyle y_{t}=\sum_{i=1}^{N}\hat{\mathcal{J}}_{i}x_{it}b_{i}+\eta _{t}
\end{equation*}%
where $\hat{\mathcal{J}}_{i}=\mathbf{\mathit{I}}\left[ \lvert t_{i,T}\rvert
>c_{p}(N,\delta )\right] ,\ $defined by (\ref{selection indicator}). Also $%
\sum_{i=1}^{N}\hat{\mathcal{J}}_{i}=\hat{k}_{T}$, where $\hat{k}_{T}$ is the
number of covariates selected by OCMT. By Theorem \ref%
{sel_consistency_theorem} the probability that the selected model contains
the signals tends to unity as $T\rightarrow \infty $. We can further write 
\begin{equation}
\textstyle y_{t}=\sum_{i=1}^{N}\hat{\mathcal{J}}_{i}x_{it}b_{i}+\eta
_{t}=\sum_{\ell =1}^{\hat{k}_{T}}\gamma _{\ell }w_{\ell t}+\eta _{t},
\label{PostOCMT}
\end{equation}%
where $\mathbf{w}_{t}=\left( w_{1t},w_{2t},\cdots ,w_{\hat{k}_{T}t}\right)
^{\prime }$. The least squares (LS) estimator of selected coefficients, $%
\boldsymbol{\gamma }_{T}=\left( \gamma _{1},\gamma _{2},\cdots ,\gamma _{%
\hat{k}_{T}}\right) ^{\prime }$, is given by 
\begin{equation}
\textstyle\hat{\boldsymbol{\gamma }}_{T}=\left( T^{-1}\sum_{t=1}^{T}\mathbf{w%
}_{t}\mathbf{w}_{t}^{\prime }\right) ^{-1}\left( T^{-1}\sum_{t=1}^{T}\mathbf{%
w}_{t}y_{t}\right) ,  \label{g-OCMT}
\end{equation}%
In establishing the rate of convergence of $\hat{\mathbf{\gamma }}_{T}$ we
distinguish between two cases: when the vector of signals, $\mathbf{x}%
_{k,t}=\left( x_{1t},x_{2t},\cdots ,x_{kt}\right) ^{\prime }$ is included in 
$\mathbf{w}_{t}$ as a subset, and when this is not the case. But we know
from Theorem \ref{sel_consistency_theorem} that the probability of the
latter tends to zero at a sufficiently fast rate. The following theorem
provides the conditions under which the estimates of the coefficients of the
selected signals and pseudo-signals of the approximating model tend to their
true mean values, defined formally below. \vspace{-0.25cm}

\begin{theorem}
\label{estimation consistency} Let the DGP for $y_{t}$, $t=1,2,\cdots ,T$,
be given by (\ref{dgp y_t}) and write down the regression model selected by
the OCMT procedure as (\ref{PostOCMT}). Suppose that Assumptions \ref{signal}%
-\ref{eigenvalues signals and pseudo signal} hold and the number of
pseudo-signals, $k_{T}^{\ast }$, grow with $T$ such that $%
k_{T}^{\ast}=\ominus (T^{d})$ with $0\leq d<\frac{1}{2}$. Consider the least
squares (LS) estimator of $\boldsymbol{\gamma }_{T}=\left( \gamma
_{1},\gamma_{2},\cdots ,\gamma _{\hat{k}_{T}}\right) ^{\prime }$, given by (%
\ref{g-OCMT}).\vspace{-0.25cm}

\begin{enumerate}
\item[(i)] If $\mathbb{E}(\beta _{it})=\beta _{i}$ for all $t$, then, 
\begin{equation*}
\left\| \hat{\boldsymbol{\gamma}}_{T} - \boldsymbol{\gamma}_{T}^{*} \right\|
= O_p \left(T^{\frac{d-1}{2}} \right),
\end{equation*}
where $\boldsymbol{\gamma}_{T}^{*} = (\gamma _{1}^{*}, \gamma _{2}^{*},
\cdots, \gamma_{\hat{k}_{T}}^{*})^{\prime} $, and \vspace{-0.25cm} 
\begin{equation*}
\left\{ 
\begin{matrix}
\gamma _{\ell}^{*} \in \boldsymbol{\beta }=(\beta _{1},\beta
_{2},\cdots,\beta _{k})^{\prime }, & \text{if } w_{\ell t}\in \mathbf{x}_{kt}
\\ 
\gamma _{\ell}^{*} = 0, & \text{otherwise}.\vspace{-0.25cm}%
\end{matrix}
\right.
\end{equation*}

\item[(ii)] If $\mathbb{E}\left( \mathbf{x}_{\tilde{k}_{T},t} \mathbf{x}_{%
\tilde{k}_{T},t}^{\prime}\right) $ is a fixed time-invariant matrix, where $%
\tilde{k}_{T} = k + k^{\ast}_{T}$, then, 
\begin{equation*}
\left\Vert \hat{\boldsymbol{\gamma }}_{T}-\boldsymbol{\gamma }_{T}^{\diamond
}\right\Vert =O_{p}\left( T^{\frac{d-1}{2}}\right) ,
\end{equation*}%
\vspace{-0.25cm} where $\boldsymbol{\gamma }_{T}^{\diamond }=(\gamma
_{1T}^{\diamond },\gamma _{2T}^{\diamond },\cdots ,\gamma _{\hat{k}%
_{T},T}^{\diamond })^{\prime }$ , and 
\begin{equation*}
\left\{ 
\begin{matrix}
\gamma _{\ell ,T}^{\diamond }\in \boldsymbol{\bar{\beta}}_{T}=(\bar{\beta}%
_{1T},\bar{\beta}_{2T},\cdots ,\bar{\beta}_{kT})^{\prime }, & \text{if }%
w_{\ell t}\in \mathbf{x}_{kt} \\ 
\gamma _{\ell ,T}^{\diamond }=0, & \text{otherwise},%
\end{matrix}%
\right.
\end{equation*}%
and $\bar{\beta}_{iT}=T^{-1}\sum_{t=1}^{T}\mathbb{E}(\beta _{it})$, $%
i=1,2,\cdots ,k$.
\end{enumerate}
\end{theorem}

\begin{remark}
{\ The above theorem builds on Theorem \ref{sel_consistency_theorem} and
establishes that in the post OCMT selected model estimated by LS only
signals will end up having non-zero limiting values, as $N$ and $%
T\rightarrow \infty $. This theorem also shows that the convergence rate of
the LS estimators depends on $d$, defined by $k_{T}^{\ast }=\ominus (T^{d})$%
, and the regular $\sqrt{T}$ rate of convergence is achieved only if $d=0$.
Similarly, \cite{lahiri2021necessary} establishes that the Lasso procedure
cannot achieve both variable selection consistency and $\sqrt{T}$%
-consistency in coefficient estimation. }
\end{remark}

\begin{remark}
The conditions of Theorem \ref{estimation consistency} are met in the case
of random coefficient models where $\beta _{it}=\beta _{i}+\sigma _{it}\xi
_{it}$, and $\xi _{it}$ are distributed independently of the signals, and
the LS estimator of $\boldsymbol{\gamma }_{T}^{\ast }$ is consistent, so
long as $0\leq d<1/2$. Interestingly, if signal and pseudo-signal variables
are generated by a stationary process, and hence they satisfy condition (ii)
of Theorem \ref{estimation consistency}, then we can extend the random
coefficient model to have time-varying means, and still estimate $%
\boldsymbol{\gamma }_{T}^{\ast }$ consistently by LS.
\end{remark}

Lastly, we consider the fit of the post OCMT selected regression in terms of
its residuals given by 
\begin{equation}
\textstyle\hat{\eta}_{t}=y_{t}-\sum_{\ell =1}^{\hat{k}_{T}}\hat{\gamma}%
_{\ell }w_{\ell t}\text{, for }t=1,2,...,T.  \label{PostOCMT error}
\end{equation}%
It is worth noting that even when all the signal variables are correctly
selected, the forecasts based on the selected model will be biased due to
parameter instability. The implications of parameter instability for the
in-sample fit of the selected regression is derived in Proposition \ref%
{obs:pop_reg_coef} of the online theory supplement, abstracting from
variable selection uncertainty. In what follows we derive the asymptotic
properties of the sum of squared residuals (SSR) of the selected model,
namely $\sum_{t=1}^{T}\hat{\eta}_{t}^{2}$, taking account of the costs
associated with variable selection uncertainty and parameter instability. To
this end we need the following assumption on the cross correlation of
parameter heterogeneity.\vspace{-0.2cm}

\begin{assumption}[Weak time dependence]
\label{weak time dependence} $h_{ij,t}=x_{it}x_{jt}(\beta _{it}-\bar{\beta}%
_{iT})(\beta _{jt}-\bar{\beta}_{jT})$ is weakly correlated over time such
that 
\begin{equation*}
\sum_{t=1}^{T}\sum_{t^{\prime }=1}^{T}cov(h_{ij,t},h_{ij,t^{\prime }})=O(T),%
\text{ for }i,j=1,2,..,k,
\end{equation*}%
where $cov(.,.)$ is the covariance operator.
\end{assumption}

\begin{remark}
Assumption \ref{weak time dependence} is a high-level assumption. Here is an
example of conditions under which this assumption holds. Suppose,
Assumptions \ref{signal} and \ref{md} hold, and the cross products of
coefficients of the signals follow martingale difference processes such that 
\begin{equation*}
\mathbb{E}\left[ \beta _{it}\beta _{jt}-\mathbb{E}(\beta _{it}\beta _{jt})|%
\mathcal{F}_{t-1}\right] =0,\text{ for }i=1,2,\cdots ,k,\ j=1,2,\cdots ,k,%
\text{ and }t=1,2,\cdots ,T.
\end{equation*}%
Then, $\sum_{t=1}^{T}\sum_{t^{\prime }=1}^{T}\text{cov}(h_{ij,t},h_{ij,t^{%
\prime }})=O(T)$. See Lemma \ref{lem:weak time dependence} in the online
theory supplement for a proof. \color{black}
\end{remark}

The following theorem establishes the limiting property of SSR of the post
OCMT selected model.

\begin{theorem}
\label{mean square error} Let the DGP for $y_{t}$, $t=1,2,\cdots ,T$ be
given by (\ref{dgp y_t}) and write down the regression model selected by the
OCMT procedure as (\ref{PostOCMT}). Suppose that Assumptions \ref{signal}-%
\ref{weak time dependence} hold and the number of pseudo-signals, $%
k_{T}^{\ast }$, grow with $T$ such that $k_{T}^{\ast }=\ominus (T^{d})$ with 
$0 \leq d < \frac{1}{2} $. Consider the residuals of the selected model,
estimated by LS and given by (\ref{PostOCMT error}). \vspace{-0.25cm}

\begin{enumerate}
\item[(i)] If $\mathbb{E}(\beta _{it})=\beta _{i}$ for all $t$, then \vspace{%
-0.2cm} 
\begin{equation}
\textstyle T^{-1}\text{SSR}=\bar{\sigma}_{u,T}^{2}+\bar{\Delta}_{\beta
,T}+O_{p}\left( T^{-\frac{1}{2}}\right) +O_{p}\left( T^{d-1}\right) ,
\label{MSE1}
\end{equation}%
where $\bar{\sigma}_{u,T}^{2}=T^{-1}\sum_{t=1}^{T}\mathbb{E}\left(
u_{t}^{2}\right) $, and $\bar{\Delta}_{\beta ,T}=T^{-1}\sum_{t=1}^{T}\text{tr%
}\left( \boldsymbol{\Sigma }_{\mathbf{x}_{k},t}\boldsymbol{\Omega }_{\beta
,t}\right) $ are non-negative, with $\boldsymbol{\Sigma }_{\mathbf{x}%
_{k},t}\equiv \left( \sigma _{ijt,x}\right) $, $\boldsymbol{\Omega }_{\beta
,t}\equiv \left( \sigma _{ijt,\beta }\right) $ for $i,j=1,2,\cdots ,k$, and $%
\sigma _{ijt,x}=\mathbb{E}\left( x_{it}x_{jt}\right) $, $\sigma _{ijt,\beta
}=\mathbb{E}\left[ (\beta _{it}-\beta _{i})(\beta _{jt}-\beta _{j})\right] $.

\item[(ii)] \vspace{-0.2cm} $\ $Let $\tilde{k}_{T}=k+k_{T}^{\ast }$ and
suppose that $\mathbb{E}\left( \mathbf{x}_{\tilde{k}_{T},t}\mathbf{x}_{%
\tilde{k}_{T},t}^{\prime }\right) $ is time-invariant (fixed). Then, 
\begin{equation}
\textstyle T^{-1}\text{SSR}=\bar{\sigma}_{u,T}^{2}+\bar{\Delta}_{\beta
,T}^{\ast }+O_{p}\left( T^{-\frac{1}{2}}\right) +O_{p}\left( T^{d-1}\right) ,
\label{MSE2}
\end{equation}%
where $\bar{\Delta}_{\beta ,T}^{\ast }=T^{-1}\sum_{t=1}^{T}\text{tr}\left( 
\boldsymbol{\Sigma }_{\mathbf{x}_{k},t}\boldsymbol{\Omega }_{\beta ,t}^{\ast
}\right) $ is non-negative, with $\boldsymbol{\Omega }_{\beta ,t}^{\ast
}\equiv \left( \sigma _{ijt,\beta }^{\ast }\right) $ for $i,j=1,2,\cdots ,k$%
, and $\sigma _{ijt,\beta }^{\ast }=\mathbb{E}\left[ (\beta _{it}-\bar{\beta}%
_{i,T})(\beta _{jt}-\bar{\beta}_{j,T})\right] $.
\end{enumerate}
\end{theorem}

\begin{remark}
\label{comparison with the oracle} The condition $d<\frac{1}{2}$ in Theorem %
\ref{mean square error} ensures that the number of pseudo-signals grows
sufficiently slowly in $T$, which in turn ensures that $T^{1-d}<T^{-\frac{1}{%
2}}$ and hence from equations (\ref{MSE1}) and (\ref{MSE2}), we can conclude
that the average of squared residuals ($T^{-1}$SSR) of the Post OCMT
selected model convergences at the same rate of $T^{-\frac{1}{2}}$ under
both scenarios (i) and (ii).
\end{remark}

Results (\ref{MSE1}) and (\ref{MSE2}) in Theorem \ref{mean square error}
show that the SSR of the selected model depends on ($i$) the unavoidable
uncertainty due to the unobserved error term, $u_{t}$, given by the term $%
\bar{\sigma}_{u,T}^{2}$, ($ii$) the cost (in terms of fit) of ignoring the
time variation in the coefficients of the signals, $\boldsymbol{\beta }_{it}$%
, $i=1,2,\cdots ,k$, as given by the term $\bar{\Delta}_{\beta ,T}$ and $%
\bar{\Delta}_{\beta ,T}^{\ast }$, respectively, and ($iii$) the $O_{p}\left(
T^{-1/2}\right) $ term due to sampling uncertainty (which will be present
even in the absence of variable selection uncertainty), and $(iv)$ the $%
O_{p}\left( T^{d-1}\right) $ term which is due to variable selection
uncertainty, and will be dominated by $O_{p}\left( T^{-1/2}\right) $ when $%
d<1/2$. Therefore, the cost of variable selection can be controlled when
using OCMT if the number of pseudo-signals, $k_{T}^{\ast }$, do not rise
faster than $\sqrt{T}$. However, to reduce the cost associated with
parameter instability more information about the nature of time variations
in $\boldsymbol{\beta }_{it}\boldsymbol{\ }$and $\sigma _{ijt,x}$ are
required. For example, $\bar{\Delta}_{\beta ,T}$ (or $\bar{\Delta}_{\beta
,T}^{\ast }$) could be lower if $\boldsymbol{\Omega }_{\beta ,t}$ is close
to zero in some periods, or if there are cancelling effects from negative $%
\sigma _{ijt,x}$ ($\sigma _{ijt,x}^{\ast }$) when $\sigma _{ijt,\beta }$ is
positive, namely $\sigma _{ijt,x}\sigma _{ijt,\beta }<0$ ($\sigma
_{ijt,x}^{\ast }\sigma _{ijt,\beta }<0$), for some $i\neq j$ and some $t$.
This finding for the in-sample fit is similar to the results for mean
squared forecast errors in the presence of breaks in the literature, such as
Proposition 2 of \cite{pesaran2007selection} or equation (20) of \cite%
{pesaran2013optimal}, where the main focus is to minimize the MSFE by
mitigating the cost of parameter instability at the expense of increased
sampling uncertainty by appropriate weighting of the observations.

\section{Monte Carlo evidence\label{sec:MC-studies}}

We use Monte Carlo (MC) techniques to compare finite sample performance of
OCMT with and without down-weighting at the selection stage, as well as
comparing the OCMT\ results with those of Lasso, A-Lasso, and boosting. In
these comparisons we consider the number of selected covariates ($\hat{k}%
_{T} $), the true positive rate (TPR), the false positive rate (FPR), and
the one-step-ahead mean squared forecast error (MSFE) of the selected
models. Sub-section \ref{simul_designs} outlines the MC designs, sub-section %
\ref{simul methods} provides a summary of how the OCMT, Lasso, A-Lasso, and
boosting procedures are implemented, and finally sub-section \ref%
{simul_results} presents the main MC findings. {Details of Lasso, A-Lasso,
and boosting procures and how they are implemented are provided in Section %
\ref{lasso, ad-lasso and cv} of the online Monte Carlo supplement.}\vspace{%
-0.2cm}

\subsection{Simulation design\label{simul_designs}}

We consider the following data generating process (DGP): 
\begin{equation*}
y_{t}=c_{t}+\rho _{y,t}y_{t-1}+\sum_{j=1}^{k}\beta _{jt}\tilde{x}_{jt}+\tau
_{u}u_{t}\text{,}
\end{equation*}%
where the four signals $\tilde{x}_{jt},$ $j=1,2,3,4$ have non-zero,
time-varying means $\mu _{jt}=\mathbb{E}\left( \tilde{x}_{jt}\right) $. To
simplify the exposition of the DGP we consider the demeaned covariates, $%
x_{jt}=\tilde{x}_{jt}-\mu _{jt}$ (so that $\mathbb{E}\left( x_{jt}\right) =0$%
), and write the DGP equivalently as 
\begin{equation}
y_{t}=d_{t}+\rho _{y,t}y_{t-1}+\sum_{j=1}^{k}\beta _{jt}x_{jt}+\tau _{u}u_{t}%
\text{,}  \label{DGPbase}
\end{equation}%
where 
\begin{equation}
d_{t}=c_{t}+\sum_{j=1}^{k}\beta _{jt}\mu _{jt}.  \label{mu}
\end{equation}%
Since $c_{t}$ is a free parameter, without loss of generality we also treat $%
\left\{ d_{t}\,,\text{ }t=1,2,...,T\right\} $ as free parameters.

For each MC replication, $r=1,2,...,R$, the target variable, $y_{t}$, is
generated as random draws using (\ref{DGPbase}). The signal variables $%
x_{jt} $, $j=1,2,3,4$, are unknown and belong to a set $\mathcal{S}%
_{Nt}=\left\{ x_{1t},x_{2t},\cdots ,x_{Nt}\right\} $. The vector of
covariates $\mathbf{x}_{t}=\left( x_{1t},x_{2t},\cdots ,x_{Nt}\right)
^{\prime }$ is generated as $\mathbf{x}_{t}=\boldsymbol{R}_{t}^{1/2}%
\boldsymbol{\varepsilon }_{t}$, where $\boldsymbol{\varepsilon }_{t}=\left(
\varepsilon _{1t},\varepsilon _{2t},\cdots ,\varepsilon _{Nt}\right)
^{\prime }$. $\left\{ \varepsilon _{it}\right\} $ are generated as AR(1)
processes with GARCH(1,1) innovations 
\begin{equation*}
\varepsilon _{it}=\rho _{i\varepsilon }\varepsilon _{i,t-1}+\left( 1-\rho
_{i\varepsilon }^{2}\right) ^{1/2}e_{\varepsilon _{i}t}\text{, for }%
t=1,2,\cdots ,T,\text{and }i=1,2,...,N,
\end{equation*}%
using the starting values $\varepsilon _{i,0}\sim IIDN\left( 0,1\right) $.
The parameters were generated heterogeneously as independent draws, $%
\rho_{i\varepsilon }\sim IIDU\left( 0,0.95\right) $. $e_{\varepsilon
_{i}t}\sim IIDN\left( 0,\sigma_{\varepsilon _{i},t}^{2}\right) $, with $%
\sigma_{\varepsilon_{i},t}^{2}$ given by 
\begin{equation*}
\sigma _{\varepsilon _{i},t}^{2}=(1-\alpha _{1\varepsilon _{i}}-\alpha
_{2\varepsilon _{i}})+\alpha _{1\varepsilon _{i}}e_{\varepsilon
_{i}t-1}^{2}+\alpha _{2\varepsilon _{i}}\sigma _{\varepsilon _{i},t-1}^{2},
\end{equation*}%
where $\alpha _{1\varepsilon _{i}}\sim IIDU(0,0.2)$, and $\alpha
_{2\varepsilon _{i}}\sim IIDU(0.6,0.75)$. The error terms, $\left\{
u_{t}\right\} _{t=1}^{T}$, in (\ref{DGPbase}) are generated as $%
IIDN(0,\sigma _{ut}^{2})$ with $\sigma _{ut}^{2}$ following the GARCH(1,1)
specification 
\begin{equation*}
\sigma _{ut}^{{2}}=(1-\alpha _{1u}-\alpha _{2u})+\alpha _{1u}u_{t-1}^{{2}%
}+\alpha _{2u}\sigma _{u,t-1}^{{2}},
\end{equation*}%
using $u_{0}\sim \mathcal{N}(0,1)$, $\alpha _{1u}=0.2$ and $\alpha
_{2u}=0.75 $.

As our baseline DGP we consider a model with stable parameters, and set $%
\beta _{jt}=1$ for $j=1,2,3,4$. We also set $c_{t}=0$ and $\mu _{jt}=1$ in (%
\ref{mu}), which yields $d_{t}=4$. In addition, we set $\rho _{y,t}=0$ when
the baseline model is static and $\rho _{y,t}=0.3$ when the baseline model
is dynamic. In the dynamic case we set $y_{0}=(1-\rho _{y,1})^{-1}d_{1}$. In
the case of models with parameter instability we consider a mixed
deterministic-stochastic model and generate $\beta _{jt}$ as 
\begin{equation*}
\beta _{jt}=b_{jt}+\tau _{\eta _{j}}\eta _{jt},\text{ for }j=1,2,3,4,
\end{equation*}%
where $b_{jt}$ are deterministic and $\eta _{jt}$ are AR(1) processes with
GARCH(1,1) innovations, 
\begin{equation*}
\eta _{jt}=\rho _{\eta j}\eta _{j,t-1}+\left( 1-\rho _{\eta j}^{2}\right)
^{1/2}e_{\eta _{j}t},\text{ }
\end{equation*}%
using the starting values $\eta _{j,0}\sim IID\mathcal{N}\left( 0,1\right) $%
, and $\rho _{\eta j}=0.5$, for all $j$. $\left\{ e_{\eta _{j}t}\right\} $
follows a normal distribution with mean zero, and variance $\sigma _{\eta
_{j}t}^{{2}}$ given by 
\begin{equation*}
\sigma _{\eta _{j}t}^{{2}}=(1-\alpha _{1\eta _{j}}-\alpha _{2\eta
_{j}})+\alpha _{1\eta _{j}}e_{\eta _{j},t-1}^{{2}}+\alpha _{2\eta
_{j}}\sigma _{\eta _{j},t-1}^{{2}}\text{, for }j=1,2,3,4,
\end{equation*}%
where $\alpha _{1\eta _{j}}=0.2$ and $\alpha _{2\eta _{j}}=0.75$. We set $%
\tau _{\eta _{j}}$ such that deterministic variations in $\beta _{jt}$ are
quite large relative to the stochastic variations. To this end we set $\tau
_{\eta _{j}}$ (using simulations) so that 
\begin{equation*}
\frac{T^{-1}\sum_{t=1}^{T}b_{jt}^{2}}{T^{-1}\sum_{t=1}^{T}\mathbb{E}\left[
\left( \beta _{jt}^{\left( r\right) }\right) ^{2}\right] }=0.95,\text{ for }%
j=1,2,3,4.
\end{equation*}

For the deterministic components of the slope coefficients ($b_{jt}$, for $%
j=1,2,3,4$), we consider the following specifications 
\begin{equation}
b_{1t}=b_{2t}=%
\begin{cases}
2 & \text{if }t\in \{1,2,\cdots ,[T/3]\}, \\ 
0 & \text{if }t\in \{[T/3]+1,[T/3]+2,\cdots ,[2T/3]\}, \\ 
1 & \text{if }t\in \{[2T/3]+1,[2T/3]+2,\cdots ,T\},%
\end{cases}
\label{bm1}
\end{equation}%
and 
\begin{equation}
b_{3t}=b_{4t}=%
\begin{cases}
0.5 & \text{if }t\in \{1,2,\cdots ,[T/2]\}, \\ 
1.5 & \text{if }t\in \{[T/2]+1,[T/2]+2,\cdots ,T\},%
\end{cases}
\label{bm3}
\end{equation}%
where $[.]$ is the nearest integer function.

We \ also set $c_{t}=0$ in (\ref{mu}) and generate the intercept as $%
d_{t}=\sum_{j=1}^{k}\beta _{jt}\mu _{jt},$ where 
\begin{equation}
\mu _{1t}=\mu _{2t}=%
\begin{cases}
0.6 & \text{if }t\in \{1,2,\cdots ,[T/3]\}, \\ 
1.5 & \text{if }t\in \{[T/3]+1,[T/3]+2,\cdots ,[2T/3]\}, \\ 
0.9 & \text{if }t\in \{[2T/3]+1,[2T/3]+2,\cdots ,T\},%
\end{cases}
\label{mu1}
\end{equation}
and \vspace{-0.3cm} 
\begin{equation}
\mu _{3t}=\mu _{4t}=%
\begin{cases}
0.9 & \text{if }t\in \{1,2,\cdots ,[T/2]\}, \\ 
1.1 & \text{if }t\in \{[T/2]+1,[T/2]+2,\cdots ,T\}.%
\end{cases}
\label{mu3}
\end{equation}%
In this design, the jumps in $b_{jt}$ and $\mu _{jt}$, for $j=1,2$, have
opposite signs and the jumps in $b_{jt}$ and $\mu _{jt}$, for $j=3,4$, have
the same sign.

The $N\times N$ correlation matrix of the covariates, $\mathbf{R}_{t}\equiv
(r_{ij,t})$, are set as $r_{ij,t}=r_{t}^{|i-j|}$, for all $i,j=1,2,\cdots ,N$%
. We allow for a break in the correlation matrix and set $r_{t}$ equal to
0.9 in the first half of the sample and 0.4 in the second half of the
sample. Also, we consider two possibilities for $\rho _{y,t}$. In the static
scenario we set $\rho _{y,t}=0$ for all $t$. In the dynamic scenario we
allow for a switch in $r_{y,t}$ and set it as 
\begin{equation}
\rho _{y,t}=%
\begin{cases}
0.2 & \text{if }t\in \{1,2,\cdots ,[T/2]\}, \\ 
0.4 & \text{if }t\in \{[T/2]+1,[T/2]+2,\cdots ,T\}.%
\end{cases}
\label{rhoyt}
\end{equation}%
For the static and dynamic models with parameter instabilities, the
parameter $\tau _{u}$ is calibrated by simulations to ensure that the
R-squared of the linear regression of $y_{t}$ on a constant term, the signal
variables $\left\{ x_{1t},x_{2t},x_{3t},x_{4t}\right\} $, and (in
experiments with $\rho _{y,t}\neq 0$) the lagged dependent variable is equal
to $30\%$ (low fit) and $50\%$ (high fit). The same value of $\tau _{u}$ is
used for the corresponding static and dynamic models without parameter
instabilities.

We base the MC results on $R=2,000$ replications, and consider $N\in \left\{
20,40,100\right\} $ and $T\in \left\{ 100,200,500\right\} $, combinations.
These choices of $\left( N,T\right) $ cover our empirical applications. For
each pair of $(N,T)$, there are four experiments in case of the models with
no parameter instabilities, and four experiments in the case of models with
parameter instabilities, corresponding to the two choices of $\tau _{u}$
(low and high fit), $\rho _{yt}$ (static to dynamic). In total, we carry out
eight different experiments.

\subsection{Selection and estimation methods using weighted and unweighted
observations\label{simul methods}}

Let $\mathbf{w}_{t}=(\mathbf{x}_{t}^{\prime },y_{t})^{\prime }$, $%
t=1,2,\cdots ,T$ be the (unweighted) set of available observations, and
denote the corresponding set of down-weighted observations by $\hat{\mathbf{w%
}}_{t}(\lambda )=\lambda ^{T-t}\mathbf{w}_{t}$ where $0<\lambda \leq 1$ is
the down-weighting coefficient. We are not arguing for the use of
exponential down-weighting -- but use it as an example. There are also
non-exponential type down-weighting schemes that one can use, e.g. \cite%
{pesaran2013optimal}. We will consider the following selection/estimation
methods: (1) OCMT with down-weighted observations $\{\hat{\mathbf{w}}%
_{t}(\lambda )\}_{t=1}^{T}$ used at both selection and estimation stages;
(2) OCMT with the unweighted observations, $\{\mathbf{w}_{t}\}_{t=1}^{T}$,
used at the selection stage and down-weighted observations, $\{\hat{\mathbf{w%
}}_{t}(\lambda )\}_{t=1}^{T}$, used at the estimation stage; (3) OCMT using
unweighted observations, $\{\mathbf{w}_{t}\}_{t=1}^{T}$, at both selection
and estimation stages; (4,5 \& 6) Lasso, A-Lasso, and boosting also using
unweighted observations, $\{\mathbf{w}_{t}\}_{t=1}^{T}$; and (7,8 \& 9)
Lasso, A-Lasso, and boosting with down-weighted observations, $\{\hat{%
\mathbf{w}}_{t}(\lambda )\}_{t=1}^{T}$ used as inputs.

We also implement a two-step procedures based on Lasso, A-Lasso and
boosting. In the first step, we apply Lasso, A-Lasso and boosting to the
original (unweighted) observations and select the variables with non-zero
coefficients. In the second step, we estimate the corresponding
post-selected model by LS using the weighted observations. Overall, the
MSFEs of these procedures were higher than that of direct application of
Lasso, A-Lasso and boosting to the weighted observations. The results are
available in Section S-2 of the online MC supplement.

We consider two sets of values for the down-weighting coefficient, $\lambda $%
: (1) Light down-weighting with $\lambda =\left\{
0.975,0.98,0.985,0.99,0.995,1\right\} $, and (2) Heavy down-weighting with $%
\lambda =\left\{ 0.95,0.96,0.97,0.98,0.99,1\right\} $. For each of the above
two sets of exponential down-weighting schemes (light/heavy) we focus on
simple average forecasts computed over the individual forecasts obtained for
each value of $\lambda $ in the set under consideration.

\subsection{Simulation results\label{simul_results}}

A summary of the main results are provided in Tables \ref{mc_tab_1} to \ref%
{mc_tab_3}, with additional summary tables highlighting the effects of
down-weighting at the selection stage, and the differences between static
versus dynamic models provided in the online MC supplement. Table \ref%
{mc_tab_1} give the number of selected covariates ($\hat{k}_{T}$), TPR and
FPR of OCMT, Lasso, A-Lasso and boosting without down-weighting. Panel A of
this table reports the results for different $N$ and $T$ combinations,
averaged across the four experiments without parameter instabilities, and
panel B of the table gives the corresponding results for the four
experiments with parameter instabilities. The results show that all the
methods under consideration have higher average TPR for models with stable
parameters compared to the ones with parameter instabilities. This is to be
expected, as the models with parameter instabilities are subject to an
additional source of uncertainty.

We further observe that the lower average TPR of OCMT in the models with
parameter instabilities is associated with a lower average number of
selected covariates, and hence a lower average FPR. On the other hand, the
other procedures tend, on average, to select more covariates in the models
with parameter instabilities and hence have a higher average FPR relative to
the models without parameter instabilities. Lastly, OCMT most of the times
selects fewer covariates relative to Lasso, A-Lasso, and boosting , while
maintaining the TPR at a similar level. As a result, OCMT has mostly the
lowest average FPR among the selection methods under consideration. Summary
Tables \ref{mc_tab_s1} and \ref{mc_tab_s2} in the online MC supplement
provide further results on the effects of down-weighting on TPR and FPR. The
results consistently show that down-weighting of observations provides no
gains for OCMT in terms of average TPR and FPR. This is also true for other
methods in majority but not all cases.

Table \ref{mc_tab_2} focusses on the one-step-ahead MSFEs and provides
comparative results on the effects of down-weighting across the methods
(OCMT, Lasso, A-Lasso and boosting). As in Table \ref{mc_tab_1}, Panel A of
Table \ref{mc_tab_2} gives average MSFEs for the four experiments without
parameter instabilities, and Panel B gives the corresponding results for the
experiments with parameter instabilities. As expected, in the absence of
parameter instabilities, using unweighted observations gives the lowest MSFE
across all the methods. Moreover, for all $N$ and $T$ combinations and
different down-weighting scenarios, the average MSFE of each method is lower
in the case of models with stable parameters as compared to those with
parameter instabilities. This observation is consistent with our finding in
Theorem \ref{mean square error} about the cost of time-variation in the
coefficients on the in-sample fit of the estimated model. As can be seen,
for models with parameter instabilities, down-weighting does improve the
forecasting performance of OCMT (with and without down-weighting in the
selection stage), Lasso, and A-Lasso. However, by comparing the MSFEs of
OCMT with and without down-weighting at the selection stage, we see that the
down-weighting at the selection stage always results in deterioration of the
forecast accuracy of OCMT, which is in line with our main theoretical
result. Last but not least, the results in Table \ref{mc_tab_2} show that
OCMT with down-weighting only at the estimation stage almost always has the
lowest average MSFE among all the methods for all choices of $N$, $T$, and
different down-weighting scenarios. In fact, in the case of experiments with
parameter instabilities OCMT with down-weighting (light or heavy) at the
estimation stage only, always beats Lasso, A-Lasso and boosting with light
or heavy down-weighting in terms of the one-step-ahead MSFE.

Table \ref{mc_tab_3} compares the performance of OCMT with the
down-weighting option at the estimation stage to that of the other
procedures, using the same set of down-weighting parameter ($\lambda $).
Specifically, we report the MSFE of Lasso, A-Lasso, and boosting relative to
that of OCMT. Since the relative MSFE ranking of OCMT, Lasso, A-Lasso, and
boosting does not appear to be affected by no/light/heavy down-weighting
options, as a summary measure, we simply average relative MSFE values across
individual experiments and the three (no/light/heavy) down-weighting
options. However, we provide the relative MSFE results for the models
without and with parameter instabilities separately, on left and right
panels of Table \ref{mc_tab_3}. Two observations stand out from this table.
First, the reported average relative MSFEs are almost always greater than
one for all the $N$ and $T$ choices, indicating that OCMT outperforms Lasso,
A-Lasso, and boosting. Second, the degree to which OCMT outperforms Lasso
and A-Lasso tends to increase with the degree of parameter instability. This
is less so if we compare OCMT\ with boosting.

Tables \ref{mc-tab-s4}, \ref{mc_tab_s6}, and \ref{mc_tab_s7} in the online
MC supplement provide further details about the performance of the methods
under consideration in static and dynamic experiments. In Table \ref%
{mc-tab-s4}, we compare the number of selected covariates, the TPR, and the
FPR of each method without down-weighting across static and dynamic models.
For various $N$ and $T$ combinations the reported results are averaged
across four experiments (with/without parameter instabilities and
with/without high-fit). The results show that all the methods tend to select
fewer covariates in the dynamic models relative to the static ones, and
hence have a lower TPR and FPR. This is expected, as in the dynamic models,
part of the variation in the target variable is explained by its own lag
rather than the signal variables. Consequently, in Tables \ref{mc_tab_s6}
and \ref{mc_tab_s7}, which are about the MSFE in static and dynamic models,
respectively, we see that all the methods have a higher MSFE in dynamic
models relative to the static ones. Additionally, the results in Tables \ref%
{mc_tab_s6} and \ref{mc_tab_s7} show that the MSFE for models with stable
parameters is always lower than the ones with parameter instabilities,
regardless of whether the model is static or not.

Overall, the results of our MC studies suggest that the OCMT procedure
without down-weighting at the selection stage is a useful method to deal
with variable selection in linear regression settings with parameter
instability.

\section{Empirical applications\label{empirical section}}

The rest of the paper considers empirical applications whereby the forecast
performance of the proposed OCMT approach with no down-weighting at the
selection stage is compared with those of Lasso and A-Lasso. In particular,
we consider the following two applications:\vspace{-0.2cm}\footnote{%
We also consider forecasting euro area quarterly output growth using the
European Central Bank (ECB) survey of professional forecasters as our third
application. The results of this application can be found in Section \ref%
{ECB survey of growth} of the online empirical supplement.}

\begin{itemize}
\item Forecasting monthly rate of price changes for 28 (out of 30) stocks in
Dow Jones using a relatively large number of financial, economic, as well as
technical indicators.\vspace{-0.2cm}

\item Forecasting quarterly output growth rates across 33 countries using
macro and financial variables.\vspace{-0.2cm}
\end{itemize}

\vspace{-0.2cm}In each application, we first compare the performance of OCMT
with and without down-weighted observations at the selection stage. We then
consider the comparative performance of OCMT (with variable selection
carried out without down-weighting) relative to Lasso and A-Lasso, with and
without down-weighting. For down-weighting we make use of exponentially
down-weighted observations, namely $\hat{x}_{it}(\lambda )=\lambda
^{T-t}x_{it}$, and $\hat{y}_{t}(\lambda )=\lambda ^{T-t}y_{t}$, where $y_{t}$
is the target variable to be forecasted, $x_{it}$, for $i=1,2,...,N$ are the
covariates in the active set, and $\lambda $ is the exponential decay
coefficient. We consider the same two sets of values for the degree of
exponential decay, $\lambda $, as in the MC section: (1) Light
down-weighting with $\lambda =\left\{ 0.975,0.98,0.985,0.99,0.995,1\right\} $%
, and (2) Heavy down-weighting with $\lambda =\left\{
0.95,0.96,0.97,0.98,0.99,1\right\} $. For each of the above two sets of
exponential down-weighting schemes we again focus on simple average
forecasts computed over the individual forecasts obtained for each value of $%
\lambda $ in the set under consideration.

For forecast evaluation we consider Mean Squared Forecasting Error (MSFE)
and Mean Directional Forecast Accuracy (MDFA), together with related pooled
versions of\ Diebold-Mariano (DM), and Pesaran-Timmermann (PT) test
statistics. A panel version of \cite{diebold2002comparing} test is proposed
by \cite{pesaran2009forecasting}. Let $q_{lt}\equiv e_{ltA}^{2}-e_{ltB}^{2}$
be the difference in the squared forecasting errors of procedures $A$ and $B$%
, for the target variable $y_{lt}$ ($l=1,2,...,L)$ and $t=1,2,...,T_{l}^{f}$%
, where $T_{l}^{f}$ is the number of forecasts for target variable $l$
(could be one or multiple step ahead) under consideration. Suppose $%
q_{lt}=\alpha _{l}+\varepsilon _{lt}$ with $\varepsilon _{lt}\sim \mathcal{N}%
(0,\sigma _{l}^{2})$. Then under the null hypothesis of $H_{0}:\alpha _{l}=0$
for all $l$ we have 
\begin{equation*}
\overline{DM}=\frac{\bar{q}}{\sqrt{V(\bar{q})}}\overset{a}{\thicksim } 
\mathcal{N}(0,1),\text{ for }T_{Lf}\rightarrow \infty, \text{ where }
T_{Lf}=\sum_{l=1}^{L}T_{l}^{f}, \bar{q}=T_{Lf}^{-1}\sum_{l=1}^{L}
\sum_{t=1}^{T_{l}^{f}}q_{lt}, \text{ and }
\end{equation*}%
\begin{equation*}
V(\bar{q})=\frac{1}{T_{Lf}^{2}}\sum_{l=1}^{L}T_{l}^{f}\hat{\sigma}_{l}^{2}, 
\text{ with }\hat{\sigma}_{l}^{2}=\frac{1}{T_{l}^{f}}
\sum_{t=1}^{T_{l}^{f}}(q_{lt}-\bar{q}_{l})^{2}\text{ and }\bar{q}_{l}=\frac{%
1 }{T_{l}^{f}}\sum_{t=1}^{T_{l}^{f}}q_{lt}.
\end{equation*}%
Note that $V(\bar{q})$ needs to be modified in the case of multiple-step
ahead forecast errors, due to the serial correlation that results in the
forecast errors from the use of over-lapping observations. There is no
adjustment needed for one-step ahead forecasting, since it is reasonable to
assume that in this case the loss differentials are serially uncorrelated.
However, to handle possible serial correlation for $h$-step ahead
forecasting with $h>1$, we can modify the panel DM test by using the
Newey-West type estimator of $\sigma _{l}^{2}$.

The $MDFA$ statistic compares the accuracy of forecasts in predicting the
direction (sign) of the target variable, and is computed as 
\begin{equation*}
MDFA=100\left\{ \frac{1}{T_{Lf}}\sum_{l=l}^{L}\sum_{t=1}^{T_{l}^{f}}\mathbf{1%
}[\text{sgn}(y_{lt}y_{lt}^{f})>0]\right\} ,
\end{equation*}%
where $\mathbf{1}(w>0)$ is the indicator function takes the value of $1$
when $w>0$ and zero otherwise, $\text{sgn}(w)$ is the sign function, $y_{lt}$
is the actual value of dependent variable at time $t$ and $y_{lt}^{f}$ is
its corresponding predicted value. To evaluate statistical significance of
the directional forecasts for each method, we also report a pooled version
of the test suggested by \cite{pesaran1992simple}: 
\begin{equation*}
PT=\frac{\hat{P}-\hat{P}^{\ast }}{\sqrt{\hat{V}(\hat{P})-\hat{V}(\hat{P}%
^{\ast })}},\text{ }
\end{equation*}%
where $\hat{P}$ is the estimator of the probability of correctly predicting
the sign of $y_{lt}$, computed by%
\begin{equation*}
\hat{P}=\frac{1}{T_{Lf}}\sum_{l=1}^{L}\sum_{t=1}^{T_{l}^{f}}\mathbf{1}[\text{
sgn}(y_{lt}y_{lt}^{f})>0],\text{ and \ }\hat{P}^{\ast }=\bar{d}_{y}\bar{d}%
_{y^{f}}+(1-\bar{d}_{y})(1-\bar{d}_{y^{f}}),\text{ with}
\end{equation*}%
\begin{equation*}
\bar{d}_{y}=\frac{1}{T_{Lf}}\sum_{l=1}^{L}\sum_{t=1}^{T_{l}^{f}}\mathbf{1}[%
\text{sgn}(y_{lt})>0],\text{ and }\bar{d}_{y^{f}}=\frac{1}{T_{Lf}}%
\sum_{l=1}^{L}\sum_{t=1}^{T_{l}^{f}}\mathbf{1}[\text{sgn}(y_{lt}^{f})>0].
\end{equation*}%
Finally,$\ \hat{V}(\hat{P})=T_{Lf}^{-1}\hat{P}^{\ast }(1-\hat{P}^{\ast }),$
and 
\begin{equation*}
\hat{V}(\hat{P^{\ast }})=\frac{1}{T_{Lf}}(2\bar{d}_{y}-1)^{2}\bar{d}%
_{y^{f}}(1-\bar{d}_{y^{f}})+\frac{1}{T_{Lf}}(2\bar{d}_{y}^{f}-1)^{2}\bar{d}%
_{y}(1-\bar{d}_{y})+\frac{4}{T_{Lf}^{2}}\bar{d}_{y}\bar{d}_{y^{f}}(1-\bar{d}%
_{y})(1-\bar{d}_{y^{f}}).
\end{equation*}%
The last term of $\hat{V}(\hat{P^{\ast }})$ is negligible and can be
ignored. Under the null hypothesis, that prediction and realization are
independently distributed, PT is asymptotically distributed as a standard
normal distribution.

\subsection{Forecasting monthly returns of stocks in Dow Jones}

In this application the focus is on forecasting one-month ahead stock
returns, defined as monthly change in natural logarithm of stock prices. We
consider stocks that were part of the Dow Jones index in 2017m12, and have
non-zero prices for at least 120 consecutive data points (10 years) over the
period 1980m1 and 2017m12. We ended up forecasting 28 blue chip stocks. 
\footnote{%
Visa and DowDuPont are excluded since they have less than 10 years of
historical price data.} Daily close prices for all the stocks are obtained
from Data Stream. For stock $i$, the price at the last trading day of each
month is used to construct the corresponding monthly stock prices, $P_{it}$.
Finally, monthly returns are computed by $r_{i,t+1}=100\ln(P_{i,t+1}/P_{it}) 
$, for $i=1,2,...,28$. For all 28 stocks we use an expanding window starting
with the observations for the first 10 years ($T=120$). The active set for
predicting $r_{i,t+1}$ consists of 40 financial, economic, and technical
variables.\footnote{%
All regressions include the intercept as the only conditioning
(pre-selected) variable.} The full list and the description of the
indicators considered can be found in Section \ref{Appendix B} of online
empirical supplement.

Overall we computed 8,659 monthly forecasts for the 28 target stocks. The
results are summarized as average forecast performances across the different
variable selection procedures. Table \ref{dow jones stocks ocmt} reports the
effects of down-weighting at the selection stage of the OCMT procedure. It
is clear that down-weighting worsens the predictive accuracy of OCMT. From
the Panel DM tests, we can also see that down-weighting at the selection
stage worsens the forecasts significantly. Panel DM test statistics is
-5.606 (-11.352) for light (heavy) versus no down-weighing at the selection
stage. Moreover, Table \ref{dow jones stocks ocmt vs lasso} shows that the
OCMT procedure with no down-weighting at the selection stage dominates
Lasso, A-Lasso and boosting in terms of MSFE and the differences are
statistically highly significant.

Further, OCMT outperforms Lasso, A-Lasso and boosting in terms of Mean
Directional Forecast Accuracy (MDFA), measured as the percent number of
correctly signed one-month ahead forecasts across all the 28 stocks over the
period 1990m2-2017m12. See Table \ref{dow jones stocks MDA}. As can be seen
from this table, OCMT\ with no down-weighting performs the best; correctly
predicting the direction of 56.057\% of 8,659 forecasts, as compared to
55.769\%, which we obtain for Lasso, A-Lasso and boosting forecast, at best.
This difference is highly significant considering the very large number of
forecasts involved. It is also of interest that the better of performance of
OCMT is achieved with a much fewer number of selected covariates as compared
to Lasso, A-Lasso and boosting. As can be seen from the last column of Table %
\ref{dow jones stocks MDA}, Lasso, A-Lasso and boosting on average select
many more covariates than OCMT (1-15 variables as compared to 0.072 for
OCMT).

So far we have focused on average performance across all the 28 stocks.
Table \ref{dow jone individual stocks summary} provides the summary results
for individual stocks, showing the relative performance of OCMT in terms of
the number of stocks, using MSFE and MDFA criteria. The results show that
OCMT performs better than Lasso, A-Lasso and boosting in the majority of the
stocks in terms of MSFE and MDFA. OCMT outperforms Lasso, A-Lasso and
boosting in at least 22 out of 28 stocks in terms of MSFE, under no
down-weighting, and almost universally when Lasso, A-Lasso and boosting are
implemented with down-weighting. Similar results are obtained when we
consider MDFA criteria, although the differences in performance are somewhat
less pronounced. Overall, we can conclude that the better average
performance of OCMT (documented in Tables \ref{dow jones stocks ocmt vs
lasso} and \ref{dow jones stocks MDA}) is not driven by a few stocks and
holds more generally.

\subsection{Forecasting quarterly output growth rates across 33 countries}

We consider one and two years ahead predictions of output growth for 33
countries (20 advanced and 13 emerging). We use quarterly data from $1979Q2$
to $2016Q4$ taken from the GVAR dataset.\footnote{%
The GVAR dataset is available at %
\url{https://sites.google.com/site/gvarmodelling/data}.} We predict $%
\Delta_{4}y_{it}=y_{it}-y_{i,t-4}$, and $\Delta _{8}y_{it}=y_{it}-y_{i,t-8},$
where $y_{it}$, is the log of real output for country $i$. We adopt the
following direct forecasting equations: 
\begin{equation*}
\Delta _{h}y_{i,t+h}=y_{i,t+h}-y_{it}=\alpha _{ih}+\lambda _{ih}\Delta
_{1}y_{it}+\boldsymbol{\beta }_{ih}^{\prime }\mathbf{x}_{it}+u_{iht},
\end{equation*}%
where we consider $h=4$ (one-year-ahead forecasts) and $h=8$
(two-years-ahead forecasts). Given the known persistence in output growth,
in addition to the intercept in the present application we also condition on
the most recent lagged output growth, denoted by $%
\Delta_{1}y_{it}=y_{it}-y_{i,t-1}$, and confine the variable selection to
list of variables set out in Table \ref{countries gdp growth active set} in
the online empirical supplement. Overall, we consider a maximum of 15
covariates in the active set covering quarterly changes in domestic
variables such as real output growth, real short term interest rate, and
long-short interest rate spread and quarterly change in the corresponding
foreign variables.

We use expanding samples, starting with the observations on the first 15
years (60 data points), and evaluate the forecasting performance of the
three methods over the period 1997Q2 to 2016Q4.

Tables \ref{countries gdp ocmt 1} and \ref{countries gdp ocmt 2},
respectively, report the MSFE of OCMT for one-year and two-year ahead
forecasts of output growth, with and without down-weighting at the selection
stage. Consistent with the previous application, down-weighting at the
selection stage worsens the forecasting accuracy. Moreover, in Tables \ref%
{countries gdp ocmt vs lasso 1} and \ref{countries gdp ocmt vs lasso 2}, we
can see that OCMT (without down-weighting at the selection stage)
outperforms Lasso, A-Lasso and boosting in two-year ahead forecasting. In
the case of one-year ahead forecasts, OCMT and Lasso are very close to each
other and both outperform A-Lasso and boosting. Table \ref{countries gdp
summary} summarizes country-specific MSFE and DM findings for OCMT relative
to Lasso, A-Lasso and boosting. The results show OCMT under-performs Lasso
in more than half of the countries for one-year ahead horizon, but
outperforms Lasso, A-Lasso and boosting in more than 70 percent of the
countries in the case of two-year ahead forecasts. It is worth noting that
while Lasso generally outperforms OCMT in the case of one-year ahead
forecasts, overall its performance is not statistically significantly
better. See Panel DM test of Table \ref{countries gdp ocmt vs lasso 1}. On
the other hand we can see from Table \ref{countries gdp ocmt vs lasso 2}
that overall OCMT significantly outperforms Lasso in the case of the
two-year ahead forecasts.

Finally in Tables \ref{countries gdp MDA & PT test 1} and \ref{countries gdp
MDA & PT test 2} we reports MDFA and PT test statistics for OCMT, Lasso,
A-Lasso and boosting. Overall, OCMT has a slightly higher MDFA and hence
predicts the direction of real output growth better than Lasso, A-Lasso and
boosting in most cases. The PT test statistics suggest that while all the
methods perform well in forecasting the direction of one-year ahead real
output growth, none of the methods considered are successful at predicting
the direction of two-year ahead output growth.

It is also worth noting that as with the previous applications, OCMT selects
very few variables from the active set (0.1 on average for both horizons,
with the maximum number of selected variables being 2 for $h=4$ and $8$). On
the other hand, Lasso on average selects 2.7 variables from the active set
for $h=4$, and $1$ variable on average for $h=8$. Maximum number of
variables selected by Lasso is $9$ and $13$ for $h=4$, $8$, respectively
(out of possible $15$). Again as to be expected, A-Lasso selects a fewer
number of variables as compared to Lasso (2.3 and 0.8 on average for $h=4,8$%
, respectively), but this does not lead to a better forecast performance in
comparison with Lasso. Boosting on average selects $2.7$ variables from the
active set for $h=4$, and $1.4$ variables on average for $h=8$.

In conclusion, down-weighting at both selection and forecasting stages
deteriorates OCMT's MSFE for both one--year and two-years ahead forecast
horizons, as compared to down-weighting only at the forecasting stage.
Moreover, light down-weighting at the forecasting stage improves forecasting
performance for both horizons. Statistically significant evidence of
forecasting skill is found for OCMT relative to Lasso only in the case of
two-years ahead forecasts. However, it is interesting that none of the big
data methods can significantly beat the simple (light down-weighted) AR(1)
baseline model.

\section{Concluding remarks\label{conclusion}}

The penalized regression approach has become the \textit{de facto} benchmark
in the literature on variable selection in the context of linear regression
models. But, barring a few exceptions (such as \citealp{kapetanios2018time}%
), these studies focus on models with stable parameters, and do not consider
the implications of parameter instabilities for variable selection.
Recently, \cite{chudik2018one} proposed OCMT as an alternative procedure to
penalized regression. One feature of the OCMT procedure is the fact that the
problem of variable selection is separated from the forecasting stage, in
contrast to the penalized regression techniques where the variable selection
and estimation are carried out simultaneously. Using OCMT one can decide
whether to use the weighted observations at the variable selection stage or
not, without preempting whether to down-weight and how to down-weight the
observations at the forecasting stage.

We have provided theoretical arguments for using the unweighted observations
at the selection stage of OCMT, and down-weighted observations at the
forecasting stage of OCMT. Our MC results as well as empirical applications
uniformly suggest that OCMT without down-weighting at the selection stage
outperforms, in terms of mean squared forecast errors, Lasso, Adaptive
Lasso, boosting, as well as when OCMT is applied with down-weighted
observations.

\begin{table}
	\caption{\footnotesize The number of selected variables  ($\hat{k}_T$), True Positive Rate (TRP), and False Positive Rate (FPR) averaged across Monte Carlo experiments with and without parameter instabilities. \bigskip}\label{mc_tab_1}
	\centering
	\renewcommand{\arraystretch}{1}%
	\footnotesize%
\begin{tabular}{cccccccccccc}
\hline\hline
& \multicolumn{3}{c}{$\hat{k}_{T}$} &  & \multicolumn{3}{c}{TPR} &  & 
\multicolumn{3}{c}{FPR} \\ \cline{2-4}\cline{6-8}\cline{10-12}
$N\backslash T$ & \textbf{100} & \textbf{150} & \textbf{200} &  & \textbf{100%
} & \textbf{150} & \textbf{200} &  & \textbf{100} & \textbf{150} & \textbf{%
200} \\ \cline{1-4}\cline{6-8}\cline{10-12}
\multicolumn{12}{l}{A. Without parameter instabilities} \\ \hline
& \multicolumn{11}{l}{OCMT} \\ \hline
\textbf{20} & 5.03 & 6.17 & 7.22 &  & 0.83 & 0.91 & 0.96 &  & 0.08 & 0.13 & 
0.17 \\ 
\textbf{40} & 4.69 & 5.98 & 6.87 &  & 0.80 & 0.91 & 0.95 &  & 0.04 & 0.06 & 
0.08 \\ 
\textbf{100} & 4.31 & 5.52 & 6.35 &  & 0.77 & 0.88 & 0.93 &  & 0.01 & 0.02 & 
0.03 \\ \hline
& \multicolumn{11}{l}{Lasso} \\ \hline
\textbf{20} & 6.82 & 7.00 & 7.20 &  & 0.84 & 0.89 & 0.93 &  & 0.17 & 0.17 & 
0.17 \\ 
\textbf{40} & 8.26 & 8.57 & 8.74 &  & 0.82 & 0.89 & 0.92 &  & 0.12 & 0.13 & 
0.13 \\ 
\textbf{100} & 10.76 & 11.00 & 10.51 &  & 0.79 & 0.87 & 0.90 &  & 0.08 & 0.08
& 0.07 \\ \hline
& \multicolumn{11}{l}{A-Lasso} \\ \hline
\textbf{20} & 5.15 & 5.35 & 5.55 &  & 0.73 & 0.80 & 0.85 &  & 0.11 & 0.11 & 
0.11 \\ 
\textbf{40} & 6.39 & 6.78 & 6.96 &  & 0.73 & 0.81 & 0.86 &  & 0.09 & 0.09 & 
0.09 \\ 
\textbf{100} & 8.65 & 9.05 & 8.83 &  & 0.72 & 0.81 & 0.86 &  & 0.06 & 0.06 & 
0.05 \\ \hline
& \multicolumn{11}{l}{Boosting} \\ \hline
\textbf{20} & 4.59 & 4.63 & 4.70 &  & 0.77 & 0.83 & 0.88 &  & 0.08 & 0.07 & 
0.06 \\ 
\textbf{40} & 6.04 & 5.79 & 5.69 &  & 0.76 & 0.83 & 0.87 &  & 0.07 & 0.06 & 
0.05 \\ 
\textbf{100} & 11.36 & 9.27 & 8.43 &  & 0.75 & 0.82 & 0.86 &  & 0.08 & 0.06
& 0.05 \\ \hline
\multicolumn{12}{l}{B. With parameter instabilities} \\ \hline
& \multicolumn{11}{l}{OCMT} \\ \hline
\textbf{20} & 4.04 & 5.07 & 5.96 &  & 0.73 & 0.85 & 0.92 &  & 0.06 & 0.08 & 
0.11 \\ 
\textbf{40} & 3.78 & 4.90 & 5.67 &  & 0.70 & 0.84 & 0.91 &  & 0.02 & 0.04 & 
0.05 \\ 
\textbf{100} & 3.54 & 4.62 & 5.26 &  & 0.66 & 0.81 & 0.88 &  & 0.01 & 0.01 & 
0.02 \\ \hline
& \multicolumn{11}{l}{Lasso} \\ \hline
\textbf{20} & 7.28 & 7.76 & 8.17 &  & 0.76 & 0.82 & 0.87 &  & 0.21 & 0.22 & 
0.23 \\ 
\textbf{40} & 9.80 & 10.60 & 11.13 &  & 0.74 & 0.82 & 0.86 &  & 0.17 & 0.18
& 0.19 \\ 
\textbf{100} & 13.68 & 14.83 & 15.56 &  & 0.70 & 0.79 & 0.83 &  & 0.11 & 0.12
& 0.12 \\ \hline
& \multicolumn{11}{l}{A-Lasso} \\ \hline
\textbf{20} & 5.49 & 5.95 & 6.30 &  & 0.65 & 0.72 & 0.78 &  & 0.15 & 0.15 & 
0.16 \\ 
\textbf{40} & 7.55 & 8.28 & 8.76 &  & 0.64 & 0.73 & 0.79 &  & 0.12 & 0.13 & 
0.14 \\ 
\textbf{100} & 10.71 & 11.85 & 12.58 &  & 0.63 & 0.73 & 0.78 &  & 0.08 & 0.09
& 0.09 \\ \hline
& \multicolumn{11}{l}{Boosting} \\ \hline
\textbf{20} & 4.59 & 4.66 & 4.75 &  & 0.68 & 0.74 & 0.79 &  & 0.09 & 0.08 & 
0.08 \\ 
\textbf{40} & 6.52 & 6.35 & 6.21 &  & 0.68 & 0.75 & 0.80 &  & 0.10 & 0.08 & 
0.08 \\ 
\textbf{100} & 12.70 & 10.73 & 10.03 &  & 0.67 & 0.74 & 0.78 &  & 0.10 & 0.08
& 0.07 \\ \hline\hline
\end{tabular}%
	\vspace{-.5cm}
	\begin{flushleft}
		\noindent 
		\scriptsize%
		\singlespacing%
		Notes: There are $k = 4$ signal variables out of N observed covariates. The reported results for OCMT, Lasso, A-Lasso, and boosting in the table are based on the original (not down-weighted) observations. Each experiment is based on 2000 Monte Carlo replications. See Section \ref{sec:MC-studies} for the detailed description of the Monte Carlo design.
	\end{flushleft}
\end{table}

\begin{table} 
	\caption{\footnotesize The effects of down-weighting on one-step-ahead MSFE of OCMT, Lasso, A-Lasso and boosting averaged across all MC experiments with and without parameter instabilities.\bigskip} \label{mc_tab_2}
	\centering
	\renewcommand{\arraystretch}{1}%

	\footnotesize%
\begin{tabular}{cccccccccccc}
\hline\hline
\multicolumn{1}{l}{Down-weighting$^{\dagger }$:} & No & Light & Heavy &  & No
& Light & Heavy &  & No & Light & Heavy \\ 
\cline{2-4}\cline{6-8}\cline{10-12}
$N\backslash T$ & \multicolumn{3}{c}{\textbf{100}} &  & \multicolumn{3}{c}{%
\textbf{150}} &  & \multicolumn{3}{c}{\textbf{200}} \\ \hline
\multicolumn{12}{l}{A. Without parameter instabilities} \\ \hline
& \multicolumn{11}{l}{OCMT(Down-weighting only at the estimation stage)} \\ 
\hline
\textbf{20} & \textbf{31.76} & 32.66 & 34.20 &  & \textbf{28.53} & 29.33 & 
31.06 &  & \textbf{26.19} & 27.17 & 28.75 \\ 
\textbf{40} & \textbf{29.13} & 29.56 & 30.51 &  & \textbf{26.72} & 27.24 & 
28.52 &  & \textbf{32.05} & 34.00 & 36.29 \\ 
\textbf{100} & \textbf{29.25} & 29.56 & 30.49 &  & \textbf{27.93} & 28.97 & 
30.69 &  & \textbf{28.94} & 29.64 & 31.48 \\ \hline
& \multicolumn{11}{l}{OCMT(Down-weighting at the variable selection and
estimation stages)} \\ \hline
\textbf{20} & \textbf{31.76} & 32.61 & 35.08 &  & \textbf{28.53} & 29.25 & 
32.19 &  & \textbf{26.19} & 27.36 & 30.67 \\ 
\textbf{40} & \textbf{29.13} & 29.46 & 31.95 &  & \textbf{26.72} & 27.21 & 
31.50 &  & \textbf{32.05} & 34.27 & 41.13 \\ 
\textbf{100} & \textbf{29.25} & 30.20 & 33.85 &  & \textbf{27.93} & 29.46 & 
36.72 &  & \textbf{28.94} & 31.19 & 40.14 \\ \hline
& \multicolumn{11}{l}{Lasso} \\ \hline
\textbf{20} & \textbf{31.82} & 33.35 & 35.49 &  & \textbf{28.59} & 29.49 & 
31.61 &  & \textbf{26.25} & 27.22 & 29.08 \\ 
\textbf{40} & \textbf{29.48} & 30.91 & 34.16 &  & \textbf{26.35} & 28.00 & 
32.31 &  & \textbf{31.78} & 33.75 & 37.80 \\ 
\textbf{100} & \textbf{30.63} & 33.29 & 37.05 &  & \textbf{28.33} & 30.90 & 
35.16 &  & \textbf{29.13} & 31.43 & 35.10 \\ \hline
& \multicolumn{11}{l}{A-Lasso} \\ \hline
\textbf{20} & \textbf{33.24} & 34.72 & 37.09 &  & \textbf{29.47} & 30.44 & 
32.85 &  & \textbf{27.01} & 27.82 & 30.13 \\ 
\textbf{40} & \textbf{31.66} & 32.87 & 36.30 &  & \textbf{27.98} & 30.08 & 
34.72 &  & \textbf{33.03} & 35.09 & 38.89 \\ 
\textbf{100} & \textbf{35.29} & 37.89 & 41.49 &  & \textbf{30.91} & 33.92 & 
38.70 &  & \textbf{31.52} & 34.37 & 38.13 \\ \hline
& \multicolumn{11}{l}{Boosting} \\ \hline
\textbf{20} & \textbf{32.69} & 35.51 & 41.25 &  & \textbf{29.51} & 31.98 & 
38.09 &  & \textbf{26.77} & 29.22 & 35.82 \\ 
\textbf{40} & \textbf{30.67} & 34.22 & 42.31 &  & \textbf{27.20} & 31.66 & 
41.76 &  & \textbf{32.90} & 40.16 & 51.66 \\ 
\textbf{100} & \textbf{33.68} & 42.00 & 48.44 &  & \textbf{29.28} & 38.67 & 
46.82 &  & \textbf{29.98} & 39.84 & 47.17 \\ \hline
\multicolumn{12}{l}{B. With parameter instabilities} \\ \hline
& \multicolumn{11}{l}{OCMT(Down-weighting only at the estimation stage)} \\ 
\hline
\textbf{20} & 35.87 & \textbf{34.94} & 35.45 &  & 31.18 & \textbf{30.17} & 
31.02 &  & 29.12 & \textbf{27.82} & 28.91 \\ 
\textbf{40} & 32.42 & \textbf{31.42} & 31.70 &  & 30.03 & \textbf{28.76} & 
29.41 &  & 35.28 & \textbf{34.78} & 36.46 \\ 
\textbf{100} & 33.31 & \textbf{32.55} & 32.83 &  & 31.66 & \textbf{30.55} & 
31.45 &  & 32.72 & \textbf{30.99} & 32.22 \\ \hline
& \multicolumn{11}{l}{OCMT(Down-weighting at the variable selection and
estimation stages)} \\ \hline
\textbf{20} & 35.87 & \textbf{35.62} & 37.29 &  & 31.18 & \textbf{31.01} & 
33.74 &  & 29.12 & \textbf{28.46} & 31.59 \\ 
\textbf{40} & 32.42 & \textbf{32.09} & 34.41 &  & 30.03 & \textbf{29.51} & 
33.94 &  & \textbf{35.28} & 35.75 & 43.09 \\ 
\textbf{100} & \textbf{33.3}1 & 33.48 & 37.29 &  & \textbf{31.66} & 32.09 & 
39.04 &  & \textbf{32.72} & 33.36 & 44.01 \\ \hline
& \multicolumn{11}{l}{Lasso} \\ \hline
\textbf{20} & \textbf{36.84} & 37.04 & 38.27 &  & 31.71 & \textbf{31.27} & 
33.02 &  & 29.80 & \textbf{28.75} & 30.35 \\ 
\textbf{40} & \textbf{33.43} & 33.81 & 36.39 &  & \textbf{30.44} & 30.47 & 
34.40 &  & 35.61 & \textbf{35.40} & 39.15 \\ 
\textbf{100} & \textbf{34.95} & 36.48 & 39.61 &  & \textbf{32.64} & 34.22 & 
37.81 &  & 33.77 & \textbf{33.67} & 37.14 \\ \hline
& \multicolumn{11}{l}{A-Lasso} \\ \hline
\textbf{20} & 38.48 & \textbf{38.26} & 39.62 &  & 32.62 & \textbf{31.93} & 
34.02 &  & 30.40 & \textbf{29.12} & 31.29 \\ 
\textbf{40} & \textbf{35.64} & 35.85 & 38.73 &  & 32.41 & \textbf{32.39} & 
36.87 &  & 37.06 & \textbf{36.64} & 40.51 \\ 
\textbf{100} & \textbf{39.82} & 41.19 & 44.04 &  & \textbf{35.67} & 37.48 & 
41.55 &  & \textbf{36.78} & 36.82 & 40.54 \\ \hline
& \multicolumn{11}{l}{Boosting} \\ \hline
\textbf{20} & \textbf{36.57} & 38.08 & 43.26 &  & \textbf{31.77} & 33.65 & 
39.96 &  & \textbf{29.29} & 30.45 & 37.32 \\ 
\textbf{40} & \textbf{33.78} & 37.36 & 45.32 &  & \textbf{29.97} & 34.37 & 
44.59 &  & \textbf{35.43} & 41.28 & 52.66 \\ 
\textbf{100} & \textbf{36.09} & 44.69 & 51.61 &  & \textbf{31.78} & 40.73 & 
49.02 &  & \textbf{33.01} & 42.10 & 49.64 \\ \hline\hline
\end{tabular}%
	\vspace{-0.3in}
	\begin{flushleft}
		\noindent 
		
		\scriptsize%
		
		\singlespacing%
		
		Notes:  The reported results are averaged across four experiments (with/without dynamics and low/high fit) for models with and without parameter instabilities. See Section \ref{sec:MC-studies} for the description of the Monte Carlo design. Full set of results is presented in the online Monte Carlo supplement.
		
		$^{\dagger }$For each of the two sets of exponential down-weighting
		(light/heavy) forecasts of the target variable are computed as the simple
		average of the forecasts obtained using the down-weighting coefficient, $%
		\lambda $.
	\end{flushleft}
\end{table}

\begin{table}
	\caption{\footnotesize One-step-ahead MSFE of Lasso, A-Lasso and boosting relative to OCMT averaged across MC experiments with and without parameter instabilities. \bigskip} \label{mc_tab_3}
	\renewcommand{\arraystretch}{1}%
	\footnotesize%
	\centering
\begin{tabular}{cccccccc}
\hline\hline
$N\backslash T$ & \textbf{100} & \textbf{150} & \textbf{200} &  & \textbf{100%
} & \textbf{150} & \textbf{200} \\ \hline
& \multicolumn{3}{l}{A. Without parameter instabilities} &  & 
\multicolumn{3}{l}{B. With parameter instabilities} \\ \cline{2-4}\cline{6-8}
& \multicolumn{7}{l}{Lasso} \\ \hline
\textbf{20} & 1.023 & 1.011 & 0.994 &  & 1.067 & 1.045 & 1.027 \\ 
\textbf{40} & 1.061 & 1.035 & 0.998 &  & 1.094 & 1.087 & 1.036 \\ 
\textbf{100} & 1.129 & 1.074 & 1.056 &  & 1.132 & 1.129 & 1.098 \\ \hline
& \multicolumn{7}{l}{A-Lasso} \\ \hline
\textbf{20} & 1.067 & 1.043 & 1.021 &  & 1.100 & 1.069 & 1.046 \\ 
\textbf{40} & 1.135 & 1.106 & 1.039 &  & 1.164 & 1.156 & 1.077 \\ 
\textbf{100} & 1.277 & 1.176 & 1.147 &  & 1.269 & 1.236 & 1.202 \\ \hline
& \multicolumn{7}{l}{Boosting} \\ \hline
\textbf{20} & 1.114 & 1.122 & 1.109 &  & 1.116 & 1.143 & 1.127 \\ 
\textbf{40} & 1.202 & 1.200 & 1.194 &  & 1.225 & 1.233 & 1.204 \\ 
\textbf{100} & 1.382 & 1.299 & 1.289 &  & 1.331 & 1.299 & 1.302 \\ 
\hline\hline
\end{tabular}%
	\vspace{-0.2in}
	\begin{flushleft}
		\noindent 
		
		\scriptsize%
		
		\singlespacing%
		
		Notes:  This table reports MSFE of Lasso, A-Lasso and boosting relative to MSFE of OCMT. Relative MSFE values are averaged across experiments and across the three options for down-weighting: no down-weighting (for all methods), light down-weighting of observations prior to Lasso, A-Lasso and boosting procedures relative to OCMT with light down-weighting only at the estimation stage, and heavy down-weighting of observations prior to Lasso, A-Lasso and boosting methods compared with OCMT with heavy down-weighting only at the estimation stage. See Section \ref{sec:MC-studies} for the description of the Monte Carlo
		design. Full set of results is presented in the online Monte Carlo supplement.
	\end{flushleft}
\end{table}

\begin{table}%
	
	\caption{{\protect\footnotesize Mean square forecast error (MSFE) and panel DM test of OCMT of one-month ahead monthly return forecasts across the 28 stocks in Dow Jones index between 1990m2 and 2017m12 (8659 forecasts)}}%
	
	\label{dow jones stocks ocmt}%
	\vspace{-0.2cm}
	
	\begin{center}
		\renewcommand{\arraystretch}{1.2}{\footnotesize 
			\begin{tabular}{cccccc}
				\hline\hline
				& \multicolumn{2}{c}{Down-weighting at$^{\dagger }$} &  &  &  \\ \cline{2-3}
				& Selection stage & Forecasting stage &  & \multicolumn{2}{c}{MSFE} \\ \hline
				(M1) & no & no &  & \multicolumn{2}{c}{61.231} \\ \hline
				\multicolumn{6}{c}{Light Down-weighting, $\lambda =\left\{
					0.975,0.98,0.985,0.99,0.995,1\right\} $} \\ \hline
				(M2) & no & yes &  & \multicolumn{2}{c}{61.641} \\ 
				(M3) & yes & yes &  & \multicolumn{2}{c}{68.131} \\ \hline
				\multicolumn{6}{c}{Heavy Down-weighting, $\lambda =\left\{
					0.95,0.96,0.97,0.98,0.99,1\right\} $} \\ \hline
				(M4) & no & yes &  & \multicolumn{2}{c}{62.163} \\ 
				(M5) & yes & yes &  & \multicolumn{2}{c}{86.073} \\ \hline
				\multicolumn{6}{c}{Pair-wise panel DM tests} \\ \cline{2-6}
				& \multicolumn{2}{c}{Light down-weighting} &  & \multicolumn{2}{c}{Heavy
					down-weighting} \\ \cline{2-3}\cline{5-6}
				& (M2) & (M3) &  & (M4) & (M5) \\ \cline{2-3}\cline{5-6}
				(M1) & -1.528 & -5.643 & (M1) & -2.459 & -11.381 \\ 
				(M2) & - & -5.606 & (M4) & - & -11.352 \\ \hline\hline
			\end{tabular}%
			\vspace{-0.2cm}}
	\end{center}
	
	\begin{flushleft}
		\scriptsize%
    	Notes: The active set consists of 40 covariates. The conditioning set only
		contains an intercept.
		
		$^{\dagger }$For each of the two sets of exponential down-weighting
		(light/heavy) forecasts of the target variable are computed as the simple
		average of the forecasts obtained using the down-weighting coefficient, $%
		\lambda $, in the \textquotedblleft light\textquotedblright\ or the
		\textquotedblleft heavy\textquotedblright\ down-weighting set under
		consideration. See footnote to Table \ref{msfe euro area ocmt}.
		
	\end{flushleft}
	
\end{table}%

\begin{table}%
	
	\caption{{\protect\footnotesize Mean square forecast error (MSFE) and panel DM test of OCMT versus Lasso, A-Lasso and boosting of one-month ahead monthly return forecasts across the 28 stocks in Dow Jones index between 1990m2 and 2017m12 (8659 forecasts)}}%
	
	\label{dow jones stocks ocmt vs lasso}%
	
	\vspace{-0.2cm}
	
	\begin{center}
		\renewcommand{\arraystretch}{1.2}{\footnotesize 
			\begin{tabular}{rccccccccc}
				\hline\hline
				& \multicolumn{9}{c}{MSFE under different down-weighting scenarios} \\ 
				\cline{2-10}
				& \multicolumn{3}{c}{No down-weighting} & \multicolumn{3}{c}{Light down-weighting$^{\dagger }$} & \multicolumn{3}{c}{Heavy down-weighting$^{\ddagger }$} \\ \cline{2-10}
				OCMT & \multicolumn{3}{c}{61.231} & \multicolumn{3}{c}{61.641} & \multicolumn{3}{c}{62.163} \\ 
				Lasso & \multicolumn{3}{c}{61.849} & \multicolumn{3}{c}{63.378} & \multicolumn{3}{c}{68.835} \\ 
				A-Lasso & \multicolumn{3}{c}{62.857} & \multicolumn{3}{c}{65.142} & \multicolumn{3}{c}{71.586} \\
				Boosting & \multicolumn{3}{c}{64.663} & \multicolumn{3}{c}{98.763} & \multicolumn{3}{c}{222.698} \\ \cline{2-10}
				& \multicolumn{9}{c}{Selected pair-wise panel DM tests} \\ \cline{2-10}
				& \multicolumn{3}{c}{No down-weighting} & \multicolumn{3}{c}{Light down-weighting} & \multicolumn{3}{c}{Heavy down-weighting} \\ \cline{2-10}
				& Lasso & A-Lasso & Boosting & Lasso & A-Lasso &Boosting& Lasso & A-Lasso &Boosting\\ \cline{2-10}
				OCMT & -0.764 & -4.063 & -7.343 & -3.318 & -5.600 & -19.053 & -7.653 & -9.722 & -30.078 \\ 
				Lasso & - & -6.192 & -6.081 & - & -8.297 & -18.519 & - & -8.947 & -29.476 \\
				A-Lasso & - & - & -3.215 & - & - & -18.084 & - & - & -29.197 \\ \hline\hline
			\end{tabular}%
			\vspace{-0.2cm}}
	\end{center}
	
	\begin{flushleft}
		\scriptsize%
		Notes: The active set consists of 40 covariates. The conditioning set
		contains only the intercept.
		
		$^{\dagger }$ Light down-weighted forecasts are computed as simple averages
		of forecasts obtained using the down-weighting coefficient, $\lambda
		=\{0.975,0.98,0.985,0.99,0.995,1\}$.
		
		$^{\ddagger }$ Heavy down-weighted forecasts are computed as simple averages
		of forecasts obtained using the down-weighting coefficient, $\lambda
		=\{0.95,0.96,0.97,0.98,0.99,1\}$. 
		
	\end{flushleft}
	
	\caption{{\protect\footnotesize Mean directional forecast accuracy (MDFA) and the average number of selected variables ($\hat{k} $) of OCMT, Lasso, A-Lasso and boosting of one-month ahead monthly return forecasts across the 28 stocks in Dow Jones index between 1990m2 and 2017m12 (8659 forecasts).}}%
	
	\label{dow jones stocks MDA}%
	
	\vspace{-0.2cm}
	
	\begin{center}
		\renewcommand{\arraystretch}{1.2}{\footnotesize 
			\begin{tabular}{lccc}
				\hline\hline
				& Down-weighting & MDFA & $\hat{k}$ \\ \hline
				OCMT & No & 56.057 & 0.072 \\ 
				& Light$^{\dagger }$ & 55.330 & 0.072 \\ 
				& Heavy$^{\ddagger }$ & 54.302 & 0.072 \\ \hline
				Lasso & No & 55.769 & 1.497 \\ 
				& Light & 54.348 & 2.120 \\ 
				& Heavy & 53.447 & 3.758 \\ \hline
				A-Lasso & No & 55.122 & 1.187 \\ 
				& Light & 53.586 & 1.610 \\ 
				& Heavy & 53.055 & 2.819 \\ \hline
				Boosting & No & 54.221 & 1.723 \\ 
				& Light & 50.872 & 8.108 \\ 
				& Heavy & 49.244 & 14.565 \\ \hline\hline
			\end{tabular}%
			\vspace{-0.2cm}}
	\end{center}
	
	\begin{flushleft}
		
		\scriptsize%
		
		Notes: The active set consists of 40 variables. The conditioning set
		contains an intercept.
		
		$^{\dagger }$ Light down-weighted forecasts are computed as simple averages
		of forecasts obtained using the down-weighting coefficient, $\lambda
		=\{0.975,0.98,0.985,0.99,0.995,1\}$.
		
		$^{\ddagger }$ Heavy down-weighted forecasts are computed as simple averages
		of forecasts obtained using the down-weighting coefficient, $\lambda
		=\{0.95,0.96,0.97,0.98,0.99,1\}$. 
		
	\end{flushleft}
	
\end{table}%

\begin{table}%
	
	\caption{{\protect\footnotesize The number of stocks out of the 28 stocks in Dow Jones index where OCMT outperforms/underperforms Lasso, A-Lasso and boosting in terms of mean square forecast error (MSFE), panel DM test and mean directional forecast accuracy (MDFA) between 1990m2 and 2017m12 (8659 forecasts).}}%
	
	\label{dow jone individual stocks summary}%
	
	\vspace{-0.2cm}
	
	\begin{center}
		\renewcommand{\arraystretch}{1}{\footnotesize 
			\begin{tabular}{lccccc}
				\hline\hline
				& \multicolumn{5}{c}{\textbf{MSFE}} \\ \cline{2-6}
				& Down- & OCMT & OCMT significantly & OCMT & OCMT significantly \\ 
				& weighting & outperforms & outperforms & underperforms & underperforms \\ 
				\hline
				Lasso & No & 22 & 3 & 6 & 2 \\ 
				& Light$^{\dagger }$ & 27 & 8 & 1 & 0 \\ 
				& Heavy$^{\ddagger }$ & 27 & 17 & 1 & 0 \\ \hline
				A-Lasso & No & 25 & 6 & 3 & 0 \\ 
				& Light & 28 & 13 & 0 & 0 \\ 
				& Heavy & 28 & 24 & 0 & 0 \\ \hline
				Boosting & No & 28 & 13 & 0 & 0 \\ 
				& Light & 28 & 28 & 0 & 0 \\ 
				& Heavy & 28 & 28 & 0 & 0 \\ \hline
				& \multicolumn{3}{c}{\textbf{MDFA}} & \multicolumn{1}{c}{} & 
				\multicolumn{1}{c}{} \\ \cline{2-4}
				& Down- & OCMT & OCMT &  & \multicolumn{1}{c}{} \\ 
				& weighting & outperforms & underperforms &  & \multicolumn{1}{c}{} \\ 
				\cline{1-4}
				Lasso & No & 11 & 10 &  & \multicolumn{1}{c}{} \\ 
				& Light & 20 & 8 &  & \multicolumn{1}{c}{} \\ 
				& Heavy & 18 & 9 &  & \multicolumn{1}{c}{} \\ \cline{1-4}
				A-Lasso & No & 16 & 7 &  & \multicolumn{1}{c}{} \\ 
				& Light & 21 & 5 &  & \multicolumn{1}{c}{} \\ 
				& Heavy & 21 & 6 &  & \multicolumn{1}{c}{} \\ \cline{1-4}
				Boosting & No & 19 & 6 &  & \multicolumn{1}{c}{} \\ 
				& Light & 21 & 3 &  & \multicolumn{1}{c}{} \\ 
				& Heavy & 21 & 3 &  & \multicolumn{1}{c}{} \\ \hline\hline
			\end{tabular}%
			\vspace{-0.2cm}}
	\end{center}
	
	\begin{flushleft}
		\scriptsize%
		Notes: The active set consists of 40 variables. The conditioning set only
		contains an intercept.
		
		$^{\dagger }$ Light down-weighted forecasts are computed as simple averages
		of forecasts obtained using the down-weighting coefficient, $\lambda
		=\{0.975,0.98,0.985,0.99,0.995,1\}$.
		
		$^{\ddagger }$ Heavy down-weighted forecasts are computed as simple averages
		of forecasts obtained using the down-weighting coefficient, $\lambda
		=\{0.95,0.96,0.97,0.98,0.99,1\}$. 
		
	\end{flushleft}
	
	\caption{{\protect\footnotesize Mean square forecast error (MSFE) and panel
			DM test of OCMT of one-year ahead output growth forecasts across 33
			countries over the period 1997Q2-2016Q4 (2607 forecasts)}}
	\label{countries gdp ocmt 1}
	\begin{center}
		\renewcommand{\arraystretch}{1}{\footnotesize \ 
			\begin{tabular}{cccccc}
				\hline\hline
				& \multicolumn{2}{c}{Down-weighting at$^{\dagger }$} & \multicolumn{3}{c}{
					MSFE ($\times 10^{4}$)} \\ \cline{4-6}
				& Selection stage & Forecasting stage & All & Advanced & Emerging \\ \hline
				(M$1$) & no & no & 11.246 & 7.277 & 17.354 \\ \hline
				\multicolumn{6}{c}{Light down-weighting, $\lambda =\left\{
					0.975,0.98,0.985,0.99,0.995,1\right\} $} \\ \hline
				(M$2$) & no & yes & 10.836 & 6.913 & 16.871 \\ 
				(M$3$) & yes & yes & 10.919 & 6.787 & 17.275 \\ \hline
				\multicolumn{6}{c}{Heavy down-weighting, $\lambda =\left\{
					0.95,0.96,0.97,0.98,0.99,1\right\} $} \\ \hline
				(M$4$) & no & yes & 11.064 & 7.187 & 17.028 \\ 
				(M$5$) & yes & yes & 11.314 & 6.906 & 18.094 \\ \hline
				\multicolumn{6}{c}{Pair-wise panel DM tests (all countries)} \\ \cline{2-6}
				& \multicolumn{2}{c}{Light down-weighting} &  & \multicolumn{2}{c}{Heavy
					down-weighting} \\ \cline{2-3}\cline{5-6}
				& (M$2$) & (M$3$) &  & (M$4$) & (M$5$) \\ \cline{2-3}\cline{5-6}
				(M$1$) & 2.394 & 1.662 & (M$1$) & 0.668 & -0.204 \\ 
				(M$2$) & - & -0.780 & (M$4$) & - & -1.320 \\ \hline\hline
			\end{tabular}
			\vspace{-0.2cm}}
	\end{center}
	\par
	\begin{flushleft}
		\scriptsize Notes: There are up to 15 macro and financial variables in
			the active set. 
		\par
		{\scriptsize $^{\dagger }$For each of the two sets of exponential
			down-weighting (light/heavy) forecasts of the target variable are computed
			as the simple average of the forecasts obtained using the down-weighting
			coefficient, $\lambda $, in the ``light'' or the ``heavy'' down-weighting set
			under consideration. }
	\end{flushleft}
\end{table}

\begin{table}
	\caption{{\protect\footnotesize Mean square forecast error (MSFE) and panel
			DM test of OCMT of two-year ahead output growth forecasts across 33
			countries over the period 1997Q2-2016Q4 (2343 forecasts)}}
	\label{countries gdp ocmt 2}
	\begin{center}
		\renewcommand{\arraystretch}{1.2}{\footnotesize \ 
			\begin{tabular}{cccccc}
				\hline\hline
				& \multicolumn{2}{c}{Down-weighting at$^{\dagger }$} & \multicolumn{3}{c}{
					MSFE ($\times 10^{4}$)} \\ \cline{4-6}
				& Selection stage & Forecasting stage & All & Advanced & Emerging \\ \hline
				(M$1$) & no & no & 9.921 & 7.355 & 13.867 \\ \hline
				\multicolumn{6}{c}{Light down-weighting, $\lambda =\left\{
					0.975,0.98,0.985,0.99,0.995,1\right\} $} \\ \hline
				(M$2$) & no & yes & 9.487 & 6.874 & 13.505 \\ 
				(M$3$) & yes & yes & 9.549 & 6.848 & 13.704 \\ \hline
				\multicolumn{6}{c}{Heavy down-weighting, $\lambda =\left\{
					0.95,0.96,0.97,0.98,0.99,1\right\} $} \\ \hline
				(M$4$) & no & yes & 9.734 & 7.027 & 13.898 \\ 
				(M$5$) & yes & yes & 10.389 & 7.277 & 15.177 \\ \hline
				\multicolumn{6}{c}{Pair-wise panel DM test (all countries)} \\ \cline{2-6}
				& \multicolumn{2}{c}{Light down-weighting} &  & \multicolumn{2}{c}{Heavy
					down-weighting} \\ \cline{2-3}\cline{5-6}
				& (M$2$) & (M$3$) &  & (M$4$) & (M$5$) \\ \cline{2-3}\cline{5-6}
				(M$1$) & 3.667 & 2.827 & (M$1$) & 0.943 & -1.664 \\ 
				(M$2$) & - & -1.009 & (M$4$) & - & -3.498 \\ \hline\hline
			\end{tabular}
			\vspace{-0.2cm}}
	\end{center}
	\par
	\begin{flushleft}
		\scriptsize \ Notes: There are up to 15 macro and financial variables in
		the active set. 
		\par
		{\scriptsize $^{\dagger }$For each of the two sets of exponential
			down-weighting (light/heavy) forecasts of the target variable are computed
			as the simple average of the forecasts obtained using the down-weighting
			coefficient, $\lambda $, in the "light" or the "heavy" down-weighting set
			under consideration.. }{\normalsize \ }
	\end{flushleft}
	
	\caption{{\protect\footnotesize Mean square forecast error (MSFE) and panel
			DM test of OCMT versus Lasso, A-Lasso and boosting for one-year ahead output
			growth forecasts across 33 countries over the period1997Q2-2016Q4 (2607
			forecasts)}}
	\label{countries gdp ocmt vs lasso 1}\vspace{-0.3cm}
	\par
	\begin{center}
		\renewcommand{\arraystretch}{1.2}{\footnotesize \ 
			\begin{tabular}{rccccccccc}
				\hline\hline
				& \multicolumn{9}{c}{MSFE under different down-weighting scenarios} \\ 
				\cline{2-10}
				& \multicolumn{3}{c}{No down-weighting} & \multicolumn{3}{c}{Light
					down-weighting$^{\dag }$} & \multicolumn{3}{c}{Heavy down-weighting$^{\ddag }
					$} \\ \cline{2-10}
				& All & Adv.* & Emer.** & All & Adv. & Emes & All & Adv. & Emes \\ 
				\cline{2-10}
				OCMT & \multicolumn{1}{r}{11.246} & \multicolumn{1}{r}{7.277} & 
				\multicolumn{1}{r}{17.354} & \multicolumn{1}{r}{10.836} & \multicolumn{1}{r}{
					6.913} & \multicolumn{1}{r}{16.871} & \multicolumn{1}{r}{11.064} & 
				\multicolumn{1}{r}{7.187} & \multicolumn{1}{r}{17.028} \\ 
				Lasso & \multicolumn{1}{r}{11.205} & \multicolumn{1}{r}{6.975} & 
				\multicolumn{1}{r}{17.714} & \multicolumn{1}{r}{10.729} & \multicolumn{1}{r}{
					6.427} & \multicolumn{1}{r}{17.347} & \multicolumn{1}{r}{11.749} & 
				\multicolumn{1}{r}{7.186} & \multicolumn{1}{r}{18.769} \\ 
				A-Lasso & \multicolumn{1}{r}{11.579} & \multicolumn{1}{r}{7.128} & 
				\multicolumn{1}{r}{18.426} & \multicolumn{1}{r}{11.153} & \multicolumn{1}{r}{
					6.548} & \multicolumn{1}{r}{18.236} & \multicolumn{1}{r}{12.254} & 
				\multicolumn{1}{r}{7.482} & \multicolumn{1}{r}{19.595} \\ 
				Boosting & \multicolumn{1}{r}{11.353} & \multicolumn{1}{r}{6.988} & 
				\multicolumn{1}{r}{18.068} & \multicolumn{1}{r}{10.886} & \multicolumn{1}{r}{
					6.401} & \multicolumn{1}{r}{17.787} & \multicolumn{1}{r}{11.868} & 
				\multicolumn{1}{r}{7.060} & \multicolumn{1}{r}{19.264} \\ \hline
				& \multicolumn{9}{c}{Pair-wise Panel DM tests (All countries)} \\ 
				\cline{2-10}
				& \multicolumn{3}{c}{No down-weighting} & \multicolumn{3}{c}{Light
					down-weighting} & \multicolumn{3}{c}{Heavy down-weighting} \\ 
				\cline{2-4}\cline{5-7}\cline{8-10}
				& Lasso & A-Lasso &Boosting& Lasso & A-Lasso &Boosting& Lasso & A-Lasso
				&Boosting\\ \cline{2-4}\cline{5-7}\cline{8-10}
				OCMT & \multicolumn{1}{r}{0.220} & \multicolumn{1}{r}{-1.079} & 
				\multicolumn{1}{r}{-0.445} & \multicolumn{1}{r}{0.486} & \multicolumn{1}{r}{
					-1.007} & \multicolumn{1}{r}{-0.195} & \multicolumn{1}{r}{-1.799} & 
				\multicolumn{1}{r}{-2.441} & \multicolumn{1}{r}{-1.920} \\ 
				Lasso & - & \multicolumn{1}{r}{-2.625} & \multicolumn{1}{r}{-1.322} & - & 
				\multicolumn{1}{r}{-3.626} & \multicolumn{1}{r}{-1.790} & - & 
				\multicolumn{1}{r}{-3.157} & \multicolumn{1}{r}{-0.894} \\ 
				A-Lasso & - & - & \multicolumn{1}{r}{1.837} & - & - & \multicolumn{1}{r}{
					2.714} & - & - & \multicolumn{1}{r}{2.271} \\ \hline\hline
			\end{tabular}
			\vspace{-0.2cm}}
	\end{center}
	\par
	\begin{flushleft}
		\scriptsize \ Notes: There are up to 15 macro and financial covariates in
		the active set. 
		\par
		{\scriptsize $^{\dagger }$ Light down-weighted forecasts are computed as
			simple averages of forecasts obtained using the down-weighting coefficient, $%
			\lambda =\{0.975,0.98,0.985,0.99,0.995,1\}$. }
		\par
		{\scriptsize $^{\ddagger }$ Heavy down-weighted forecasts are computed as
			simple averages of forecasts obtained using the down-weighting coefficient, $%
			\lambda =\{0.95,0.96,0.97,0.98,0.99,1\}$. }
		\par
		{\scriptsize $^{\ast }$ Adv. stands for advanced economies. }
		\par
		{\scriptsize $^{\ast \ast }$ Emer. stands for emerging economies. }
		\par
		{\normalsize \ }
	\end{flushleft}
	
\end{table}

\begin{table}
	\caption{{\protect\footnotesize Mean square forecast error (MSFE) and panel
			DM test of OCMT versus Lasso, A-Lasso and boosting of two-year ahead output
			growth forecasts across 33 countries over the period1997Q2-2016Q4 (2343
			forecasts)}}
	\label{countries gdp ocmt vs lasso 2}\vspace{-0.3cm}
	\par
	\begin{center}
		\renewcommand{\arraystretch}{1}{\footnotesize \ 
			
			\begin{tabular}{rccccccccc}
				\hline\hline
				& \multicolumn{9}{c}{MSFE under different down-weighting scenarios} \\ 
				\cline{2-10}
				& \multicolumn{3}{c}{No down-weighting} & \multicolumn{3}{c}{Light
					down-weighting$^{\dag }$} & \multicolumn{3}{c}{Heavy down-weighting$^{\ddag }
					$} \\ \cline{2-10}
				& All & Adv.* & Emer.** & All & Adv. & Emes & All & Adv. & Emes \\ 
				\cline{2-10}
				OCMT & \multicolumn{1}{r}{9.921} & \multicolumn{1}{r}{7.355} & 
				\multicolumn{1}{r}{13.867} & \multicolumn{1}{r}{9.487} & \multicolumn{1}{r}{
					6.874} & \multicolumn{1}{r}{13.505} & \multicolumn{1}{r}{9.734} & 
				\multicolumn{1}{r}{7.027} & \multicolumn{1}{r}{13.898} \\ 
				Lasso & \multicolumn{1}{r}{10.151} & \multicolumn{1}{r}{7.583} & 
				\multicolumn{1}{r}{14.103} & \multicolumn{1}{r}{9.662} & \multicolumn{1}{r}{
					7.099} & \multicolumn{1}{r}{13.605} & \multicolumn{1}{r}{10.202} & 
				\multicolumn{1}{r}{7.428} & \multicolumn{1}{r}{14.469} \\ 
				A-Lasso & \multicolumn{1}{r}{10.580} & \multicolumn{1}{r}{7.899} & 
				\multicolumn{1}{r}{14.705} & \multicolumn{1}{r}{10.090} & \multicolumn{1}{r}{
					7.493} & \multicolumn{1}{r}{14.087} & \multicolumn{1}{r}{11.008} & 
				\multicolumn{1}{r}{8.195} & \multicolumn{1}{r}{15.336} \\ 
				Boosting & \multicolumn{1}{r}{10.182} & \multicolumn{1}{r}{7.600} & 
				\multicolumn{1}{r}{14.154} & \multicolumn{1}{r}{9.818} & \multicolumn{1}{r}{
					7.231} & \multicolumn{1}{r}{13.796} & \multicolumn{1}{r}{11.040} & 
				\multicolumn{1}{r}{8.213} & \multicolumn{1}{r}{15.391} \\ \hline
				& \multicolumn{9}{c}{Pair-wise Panel DM tests (All countries)} \\ 
				\cline{2-10}
				& \multicolumn{3}{c}{No down-weighting} & \multicolumn{3}{c}{Light
					down-weighting} & \multicolumn{3}{c}{Heavy down-weighting} \\ 
				\cline{2-4}\cline{5-7}\cline{8-10}
				& Lasso & A-Lasso & Boosting & Lasso & A-Lasso & Boosting & Lasso & A-Lasso
				& Boosting \\ \cline{2-4}\cline{5-7}\cline{8-10}
				OCMT & \multicolumn{1}{r}{-2.684} & \multicolumn{1}{r}{-4.200} & 
				\multicolumn{1}{r}{-2.681} & \multicolumn{1}{r}{-2.137} & \multicolumn{1}{r}{
					-4.015} & \multicolumn{1}{r}{-2.933} & \multicolumn{1}{r}{-3.606} & 
				\multicolumn{1}{r}{-4.789} & \multicolumn{1}{r}{-4.923} \\ 
				Lasso & - & \multicolumn{1}{r}{-5.000} & \multicolumn{1}{r}{-0.430} & - & 
				\multicolumn{1}{r}{-4.950} & \multicolumn{1}{r}{-2.317} & - & 
				\multicolumn{1}{r}{-4.969} & \multicolumn{1}{r}{-4.588} \\ 
				A-Lasso & - & - & \multicolumn{1}{r}{3.778} & - & - & \multicolumn{1}{r}{
					3.661} & - & - & \multicolumn{1}{r}{-0.252} \\ \hline\hline
			\end{tabular}

			\vspace{-0.2cm}}
	\end{center}
	\par
	\begin{flushleft}
		\scriptsize \ Notes: There are up to 15 macro and financial covariates in
		the active set. 
		\par
		{\scriptsize $^{\dagger }$ Light down-weighted forecasts are computed as
			simple averages of forecasts obtained using the down-weighting coefficient, $%
			\lambda =\{0.975,0.98,0.985,0.99,0.995,1\}$. }
		\par
		{\scriptsize $^{\ddagger }$ Heavy down-weighted forecasts are computed as
			simple averages of forecasts obtained using the down-weighting coefficient, $%
			\lambda =\{0.95,0.96,0.97,0.98,0.99,1\}$. }
		\par
		{\scriptsize $^{\ast }$ Adv. stands for advanced economies.}
		\par
		{ $^{\ast \ast }$ Emer. stands for emerging economies.}
	\end{flushleft}
	
	\caption{{\protect\footnotesize The number of countries out of the 33
			countries where OCMT outperforms/underperforms Lasso, A-Lasso and boosting in
			terms of mean square forecast error (MSFE) and panel DM test over the period
			1997Q2 -2016Q4}}
	\label{countries gdp summary}
	\begin{center}
		\setlength{\tabcolsep}{2pt}\renewcommand{\arraystretch}{1}{\footnotesize 
			\begin{tabular}{cccccc}
				\hline\hline
				&  &  & OCMT &  & OCMT \\ 
				& Down- & OCMT & significantly & OCMT & significantly \\ 
				& weighting & outperforms & outperforms & underperforms & underperforms \\ 
				\hline
				& \multicolumn{5}{c}{One-years-ahead horizon ($h=4$ quarters)} \\ \hline
				Lasso & No & \multicolumn{1}{r}{13} & \multicolumn{1}{r}{0} & 
				\multicolumn{1}{r}{20} & \multicolumn{1}{r}{3} \\ 
				& Light$^{\dag }$ & \multicolumn{1}{r}{12} & \multicolumn{1}{r}{1} & 
				\multicolumn{1}{r}{21} & \multicolumn{1}{r}{3} \\ 
				& Heavy$^{\ddag }$ & \multicolumn{1}{r}{17} & \multicolumn{1}{r}{1} & 
				\multicolumn{1}{r}{16} & \multicolumn{1}{r}{3} \\ \hline
				A-Lasso & No & \multicolumn{1}{r}{16} & \multicolumn{1}{r}{1} & 
				\multicolumn{1}{r}{17} & \multicolumn{1}{r}{2} \\ 
				& Light & \multicolumn{1}{r}{14} & \multicolumn{1}{r}{2} & 
				\multicolumn{1}{r}{19} & \multicolumn{1}{r}{2} \\ 
				& Heavy & \multicolumn{1}{r}{19} & \multicolumn{1}{r}{1} & 
				\multicolumn{1}{r}{14} & \multicolumn{1}{r}{0} \\ \hline
				Boosting & No & \multicolumn{1}{r}{11} & \multicolumn{1}{r}{1} & 
				\multicolumn{1}{r}{22} & \multicolumn{1}{r}{3} \\ 
				& Light & \multicolumn{1}{r}{11} & \multicolumn{1}{r}{1} & 
				\multicolumn{1}{r}{22} & \multicolumn{1}{r}{3} \\ 
				& Heavy & \multicolumn{1}{r}{17} & \multicolumn{1}{r}{1} & 
				\multicolumn{1}{r}{16} & \multicolumn{1}{r}{1} \\ \hline
				& \multicolumn{5}{c}{Two-years-ahead horizon ($h=8$ quarters)} \\ \hline
				Lasso & No & \multicolumn{1}{r}{24} & \multicolumn{1}{r}{1} & 
				\multicolumn{1}{r}{9} & \multicolumn{1}{r}{0} \\ 
				& Light & \multicolumn{1}{r}{25} & \multicolumn{1}{r}{1} & 
				\multicolumn{1}{r}{8} & \multicolumn{1}{r}{1} \\ 
				& Heavy & \multicolumn{1}{r}{25} & \multicolumn{1}{r}{1} & 
				\multicolumn{1}{r}{8} & \multicolumn{1}{r}{0} \\ \hline
				A-Lasso & No & \multicolumn{1}{r}{25} & \multicolumn{1}{r}{2} & 
				\multicolumn{1}{r}{8} & \multicolumn{1}{r}{0} \\ 
				& Light & \multicolumn{1}{r}{28} & \multicolumn{1}{r}{3} & 
				\multicolumn{1}{r}{5} & \multicolumn{1}{r}{1} \\ 
				& Heavy & \multicolumn{1}{r}{30} & \multicolumn{1}{r}{3} & 
				\multicolumn{1}{r}{3} & \multicolumn{1}{r}{0} \\ \hline
				Boosting & No & \multicolumn{1}{r}{23} & \multicolumn{1}{r}{2} & 
				\multicolumn{1}{r}{10} & \multicolumn{1}{r}{0} \\ 
				& Light & \multicolumn{1}{r}{25} & \multicolumn{1}{r}{1} & 
				\multicolumn{1}{r}{8} & \multicolumn{1}{r}{0} \\ 
				& Heavy & \multicolumn{1}{r}{32} & \multicolumn{1}{r}{4} & 
				\multicolumn{1}{r}{1} & \multicolumn{1}{r}{0} \\ \hline\hline
			\end{tabular}
			\vspace{-0.2cm}}
	\end{center}
	\par
	\begin{flushleft}
		\scriptsize Notes: There are up to 15 macro and financial covariates in
		the active set. 
		\par
		{\scriptsize $^{\dagger }$Light down-weighted forecasts are computed as
			simple averages of forecasts obtained using the down-weighting coefficient, $%
			\lambda =\{0.975,0.98,0.985,0.99,0.995,1\}$. }
		\par
		{\scriptsize $^{\ddagger }$ Heavy down-weighted forecasts are computed as
			simple averages of forecasts obtained using the down-weighting coefficient, $%
			\lambda =\{0.95,0.96,0.97,0.98,0.99,1\}$. }
	\end{flushleft}
\end{table}

\begin{table}
	\caption{{\protect\footnotesize Mean directional forecast accuracy (MDFA)
			and PT test of OCMT, Lasso, A-Lasso and boosting for one-year ahead output
			growth forecasts over the period 1997Q2-2016Q4 (2607 forecasts)}}
	\label{countries gdp MDA & PT test 1}
	\begin{center}
		\renewcommand{\arraystretch}{1.2}{\footnotesize 
			\begin{tabular}{ccccccccc}
				\hline\hline
				& Down- & \multicolumn{3}{c}{MDFA} &  & \multicolumn{3}{c}{PT tests} \\ 
				\cline{3-5}\cline{7-9}
				& weighting & All & Advanced & Emerging &  & All & Advanced & Emerging \\ 
				\hline
				OCMT & No & 87.6 & 87.4 & 88.0 &  & 8.12 & 7.40 & 3.48 \\ 
				& Light$^{\dag }$ & 87.4 & 87.1 & 87.8 &  & 7.36 & 6.95 & 2.53 \\ 
				& Heavy$^{\ddag }$ & 86.8 & 86.3 & 87.5 &  & 6.25 & 5.93 & 1.95 \\ \hline
				Lasso & No & 86.2 & 86.7 & 85.3 &  & 9.64 & 9.15 & 3.80 \\ 
				& Light & 87.1 & 87.1 & 87.1 &  & 8.12 & 8.22 & 2.26 \\ 
				& Heavy & 86.0 & 85.8 & 86.4 &  & 6.24 & 6.43 & 1.40 \\ \hline
				A-Lasso & No & 87.3 & 87.3 & 87.2 &  & 10.80 & 9.91 & 4.75 \\ 
				& Light & 86.5 & 86.6 & 86.4 &  & 8.25 & 8.36 & 2.48 \\ 
				& Heavy & 85.5 & 85.3 & 85.7 &  & 6.84 & 6.92 & 1.88 \\ \hline
				Boosting & No & 86.7 & 87.1 & 86.0 &  & 8.17 & 8.39 & 3.06 \\ 
				& Light & 86.6 & 86.6 & 86.6 &  & 7.43 & 5.48 & 7.30 \\ 
				& Heavy & 85.4 & 85.6 & 85.1 &  & 5.66 & 6.06 & 1.44 \\ \hline\hline
			\end{tabular}

			\vspace{-0.2cm}}
	\end{center}
	\par
	\begin{flushleft}
		\scriptsize Notes: There are up to 15 macro and financial variables in
		the active set. 
		\par
		{\scriptsize $^{\dagger }$ Light down-weighted forecasts are computed as
			simple averages of forecasts obtained using the down-weighting coefficient, $%
			\lambda =\{0.975,0.98,0.985,0.99,0.995,1\}$. }
		\par
		{\scriptsize $^{\ddagger }$ Heavy down-weighted forecasts are computed as
			simple averages of forecasts obtained using the down-weighting coefficient, $%
			\lambda =\{0.95,0.96,0.97,0.98,0.99,1\}$.}
	\end{flushleft}
	
	\caption{{\protect\footnotesize Mean directional forecast accuracy (MDFA)
			and PT test of OCMT, Lasso, A-Lasso and boosting for two-year ahead output
			growth forecasts over the period 1997Q2-2016Q4 (2343 forecasts)}}
	\label{countries gdp MDA & PT test 2}
	\begin{center}
		\renewcommand{\arraystretch}{1.2}{\footnotesize 
			\begin{tabular}{ccccccccc}
				\hline\hline
				& Down- & \multicolumn{3}{c}{MDFA} &  & \multicolumn{3}{c}{PT tests} \\ 
				\cline{3-5}\cline{7-9}
				& weighting & All & Advanced & Emerging &  & All & Advanced & Emerging \\ 
				\hline
				OCMT & No & 88.0 & 86.7 & 89.9 &  & 0.52 & 0.00 & 0.47 \\ 
				& Light$^{\dag }$ & 87.7 & 86.6 & 89.3 &  & 1.11 & 0.39 & 0.94 \\ 
				& Heavy$^{\ddag }$ & 87.0 & 85.8 & 88.8 &  & 0.50 & 0.89 & 0.34 \\ \hline
				Lasso & No & 87.2 & 86.2 & 88.7 &  & 0.77 & 0.60 & 0.66 \\ 
				& Light & 87.5 & 86.3 & 89.4 &  & 0.07 & 0.79 & 0.88 \\ 
				& Heavy & 86.8 & 85.5 & 88.8 &  & 1.54 & 1.87 & 0.34 \\ \hline
				A-Lasso & No & 87.0 & 85.6 & 89.2 &  & 0.33 & 0.13 & 1.00 \\ 
				& Light & 87.1 & 85.9 & 88.9 &  & 1.03 & 1.82 & 1.10 \\ 
				& Heavy & 86.2 & 84.8 & 88.4 &  & 1.53 & 1.92 & 0.62 \\ \hline
				Boosting & No & 87.3 & 85.8 & 89.7 &  & 0.63 & 0.19 & 1.44 \\ 
				& Light & 87.6 & 86.7 & 89.1 &  & 2.23 & 3.77 & 1.05 \\ 
				& Heavy & 86.2 & 84.9 & 88.1 &  & 1.47 & 2.07 & 0.79 \\ \hline\hline
		\end{tabular}}
	\end{center}
	\par
	\begin{flushleft}
		\scriptsize Notes: There are up to 15 macro and financial variables in
		the active set.
		\par
		{\scriptsize $^{\dagger }$Light down-weighted forecasts are computed as
			simple averages of forecasts obtained using the down-weighting coefficient, $%
			\lambda =\{0.975,0.98,0.985,0.99,0.995,1\}$. }
		\par
		{\scriptsize $^{\ddagger }$ Heavy down-weighted forecasts are computed as
			simple averages of forecasts obtained using the down-weighting coefficient, $%
			\lambda =\{0.95,0.96,0.97,0.98,0.99,1\}$. }
	\end{flushleft}
\end{table}

\newpage

{\ {\small \setstretch{1.15} 
\bibliographystyle{chicago}
\bibliography{reference_forecasting}
}}

\newpage 

\setcounter{section}{0} \renewcommand{\thesection}{S-\arabic{section}}

\setcounter{table}{0} \renewcommand{\thetable}{S3.\arabic{table}}

\setcounter{page}{1} \renewcommand{\thepage}{S.\arabic{page}}

\setcounter{equation}{0} \renewcommand{\theequation}{S.\arabic{equation}}

\setcounter{proposition}{0} \renewcommand{\theproposition}{S.%
\arabic{proposition}}

\setcounter{corollary}{0} \renewcommand{\thecorollary}{S.\arabic{corollary}}

\setcounter{proofcorollary}{0} \renewcommand{\theproofcorollary}{S.%
\arabic{proofcorollary}} \setcounter{proofproposition}{0} %
\renewcommand{\theproofproposition}{S.\arabic{proofproposition}}

\quad \vspace{0.05in}

\begin{center}
\textbf{\ {\large Online Theory Supplement to} }\\[0pt]

\textbf{{\large {``Variable Selection in High Dimensional Linear Regressions
with Parameter Instability''}}} \\[0pt]

Alexander Chudik

Federal Reserve Bank of Dallas\bigskip

M. Hashem Pesaran

University of Southern California, USA and Trinity College, Cambridge,
UK\bigskip

Mahrad Sharifvaghefi

University of Pittsburgh\\[0pt]
\bigskip \bigskip

\today\bigskip
\end{center}

\noindent This online theory supplement has three sections. Section \ref%
{appendix A} provides the proofs of Theorems 1 to 3, and additional
propositions and corollaries. Section \ref{main_lemmas} establishes the main
lemmas needed for the proof of the theorems in Section \ref{appendix A}.
Section \ref{Complementary_lemmas} contains the complementary lemmas needed
for the proofs of the main lemmas in the previous section.

\textbf{Notations:} Generic finite positive constants are denoted by $C_{i}$
for $i=1,2,\cdots $ and $c $. They can take different values in different
instances. $\lVert \mathbf{A}\rVert _{2}$, $\lVert \mathbf{A}\rVert _{F}$, $%
\lVert \mathbf{A}\rVert _{\infty}$ and $\lVert \mathbf{A}\rVert _{1}$ denote
the spectral, Frobenius, row, and column norms of matrix $\mathbf{A}$,
respectively. $\lambda_{i}(\mathbf{A}) $ denotes the $i^{th} $ eigenvalue of
a square matrix $\mathbf{A} $. $\text{tr}(\mathbf{A})$ and $\text{det}(%
\mathbf{A})$ are the trace and determinant of a square matrix $\mathbf{A}$,
respectively. $\left\| \mathbf{x} \right\| $ denotes the $\ell_{2} $ norm of
vector $\mathbf{x} $. If $\{f_{n}\}_{n=1}^{\infty }$ is any real sequence
and $\{g_{n}\}_{n=1}^{\infty }$ is a sequence of positive real numbers, then 
$f_{n}=O(g_{n})$, if there exists a positive constant $C_{0}$ and $n_0 $
such that $\lvert f_{n}\rvert /g_{n}\leq C_{0}$ for all $n > n_0$. $%
f_{n}=o(g_{n})$ if $f_{n}/g_{n}\rightarrow 0$ as $n\rightarrow \infty $. If $%
\{f_{n}\}_{n=1}^{\infty }$ and $\{g_{n}\}_{n=1}^{\infty }$ are both positive
sequences of real numbers, then $f_{n}=\ominus (g_{n})$ if there exist $%
n_{0}\geq 1$ and positive constants $C_{0}$ and $C_{1}$, such that $%
\inf_{n\geq n_{0}}\left( f_{n}/g_{n}\right) \geq C_{0}$ and $\sup_{n\geq
n_{0}}\left( f_{n}/g_{n}\right) \leq C_{1}$. If $\{f_{n}\}_{n=1}^{\infty }$
is a sequence of random variables and $\{g_{n}\}_{n=1}^{\infty }$ is a
sequence of positive real numbers, then $f_{n}=O_p(g_{n})$, if for any $%
\varepsilon > 0 $, there exists a positive constant $B_{\varepsilon}$ and $%
n_{\varepsilon} $ such that $\Pr \left(\lvert f_{n}\rvert > g_{n}
B_{\varepsilon} \right) < \varepsilon $ for all $n > n_{\varepsilon}$.


\section{Proof of the Theorems}

\label{appendix A}

This section provides the proofs of Theorems 1 to 3. The proofs are based on
lemmas presented in Section \ref{main_lemmas}. Among these, Lemmas \ref%
{t_test_bound} and \ref{reg coef} are key. For each covariate $i=1,2,\cdots
,N$, Lemmas \ref{t_test_bound} establishes exponential probability
inequalities for the t-ratio multiple tests conditional on the average net
effect, $\bar{\theta}_{i,T}$, being either of the order $\ominus
(T^{-\varepsilon _{i}})$ for some $\varepsilon _{i}>1/2$, or of the order $%
\ominus (T^{-\vartheta _{i}})$, for some $0\leq \vartheta _{i}<1/2$. For DGP
given by (\ref{dgp y_t}), Lemma \ref{reg coef} provides asymptotic
properties of LS estimator of coefficients and SSR of a regression model
that includes all the signals and pseudo-signals. This lemma establishes
that the coefficients of pseudo-signals estimated by LS converges to zero so
long as $k_{T}^{\ast }=\ominus (T^{d})$ grows at a slow rate relative to $T$%
, i.e. $0\leq d<1/2$. This lemma also shows that the SSR of the regression
model converges to that of the oracle model, which includes only the signals.

\begin{flushleft}
\textbf{Additional notations and definitions: }Throughout this section we
consider the following events: 
\begin{equation}
\textstyle\mathcal{A}_{0}=\mathcal{H}\cap \mathcal{G},\ \text{where}\ 
\mathcal{H}=\left\{ \sum_{i=1}^{k}\hat{\mathcal{J}}_{i}=k\right\} \ \text{%
and }\ \mathcal{G}=\left\{ \sum_{i=k+k_{T}^{\ast }+1}^{N}\hat{\mathcal{J}}%
_{i}=0\right\} ,  \label{approx_event}
\end{equation}%
where $\{\hat{\mathcal{J}}_{i}\text{ for }i=1,2,\cdots ,N\}$ are the
selection indicators defined by (\ref{selection indicator}). $\mathcal{A}%
_{0} $ is the event of selecting the approximating model, defined by (\ref%
{approx_model_sel_def}). $\mathcal{H}$ is the event that all signals are
selected, and $\mathcal{G}$ is the event that no noise variable is selected.
To simplify the exposition, with slight abuse of notation, we denote the
probability of an event $\mathcal{E}$ conditional on $\bar{\theta}_{i,T}$
being of order $\ominus (T^{-a})$ by $\Pr [\mathcal{E}|\bar{\theta}%
_{i,T}=\ominus (T^{-a})],$where $a $ is a nonnegative constant.
\end{flushleft}

\subsection{Proof of Theorem \protect\ref{sel_consistency_theorem}\label{A1}}

{\ To establish result (\ref{approx_model_selection}), first note that $%
\mathcal{A}_{0}^{c}=\mathcal{H}^{c}\cup \mathcal{G}^{c}$ and hence ($%
\mathcal{H}^{c}$ denotes the complement of $\mathcal{H}$) 
\begin{equation}
\Pr (\mathcal{A}_{0}^{c})=\Pr (\mathcal{H}^{c})+\Pr (\mathcal{G}^{c})-\Pr (%
\mathcal{H}^{c}\cap \mathcal{G}^{c})\leq \Pr (\mathcal{H}^{c})+\Pr (\mathcal{%
\ G}^{c}),  \label{p1}
\end{equation}%
where $\mathcal{H}$ and $\mathcal{G}$ are given by (\ref{approx_event}). We
also have $\mathcal{H}^{c}=\{\sum_{i=1}^{k}\hat{\mathcal{J}}_{i}<k\}$ and $%
\mathcal{G}^{c}=\{\sum_{i=k+k_{T}^{\ast }+1}^{N}\hat{\mathcal{J}}_{i}>0\}$.
Let's consider $\Pr \left( \mathcal{H}^{c}\right) $ and $\Pr \left( \mathcal{%
\ G}^{c}\right) $ in turn. We have 
$\textstyle\Pr (\mathcal{H}^{c})\leq \sum_{i=1}^{k}\Pr (\hat{\mathcal{J}}%
_{i}=0)$. 
But for any signal 
\begin{equation*}
\textstyle\Pr (\hat{\mathcal{J}}_{i}=0)=\Pr \left[ \lvert t_{i,T}\rvert
<c_{p}(N,\delta )|\bar{\theta}_{i,T}=\ominus (T^{-\vartheta _{i}})\right]
=1-\Pr \left[ \lvert t_{i,T}\rvert >c_{p}(N,\delta )|\bar{\theta}%
_{i,T}=\ominus (T^{-\vartheta _{i}})\right] ,
\end{equation*}%
where $0\leq \vartheta _{i}<1/2$ and hence by Lemma \ref{t_test_bound}, we
can conclude that there exist sufficiently large finite positive constants $%
C_{0}$ and $C_{1}$ such that 
$\textstyle\Pr (\hat{\mathcal{J}}_{i}=0)=O\left[ \exp (-C_{0}T^{C_{1}})%
\right] $. 
Since by Assumption \ref{signal}, the number of signals is finite we can
further conclude that 
\begin{equation}
\textstyle\Pr (\mathcal{H}^{c})=O\left[ \exp (-C_{0}T^{C_{1}})\right] ,
\label{p2}
\end{equation}%
for some finite positive constants $C_{0}$ and $C_{1}$. In the next step
note that 
\begin{equation*}
\textstyle\Pr (\mathcal{G}^{c})=\Pr \left( \sum_{i=k+k_{T}^{\ast }+1}^{N}%
\hat{\mathcal{J}}_{i}>0\right) \leq \sum_{i=k+k_{T}^{\ast }+1}^{N}\Pr \left( 
\hat{\mathcal{J}}_{i}=1\right) .
\end{equation*}%
But for any noise variable $\textstyle\Pr (\hat{\mathcal{J}}_{i}=1)=\Pr %
\left[ \lvert t_{i,T}\rvert >c_{p}(N,\delta )|\bar{\theta}_{i,T}=\ominus
(T^{-\epsilon _{i}})\right] ,$ where $\epsilon _{i}\geq 1/2$ and hence by
Lemma \ref{t_test_bound}, we can conclude that there exist sufficiently
large finite positive constants $C_{0}$ and $C_{1}$ such that for any $0<\pi
<1$, $\textstyle\Pr (\hat{\mathcal{J}}_{i}=1)\leq \exp \left[ -\frac{(1-\pi
)^{2}\bar{\sigma}_{\eta _{i},T}^{2}\bar{\sigma}_{x_{i},T}^{2}c_{p}^{2}(N,%
\delta )}{2\bar{\omega}_{iy,T}^{2}(1+d_{T})^{2}}\right] +\exp
(-C_{0}T^{C_{1}})$, in which $\bar{\sigma}_{x_{i},T}^{2}=T^{-1}\sum_{t=1}^{T}%
\mathbb{E}(x_{it}^{2})$, $\bar{\omega}_{iy,T}^{2}=T^{-1}\sum_{t=1}^{T}%
\mathbb{E}(x_{it}^{2}y_{t}^{2}|\mathcal{F}_{t-1})$, $\bar{\sigma}_{\eta
_{i},T}^{2}=T^{-1}\sum_{t=1}^{T}\mathbb{E}(\eta _{it}^{2})$, }$\eta
_{it}=y_{t}-\phi _{i,T}x_{it}$, and $\phi _{i,T}$ is defined in (\ref%
{phiiTdef}). {Therefore, 
\begin{equation*}
\textstyle\Pr (\mathcal{G}^{c})\leq N\exp \left[ -\frac{\mathcal{X}%
_{NT}(1-\pi )^{2}c_{p}^{2}(N,\delta )}{2(1+d_{T})^{2}}\right] +N\exp
(-C_{0}T^{C_{1}}),
\end{equation*}%
where $\mathcal{X}_{NT}=\inf_{i\in \{k+k^{\ast }+1,k+k^{\ast }+2,\cdots ,N\}}%
\frac{\bar{\sigma}_{\eta _{i},T}^{2}\bar{\sigma}_{x_{i},T}^{2}}{\bar{\omega}%
_{iy,T}^{2}}$. By result (II) of Lemma \ref{cv_lemma} in online theory
supplement we can further conclude that for any $0<\pi <1$, 
\begin{equation}
\textstyle\Pr (\mathcal{G}^{c})=O\left( N^{1-\mathcal{X}_{NT}\left( \frac{%
1-\pi }{1+d_{T}}\right) ^{2}\delta }\right) +O\left[ N\exp (-C_{0}T^{C_{1}})%
\right] ,  \label{p3}
\end{equation}%
Using (\ref{p2}) and (\ref{p3}) in (\ref{p1}), we obtain $\Pr (\mathcal{A}%
_{0}^{c})=O\left( N^{1-\mathcal{X}_{NT}\left( \frac{1-\pi }{1+d_{T}}\right)
^{2}\delta }\right) +O\left[ N\exp (-C_{0}T^{C_{1}})\right] $ and $\Pr (%
\mathcal{A}_{0})=1-O\left( N^{1-\mathcal{X}_{NT}\left( \frac{1-\pi }{1+d_{T}}%
\right) ^{2}\delta }\right) -O\left[ N\exp (-C_{0}T^{C_{1}})\right] $, which
completes the proof. }

\subsection{Proof of Theorem \protect\ref{estimation consistency}\label{A2}}

For any $B>0$, 
\begin{align*}
\Pr \left( T^{\frac{1-d}{2}}\left\Vert \hat{\boldsymbol{\gamma }}_{T}- 
\boldsymbol{\gamma }_{T}^{\ast }\right\Vert >B\right) =& \Pr \left( T^{\frac{
1-d}{2}}\left\Vert \hat{\boldsymbol{\gamma }}_{T}-\boldsymbol{\gamma }
_{T}^{\ast }\right\Vert >B|\mathcal{A}_{0}\right) \Pr \left( \mathcal{A}
_{0}\right) + \\
& \Pr \left( T^{\frac{1-d}{2}}\left\Vert \hat{\boldsymbol{\gamma }}_{T}- 
\boldsymbol{\gamma }_{T}^{\ast }\right\Vert >B|\mathcal{A}_{0}^{c}\right)
\Pr \left( \mathcal{A}_{0}^{c}\right) .
\end{align*}%
Since $\Pr \left( T^{\frac{1-d}{2}}\left\Vert \hat{\boldsymbol{\gamma }}_{T}-%
\boldsymbol{\gamma }_{T}^{\ast }\right\Vert >B|\mathcal{A}_{0}^{c}\right) $
and $\Pr \left( \mathcal{A}_{0}\right) $ are less than or equal to one, we
can further write, 
\begin{equation*}
\Pr \left( T^{\frac{1-d}{2}}\left\Vert \hat{\boldsymbol{\gamma }}_{T}- 
\boldsymbol{\gamma }_{T}^{\ast }\right\Vert >B\right) \leq \Pr \left( T^{ 
\frac{1-d}{2}}\left\Vert \hat{\boldsymbol{\gamma }}_{T}-\boldsymbol{\gamma }
_{T}^{\ast }\right\Vert >B|\mathcal{A}_{0}\right) +\Pr \left( \mathcal{A}
_{0}^{c}\right) .
\end{equation*}%
By conditioning on $\mathcal{A}_{0}$ the dimension of vector $\hat{ 
\boldsymbol{\gamma }}_{T}$ is at most equal to $k+k_{T}^{\ast}$ and by
assumption $k_{T}^{\ast}=\ominus (T^{d})$ where $0\leq d<1/2$. Therefore, by
Lemma \ref{reg coef} in online theory supplement, conditional on $\mathcal{A}%
_{0} $, $\left\Vert \hat{\boldsymbol{\gamma }}_{T}-\boldsymbol{\gamma }%
_{T}^{\ast }\right\Vert $ is $O_{p}\left( T^{\frac{d-1}{2}}\right) $. By
Theorem 1, we also have $\lim_{T\rightarrow \infty }\Pr \left( \mathcal{A}%
_{0}^{c}\right) =0$. Hence, for any $\varepsilon >0$, there exists $%
B_{\varepsilon }>0$ and $T_{\varepsilon }>0$ such that 
\begin{equation*}
\Pr \left( T^{\frac{1-d}{2}}\left\Vert \hat{\boldsymbol{\gamma }}_{T}- 
\boldsymbol{\gamma }_{T}^{\ast }\right\Vert >B_{\varepsilon }|\mathcal{A}
_{0}\right) +\Pr \left( \mathcal{A}_{0}^{c}\right) <\varepsilon \text{ for
all }T>T_{\varepsilon }.
\end{equation*}%
Therefore, $\Pr \left( T^{\frac{1-d}{2}}\left\Vert \hat{\boldsymbol{\gamma }}%
_{T}-\boldsymbol{\gamma }_{T}^{\ast }\right\Vert >B_{\varepsilon }\right)
<\varepsilon \text{ for all }T>T_{\varepsilon }, $ and we conclude that 
\begin{equation*}
\left\Vert \hat{\boldsymbol{\gamma }}_{T}-\boldsymbol{\gamma }_{T}^{\ast
}\right\Vert =O_{P}\left( T^{\frac{d-1}{2}}\right) ,  \label{part2T2}
\end{equation*}%
as required. Similar lines of arguments can be used to show that if $\mathbb{%
E}\left( \mathbf{x}_{\tilde{k}_{T},t}\mathbf{x}_{\tilde{k}_{T},t}^{\prime
}\right) $ is a fixed time-invariant matrix, then 
$\left\Vert \hat{\boldsymbol{\gamma }}_{T}-\boldsymbol{\gamma }%
_{T}^{\diamond }\right\Vert =O_{P}\left( T^{\frac{d-1}{2}}\right)$, 
which completes the proof.


\subsection{Proof of Theorem \protect\ref{mean square error}\label{A3}}

Let $D_{T}=T^{-1}\sum_{t=1}^{T}\hat{\eta}_{t}^{2}-\left(\bar{\Delta}%
_{\beta,T} +\bar{\sigma}_{u,T}^{2}\right) $. For any $B>0$, 
\begin{align*}
\Pr \left( T^{\frac{1}{2}}\left\vert D_{T}\right\vert >B\right) =& \Pr
\left( T^{\frac{1}{2}}\left\vert D_{T}\right\vert >B|\mathcal{A}_{0}\right)
\Pr \left( \mathcal{A}_{0}\right) + \Pr \left( T^{\frac{1}{2}}\left\vert
D_{T}\right\vert >B|\mathcal{A} _{0}^{c}\right) \Pr \left( \mathcal{A}%
_{0}^{c}\right).
\end{align*}%
Since $\Pr \left( T^{\frac{1}{2}}\left\vert D_{T}\right\vert >B|\mathcal{A}%
_{0}^{c}\right) $ and $\Pr \left( \mathcal{A}_{0}\right) $ are less than or
equal to one, we can further write, 
\begin{equation*}
\Pr \left( T^{\frac{1}{2}}\left\vert D_{T}\right\vert >B\right) \leq \Pr
\left( T^{\frac{1}{2}}\left\vert D_{T}\right\vert >B|\mathcal{A}_{0}\right)
+\Pr \left( \mathcal{A}_{0}^{c}\right) .
\end{equation*}%
By conditioning on $\mathcal{A}_{0}$, the number of selected covariates is
at most equal to $k+k_{T}^{\ast }$ and by assumption $k_{T}^{\ast }=\ominus
(T^{d})$, where $0\leq d<1/2$. Therefore, by Lemma \ref{reg coef} in online
theory supplement, conditional on $\mathcal{A}_{0}$, $D_{T}$ is $O_{p}\left(
T^{- \frac{1}{2}}\right) $. By Theorem 1, we also have $\lim_{T\rightarrow
\infty }\Pr \left( \mathcal{A}_{0}^{c}\right) =0$. Hence, for any $%
\varepsilon >0$, there exists $B_{\varepsilon }>0$ and $T_{\varepsilon }>0$
such that 
$\Pr \left( T^{\frac{1}{2}}\left\vert D_{T}\right\vert >B_{\varepsilon }| 
\mathcal{A}_{0}\right) +\Pr \left( \mathcal{A}_{0}^{c}\right) <\varepsilon, 
\text{ for all }T>T_{\varepsilon }$. 
Therefore, $\Pr \left( T^{\frac{1}{2}}\left\vert D_{T}\right\vert
>B_{\varepsilon }\right) <\varepsilon \text{ for all }T>T_{\varepsilon }, $
and we conclude that 
\begin{equation*}
T^{-1}\sum_{t=1}^{T}\hat{\eta}_{t}^{2}-\left(\bar{\Delta}_{\beta,T} +\bar{
\sigma}_{u,T}^{2}\right) =O_{p}\left( T^{-\frac{1}{2}}\right).
\end{equation*}%
Furthermore, by Lemma \ref{reg coef}, $\bar{\Delta}_{\beta,T} $ is
non-negative. Following similar lines of argument we get that if $\mathbb{E}%
\left( \mathbf{x}_{\tilde{k}_{T},t}\mathbf{x}_{\tilde{k}_{T},t}^{\prime
}\right) $ is a fixed time-invariant matrix, then, 
\begin{equation*}
T^{-1}\sum_{t=1}^{T}\hat{\eta}_{t}^{2}-\left(\bar{\Delta}_{\beta,T}^{\ast}+%
\bar{\sigma}_{u,T}^{2}\right) =O_{p}\left( T^{-\frac{1}{2}}\right) ,
\end{equation*}%
with $\bar{\Delta}_{\beta,T}^{\ast} \geq 0 $ which completes the proof.


\subsection{Propositions and corollaries \label{APC}}

\begin{proposition}
\label{obs:pop_reg_coef} Suppose the target variable $y_{t}$ is generated
according to (\ref{dgp y_t}), and Assumptions \ref{signal}-\ref{eigenvalues
signals and pseudo signal} hold. Consider the following regression equation: 
\begin{equation}
y_{t}=\sum_{i=1}^{k}x_{it}\gamma _{iT}+\eta _{t}=\mathbf{x}_{k,t}^{\prime }%
\boldsymbol{\gamma }_{T}+\eta _{t},\text{ }t=1,2....,T  \label{eq:model_yt}
\end{equation}%
where $\boldsymbol{\gamma }_{T}$ is defined by 
\begin{equation}
\boldsymbol{\gamma }_{T}=\arg \min_{\mathbf{b}}T^{-1}\sum_{t=1}^{T}\mathbb{E}%
\left( y_{t}-\mathbf{x}_{k,t}^{\prime }\mathbf{b}\right) ^{2}.  \label{gamma}
\end{equation}%
Then there exists a positive constnt $\epsilon \geq 1/2$, such that 
\begin{equation*}
\boldsymbol{\gamma }_{T}=\left[ T^{-1}\sum_{t=1}^{T}\mathbb{E}\left( \mathbf{%
x}_{k,t}\mathbf{x}_{k,t}^{\prime }\right) \right] ^{-1}T^{-1}\sum_{t=1}^{T}%
\sum_{i=1}^{k}\mathbb{E}\left( \mathbf{x}_{k,t}x_{it}\right) \mathbb{E}%
\left( \beta _{it}\right) +d_{T}\boldsymbol{\tau }_{k},
\end{equation*}%
where $d_{T}=O\left( T^{-\epsilon }\right) $, and $\boldsymbol{\tau }_{k}$
is the $k\times 1$ vector of ones. Also, if the expected value of $\beta
_{it}$ for $i=1,2,\cdots ,k$ are time-invariant, i.e., $\mathbb{E}\left(
\beta _{it}\right) =\beta _{i}$, then $\gamma _{iT}=\beta _{i}+d_{T}$ for $%
i=1,2,\cdots ,k$ and there exists $\varrho \geq 1$ such that 
\begin{equation*}
\mathbb{E}\left( \eta _{t}^{2}\right) =\Delta _{\beta ,t}+\mathbb{E}%
(u_{t}^{2})+e_{T},
\end{equation*}%
where $e_{T}=O\left( T^{-\varrho }\right) $ 
\begin{equation}
\Delta _{\beta ,t}=\sum_{i=1}^{k}\sum_{j=1}^{k}\sigma _{ijt,x}\sigma
_{ij,\beta }=\text{tr}\left( \boldsymbol{\Sigma }_{\mathbf{x}_{k},t}%
\boldsymbol{\Omega }_{\beta ,t}\right) \geq 0,  \label{Delta}
\end{equation}%
$\boldsymbol{\Sigma }_{\mathbf{x}_{k},t}\equiv \left( \sigma _{ijt,x}\right) 
$, $\boldsymbol{\Omega }_{\beta ,t}\equiv \left( \sigma _{ijt,\beta }\right) 
$, for $i,j=1,2,\cdots ,k$, $\sigma _{ijt,x}=\mathbb{E}\left(
x_{it}x_{jt}\right) $, and $\sigma _{ijt,\beta }=\mathbb{E}\left[ (\beta
_{it}-\beta _{i})(\beta _{jt}-\beta _{j})\right] $.

\noindent Alternatively, if the covariance matrix of the signals are
time-invariant, i.e., $\mathbb{E}\left( \mathbf{x}_{k,t}\mathbf{x}%
_{k,t}^{\prime }\right) =\boldsymbol{\Sigma }_{x_{k}}$, then $\gamma _{iT}=%
\bar{\beta}_{iT}+d_{T}$ for $i=1,2,\cdots ,k$, where $\bar{\beta}%
_{iT}=T^{-1}\sum_{t=1}^{T}\mathbb{E}\left( \beta _{it}\right) $, and there
exists $\varrho \geq 1$ such that 
\begin{equation*}
\mathbb{E}\left( \eta _{t}^{2}\right) =\Delta _{\beta ,t}^{\ast }+\mathbb{E}%
(u_{t}^{2})+e_{T}
\end{equation*}%
where $e_{T}=O\left( T^{-\varrho }\right) $ 
\begin{equation}
\Delta _{\beta ,t}^{\ast }=\sum_{i=1}^{k}\sum_{j=1}^{k}\sigma _{ij,x}\sigma
_{ijt,\beta }^{\ast }=\text{tr}\left( \boldsymbol{\Omega }_{\beta ,t}^{\ast }%
\boldsymbol{\Sigma }_{\mathbf{x}_{k}}\right) \geq 0,  \label{Delta*}
\end{equation}%
$\boldsymbol{\Omega }_{\beta ,t}^{\ast }\equiv \left( \sigma _{ijt,\beta
}^{\ast }\right) ,$ for $i,j=1,2,\cdots ,k$, and $\sigma _{ijt,\beta }^{\ast
}=\mathbb{E}\left[ (\beta _{it}-\bar{\beta}_{i,T})(\beta _{jt}-\bar{\beta}%
_{j,T})\right] $.
\end{proposition}

\begin{remark}
Proposition \ref{obs:pop_reg_coef} shows that in a linear regression model
that does not consider parameter instability, the deviation of each
coefficient from the simple time-average of the corresponding coefficients
in the DGP approaches zero. Moreover, $\Delta _{\beta ,t}\geq 0$ and $\Delta
_{\beta ,t}^{\ast }\geq 0$ represent the costs, in mean squared error sense,
of neglecting parameter instability.
\end{remark}


\begin{corollary}
\label{cor:delta_choice} Let $y_{t}$ for $t=1,2,\cdots ,T$ be generated by (%
\ref{dgp y_t}), and consider the active set $\mathcal{S}_{Nt}=%
\{x_{1t},x_{2t},\cdots ,x_{Nt}\}$ which contains $k$ signals, $k_{T}^{\ast }$
pseudo-signals, and $N-k-k_{T}^{\ast }$ noise variables. Suppose Assumptions %
\ref{signal}-\ref{subg} hold and the noise variables, $x_{it}$ $%
i=k+k_{T}^{\ast }+1,k+k_{T}^{\ast }+2,\cdots ,N$, are independent of the
target, $y_{t}$, and have time-invaraint unconditional variances, $\mathbb{V}%
\left( x_{it}^{2}\right) =\sigma _{i}^{2}$ for $i=k+k_{T}^{\ast
}+1,k+k_{T}^{\ast }+2,\cdots ,N$, and $N=\ominus (T^{\kappa })$ with $\kappa
>0$. Then, there exist finite positive constants $C_{0}$ and $C_{1}$ such
that, for any $\pi $ in $(0,1)$ and any null sequence $d_{T}>0$, the
probability of selecting the approximating model $\mathcal{A}_{0}$, defined
by (\ref{approx_model_sel_def}), by the OCMT procedure with the critical
value function $c_{p}(N,\delta )$ given by (\ref{cv_function}), for some $%
\delta >0$, is given by 
\begin{equation}
\Pr (\mathcal{A}_{0})=1-O\left[ T^{\kappa \left( 1-\left( \frac{1-\pi }{%
1+d_{T}}\right) ^{2}\delta \right) }\right] -O\left[ T^{\kappa }\exp \left(
-C_{0}T^{C_{1}}\right) \right] .  \label{eq: approx_model_selection_v2}
\end{equation}
\end{corollary}

\begin{remark}
Corollary \ref{cor:delta_choice} shows that if we further assume that the
noise variables are independent of $y_{t} $ and their variance does not
change over time, then for any $\delta > 1 $, the OCMT procedure
consistently selects the approximating model.
\end{remark}

\subsection{Proof of propositions and corollaries \label{Proof_APC}}

\begin{proofproposition}
Since the objective function for $\boldsymbol{\gamma }_{T}$ is convex and,
by Assumption \ref{eigenvalues signals and pseudo signal}, $%
T^{-1}\sum_{t=1}^{T}\mathbb{E}\left( \mathbf{x}_{k,t}\mathbf{x}%
_{k,t}^{\prime }\right) $ is invertible, then by the first-order condition
of the minimization in (\ref{gamma}) we have 
\begin{equation*}
\boldsymbol{\gamma }_{T}=\left[ T^{-1}\sum_{t=1}^{T}\mathbb{E}\left( \mathbf{%
x}_{k,t}\mathbf{x}_{k,t}^{\prime }\right) \right] ^{-1}T^{-1}\sum_{t=1}^{T}%
\mathbb{E}\left( \mathbf{x}_{k,t}y_{t}\right) .
\end{equation*}
Substituting $y_{t}$ from (\ref{dgp y_t}), now yields 
\begin{equation*}
\boldsymbol{\gamma }_{T}=\left[ T^{-1}\sum_{t=1}^{T}\mathbb{E}\left( \mathbf{%
x}_{k,t}\mathbf{x}_{k,t}^{\prime }\right) \right] ^{-1}T^{-1}\sum_{t=1}^{T}%
\sum_{i=1}^{k}\mathbb{E}\left( \mathbf{x}_{k,t}x_{it}\beta _{it}\right) +%
\left[ T^{-1}\sum_{t=1}^{T}\mathbb{E}\left( \mathbf{x}_{k,t}\mathbf{x}%
_{k,t}^{\prime }\right) \right] ^{-1}T^{-1}\sum_{t=1}^{T}\mathbb{E}\left( 
\mathbf{x}_{k,t}u_{it}\right) .
\end{equation*}%
By part (c) of Assumptions \ref{md}, all the elements of the $k\times 1$
vector $T^{-1}\sum_{t=1}^{T}\mathbb{E}\left( \mathbf{x}_{k,t}u_{it}\right) $
are $O\left( T^{-\epsilon }\right) $ for some $\epsilon \geq 1/2$. Moreover,
by Assumptions \ref{subg} and \ref{eigenvalues signals and pseudo signal},
all the element of $k\times k$ matrix $\left[ T^{-1}\sum_{t=1}^{T}\mathbb{E}%
\left( \mathbf{x}_{k,t}\mathbf{x}_{k,t}^{\prime }\right) \right] ^{-1}$ are
finite fixed numbers. Since, by Assumption \ref{signal}, the number of
signals, $k$, is a finite fixed number, we can further conclude that all the
elements of $k\times 1$ vector, 
\begin{equation*}
\left[ T^{-1}\sum_{t=1}^{T}\mathbb{E}\left( \mathbf{x}_{k,t}\mathbf{x}%
_{k,t}^{\prime }\right) \right] ^{-1}T^{-1}\sum_{t=1}^{T}\mathbb{E}\left( 
\mathbf{x}_{k,t}u_{it}\right) ,
\end{equation*}%
are $O\left( T^{-\epsilon }\right) $ for some $\epsilon \geq 1/2$ and
consequently we can write 
\begin{equation*}
\boldsymbol{\gamma }_{T}=\left[ T^{-1}\sum_{t=1}^{T}\mathbb{E}\left( \mathbf{%
x}_{k,t}\mathbf{x}_{k,t}^{\prime }\right) \right] ^{-1}T^{-1}\sum_{t=1}^{T}%
\sum_{i=1}^{k}\mathbb{E}\left( \mathbf{x}_{k,t}x_{it}\beta _{it}\right)
+d_{T}\boldsymbol{\tau }_{k},
\end{equation*}

where $d_{T}=O\left( T^{-\epsilon }\right) $ for some $\epsilon \geq 1/2$
and $\boldsymbol{\tau }_{k}$ is the $k\times 1$ vector of ones. By
Assumption \ref{signal}, $\beta _{it}$ is independent of $x_{jt}$ for all $%
i,j=1,2,\cdots ,k $, therefore, 
\begin{equation*}
\boldsymbol{\gamma }_{T}=\left[ T^{-1}\sum_{t=1}^{T}\mathbb{E}\left( \mathbf{%
\ x}_{k,t}\mathbf{x}_{k,t}^{\prime }\right) \right] ^{-1}T^{-1}%
\sum_{t=1}^{T}\sum_{i=1}^{k}\mathbb{E}\left( \mathbf{x}_{k,t}x_{it}\right) 
\mathbb{E}\left( \beta _{it}\right) +d_{T}\boldsymbol{\tau }_{k}.
\end{equation*}%
Consider first the case where $\mathbb{E}\left( \beta _{it}\right) $ is
time-invariant and set $\mathbb{E}\left( \beta _{it}\right) =\beta _{i}$.
Then 
\begin{align*}
\boldsymbol{\gamma }_{T}=& \left[ T^{-1}\sum_{t=1}^{T}\mathbb{E}\left( 
\mathbf{x}_{k,t}\mathbf{x}_{k,t}^{\prime }\right) \right] ^{-1}T^{-1}%
\sum_{t=1}^{T}\sum_{i=1}^{k}\mathbb{E}\left( \mathbf{x}_{k,t}x_{it}\right)
\beta _{i}+d_{T}\boldsymbol{\tau }_{k} \\
=& \left[ T^{-1}\sum_{t=1}^{T}\mathbb{E}\left( \mathbf{x}_{k,t}\mathbf{x}%
_{k,t}^{\prime }\right) \right] ^{-1}T^{-1}\sum_{t=1}^{T}\mathbb{E}\left( 
\mathbf{x}_{k,t}\sum_{i=1}^{k}x_{it}\beta _{i}\right) +d_{T}\boldsymbol{\tau 
}_{k} \\
=& \left[ T^{-1}\sum_{t=1}^{T}\mathbb{E}\left( \mathbf{x}_{k,t}\mathbf{x}%
_{k,t}^{\prime }\right) \right] ^{-1}T^{-1}\sum_{t=1}^{T}\mathbb{E}\left( 
\mathbf{x}_{k,t}\mathbf{x}_{k,t}^{\prime }\boldsymbol{\beta }\right) +d_{T}%
\boldsymbol{\tau }_{k} \\
=& \left[ T^{-1}\sum_{t=1}^{T}\mathbb{E}\left( \mathbf{x}_{k,t}\mathbf{x}%
_{k,t}^{\prime }\right) \right] ^{-1}\left[ T^{-1}\sum_{t=1}^{T}\mathbb{E}%
\left( \mathbf{x}_{k,t}\mathbf{x}_{k,t}^{\prime }\right) \right] \boldsymbol{%
\ \beta }+d_{T}\boldsymbol{\tau }_{k}=\boldsymbol{\beta }+d_{T}\boldsymbol{\
\tau }_{k},
\end{align*}%
where $\boldsymbol{\beta }=\left( \beta _{1},\beta _{2},\cdots ,\beta
_{k}\right) ^{\prime }$. So, in this case $\gamma _{iT}$ would converge to
the expected value of $\beta _{it}$ at $T\rightarrow \infty $. Moreover, 
\begin{equation*}
\eta _{t}=y_{t}-\sum_{i=1}^{k}x_{it}(\beta _{i}+d_{T}).
\end{equation*}%
By substituting for $y_{t}$ from (\ref{dgp y_t}), we have 
\begin{equation*}
\eta _{t}=\sum_{i=1}^{k}x_{it}\left( \beta _{it}-\beta _{i}\right)
+u_{t}+d_{T}\sum_{i=1}^{k}x_{it}.
\end{equation*}%
Therefore, by Assumptions \ref{signal} and \ref{md}, 
\begin{equation*}
\mathbb{E}\left( \eta _{t}^{2}\right) =\sum_{i=1}^{k}\sum_{j=1}^{k}\sigma
_{ijt,x}\sigma _{ijt,\beta }+\mathbb{E}(u_{t}^{2})+e_{T}.
\end{equation*}%
where $e_{T}=O\left( T^{-\varrho }\right) $ for some $\varrho \geq 1$, $%
\sigma _{ijt,x}=\mathbb{E}\left( x_{it}x_{jt}\right) $, $\sigma _{ijt,\beta
}=\mathbb{E}\left[ (\beta _{it}-\beta _{i})(\beta _{jt}-\beta _{j})\right] $
. We further have 
\begin{equation*}
\Delta _{\beta ,t}=\sum_{i=1}^{k}\sum_{j=1}^{k}\sigma _{ijt,x}\sigma
_{ijt,\beta }=\text{tr}\left( \boldsymbol{\Omega }_{\beta ,t}\boldsymbol{\
\Sigma }_{\mathbf{x}_{k},t}\right) ,
\end{equation*}%
where $\boldsymbol{\Omega }_{\beta ,t}\equiv \left( \sigma _{ijt,\beta
}\right) $ and $\boldsymbol{\Sigma }_{\mathbf{x}_{k},t}\equiv \left( \sigma
_{ijt,x}\right) $ for $i,j=1,2,\cdots ,k$. By result 9(b) on page 44 of \cite%
{Lutkepohl1996handbook}, we can further write 
\begin{equation*}
\text{tr}\left( \boldsymbol{\Omega }_{\beta ,t}\boldsymbol{\Sigma }_{\mathbf{%
\ x}_{k},t}\right) \geq k\left[ \text{det}\left( \boldsymbol{\Omega }_{\beta
,t}\right) \right] ^{1/k}\left[ \text{det}\left( \boldsymbol{\Sigma }_{%
\mathbf{x}_{k},t}\right) \right] ^{1/k}.
\end{equation*}%
But $k$ is a finite fixed integer. Furthermore, $\text{det}\left( 
\boldsymbol{\Omega }_{\beta ,t}\right) \geq 0$ and $\text{det}\left( 
\boldsymbol{\Sigma }_{\mathbf{x}_{k},t}\right) >0$, since $\boldsymbol{\
\Omega }_{\beta ,t}$ and $\boldsymbol{\Sigma }_{\mathbf{x}_{k},t}$ are
positive semi-definite and positive definite matrices, respectively. So, we
can conclude that $\Delta _{\beta ,t}\geq 0$.

Consider now a second case where $\mathbb{E}\left( \mathbf{x}_{k,t}\mathbf{x}
_{k,t}^{\prime }\right) $ is time-invariant and set $\mathbb{E}\left( 
\mathbf{x}_{k,t}\mathbf{x}_{k,t}^{\prime }\right) =\boldsymbol{\Sigma }_{{x}
_{k}}$. Then 
\begin{align*}
\boldsymbol{\gamma }_{T}=& \left[ T^{-1}\sum_{t=1}^{T}\mathbb{E}\left( 
\mathbf{x}_{k,t}\mathbf{x}_{k,t}^{\prime }\right) \right] ^{-1}%
\sum_{i=1}^{k} \mathbb{E}\left( \mathbf{x}_{k,t}x_{it}\right) \left[
T^{-1}\sum_{t=1}^{T} \mathbb{E}\left( \beta _{it}\right) \right] +d_{T}%
\boldsymbol{\tau }_{k} \\
=& \left[ T^{-1}\sum_{t=1}^{T}\mathbb{E}\left( \mathbf{x}_{k,t}\mathbf{x}
_{k,t}^{\prime }\right) \right] ^{-1}\sum_{i=1}^{k}\mathbb{E}\left( \mathbf{%
x }_{k,t}x_{it}\bar{\beta}_{iT}\right) +d_{T}\boldsymbol{\tau }_{k} \\
=& \left[ T^{-1}\sum_{t=1}^{T}\mathbb{E}\left( \mathbf{x}_{k,t}\mathbf{x}
_{k,t}^{\prime }\right) \right] ^{-1}\mathbb{E}\left( \mathbf{x}_{k,t} 
\mathbf{x}_{k,t}^{\prime }\bar{\boldsymbol{\beta }}_{T}\right) +d_{T} 
\boldsymbol{\tau }_{k} \\
=& \left[ T^{-1}\sum_{t=1}^{T}\mathbb{E}\left( \mathbf{x}_{k,t}\mathbf{x}
_{k,t}^{\prime }\right) \right] ^{-1}\left[ T^{-1}\sum_{t=1}^{T}\mathbb{E}
\left( \mathbf{x}_{k,t}\mathbf{x}_{k,t}^{\prime }\right) \right] \bar{ 
\boldsymbol{\beta }}_{T}+d_{T}\boldsymbol{\tau }_{k}=\bar{\boldsymbol{\beta }
}_{T}+d_{T}\boldsymbol{\tau }_{k}
\end{align*}
where $\bar{\boldsymbol{\beta }}_{T}=\left( \bar{\beta}_{1T},\bar{\beta}
_{2T},\cdots ,\bar{\beta}_{kT}\right) ^{\prime }$ and $\bar{\beta}
_{iT}=T^{-1}\sum_{t=1}^{k}\mathbb{E}\left( \beta _{it}\right) $. So, in this
case $\gamma _{iT}$ would converge to the simple average of expected value
of $\beta _{it}$ across time. Moreover, 
\begin{equation}
\eta _{t}=y_{t}-\sum_{i=1}^{k}x_{it}(\bar{\beta}_{iT}+d_{T})
\end{equation}
By substituting for $y_{t}$ from (\ref{dgp y_t}), we have 
\begin{equation}
\eta _{t}=\sum_{i=1}^{k}x_{it}\left( \beta _{it}-\bar{\beta}_{iT}\right)
+u_{t}+d_{T}\sum_{i=1}^{k}x_{it}.
\end{equation}
Therefore, by Assumptions \ref{signal} and \ref{md}, 
\begin{equation}
\mathbb{E}\left( \eta _{t}\right) ^{2}=\sum_{i=1}^{k}\sum_{j=1}^{k}\sigma
_{ij,x}\sigma _{ijt,\beta }^{\ast }+\mathbb{E}(u_{t}^{2})+e_{T},
\end{equation}
where $e_{T}=O\left( T^{-\varrho }\right) $ for some $\varrho \geq 1$, $%
\sigma _{ij,x}=\mathbb{E}\left( x_{it}x_{jt}\right) $ and $\sigma
_{ijt,\beta }^{\ast }=\mathbb{E}\left[ (\beta _{it}-\bar{\beta}_{i,T})(\beta
_{jt}-\bar{\beta}_{j,T})\right] $. We further have 
\begin{equation*}
\Delta _{\beta ,t}^{\ast }=\sum_{i=1}^{k}\sum_{j=1}^{k}\sigma _{ij,x}\sigma
_{ijt,\beta }^{\ast }=\text{tr}\left( \boldsymbol{\Omega }_{\beta ,t}^{\ast
} \boldsymbol{\Sigma }_{\mathbf{x}_{k}}\right)
\end{equation*}
where $\boldsymbol{\Omega }_{\beta ,t}^{\ast }\equiv \left( \sigma
_{ijt,\beta }^{\ast }\right) $ and $\boldsymbol{\Sigma }_{\mathbf{x}
_{k}}\equiv \left( \sigma _{ij,x}\right) $ for $i,j=1,2,\cdots ,k$. By
result 9(b) on page 44 of \cite{Lutkepohl1996handbook}, we can further write 
\begin{equation*}
\text{tr}\left( \boldsymbol{\Omega }_{\beta ,t}^{\ast }\boldsymbol{\Sigma }%
_{ \mathbf{x}_{k}}\right) \geq k\left[ \text{det}\left( \boldsymbol{\Omega }
_{\beta ,t}^{\ast }\right) \right] ^{1/k}\left[ \text{det}\left( \boldsymbol{%
\ \Sigma }_{\mathbf{x}_{k}}\right) \right] ^{1/k}.
\end{equation*}
But $k$ is a finite fixed integer. Furthermore, $\text{det}\left( 
\boldsymbol{\Omega }_{\beta ,t}^{\ast }\right) \geq 0$ and $\text{det}\left( 
\boldsymbol{\Sigma }_{\mathbf{x}_{k},t}\right) >0$, since $\boldsymbol{\
\Omega }_{\beta ,t}^{\ast }$ and $\boldsymbol{\Sigma }_{\mathbf{x}_{k},t}$
are positive semi-definite and positive definite matrices, respectively. So,
we can conclude that $\Delta _{\beta ,t}^{\ast }\geq 0$.
\end{proofproposition}

\begin{proofcorollary}
By Theorem \ref{sel_consistency_theorem}, we have that under Assumptions \ref%
{signal}-\ref{subg}, there exist finite positive constants $C_{0}$ and $%
C_{1} $ such that, for any $0<\pi <1$, the probability of selecting the
approximating model $\mathcal{A}_{0}$, as defined by (\ref%
{approx_model_sel_def}), is given by 
\begin{equation}
\Pr (\mathcal{A}_{0})=1-O\left[ T^{\kappa \left( 1-\mathcal{X}_{NT}\left( 
\frac{1-\pi }{1+d_{T}}\right) ^{2}\delta \right) }\right] -O\left[ T^{\kappa
}\exp \left( -C_{0}T^{C_{1}}\right) \right] ,
\label{eq: approx_model_selection}
\end{equation}%
where 
\begin{equation*}
\textstyle\mathcal{X}_{NT}=\inf_{i\in \{k+k^{\ast }+1,\cdots ,N\}}\frac{\bar{%
\sigma}_{\eta _{i},T}^{2}\bar{\sigma}_{x_{i},T}^{2}}{\bar{\omega}_{iy,T}^{2}}%
.
\end{equation*}%
with $\bar{\sigma}_{x_{i},T}^{2}=T^{-1}\sum_{t=1}^{T}\mathbb{E}(x_{it}^{2})$%
, $\bar{\omega}_{iy,T}^{2}=T^{-1}\sum_{t=1}^{T}\mathbb{E}%
(x_{it}^{2}y_{t}^{2}|\mathcal{F}_{t-1})$, $\bar{\sigma}_{\eta
_{i},T}^{2}=T^{-1}\sum_{t=1}^{T}\mathbb{E}(\eta _{it}^{2})$, $\eta
_{it}=y_{t}-\phi _{i,T}x_{it}$, and $\phi _{i,T}$ is defined in (\ref%
{phiiTdef}). Note that, 
\begin{align*}
\bar{\sigma}_{\eta _{i},T}^{2}& =T^{-1}\sum_{t=1}^{T}\mathbb{E}\left(
y_{t}^{2}\right) +\phi _{i,T}^{2}\left[ T^{-1}\sum_{t=1}^{T}\mathbb{E}\left(
x_{it}^{2}\right) \right] -2\phi _{i,T}\left[ T^{-1}\sum_{t=1}^{T}\mathbb{E}%
\left( x_{it}y_{t}\right) \right] \\
& =\bar{\sigma}_{y,T}^{2}+\phi _{i,T}^{2}\bar{\sigma}_{x_{i},T}^{2}-2\phi
_{i,T}\bar{\theta}_{i,T}=\bar{\sigma}_{y,T}^{2}-\phi _{i,T}^{2}\bar{\sigma}%
_{x_{i},T}^{2}.
\end{align*}%
But, $x_{it}$ for all $i\in \{k+k_{T}^{\ast }+1,k+k_{T}^{\ast }+2,\cdots
,N_{T}\}$ are independent of $y_{t}$ and hence $\phi _{i,T}=0$.
Consequently, $\bar{\sigma}_{\eta _{i},T}^{2}=\bar{\sigma}_{y,T}^{2}$ for $%
i\in \{k+k_{T}^{\ast }+1,k+k_{T}^{\ast }+2,\cdots ,N_{T}\}$. Moreover, since 
$x_{it}$ for all $i\in \{k+k_{T}^{\ast }+1,k+k_{T}^{\ast }+2,\cdots ,N_{T}\}$
are independent of $y_{t}$, we can write 
\begin{equation*}
\bar{\omega}_{iy,T}^{2}=T^{-1}\sum_{t=1}^{T}\mathbb{E}\left( x_{it}^{2}|%
\mathcal{F}_{t-1}\right) \mathbb{E}\left( y_{t}^{2}|\mathcal{F}_{t-1}\right)
=T^{-1}\sum_{t=1}^{T}\mathbb{E}\left( x_{it}^{2}\right) \mathbb{E}\left(
y_{t}^{2}\right) ,
\end{equation*}%
for $i\in \{k+k_{T}^{\ast }+1,k+k_{T}^{\ast }+2,\cdots ,N_{T}\}$. Therefore, 
\begin{equation*}
\textstyle\mathcal{X}_{NT}=\inf_{i\in \{k+k^{\ast }+1,\cdots ,N\}}\frac{\bar{%
\sigma}_{y,T}^{2}\bar{\sigma}_{x_{i},T}^{2}}{T^{-1}\sum_{t=1}^{T}\mathbb{E}%
\left( x_{it}^{2}\right) \mathbb{E}\left( y_{t}^{2}\right) },
\end{equation*}%
In cases where $\mathbb{E}(x_{it}^{2})$ for $i\in \{k+k_{T}^{\ast
}+1,k+k_{T}^{\ast }+2,\cdots ,N_{T}\}$ are time-invariant, we can conclude
that $\mathcal{X}_{NT}=1$ and hence the probability of selecting the
approximating model is given by 
\begin{equation*}
\Pr (\mathcal{A}_{0})=1-O\left[ T^{\kappa \left( 1-\left( \frac{1-\pi }{%
1+d_{T}}\right) ^{2}\delta \right) }\right] -O\left[ T^{\kappa }\exp \left(
-C_{0}T^{C_{1}}\right) \right] ,
\end{equation*}%
as required. Note that $d_{T}\rightarrow 0$ and $T\rightarrow \infty $ and $%
\pi $ is an arbitrary constant between zero and one.
\end{proofcorollary}

Since $\pi $ is an arbitrary constant between zero and one, result (\ref{eq:
approx_model_selection_v2}) of Corollary \ref{cor:delta_choice} implies that
for any $\delta >1$, we can select an approximating model with probability
approaching one as $N$ and $T$ grows to infinity.%

\section{Main lemmas}

\label{main_lemmas}

\begin{lemma}
\label{md y} Let $y_{t}$ be a target variable generated by equation (\ref%
{dgp y_t}), and $x_{it}$ be a covariate in the active set $\mathcal{S}%
_{Nt}=\{x_{1t},x_{2t}, \cdots ,x_{Nt}\}$. Under Assumptions \ref{signal},
and \ref{md}, we have 
\begin{equation*}
\mathbb{E}\left[ y_{t}x_{it}-\mathbb{E}(y_{t}x_{it})|\mathcal{F}_{t-1}\right]
=0, \text{ for } i=1,2,\cdots ,N,
\end{equation*}
and 
\begin{equation*}
\mathbb{E}\left[ y_{t}^{2}-\mathbb{E}(y_{t}^{2})|\mathcal{F}_{t-1}\right] =0.
\end{equation*}
\end{lemma}

\begin{proof}
For $i=1,2,\cdots ,N$, we have 
\begin{equation*}
\mathbb{E}(y_{t}x_{it}|\mathcal{F}_{t-1})\textstyle= \sum_{j=1}^{k}\mathbb{E}%
(\beta_{jt}|\mathcal{F}_{t-1}) \mathbb{E}(x_{jt}x_{it}|\mathcal{F}_{t-1})+%
\mathbb{E}(u_{t}x_{it}| \mathcal{F}_{t-1}).
\end{equation*}

By Assumption \ref{md}, $\mathbb{E}(\beta _{jt}| \mathcal{F}_{t-1})=\mathbb{E%
}(\beta _{jt})$, $\mathbb{E}(x_{jt}x_{it}| \mathcal{F}_{t-1})=\mathbb{E}%
(x_{jt}x_{it})$, and $\mathbb{E}(u_{t}x_{it}| \mathcal{F}_{t-1})=\mathbb{E}%
(u_{t}x_{it})$. Therefore, 
\begin{equation*}
\textstyle\mathbb{E}(y_{t}x_{it}|\mathcal{F}_{t-1})= \sum_{j=1}^{k} \mathbb{E%
}(\beta _{jt})\mathbb{E}(x_{jt}x_{it})+\mathbb{E}(u_{t}x_{it})= \mathbb{E}%
(y_{t}x_{it}).
\end{equation*}
Also to establish the last result, note that $y_{t}$ can be written as 
\begin{equation*}
\textstyle y_{t}= \sum_{j=1}^{k}\beta _{jt}x_{jt}+u_{t} = \mathbf{x}%
_{k,t}^{\prime }\boldsymbol{\beta}_{t} + u_{t}\text{,}
\end{equation*}
where $\mathbf{x}_{k,t}=(x_{1t},x_{2t},\cdots ,x_{kt})^{\prime }$, and $%
\boldsymbol{\beta }_{t}=(\beta _{1t},\beta _{2t},\cdots ,\beta
_{kt})^{\prime }$. Hence, 
\begin{equation*}
\begin{split}
\mathbb{E}(y_{t}^{2}|\mathcal{F}_{t-1})& \textstyle=\mathbb{E}(\boldsymbol{%
\beta}_{t}^{\prime }|\mathcal{F}_{t-1})\mathbb{E}(\mathbf{x}_{t}\mathbf{x}
_{t}^{\prime }|\mathcal{F}_{t-1})\mathbb{E}(\boldsymbol{\beta}_{t}| \mathcal{%
F}_{t-1})+\mathbb{E}(u_{t}^{2}|\mathcal{F}_{t-1})+2\mathbb{E}( \boldsymbol{%
\beta}_{t}^{\prime }|\mathcal{F}_{t-1})\mathbb{E}(\mathbf{x}_{t}u_{t}|%
\mathcal{F}_{t-1}) \\
& \textstyle=\mathbb{E}(\boldsymbol{\beta}_{t}^{\prime })\mathbb{E}(\mathbf{x%
}_{t}\mathbf{x}_{t}^{\prime })\mathbb{E}(\boldsymbol{\beta}_{t})+ \mathbb{E}%
(u_{t}^{2})+2\mathbb{E}(\boldsymbol{\beta}_{t}^{\prime })\mathbb{E}(\mathbf{x%
}_{t}u_{t})=\mathbb{E}(y_{t}^{2}).
\end{split}%
\end{equation*}
\end{proof}


\begin{lemma}
\label{subg y} Let $y_{t}$ be a target variable generated by equation (\ref%
{dgp y_t}). Under Assumptions \ref{subg}-\ref{signal}, for any value of $%
\alpha >0$, there exist some positive constants $C_{0}$ and $C_{1}$ such
that 
\begin{equation*}
\sup_{t}\Pr (\lvert y_{t}\rvert >\alpha )\leq C_{0}\exp (-C_{1}\alpha
^{s/2}).
\end{equation*}
\end{lemma}

\begin{proof}
Note that 
\begin{equation*}
\textstyle\lvert y_{t}\rvert \leq \sum_{j=1}^{k}\lvert \beta
_{jt}x_{jt}\rvert +\lvert u_{t}\rvert .
\end{equation*}
Therefore, 
\begin{equation*}
\textstyle\Pr (\lvert y_{t}\rvert >\alpha )\leq \sum_{j=1}^{k}\lvert \beta
_{jt}x_{jt}\rvert +\lvert u_{t}\rvert >\alpha ),
\end{equation*}
and by Lemma \ref{prob_sum} for any $0<\pi _{i}<1$, $i=1,2,\cdots ,k +1$,
with $\sum_{i=1}^{k+1}\pi_{j}=1$, we can further write 
\begin{equation*}
\Pr (\lvert y_{t}\rvert >\alpha )\textstyle\leq \sum_{j=1}^{k}\Pr (\lvert
\beta _{jt}x_{jt}\rvert >\pi _{j}\alpha )+\Pr (\lvert u_{t}\rvert
>\pi_{k+1}\alpha ).
\end{equation*}

Moreover, by Lemma \ref{prob_product}, we have 
\begin{equation*}
\begin{split}
\Pr (\lvert \beta _{jt}x_{jt}\rvert >\pi _{j}\alpha )& \leq \Pr [\lvert
x_{jt}\rvert >(\pi _{j}\alpha )^{1/2}]+\Pr [\lvert \beta _{jt}\rvert >(\pi
_{i}\alpha )^{1/2}],
\end{split}%
\end{equation*}
and hence 
\begin{equation*}
\begin{split}
\Pr (\lvert y_{t}\rvert >\alpha )& \textstyle\leq \sum_{j=1}^{k}\Pr [\lvert
x_{jt}\rvert >(\pi _{j}\alpha )^{1/2}]+\sum_{j=1}^{k}\Pr [\lvert \beta
_{jt}\rvert >(\pi _{j}\alpha )^{1/2}]+\Pr (\lvert u_{t}\rvert >\pi
_{k+1}\alpha ),
\end{split}%
\end{equation*}

Therefore, under Assumptions \ref{subg}-\ref{signal}, we can conclude that
for any value of $\alpha >0$, there exist some positive constants $C_{0}$
and $C_{1}$ such that 
\begin{equation*}
\sup_{t}\Pr (\lvert y_{t}\rvert >\alpha )\leq C_{0}\exp (-C_{1}\alpha
^{s/2}) \text{.}
\end{equation*}
\end{proof}


\begin{lemma}
\label{conditional_corr_x_i_x_j} Let $x_{it}$ be a covariate in the active
set, $\mathcal{S}_{Nt}=\{x_{1t},x_{2t},\cdots ,x_{Nt}\}$. Suppose
Assumptions \ref{md}-\ref{subg} hold and $\zeta_T = \ominus(T^{\lambda})$
for some $\lambda > 0$. Then, if $0 < \lambda \leq (s+2)/(s+4)$, for any $0
<\pi < 1$, 
\begin{equation*}
\Pr (|\mathbf{x}_{i}^{\prime}\mathbf{x}_{j}-\mathbb{E}(\mathbf{x}%
_{i}^{\prime }\mathbf{x}_{j})|>\zeta _{T})\leq \exp\left(- \frac{(1-\pi)^{2}
\zeta_{T}^2}{2 T \bar{\omega}_{i j,T}^2 } \right),
\end{equation*}%
where, $\mathbf{x}_{i}=(x_{i1},x_{i2},\cdots ,x_{iT})^{\prime }$ and $\bar{%
\omega}_{i j,T}^2 = T^{-1}\sum_{t=1}^{T} \mathbb{E}\left(x_{it}^2 x_{jt}^2 | 
\mathcal{F}_{t-1}\right) $. Also, if $\lambda >(s+2)/(s+4)$, there exists a
finite positive constant $C_{1}$, 
\begin{equation*}
\Pr (|\mathbf{x}_{i}^{\prime }\mathbf{x}_{j}-\mathbb{E}(\mathbf{x}%
_{i}^{\prime }\mathbf{x}_{j})|>\zeta_{T}) \leq \exp \left(
-C_{1}\zeta_{T}^{s/(s+1)}\right),
\end{equation*}
for all $i,j = 1, 2, \cdots, N $. .
\end{lemma}

\begin{proof}
Note that $[\mathbf{x}_{i}^{\prime}\mathbf{x}_{j} -\mathbb{E}(\mathbf{x}%
_{i}^{\prime} \mathbf{x}_{j})] = \sum_{t=1}^{T}[x_{it} x_{jt} - \mathbb{E}%
(x_{it} x_{jt})]$, for all $i$ and $j$. By Assumption \ref{md} we have 
\begin{equation*}
\mathbb{E}\left[ x_{it}x_{jt}-\mathbb{E}(x_{it}x_{jt})|\mathcal{F}_{t-1} %
\right] =0,
\end{equation*}
for $i,j=1,2,\cdots ,N$. Moreover, by Assumption \ref{subg}, for all $i =
1,2, \cdots, N $ and $\alpha > 0$, there exist some finite positive
constants $C_0$ and $C_1 $ such that, 
\begin{equation*}
\sup_{t}\Pr (|x_{i t}|>\alpha )\leq C_{0}\exp (-C_{1}\alpha^{s}).
\end{equation*}
Therefore, by Lemma \ref{exp_tail_prod}, for all $i,j = 1,2, \cdots, N $ and 
$\alpha > 0$, 
\begin{equation*}
\sup_{t}\Pr (|x_{it}x_{jt}|>\alpha )\leq C_{0}\exp (-C_{1}\alpha ^{s/2}).
\end{equation*}
Hence, by Lemma \ref{mart_diff_proc_exp_tail}, if $0<\lambda \leq
(s+2)/(s+4) $, for any $0 < \pi < 1$, 
\begin{equation*}
\textstyle \Pr (|\mathbf{x}_{i}^{\prime}\mathbf{x}_{j} -\mathbb{E}(\mathbf{x}%
_{i}^{\prime} \mathbf{x}_{j})|>\zeta_{T}) \leq \exp\left(- \frac{(1-\pi)^{2}
\zeta_{T}^2}{2 T \bar{\omega}_{i j,T}^2 } \right).
\end{equation*}
Moreover, if $\lambda >(s+2)/(s+4)$, then there exists a finite positive
constant $C_{1}$, such that 
\begin{equation*}
\textstyle\Pr (|\mathbf{x}_{i}^{\prime}\mathbf{x}_{j} -\mathbb{E}(\mathbf{x}%
_{i}^{\prime} \mathbf{x}_{j})|>\zeta_{T}) \leq \exp \left( -C_{1}\zeta
_{T}^{s/(s+1)}\right).
\end{equation*}
\end{proof}


\begin{lemma}
\label{conditional_corr_xy} Let $y_{t}$ be a target variable generated by
the DGP given by (\ref{dgp y_t}) and $x_{it}$ be a covariate in the active
set, $\{x_{1t},x_{2t},\cdots ,x_{Nt}\}$. Suppose Assumptions \ref{signal}-%
\ref{subg} hold and $\zeta_T =\ominus(T^{\lambda})$ for some $\lambda > 0$.
Then, if $0<\lambda \leq (s+4)/(s+8)$, for any $0 < \pi < 1$, 
\begin{equation*}
\Pr(|\mathbf{x}_{i}^{\prime }\mathbf{y}-\theta_{i,T}| > \zeta _{T})\leq
\exp\left(- \frac{(1-\pi)^{2} \zeta_{T}^2}{2 T \bar{\omega}_{iy,T}^2 }
\right),
\end{equation*}%
where $\mathbf{x}_{i}=(x_{i1},x_{i2},\cdots,x_{iT})^{\prime }$, $\mathbf{y}%
=(y_{1},y_{2},\cdots,y_{T})^{\prime }$, $\theta _{i,T}=T\bar{\theta}_{i,T}=%
\mathbb{E}(\mathbf{x}_{i}^{\prime }\mathbf{y})$ and $\bar{\omega}_{i y,T}^2
= T^{-1} \sum_{t=1}^{T} \mathbb{E}\left(x_{it}^2 y_{t}^2 | \mathcal{F}_{t-1}
\right) $. Also, if $\lambda >(s+4)/(s+8)$, there exists a finite positive
constant $C_1$ such that 
\begin{equation*}
\Pr (|\mathbf{x}_{i}^{\prime }\mathbf{y}-\theta_{i,T} | > \zeta _{T}) \leq
\exp \left(-C_{1}\zeta_{T}^{s/(s+1)}\right),
\end{equation*}%
for all $i=1,2,\cdots ,N$.
\end{lemma}

\begin{proof}
Note that $[\mathbf{x}_{i}^{\prime}\mathbf{y} - \theta_{i,T}] =
\sum_{t=1}^{T}[x_{it} y_{t} - \mathbb{E}(x_{it} y_{t})]$, for all $i$. By
Lemma \ref{md y} 
\begin{equation*}
\mathbb{E}\left[ x_{it}y_{t}-\mathbb{E}(x_{it}y_{t})|\mathcal{F}_{t-1} %
\right] =0,
\end{equation*}
for $i=1,2,\cdots ,N$. Moreover, by Assumption \ref{subg}, for all $i = 1,2,
\cdots, N $ and $\alpha > 0$, there exist some finite positive constants $%
C_0 $ and $C_1 $ such that, 
\begin{equation*}
\sup_{t}\Pr (|x_{i t}|>\alpha )\leq C_{0}\exp (-C_{1}\alpha^{s}).
\end{equation*}
Also, by Lemma \ref{subg y}, there exist some finite positive constants $C_0$
and $C_1 $ such that, 
\begin{equation*}
\sup_{t}\Pr (\lvert y_{t}\rvert >\alpha )\leq C_{0}\exp (-C_{1}\alpha ^{s/2})%
\text{.}
\end{equation*}
Therefore, by Lemma \ref{exp_tail_prod}, for all $i = 1,2, \cdots, N $ and $%
\alpha > 0$, 
\begin{equation*}
\sup_{t}\Pr (|x_{it}y_{t}|>\alpha )\leq C_{0}\exp (-C_{1}\alpha ^{s/4}).
\end{equation*}
Hence, by Lemma \ref{mart_diff_proc_exp_tail}, if $0<\lambda \leq
(s+4)/(s+8) $, for any $0 < \pi < 1$, 
\begin{equation*}
\textstyle \Pr (|\mathbf{x}_{i}^{\prime}\mathbf{y} -\theta_{i,T}|>\zeta_{T})
\leq \exp\left(- \frac{(1-\pi)^{2} \zeta_{T}^2}{2 T \bar{\omega}_{i y,T}^2 }
\right).
\end{equation*}
Moreover, if $\lambda >(s+4)/(s+8)$, there exists a finite positive constant 
$C_{1}$, 
\begin{equation*}
\textstyle\Pr (|\mathbf{x}_{i}^{\prime}\mathbf{y}_{j}
-\theta_{i,T}|>\zeta_{T}) \leq \exp \left( -C_{1}\zeta _{T}^{s/(s+1)}\right).
\end{equation*}
\end{proof}


\begin{lemma}
\label{variance eta} Let $y_{t}$ be a target variable generated by equation (%
\ref{dgp y_t}) and $x_{it}$ be a covariate in the active set, $\mathcal{S}%
_{Nt}=\{x_{1t},x_{2t},\cdots ,x_{Nt}\}$. Suppose Assumptions \ref{signal}-%
\ref{subg} hold and $\zeta_T =\ominus(T^{\lambda})$ for some $\lambda > 0$.
Consider the projection regression of $y_{t}$ on $x_{it}$ as 
\begin{equation*}
y_{t}=\phi _{i,T}x_{it}+\eta _{it},
\end{equation*}%
where the projection coefficient $\phi _{i,T}$ is given by (\ref{phiiTdef}).
Then, if $0<\lambda \leq(s+4)/(s+8)$, there exist sufficiently large
positive constants $C_{0}$, $C_{1}$ and $C_{2}$ such that 
\begin{equation*}
\Pr \left[ \lvert \boldsymbol{\eta }_{i}^{\prime }\mathbf{M}_{x_{i}}%
\boldsymbol{\eta }_{i}-\mathbb{E}(\boldsymbol{\eta }_{i}^{\prime }%
\boldsymbol{\eta }_{i})\rvert >\zeta _{T}\right] \leq \exp
(-C_{0}T^{-1}\zeta _{T}^{2})+\exp (-C_{1}T^{C_{2}}),
\end{equation*}%
where $\boldsymbol{\eta }_{i}=(\eta _{i1},\eta _{i2},\cdots ,\eta
_{iT})^{\prime }$ and $\mathbf{M}_{x_{i}}=\mathbf{I}-T^{-1}\mathbf{x}%
_{i}(T^{-1}\mathbf{x}_{i}^{\prime }\mathbf{x}_{i})^{-1}\mathbf{x}%
_{i}^{\prime }$ with $\mathbf{x}_{i}=(x_{i1},x_{i2},\cdots ,x_{iT})^{\prime
} $. Also, if $\lambda >(s+4)/(s+8)$, there exist sufficiently large
positive constants $C_{0}$, $C_{1}$ and $C_{2}$ such that 
\begin{equation*}
\Pr \left[ \lvert \boldsymbol{\eta }_{i}^{\prime }\mathbf{M}_{x_{i}}%
\boldsymbol{\eta }_{i}-\mathbb{E}(\boldsymbol{\eta }_{i}^{\prime }%
\boldsymbol{\eta }_{i})\rvert >\zeta _{T}\right] \leq \exp (-C_{0}\zeta
_{T}^{s/(s+1)})+\exp (-C_{1}T^{C_{2}}),
\end{equation*}%
for all $i=1,2,\cdots ,N$.
\end{lemma}

\begin{proof}
Note that $\boldsymbol{\eta }_{i}^{\prime }\mathbf{M}_{x_{i}}\boldsymbol{%
\eta }_{i}=\mathbf{y}^{\prime }\mathbf{M}_{x_{i}}\mathbf{y}$, where $\mathbf{%
y}=(y_{1},y_{2},\cdots ,y_{T})^{\prime }$. By Assumption \ref{md}, we have 
\begin{equation*}
\mathbb{E}\left[ x_{it}^2-\mathbb{E}(x_{it}^2)|\mathcal{F}_{t-1} \right] =0,
\end{equation*}
for $i=1,2,\cdots ,N$. By Lemma \ref{md y}, we also have 
\begin{equation*}
\mathbb{E}\left[ y_{t}x_{it}-\mathbb{E}(y_{t}x_{it})|\mathcal{F}_{t-1}\right]
=0,
\end{equation*}%
for $i=1,2,\cdots ,N$, and 
\begin{equation*}
\mathbb{E}\left[ y_{t}^{2}-\mathbb{E}(y_{t}^{2})|\mathcal{F}_{t-1}\right] =0.
\end{equation*}%
Moreover, by Assumption \ref{subg}, for all $i = 1,2, \cdots, N $ and $%
\alpha > 0$, there exist some finite positive constants $C_0$ and $C_1 $
such that, 
\begin{equation*}
\sup_{t}\Pr (|x_{i t}|>\alpha )\leq C_{0}\exp (-C_{1}\alpha^{s}).
\end{equation*}
Also, by Lemma \ref{subg y}, there exist some finite positive constants $C_0$
and $C_1 $ such that, 
\begin{equation*}
\sup_{t}\Pr (\lvert y_{t}\rvert >\alpha )\leq C_{0}\exp (-C_{1}\alpha ^{s/2})%
\text{.}
\end{equation*}
Therefore by Lemma \ref{sum_martig_diff_prod_x}, we can conclude that there
exist sufficiently large positive constants $C_{0}$, $C_{1}$, and $C_{2}$
such that if $0<\lambda \leq (s+4)/(s+8)$, then 
\begin{equation*}
\Pr \left[ \lvert \boldsymbol{\eta }_{i}^{\prime }\mathbf{M}_{x_{i}}%
\boldsymbol{\eta }_{i}-\mathbb{E}(\boldsymbol{\eta }_{i}^{\prime }%
\boldsymbol{\eta }_{i})\rvert >\zeta _{T}\right] \leq \exp
(-C_{0}T^{-1}\zeta _{T}^{2})+\exp (-C_{1}T^{C_{2}}),
\end{equation*}%
and if $\lambda >(s+4)/(s+8)$, then 
\begin{equation*}
\Pr \left[ \lvert \boldsymbol{\eta }_{i}^{\prime }\mathbf{M}_{x_{i}}%
\boldsymbol{\eta }_{i}-\mathbb{E}(\boldsymbol{\eta }_{i}^{\prime }%
\boldsymbol{\eta }_{i})\rvert >\zeta _{T}\right] \leq \exp (-C_{0}\zeta
_{T}^{s/(s+1)})+\exp (-C_{1}T^{C_{2}}),
\end{equation*}%
for all $i=1,2,\cdots ,N$.
\end{proof}


\begin{lemma}
\label{t_test_bound} Let $y_{t}$ be a target variable generated by equation (%
\ref{dgp y_t}) and $x_{it}$ be a covariate in the active set, $\mathcal{S}%
_{Nt}=\{x_{1t},x_{2t},\cdots ,x_{Nt}\}$. Suppose Assumptions \ref{signal}-%
\ref{subg} hold and consider the projection regression of $y_{t}$ on $x_{it}$
as 
\begin{equation}
y_{t}=\phi _{i,T}x_{it}+\eta _{it},  \label{etait}
\end{equation}%
where $\phi _{i,T}$ is given in (\ref{phiiTdef}). Define, 
\begin{equation*}
t_{i,T}=\frac{T^{-1/2}\mathbf{x}_{i}^{\prime }\mathbf{y}}{\sqrt{T^{-1}%
\boldsymbol{\eta }_{i}^{\prime }\mathbf{M}_{x_{i}}\boldsymbol{\eta }_{i}}%
\sqrt{T^{-1}\mathbf{x}_{i}^{\prime }\mathbf{x}_{i}}},
\end{equation*}%
where $\mathbf{x}_{i}=(x_{i1},x_{i2},\cdots ,x_{iT})^{\prime }$, $\mathbf{y}%
=(y_{1},y_{2}\allowbreak ,\cdots ,y_{T})^{\prime }$, $\boldsymbol{\eta }%
_{i}=(\eta _{i1},\eta _{i2},\cdots ,\eta _{iT})^{\prime }$ and $\mathbf{M}%
_{x_{i}}=\mathbf{I}-T^{-1}\mathbf{x}_{i}(T^{-1}\mathbf{x}_{i}^{\prime }%
\mathbf{x}_{i})^{-1}\mathbf{x}_{i}^{\prime }$. Then, there exist
sufficiently large finite positive constants $C_{0}$ and $C_{1}$ such that
for any $0<\pi <1$, any null sequence $d_{T} >0 $, and $\epsilon _{i}\geq 
\frac{1}{2}$ 
\begin{equation*}
\Pr \left[ \lvert t_{i,T}\rvert >c_{p}(N,\delta )|\theta _{i,T}=\ominus
(T^{1-\epsilon _{i}})\right] \leq \exp \left[ -\frac{(1-\pi )^{2}\bar{\sigma}%
_{\eta _{i},T}^{2}\bar{\sigma}_{x_{i},T}^{2}c_{p}^{2}(N,\delta )}{2\bar{%
\omega}_{iy,T}^{2}(1+d_{T})^{2}}\right] +\exp (-C_{0}T^{C_{1}}),
\end{equation*}%
where $c_{p}(N,\delta )$ is defined by (\ref{cv_function}), $\theta _{i,T}=T%
\bar{\theta}_{i,T}=\mathbb{E}(\mathbf{x}_{i}^{\prime }\mathbf{y})$, $\bar{%
\sigma}_{\eta _{i},T}^{2}=T^{-1}\sum_{t=1}^{T}\mathbb{E}\left( \eta
_{it}^{2}\right) $, $\bar{\sigma}_{x_{i},T}^{2}=T^{-1}\sum_{t=1}^{T}\mathbb{E%
}\left( x_{it}^{2}\right) $ and $\bar{\omega}_{iy,T}^{2}=T^{-1}\sum_{t=1}^{T}%
\mathbb{E}\left( x_{it}^{2}y_{t}^{2}|\mathcal{F}_{t-1}\right) $. Also, if $%
c_{p}(N,\delta )=o(T^{1/2-\vartheta -c})$ for any $0\leq \vartheta_{i} <1/2$
and a finite positive constant $c$, then, there exist some finite positive
constants $C_{0}$ and $C_{1}$ such that 
\begin{equation*}
\textstyle\Pr \left[ \lvert t_{i,T}\rvert >c_{p}(N,\delta )|\theta
_{i,T}=\ominus (T^{1-\vartheta _{i}})\right] \geq 1-\exp (-C_{0}T^{C_{1}}).
\end{equation*}
\end{lemma}

\begin{proof}
We have $\lvert t_{i,T}\rvert =\mathcal{A}_{iT}\mathcal{B}_{iT}$, where, 
\begin{equation*}
\mathcal{A}_{iT}=\frac{|T^{-1/2}\mathbf{x}_{i}^{\prime }\mathbf{y}|}{\bar{%
\sigma} _{\eta _{i,T}} \bar{\sigma}_{x_{i,T}}},
\end{equation*}%
and 
\begin{equation*}
\mathcal{B}_{iT}=\frac{\bar{\sigma}_{\eta _{i},T}\bar{\sigma}_{x_{i},T}}{%
\sqrt{T^{-1}\boldsymbol{\eta }_{i}^{\prime }\mathbf{M}_{x_{i}}\boldsymbol{%
\eta }_{i}}\sqrt{T^{-1}\mathbf{x}_{i}^{\prime }\mathbf{x}_{i}}}.
\end{equation*}%
In the first case where $\theta _{i,T}=\ominus (T^{1-\epsilon _{i}})$ for
some $\epsilon _{i}\geq 1/2$, by using Lemma \ref{prob_product} we have 
\begin{equation*}
\begin{split}
\Pr \left[ \lvert t_{i,T}\rvert >c_{p}(n,\delta )|\theta _{i,T}=\ominus
(T^{1-\epsilon _{i}})\right] & \leq \Pr \left[ \mathcal{A}%
_{iT}>c_{p}(N,\delta )/(1+d_{T})|\theta _{i,T}=\ominus (T^{1-\epsilon _{i}})%
\right] + \\
& \quad \ \Pr \left[ \mathcal{B}_{iT}>1+d_{T}|\theta _{i,T}=\ominus
(T^{1-\epsilon _{i}})\right] \text{,}
\end{split}%
\end{equation*}%
where $d_{T}\rightarrow 0$ as $T\rightarrow \infty $. By using Lemma \ref%
{prob_inv_sqrt_rand}, 
\begin{equation*}
\begin{split}
& \Pr \left[ \mathcal{B}_{iT}>1+d_{T}|\theta _{i,T}=\ominus (T^{1-\epsilon
_{i}})\right] \\
& \qquad \leq \Pr \left( \lvert \frac{\bar{\sigma}_{\eta _{i},T}\bar{\sigma}%
_{x_{i},T}}{\sqrt{T^{-1}\boldsymbol{\eta }_{i}^{\prime }\mathbf{M}_{x_{i}}%
\boldsymbol{\eta }_{i}}\sqrt{T^{-1}\mathbf{x}_{i}^{\prime }\mathbf{x}_{i}}}%
-1\rvert >d_{T}|\theta _{i,T}=\ominus (T^{1-\epsilon _{i}})\right) \\
& \qquad \leq \Pr \left( \lvert \frac{(T^{-1}\boldsymbol{\eta }_{i}^{\prime }%
\mathbf{M}_{x_{i}}\boldsymbol{\eta }_{i})(T^{-1}\mathbf{x}_{i}^{\prime }%
\mathbf{x}_{i})}{\bar{\sigma}_{\eta _{i},T}^{2}\bar{\sigma}_{x_{i},T}^{2}}%
-1\rvert >d_{T}|\theta _{i,T}=\ominus (T^{1-\epsilon _{i}})\right) \\
& \textstyle\qquad =\Pr \left[ \mathcal{M}_{iT}+\mathcal{R}_{iT}+\mathcal{M}%
_{iT}\mathcal{R}_{iT}>d_{T}|\theta _{i,T}=\ominus (T^{1-\epsilon _{i}})%
\right]
\end{split}%
\end{equation*}%
where $\mathcal{M}_{iT}=|(T^{-1}\mathbf{x}_{i}^{\prime }\mathbf{x}_{i})/\bar{%
\sigma}_{x_{i},T}^{2}-1|$ and $\mathcal{R}_{iT}=|(T^{-1}\boldsymbol{\eta }%
_{i}^{\prime }\mathbf{M}_{x_{i}}\boldsymbol{\eta }_{i})/\bar{\sigma}_{\eta
_{i},T}^{2}-1|$. By using Lemmas \ref{prob_sum} and \ref{prob_product} , for
any values of $0<\pi _{i}<1$ with $\sum_{i=1}^{3}\pi _{i}=1$ and a strictly
positive constant, $c$, we have 
\begin{equation*}
\begin{split}
& \Pr \left[ \mathcal{B}_{iT}>1+d_{T}|\theta _{i,T}=\ominus (T^{1-\epsilon
_{i}})\right] \\
& \textstyle\qquad \leq \Pr \left[ \mathcal{M}_{iT}>\pi _{1}d_{T}|\theta
_{i,T}=\ominus (T^{1-\epsilon _{i}})\right] +\Pr \left[ \mathcal{R}_{iT}>\pi
_{2}d_{T}|\theta _{i,T}=\ominus (T^{1-\epsilon _{i}})\right] + \\
& \textstyle\qquad \quad \ \Pr \left[ \mathcal{M}_{iT}>\frac{\pi _{3}}{c}%
d_{T}|\theta _{i,T}=\ominus (T^{1-\epsilon _{i}})\right] +\Pr \left[ 
\mathcal{R}_{iT}>c|\theta _{i,T}=\ominus (T^{1-\epsilon _{i}})\right] .
\end{split}%
\end{equation*}%
First, consider $\Pr \left[ \mathcal{M}_{iT}>\pi _{1}d_{T}|\theta
_{i,T}=\ominus (T^{1-\epsilon _{i}})\right] $, and note that 
\begin{equation*}
\Pr \left[ \mathcal{M}_{iT}>\pi _{1}d_{T}|\theta _{i,T}=\ominus
(T^{1-\epsilon _{i}})\right] =\Pr \left[ |\mathbf{x}_{i}^{\prime }\mathbf{x}%
_{i}-\mathbb{E}(\mathbf{x}_{i}^{\prime }\mathbf{x}_{i})|>\pi _{1}\bar{\sigma}%
_{x_{i},T}^{2}Td_{T}|\theta _{i,T}=\ominus (T^{1-\epsilon _{i}})\right] .
\end{equation*}%
Therefore, by Lemma \ref{conditional_corr_x_i_x_j}, there exist some
constants $C_{0}$ and $C_{1}$ such that, 
\begin{equation*}
\Pr \left[ \mathcal{M}_{iT}>\pi _{1}d_{T}|\theta _{i,T}=\ominus
(T^{1-\epsilon _{i}})\right] \leq \exp (-C_{0}T^{C_{1}}).
\end{equation*}%
Similarly, 
\begin{equation*}
\textstyle\Pr \left[ \mathcal{M}_{iT}>\frac{\pi _{3}}{c}d_{T}|\theta
_{i,T}=\ominus (T^{1-\epsilon _{i}})\right] \leq \exp (-C_{0}T^{C_{1}}).
\end{equation*}%
Also note that 
\begin{equation*}
\Pr \left[ \mathcal{R}_{iT}>\pi _{2}d_{T}|\theta _{i,T}=\ominus
(T^{1-\epsilon _{i}})\right] =\Pr \left[ |\boldsymbol{\eta }_{i}^{\prime }%
\mathbf{M}_{x_{i}}\boldsymbol{\eta }_{i}-\mathbb{E}(\boldsymbol{\eta }%
_{i}^{\prime }\boldsymbol{\eta }_{i})|>\pi _{2}\bar{\sigma}_{\eta
_{i},T}^{2}Td_{T}|\theta _{i,T}=\ominus (T^{1-\epsilon _{i}})\right] .
\end{equation*}%
Therefore, by Lemma \ref{variance eta}, there exist some constants $C_{0}$
and $C_{1}$ such that, 
\begin{equation*}
\Pr \left[ \mathcal{R}_{iT}>\pi _{2}d_{T}|\theta _{i,T}=\ominus
(T^{1-\epsilon _{i}})\right] \leq \exp (-C_{0}T^{C_{1}}).
\end{equation*}%
Similarly, 
\begin{equation*}
\textstyle\Pr \left[ \mathcal{R}_{iT}>c|\theta _{i,T}=\ominus (T^{1-\epsilon
_{i}})\right] \leq \exp (-C_{0}T^{C_{1}}).
\end{equation*}%
Therefore, we can conclude that there exist some constants $C_{0}$ and $%
C_{1} $ such that, 
\begin{equation*}
\Pr \left[ \mathcal{B}_{iT}>1+d_{T}|\theta _{i,T}=\ominus (T^{1-\epsilon
_{i}})\right] \leq \exp (-C_{0}T^{C_{1}})
\end{equation*}%
Now consider $\Pr \left[ \mathcal{A}_{iT}>c_{p}(N,\delta )/(1+d_{T})|\theta
_{i,T}=\ominus (T^{1-\epsilon _{i}})\right] $, which is equal to 
\begin{eqnarray*}
&&\Pr \left( \frac{\left\vert \mathbf{x}_{i}^{\prime }\mathbf{y}-\theta
_{i,T}+\theta _{i,T}\right\vert }{\bar{\sigma}_{\eta _{i},T}\bar{\sigma}%
_{x_{i},T}}>T^{1/2}\frac{c_{p}(N,\delta )}{1+d_{T}}|\theta _{i,T}=\ominus
(T^{1-\epsilon _{i}})\right) \\
&&\qquad \ \leq \Pr \left( |\mathbf{x}_{i}^{\prime }\mathbf{y}-\theta
_{i,T}|>\frac{\bar{\sigma}_{\eta _{i},T}\bar{\sigma}_{x_{i},T}}{1+d_{T}}%
T^{1/2}c_{p}(N,\delta )-\lvert \theta _{i,T}\rvert \ |\theta _{i,T}=\ominus
(T^{1-\epsilon _{i}})\right) \text{.}
\end{eqnarray*}%
Note that since $\epsilon _{i}\geq 1/2$ and $c_{p}(N,\delta )\rightarrow
\infty $ as N and consequently T goes to infinity, the first term on the
right hand side of the inequality dominate the second one. Moreover, Since $%
c_{p}(N,\delta )=o(T^{\lambda })$ for all values of $\lambda >0$, by Lemma %
\ref{conditional_corr_xy}, for any $0<\pi <1$, 
\begin{equation*}
\textstyle\Pr \left[ |\mathbf{x}_{i}^{\prime }\mathbf{y}|>\frac{\bar{\sigma}%
_{\eta _{i},T}\bar{\sigma}_{x_{i},T}}{1+d_{T}}T^{1/2}c_{p}(N,\delta )|\theta
_{i,T}=\ominus (T^{1-\epsilon _{i}})\right] \leq \exp \left[ -\frac{(1-\pi
)^{2}\bar{\sigma}_{\eta _{i},T}^{2}\bar{\sigma}_{x_{i},T}^{2}c_{p}^{2}(N,%
\delta )}{2\bar{\omega}_{iy,T}^{2}(1+d_{T})^{2}}\right] .
\end{equation*}

Given the probability upper bound for $\mathcal{A}_{iT}$ and $\mathcal{B}%
_{iT}$, we can conclude that there exist some finite positive constants $%
C_{0}$ and $C_{1}$ such that for any $0 < \pi < 1$ 
\begin{equation*}
\Pr \left[ \lvert t_{i,T}\rvert >c_{p}(N,\delta )|\theta
_{i,T}=\ominus(T^{1-\epsilon _{i}})\right] \leq \exp \left[-\frac{%
(1-\pi)^{2} \bar{\sigma}_{\eta_{i},T}^{2} \bar{\sigma}_{x_{i},T}^{2}
c_{p}^{2}(N,\delta)}{2 \bar{\omega}_{iy, T}^{2} (1+d_{T})^{2}} \right] +
\exp (-C_{0}T^{C_{1}}).
\end{equation*}%
Let's consider the next case where $\theta _{i,T}=\ominus (T^{1-\vartheta
_{i}})$ for some $0\leq \vartheta _{i}<1/2$. We know that 
\begin{equation*}
\Pr \left[ \lvert t_{i,T}\rvert >c_{p}(N,\delta )|\theta _{i,T}=\ominus
(T^{1-\vartheta _{i}})\right] =1-\Pr \left[ \lvert t_{i,T} \rvert
<c_{p}(N,\delta )|\theta _{i,T}=\ominus (T^{1-\vartheta _{i}})\right] .
\end{equation*}%
By Lemma \ref{prob_product_two}, 
\begin{equation*}
\begin{split}
& \Pr \left[ \lvert t_{i,T}\rvert <c_{p}(N,\delta )|\theta _{i,T}=\ominus
(T^{1-\vartheta _{i}})\right] \leq \Pr \left[ \mathcal{A}_{iT}<\sqrt{1+d_{T}}%
c_{p}(N,\delta )|\theta _{i,T}=\ominus (T^{1-\vartheta _{i}})\right] + \\
& \qquad \qquad \qquad \qquad \qquad \qquad \ \Pr \left[ \mathcal{B}_{iT}<1/%
\sqrt{1+d_{T}}|\theta _{i,T}=\ominus (T^{1-\vartheta _{i}})\right] .
\end{split}%
\end{equation*}%
Since $\theta _{i,T}=\ominus (T^{1-\vartheta _{i}})$, for some $0\leq
\vartheta _{i}<1/2$ and $c_{p}(N,\delta )=o(T^{1/2-\vartheta -c})$, for any $%
0 \leq \vartheta < 1/2 $, $|\theta _{i,T}| - \bar{\sigma}_{\eta_{i},T} \bar{%
\sigma}_{x_{i},T}[(1+d_{T})T]^{1/2}c_{p}(N,\delta ) = \ominus
(T^{1-\vartheta _{i}})>0$ and by Lemma \ref{prob_rand_sum_constnt}, we have 
\begin{equation*}
\begin{split}
& \Pr \left[ \mathcal{A}_{iT}<\sqrt{1+d_{T}}c_{p}(N,\delta )|\theta_{i,T} =
\ominus (T^{1-\vartheta_{i}})\right] \\
& \qquad = \Pr \left[ \frac{|T^{-1/2}\mathbf{x}_{i}^{\prime}\mathbf{y}
-T^{-1/2}\theta_{i,T} + T^{-1/2} \theta_{i,T}|}{\bar{\sigma}_{\eta_{i},T} 
\bar{\sigma}_{x_{i},T}} < \sqrt{1+d_{T}} c_{p}(N,\delta) | \theta_{i,T} =
\ominus (T^{1-\vartheta _{i}}) \right] \\
& \qquad \leq \Pr \left[ |\mathbf{x}_{i}^{\prime }\mathbf{y}-\theta
_{i,T}|>|\theta_{i,T}|- \bar{\sigma}_{\eta_{i},T}\bar{\sigma}%
_{x_{i},T}[(1+d_{T})T]^{1/2}c_{p}(N,\delta )|\theta _{i,T}=\ominus
(T^{1-\vartheta _{i}})\right] .
\end{split}%
\end{equation*}%
Therefore, by Lemma \ref{conditional_corr_xy}, there exist some finite
positive constants $C_{0}$ and $C_{1}$ such that, 
\begin{equation*}
\Pr \left[ |\mathbf{x}_{i}^{\prime }\mathbf{y}-\theta _{i,T}|>|\theta
_{i,T}|-\bar{\sigma}_{\eta_{i},T} \bar{\sigma}%
_{x_{i},T}[(1+d_{T})T]^{1/2}c_{p}(N,\delta )|\theta _{i,T}=\ominus
(T^{1-\vartheta _{i}})\right] \leq \exp (-C_{0}T^{C_{1}}),
\end{equation*}%
and therefore 
\begin{equation*}
\Pr \left[ \mathcal{A}_{iT}<\sqrt{1+d_{T}}c_{p}(N,\delta )|\theta
_{i,T}=\ominus (T^{1-\vartheta _{i}})\right] \leq \exp (-C_{0}T^{C_{1}}).
\end{equation*}%
Now let consider the probability of $\mathcal{B}_{iT}$, 
\begin{equation*}
\begin{split}
& \Pr \left( \mathcal{B}_{iT}<1/\sqrt{1+d_{T}}|\theta _{i,T}=\ominus
(T^{1-\vartheta _{i}})\right) \\
& =\Pr \left( \frac{\bar{\sigma}_{\eta_{i},T} \bar{\sigma}_{x_{i},T}}{\sqrt{%
T^{-1} \boldsymbol{\eta }_{i}^{\prime }\mathbf{M}_{x_{i}}\boldsymbol{\eta }%
_{i}}\sqrt{T^{-1}\mathbf{x}_{i}^{\prime }\mathbf{x}_{i}}}<\frac{1}{\sqrt{%
1+d_{T}}}|\theta _{i,T}=\ominus (T^{1-\vartheta _{i}})\right) \\
& \qquad =\Pr \left( \frac{(T^{-1}\boldsymbol{\eta }_{i}^{\prime }\mathbf{M}%
_{x_{i}}\boldsymbol{\eta }_{i})(T^{-1}\mathbf{x}_{i}^{\prime }\mathbf{x}_{i})%
}{\bar{\sigma}_{\eta_{i},T}^{2}\bar{\sigma}_{x_{i},T}^{2}}>1+d_{T}|\theta
_{i,T}=\ominus (T^{1-\vartheta _{i}})\right) \\
& \textstyle\qquad \leq \Pr (\mathcal{M}_{iT}+\mathcal{R}_{iT}+\mathcal{M}%
_{iT}\mathcal{R}_{iT}>d_{T}|\theta _{i,T}=\ominus (T^{1-\vartheta _{i}}))%
\text{,}
\end{split}%
\end{equation*}%
where $\mathcal{M}_{iT}=|(T^{-1}\mathbf{x}_{i}^{\prime}\mathbf{x}_{i})/\bar{%
\sigma}_{x_{i},T}^{2}-1|$ and $\mathcal{R}_{iT}=|(T^{-1}\boldsymbol{\eta}%
_{i}^{\prime }\mathbf{M}_{x_{i}}\boldsymbol{\eta }_{i})/\bar{\sigma}_{\eta
_{i},T}^{2}-1|$. By using Lemmas \ref{prob_sum} and \ref{prob_product} , for
any values of $0<\pi _{i}<1$ with $\sum_{i=1}^{3}\pi _{i}=1$ and a positive
constant, $c$, we have 
\begin{equation*}
\begin{split}
& \Pr \left[ \mathcal{B}_{iT}<1/\sqrt{1+d_{T}}|\theta _{i,T}=\ominus
(T^{1-\vartheta _{i}})\right] \\
& \textstyle\qquad \leq \Pr \left[ \mathcal{M}_{iT}>\pi _{1}d_{T}|\theta
_{i,T}=\ominus (T^{1-\vartheta _{i}})\right] +\Pr \left[ \mathcal{R}%
_{iT}>\pi _{2}d_{T}|\theta _{i,T}=\ominus (T^{1-\vartheta _{i}})\right] + \\
& \textstyle\qquad \quad \ \Pr \left[ \mathcal{M}_{iT}>\frac{\pi _{3}}{c}%
d_{T}|\theta _{i,T}=\ominus (T^{1-\vartheta _{i}})\right] +\Pr \left[ 
\mathcal{R}_{iT}>c|\theta _{i,T}=\ominus (T^{1-\vartheta _{i}})\right] .
\end{split}%
\end{equation*}%
Let's first consider the $\Pr \left[ \mathcal{M}_{iT}>\pi _{1}d_{T}|\theta
_{i,T}=\ominus (T^{1-\vartheta _{i}})\right] $. Note that 
\begin{equation*}
\Pr \left[ \mathcal{M}_{iT}>\pi _{1}d_{T}|\theta _{i,T}=\ominus
(T^{1-\vartheta _{i}})\right] =\Pr \left[ |\mathbf{x}_{i}^{\prime }\mathbf{x}%
_{i}-\mathbb{E}(\boldsymbol{x}_{i}^{\prime }\boldsymbol{x}_{i})|>\pi _{1}%
\bar{\sigma}_{x_{i},T}^{2}Td_{T}|\theta _{i,T}=\ominus (T^{1-\vartheta _{i}})%
\right] .
\end{equation*}%
So, by Lemma \ref{conditional_corr_x_i_x_j}, we know that there exist some
constants $C_{0}$ and $C_{1}$ such that, 
\begin{equation*}
\Pr \left[ \mathcal{M}_{iT}>\pi _{1}d_{T}|\theta _{i,T}=\ominus
(T^{1-\vartheta _{i}})\right] \leq \exp (-C_{0}T^{C_{1}}).
\end{equation*}%
Similarly, 
\begin{equation*}
\textstyle\Pr \left[ \mathcal{M}_{iT}>\frac{\pi _{3}}{c}d_{T}|\theta
_{i,T}=\ominus (T^{1-\vartheta _{i}})\right] \leq \exp (-C_{0}T^{C_{1}}).
\end{equation*}%
Also note that 
\begin{equation*}
\Pr \left[ \mathcal{R}_{iT}>\pi _{2}d_{T}|\theta _{i,T}=\ominus
(T^{1-\vartheta _{i}})\right] =\Pr \left[ |\boldsymbol{\eta }_{i}^{\prime }%
\mathbf{M}_{x_{i}}\boldsymbol{\eta }_{i}-\mathbb{E}(\boldsymbol{\eta }%
_{i}^{\prime }\boldsymbol{\eta }_{i})|>\pi _{2}\bar{\sigma}%
_{\eta_{i},T}^{2}Td_{T}|\theta _{i,T}=\ominus (T^{1-\vartheta _{i}})\right] .
\end{equation*}%
Therefore, by Lemma \ref{variance eta}, there exist some constants $C_{0}$
and $C_{1}$ such that, 
\begin{equation*}
\Pr (\mathcal{R}_{iT}>\pi _{2}d_{T}|\theta _{i,T} = \ominus (T^{1-\vartheta
_{i}}))\leq \exp (-C_{0}T^{C_{1}}).
\end{equation*}%
Similarly, 
\begin{equation*}
\textstyle\Pr (\mathcal{R}_{iT}>c|\theta _{i,T} =\ominus (T^{1-\vartheta
_{i}}))\leq \exp (-C_{0}T^{C_{1}}).
\end{equation*}%
Therefore, we can conclude that there exist some constants $C_{0}$ and $%
C_{1} $ such that, 
\begin{equation*}
\Pr \left[ \mathcal{B}_{iT}<1/\sqrt{1+d_{T}}|\theta _{i,T}=\ominus
(T^{1-\vartheta _{i}})\right] \leq \exp (-C_{0}T^{C_{1}}).
\end{equation*}

So, overall we conclude that 
\begin{equation*}
\begin{split}
& \Pr \left[ \lvert t_{i,T}\rvert >c_{p}(N,\delta )|\theta _{i,T}=\ominus
(T^{1-\vartheta _{i}})\right] \\
& \qquad \qquad =1-\Pr \left[ \lvert t_{i,T} \rvert <c_{p}(N,\delta )|\theta
_{i,T}=\ominus (T^{1-\vartheta _{i}})\right] \geq 1-\exp (-C_{0}T^{C_{1}}).
\end{split}%
\end{equation*}
\end{proof}



\begin{lemma}
\label{reg coef} Suppose $y_{t}$ are generated by 
\begin{equation}
y_{t}=\sum_{i=1}^{k}x_{it}\beta _{it}+u_{t}\text{ for }t=1,2,\cdots ,T,
\label{dgp yt supp}
\end{equation}%
and consider the LS estimator of the following regression augmented with the
additional $l_{T}$ regressors from the active set: 
\begin{equation*}
y_{t}=\mathbf{x}_{kt}^{\prime }\boldsymbol{\phi }+\mathbf{s}_{t}^{\prime }%
\boldsymbol{\delta }_{T}+\eta _{t},
\end{equation*}%
where $\mathbf{x}_{kt}=(x_{1t},x_{2t},\cdots ,x_{kt})^{\prime }$, is the $%
k\times 1$ vector of signals, $\mathbf{s}_{t}$ is the $l_{T}\times 1$ vector
of additional regressors, $\boldsymbol{\phi }=(\phi _{1},\phi _{2},\cdots
,\phi _{k})^{\prime }$ and $\boldsymbol{\delta }=(\delta _{1},\delta
_{2},\cdots ,\delta _{l_{T}})^{\prime }$ are the associated coefficients.
The LS estimator of $\boldsymbol{\ \gamma }_{T}=(\boldsymbol{\phi }^{\prime
},\boldsymbol{\delta }_{T}^{\prime })^{\prime }$ is 
\begin{equation}
\hat{\boldsymbol{\gamma }}_{T}=\left( T^{-1}\mathbf{W}^{\prime }\mathbf{W}%
\right) ^{-1}\left( T^{-1}\mathbf{W}^{\prime }\mathbf{y}\right),
\label{reg ols est}
\end{equation}%
where $\mathbf{W}=(\mathbf{w}_{1},\mathbf{w}_{2},\cdots ,\mathbf{w}%
_{T})^{\prime }$, $\mathbf{w}_{t}=\left( \mathbf{x}_{kt}^{\prime },\mathbf{s}%
_{t}^{\prime }\right) ^{\prime }$ and $\mathbf{y}=(y_{1},y_{2},\cdots
,y_{T})^{\prime }$. The model error is 
\begin{equation}
\boldsymbol{\hat{\eta}}=\mathbf{y}-\mathbf{W}\boldsymbol{\hat{\gamma}}_{T}.
\label{model error}
\end{equation}%
Suppose that $\lambda _{\min }\left[ T^{-1}\mathbb{E}(\mathbf{W^{\prime }}%
\mathbf{W})\right] >c>0$, and $l_{T}=\ominus (T^{d})$, where $0\leq d<\frac{1%
}{2}$. Moreover suppose that Assumptions \ref{signal}-\ref{subg} holds. Now,

\begin{enumerate}
\item[(i)] If $\mathbb{E}(\beta _{it})=\beta _{i}$ for all $t$, then 
\begin{equation*}
\left\Vert \hat{\boldsymbol{\gamma }}_{T}-\boldsymbol{\gamma }_{T}^{\ast
}\right\Vert =O_{p}\left( T^{-\frac{1- d}{2}}\right) ,
\label{ols_gamma_consistency1}
\end{equation*}%
where $\boldsymbol{\gamma }_{T}^{\ast }=(\boldsymbol{\beta }^{\prime },%
\mathbf{\ 0}_{l_{T}}^{\prime })^{\prime }$ and $\boldsymbol{\beta }=(\beta
_{1},\beta _{2},\cdots ,\beta _{k})^{\prime }$. Under Assumption \ref{weak
time dependence} we also have 
\begin{equation*}
T^{-1}\boldsymbol{\hat{\eta}}^{\prime }\boldsymbol{\hat{\eta}}=\bar{\sigma}%
_{u,T}^{2}+\bar{\Delta}_{\beta,T}+O_{p}\left( \frac{1}{\sqrt{T}}\right)
+O_{p}\left( T^{-(1-d)}\right) ,
\end{equation*}%
where $\bar{\sigma}_{u,T}^{2}=T^{-1}\sum_{t=1}^{T}\mathbb{E}%
\left(u_t^2\right) $, and $\bar{\Delta}_{\beta,T} = T^{-1} \sum_{t=1}^{T} 
\text{tr} \left( \boldsymbol{\Sigma}_{\mathbf{x}_k,t} \boldsymbol{\Omega}%
_{\beta,t} \right)$ are non-negative, with $\boldsymbol{\Sigma}_{\mathbf{x}%
_k,t} \equiv \left(\sigma_{ijt,x}\right)$, $\boldsymbol{\Omega}_{\beta,t}
\equiv \left(\sigma_{ijt,\beta}\right)$ for $i,j = 1,2, \cdots, k$, and $%
\sigma_{ijt,x} = \mathbb{E}\left( x_{it}x_{jt}\right)$, $\sigma_{ijt,\beta }=%
\mathbb{E}\left[(\beta_{it} - \beta_{i})(\beta_{jt} - \beta_{j})\right]$.

\item[(ii)] If $\mathbb{E}\left( \mathbf{w}_{t}\mathbf{w}_{t}^{\prime
}\right) $ is time-invariant, then 
\begin{equation*}
\left\Vert \hat{\boldsymbol{\gamma }}_{T}-\boldsymbol{\gamma }_{T}^{\diamond
}\right\Vert =O_{p}\left( T^{-\frac{1-d}{2}}\right) ,
\label{ols_gamma_consistency2}
\end{equation*}%
where $\boldsymbol{\gamma }_{T}^{\diamond }=(\boldsymbol{\bar{\beta}}%
_{T}^{\prime },\mathbf{0}_{l_{T}}^{\prime })^{\prime }$, $\boldsymbol{\bar{%
\beta}}_{T}=(\bar{\beta}_{1T},\bar{\beta}_{2T},\cdots ,\bar{\beta}%
_{kT})^{\prime }$, and $\bar{\beta}_{iT}=T^{-1}\sum_{t=1}^{T}\mathbb{E}%
(\beta _{it})$. If Assumption \ref{weak time dependence} also holds, then 
\begin{equation*}
T^{-1}\boldsymbol{\hat{\eta}}^{\prime }\boldsymbol{\hat{\eta}}= \bar{\sigma}%
_{u,T}^{2}+\bar{\Delta}_{\beta,T}^{\ast}+ O_{p}\left( \frac{1}{\sqrt{T}}%
\right) +O_{p}\left( T^{-(1-d)} \right) ,
\end{equation*}%
where $\bar{\Delta}_{\beta,T}^{\ast} = T^{-1} \sum_{t=1}^{T} \text{tr}
\left( \boldsymbol{\Sigma}_{\mathbf{x}_k,t} \boldsymbol{\Omega}%
_{\beta,t}^{\ast} \right)$ is non-negative, with $\boldsymbol{\Omega}%
_{\beta,t}^{\ast} \equiv \left(\sigma_{ijt,\beta}^{\ast}\right)$ for $i,j =
1,2, \cdots, k$, and $\sigma_{ijt,\beta }^{\ast}=\mathbb{E}\left[ (\beta
_{it}-\bar{\beta}_{i,T})(\beta _{jt}-\bar{\beta}_{j,T})\right]$.
\end{enumerate}
\end{lemma}

\begin{proof}
In the first scenario, where $\mathbb{E}(\beta _{it})=\beta _{i}$ for all $t$
, we can write (\ref{dgp yt supp}) as 
\begin{equation*}
y_{t}=\sum_{i=1}^{k}x_{it}\beta _{i}+\sum_{i=1}^{k}x_{it}\left( \beta
_{it}-\beta _{i}\right) +u_{t}=\sum_{i=1}^{k}x_{it}\beta
_{i}+\sum_{i=1}^{k}r_{it}+u_{t}=\mathbf{x}_{kt}^{\prime }\boldsymbol{\beta }+%
\mathbf{r}_{t}^{\prime }\boldsymbol{\tau }+u_{t},
\end{equation*}%
where $r_{it}=x_{it}\left( \beta _{it}-\beta _{i}\right) $, $\mathbf{r}%
_{t}=(r_{1t},r_{2t},\cdots ,r_{kt})^{\prime }$, and $\boldsymbol{\tau }$ is
a $k\times 1$ vector of ones. We can further write the DGP in a following
matrix format, 
\begin{equation}
\mathbf{y}=\mathbf{X}_{k}\boldsymbol{\beta }+\mathbf{R}\boldsymbol{\tau}+%
\mathbf{u},  \label{dgp matrix format}
\end{equation}%
where $\mathbf{X}_{k}=(\mathbf{x}_{k1},\mathbf{x}_{k2},\cdots ,\mathbf{x}%
_{kT})^{\prime }$, $\mathbf{R}=(\mathbf{r}_{1},\mathbf{r}_{2},\cdots ,%
\mathbf{\ \ r}_{T})^{\prime }$ and $\mathbf{u}=(u_{1},u_{2},\cdots
,u_{T})^{\prime }$. By substituting (\ref{dgp matrix format}) into (\ref{reg
ols est}), we obtain 
\begin{equation*}
\hat{\boldsymbol{\gamma }}_{T}=\left( T^{-1}\mathbf{W}^{\prime }\mathbf{W}%
\right) ^{-1}\left( T^{-1}\mathbf{W}^{\prime }\mathbf{X}_{k}\boldsymbol{%
\beta }\right) +\left( T^{-1}\mathbf{W}^{\prime }\mathbf{W}\right)
^{-1}\left( T^{-1}\mathbf{W}^{\prime }\mathbf{R}\boldsymbol{\tau }\right)
+\left( T^{-1}\mathbf{W}^{\prime }\mathbf{W}\right) ^{-1}\left( T^{-1}%
\mathbf{W}^{\prime }\mathbf{u}\right) ,  \label{expanded ols}
\end{equation*}%
where $\mathbf{W=(X}_{k}\mathbf{,S)}$, and $\mathbf{S}=(\mathbf{s}_{1},%
\mathbf{s}_{2},\cdots ,\mathbf{s}_{T})^{\prime }$. Since $\boldsymbol{\gamma 
}_{T}^{\ast }=(\boldsymbol{\beta }^{\prime },\mathbf{0}_{l_{T}}^{\prime
})^{\prime }$, $\mathbf{X}_{k}\boldsymbol{\beta }=\mathbf{X}_{k}\boldsymbol{%
\beta }+\mathbf{S}\mathbf{0}_{l_{T}}=\mathbf{W}\boldsymbol{\gamma }%
_{T}^{\ast }$, which in turn allows us to write the above result as: 
\begin{equation*}
\hat{\boldsymbol{\gamma }}_{T}=\left( T^{-1}\mathbf{W}^{\prime }\mathbf{W}%
\right) ^{-1}\left( T^{-1}\mathbf{W}^{\prime }\mathbf{W}\right) \boldsymbol{%
\gamma }_{T}^{\ast }+\left( T^{-1}\mathbf{W}^{\prime }\mathbf{W}\right)
^{-1}\left( T^{-1}\mathbf{W}^{\prime }\mathbf{R}\boldsymbol{\tau }\right)
+\left( T^{-1}\mathbf{W}^{\prime }\mathbf{W}\right) ^{-1}\left( T^{-1}%
\mathbf{W}^{\prime }\mathbf{u}\right) ,
\end{equation*}%
and hence 
\begin{equation}
\hat{\boldsymbol{\gamma }}_{T}-\boldsymbol{\gamma }_{T}^{\ast }=\left( T^{-1}%
\mathbf{W}^{\prime }\mathbf{W}\right) ^{-1}\left( T^{-1}\mathbf{W}^{\prime }%
\mathbf{R}\boldsymbol{\tau }\right) +\left( T^{-1}\mathbf{W}^{\prime }%
\mathbf{W}\right) ^{-1}\left( T^{-1}\mathbf{W}^{\prime }\mathbf{u}\right) .
\label{gamma hat minus gamma star}
\end{equation}%
We can further write 
\begin{align*}
\hat{\boldsymbol{\gamma }}_{T}-\boldsymbol{\gamma }_{T}^{\ast }=& \left\{
\left( T^{-1}\mathbf{W}^{\prime }\mathbf{W}\right) ^{-1}-\left[ \mathbb{E}%
\left( T^{-1}\mathbf{W}^{\prime }\mathbf{W}\right) \right] ^{-1}\right\}
\left( T^{-1} \mathbf{W}^{\prime }\mathbf{R}\boldsymbol{\tau } \right) + \\
& \left[ \mathbb{E}\left( T^{-1}\mathbf{W}^{\prime }\mathbf{W}\right) \right]
^{-1}\left(T^{-1} \mathbf{W}^{\prime }\mathbf{R}\boldsymbol{\ \tau }\right) +
\\
& \left\{ \left( T^{-1}\mathbf{W}^{\prime }\mathbf{W}\right) ^{-1}-\left[ 
\mathbb{E}\left( T^{-1}\mathbf{W}^{\prime }\mathbf{W}\right) \right]
^{-1}\right\} \left\{ T^{-1}\left[ \left( \mathbf{W}^{\prime }\mathbf{u}%
\right) -\mathbb{E}\left( \mathbf{W}^{\prime }\mathbf{u}\right) \right]
\right\} + \\
& \left\{ \left( T^{-1}\mathbf{W}^{\prime }\mathbf{W}\right) ^{-1}-\left[ 
\mathbb{E}\left( T^{-1}\mathbf{W}^{\prime }\mathbf{W}\right) \right]
^{-1}\right\} \left[ T^{-1}\mathbb{E}\left( \mathbf{W}^{\prime }\mathbf{u}%
\right) \right] + \\
& \left[ \mathbb{E}\left( T^{-1}\mathbf{W}^{\prime }\mathbf{W}\right) \right]
^{-1}\left\{ T^{-1}\left[ \left( \mathbf{W}^{\prime }\mathbf{u}\right) -%
\mathbb{E}\left( \mathbf{W}^{\prime }\mathbf{u}\right) \right] \right\} + \\
& \left[ \mathbb{E}\left( T^{-1}\mathbf{W}^{\prime }\mathbf{W}\right) \right]
^{-1}\left[ T^{-1} \mathbb{E}\left( \mathbf{W}^{\prime }\mathbf{u}\right) %
\right]. \\
\end{align*}%
Hence, by the sub-additive property of norms and Lemma \ref%
{F_norm_submultiplicative}, we have 
\begin{align*}
\left\Vert \hat{\boldsymbol{\gamma }}_{T}-\boldsymbol{\gamma }_{T}^{\ast
}\right\Vert \leq & \left\Vert \left( T^{-1}\mathbf{W}^{\prime }\mathbf{W}%
\right) ^{-1}-\left[ \mathbb{E}\left( T^{-1}\mathbf{W}^{\prime }\mathbf{W}%
\right) \right] ^{-1}\right\Vert _{F}\left\Vert T^{-1} \mathbf{W}^{\prime }%
\mathbf{R}\boldsymbol{\tau } \right\Vert + \\
& \left\Vert \left[ \mathbb{E}\left( T^{-1}\mathbf{W}^{\prime }\mathbf{W}%
\right) \right] ^{-1}\right\Vert _{2}\left\Vert T^{-1}\mathbf{W}^{\prime }%
\mathbf{R}\boldsymbol{\tau } \right\Vert + \\
& \left\Vert \left( T^{-1}\mathbf{W}^{\prime }\mathbf{W}\right) ^{-1}-\left[ 
\mathbb{E}\left( T^{-1}\mathbf{W}^{\prime }\mathbf{W}\right) \right]
^{-1}\right\Vert _{F}\left\Vert T^{-1}\left[ \left( \mathbf{W}^{\prime }%
\mathbf{u}\right) -\mathbb{E}\left( \mathbf{W}^{\prime }\mathbf{u}\right) %
\right] \right\Vert _{}+ \\
& \left\Vert \left( T^{-1}\mathbf{W}^{\prime }\mathbf{W}\right) ^{-1}-\left[ 
\mathbb{E}\left( T^{-1}\mathbf{W}^{\prime }\mathbf{W}\right) \right]
^{-1}\right\Vert _{F}\left\Vert T^{-1}\mathbb{E}\left( \mathbf{W}^{\prime }%
\mathbf{u}\right) \right\Vert + \\
& \left\Vert \left[ \mathbb{E}\left( T^{-1}\mathbf{W}^{\prime }\mathbf{W}%
\right) \right] ^{-1}\right\Vert _{2}\left\Vert T^{-1}\left[ \left( \mathbf{%
\ W}^{\prime }\mathbf{u}\right) -\mathbb{E}\left( \mathbf{W}^{\prime }%
\mathbf{\ u}\right) \right] \right\Vert + \\
& \left\Vert \left[ \mathbb{E}\left( T^{-1}\mathbf{W}^{\prime }\mathbf{W}%
\right) \right] ^{-1}\right\Vert_{2} \left\Vert T^{-1}\mathbb{E}\left( 
\mathbf{W}^{\prime }\mathbf{\ u}\right) \right\Vert
\end{align*}%
%
%
%
%
%
%
By Assumption \ref{md} 
\begin{equation*}
\left \Vert T^{-1}\mathbb{E}\left( \mathbf{W}^{\prime }\mathbf{u}\right)
\right\Vert = \left \Vert T^{-1}\sum_{t=1}^{T}\mathbb{E}(\mathbf{w}%
_{t}u_{t}) \right \Vert = O\left(T^{-\frac{2\epsilon-d}{2} }\right),
\end{equation*}%
where $\epsilon \geq 1/2 $.

Assumptions \ref{md} and \ref{subg} imply that $\mathbf{W}$ and $\mathbf{u}$
satisfy condition (i) and (ii) of Lemma \ref{bound on expected sample cov
dev}, and by Lemmas \ref{bound on expected sample cov dev} and \ref{Op
sample cov dev}, 
\begin{equation*}
\left\Vert T^{-1}\left[\mathbf{W}^{\prime }\mathbf{u} -\mathbb{E}\left( 
\mathbf{W}^{\prime }\mathbf{\ u}\right)\right] \right\Vert =O_{p}\left( T^{-%
\frac{1-d}{2}}\right).
\end{equation*}%
Similarly, 
\begin{equation*}
\left\Vert T^{-1}\left[ \left( \mathbf{W}^{\prime }\mathbf{W}\right) -%
\mathbb{E}\left( \mathbf{W}^{\prime }\mathbf{W}\right) \right] \right\Vert
_{F}=O_{p}\left(T^{-(1/2-d)}\right) ,
\end{equation*}%
and since $l_{T}=\ominus (T^{d})$ with $0\leq d<1/2$, by Lemma \ref{Op inv
sample cov dev}, 
\begin{equation*}
\left\Vert \left( T^{-1}\mathbf{W}^{\prime }\mathbf{W}\right) ^{-1}-\left[ 
\mathbb{E}\left( T^{-1}\mathbf{W}^{\prime }\mathbf{W}\right) \right]
^{-1}\right\Vert _{F}=O_{p}\left(T^{-(1/2-d)}\right).
\end{equation*}%
Now consider $\left\Vert T^{-1}\mathbf{W}^{\prime }\mathbf{R}\boldsymbol{\
\tau }\right\Vert $. Note that the row $j$ and column $i$ of $l_{T}\times p$
matrix $T^{-1}\mathbf{W}^{\prime }\mathbf{R}$ is equal to $%
T^{-1}\sum_{t=1}^{T}w_{jt}r_{it}$. Hence the $j^{\text{th}}$ element of $%
l_{T}\times 1$ vector $T^{-1}\mathbf{W}^{\prime }\mathbf{R}\boldsymbol{\tau }
$ is equal $T^{-1}\sum_{i=1}^{k}\sum_{t=1}^{T}w_{jt}r_{it}$. In other words, 
$T^{-1}\mathbf{W}^{\prime }\mathbf{R}\boldsymbol{\tau }=T^{-1}\sum_{i=1}^{k}%
\sum_{t=1}^{T}\mathbf{w}_{t}r_{it}$. Therefore, (recalling that $%
r_{it}=x_{it}\left( \beta _{it}-\beta _{i}\right) $) 
\begin{align*}
\left\Vert T^{-1}\mathbf{W}^{\prime }\mathbf{R}\boldsymbol{\tau }\right\Vert
^{2}=& \left\Vert T^{-1}\sum_{i=1}^{k}\sum_{t=1}^{T}(\mathbf{w}%
_{t}r_{it})\right\Vert ^{2}\leq \sum_{i=1}^{k}\left\Vert T^{-1}\sum_{t=1}^{T}%
\mathbf{w}_{t}x_{it}\left( \beta _{it}-\beta _{i}\right) \right\Vert ^{2} \\
=& T^{-2}\sum_{i=1}^{k}\sum_{t=1}^{T}\sum_{t^{\prime }=1}^{T}\mathbf{w}%
_{t}^{\prime }\mathbf{w}_{t^{\prime }}x_{it}x_{it^{\prime }}\left( \beta
_{it}-\beta _{i}\right) \left( \beta _{it^{\prime }}-\beta _{i}\right) \\
=& T^{-2}\sum_{i=1}^{k}\sum_{t=1}^{T}\sum_{t^{\prime }=1}^{T}\sum_{\ell
=1}^{k+l_{T}}w_{\ell t}w_{\ell t^{\prime }}x_{it}x_{it^{\prime }}\left(
\beta _{it}-\beta _{i}\right) \left( \beta _{it^{\prime }}-\beta _{i}\right)
. \\
&
\end{align*}%
Since, by Assumption \ref{md}, $\beta _{it}$ for $i=1,2,\cdots ,k$ are
distributed independently of $\mathbf{w}_{t}$ for $t=1,2,\cdots ,T$, we can
further write, 
\begin{align*}
& \mathbb{E}\left\Vert T^{-1}\mathbf{W}^{\prime }\mathbf{R}\boldsymbol{\tau }%
\right\Vert ^{2}\leq T^{-2}\sum_{i=1}^{k}\sum_{t=1}^{T}\sum_{t^{\prime
}=1}^{T}\sum_{\ell =1}^{k+\ell _{T}}\mathbb{E}\left( w_{\ell t}w_{\ell
t^{\prime }}x_{it}x_{it^{\prime }}\right) \mathbb{E}\left[ \left( \beta
_{it}-\beta _{i}\right) \left( \beta _{it^{\prime }}-\beta _{i}\right) %
\right] \\
& \qquad \leq T^{-2}\sum_{i=1}^{k}\sum_{t=1}^{T}\sum_{t^{\prime
}=1}^{T}\sum_{\ell =1}^{k+\ell _{T}}\left\vert \mathbb{E}\left( w_{\ell
t}w_{\ell t^{\prime }}x_{it}x_{it^{\prime }}\right) \right\vert \times
\left\vert \mathbb{E}\left[ \left( \beta _{it}-\beta _{i}\right) \left(
\beta _{it^{\prime }}-\beta _{i}\right) \right] \right\vert \\
& \qquad \leq T^{-2}\left( k+\ell _{T}\right) sup_{i,\ell ,t,t^{\prime
}}\left\vert \mathbb{E}\left( w_{\ell t}w_{\ell t^{\prime
}}x_{it}x_{it^{\prime }}\right) \right\vert
\sum_{i=1}^{k}\sum_{t=1}^{T}\sum_{t^{\prime }=1}^{T}\left\vert \mathbb{E}%
\left[ \left( \beta _{it}-\beta _{i}\right) \left( \beta _{it^{\prime
}}-\beta _{i}\right) \right] \right\vert
\end{align*}%
Since $\mathbf{W}$ satisfy condition (i) of Lemma \ref{bound on expected
sample cov dev}, we have $\sup_{i,\ell ,t,t^{\prime }}\left\vert \mathbb{E}%
(w_{\ell t}w_{\ell t^{\prime }}x_{it}x_{it^{\prime }})\right\vert <C<\infty $%
. Also, note that for any $t^{\prime }<t$, 
\begin{equation*}
\mathbb{E}\left[ \left( \beta _{it}-\beta _{i}\right) \left( \beta
_{it^{\prime }}-\beta _{i}\right) \right] =\mathbb{E}\left[ \left( \beta
_{it^{\prime }}-\beta _{i}\right) \mathbb{E}\left( \beta _{it}-\beta _{i}|%
\mathcal{F}_{t-1}\right) \right] ,
\end{equation*}%
and by Assumption \ref{md}, $\mathbb{E}\left( \beta _{it}-\beta _{i}|%
\mathcal{F}_{t-1}\right) =0$. Therefore, 
\begin{align*}
\sum_{t=1}^{T}\sum_{t^{\prime }=1}^{T}\left\vert \mathbb{E}\left[ \left(
\beta _{it}-\beta _{i}\right) \left( \beta _{it^{\prime }}-\beta _{i}\right) %
\right] \right\vert & =\sum_{t=1}^{T}\left\vert \mathbb{E}\left[ \left(
\beta _{it}-\beta _{i}\right) ^{2}\right] \right\vert
+2\sum_{t=2}^{T}\sum_{t^{\prime }=1}^{t}\left\vert \mathbb{E}\left[ \left(
\beta _{it}-\beta _{i}\right) \left( \beta _{it^{\prime }}-\beta _{i}\right) %
\right] \right\vert \\
& =\sum_{t=1}^{T}\left\vert \mathbb{E}\left[ \left( \beta _{it}-\beta
_{i}\right) ^{2}\right] \right\vert =O(T).
\end{align*}%
Since, by Assumption \ref{signal}, $k$ is also a finite fixed integer, we
conclude that 
\begin{equation*}
\mathbb{E}\left\Vert T^{-1}\mathbf{W}^{\prime }\mathbf{R}\boldsymbol{\tau }%
\right\Vert ^{2}=O\left( T^{-(1-d)}\right) ,
\end{equation*}%
and hence, by Lemma \ref{Op sample cov dev}, 
\begin{equation*}
\left\Vert T^{-1}\mathbf{W}^{\prime }\mathbf{R}\boldsymbol{\tau }\right\Vert
=O_{p}\left( T^{-\frac{1-d}{2}}\right) .
\end{equation*}%
So, we can conclude that 
\begin{equation*}
\left\Vert \hat{\boldsymbol{\gamma }}_{T}-\boldsymbol{\gamma }_{T}^{\ast
}\right\Vert =O_{p}\left( T^{-\frac{1-d}{2}}\right) ,
\end{equation*}%
as required.

In the next step, consider the mean squared errors of the model, $T^{-1}%
\boldsymbol{\hat{\eta}}_{T}^{\prime }\boldsymbol{\hat{\eta}}_{T}$. By
substituting $y$ from (\ref{dgp matrix format}) into equation (\ref{model
error}) for the model error, we have 
\begin{equation*}
\boldsymbol{\hat{\eta}}=\mathbf{y}-\mathbf{W}\boldsymbol{\hat{\gamma}}_{T}=%
\mathbf{X}_{k}\boldsymbol{\beta }+\mathbf{R}\boldsymbol{\tau }+\mathbf{u}-%
\mathbf{W}\boldsymbol{\hat{\gamma}}_{T}.
\end{equation*}%
Since $\mathbf{X}_{k}\boldsymbol{\beta }=\mathbf{W}\boldsymbol{\gamma }%
_{T}^{\ast }$, where $\boldsymbol{\gamma }_{T}^{\ast }=(\boldsymbol{\beta }%
^{\prime },\mathbf{0}_{l_{T}}^{\prime })^{\prime }$, we can further write, 
\begin{equation*}
\boldsymbol{\hat{\eta}}=\mathbf{R}\boldsymbol{\tau }+\mathbf{u}-\mathbf{\ W}%
\left( \boldsymbol{\hat{\gamma}}_{T}-\boldsymbol{\gamma }_{T}^{\ast }\right)
.
\end{equation*}%
Therefore, 
\begin{align*}
T^{-1}\boldsymbol{\hat{\eta}}^{\prime }\boldsymbol{\hat{\eta}}& =T^{-1}\left[
\mathbf{R}\boldsymbol{\tau }+\mathbf{u}-\mathbf{W}\left( \boldsymbol{\ \ 
\hat{\gamma}}_{T}-\boldsymbol{\gamma }_{T}^{\ast }\right) \right] ^{\prime }%
\left[ \mathbf{R}\boldsymbol{\tau }+\mathbf{u}-\mathbf{W}\left( \boldsymbol{%
\hat{\gamma}}_{T}-\boldsymbol{\gamma }_{T}^{\ast }\right) \right] \\
& =T^{-1}\left( \mathbf{R}\boldsymbol{\tau }+\mathbf{u}\right) ^{\prime
}\left( \mathbf{R}\boldsymbol{\tau }+\mathbf{u}\right) +T^{-1}\left[ \mathbf{%
W}\left( \boldsymbol{\hat{\gamma}}_{T}-\boldsymbol{\gamma }_{T}^{\ast
}\right) \right] ^{\prime }\left[ \mathbf{W}\left( \boldsymbol{\ \hat{\gamma}%
}_{T}-\boldsymbol{\gamma }_{T}^{\ast }\right) \right] - \\
& \qquad 2T^{-1}\left[ \mathbf{W}\left( \boldsymbol{\hat{\gamma}}_{T}-%
\boldsymbol{\gamma }_{T}^{\ast }\right) \right] ^{\prime }\left( \mathbf{R}%
\boldsymbol{\tau }+\mathbf{u}\right) \\
& =T^{-1}\left( \boldsymbol{\tau }^{\prime }\mathbf{R}^{\prime }\mathbf{R}%
\boldsymbol{\tau }+\mathbf{u}^{\prime }\mathbf{u}\right) +2T^{-1}\boldsymbol{%
\tau }^{\prime }\mathbf{R}^{\prime }\mathbf{u}+\left( \boldsymbol{\hat{\gamma%
}}_{T}-\boldsymbol{\gamma }_{T}^{\ast }\right) ^{\prime }\left( T^{-1}%
\mathbf{W}^{\prime }\mathbf{W}\right) \left( \boldsymbol{\hat{\gamma}}_{T}-%
\boldsymbol{\gamma }_{T}^{\ast }\right) - \\
& \qquad 2\left( \boldsymbol{\hat{\gamma}}_{T}-\boldsymbol{\gamma }%
_{T}^{\ast }\right) ^{\prime }\left[ T^{-1}\left( \mathbf{W}^{\prime }%
\mathbf{\ R}\boldsymbol{\tau }+\mathbf{W}^{\prime }\mathbf{u}\right) \right]
.
\end{align*}%
By substituting for $\boldsymbol{\hat{\gamma}}_{T}-\boldsymbol{\gamma }%
_{T}^{\ast }$ from (\ref{gamma hat minus gamma star}), we get 
\begin{align*}
T^{-1}\boldsymbol{\hat{\eta}}^{\prime }\boldsymbol{\hat{\eta}}&
=T^{-1}\left( \boldsymbol{\tau }^{\prime }\mathbf{R}^{\prime }\mathbf{R}%
\boldsymbol{\ \tau }+\mathbf{u}^{\prime }\mathbf{u}\right) +2T^{-1}%
\boldsymbol{\tau }^{\prime }\mathbf{R}^{\prime }\mathbf{u}+ \\
& \qquad \left[ T^{-1}\left( \mathbf{W}^{\prime }\mathbf{R}\boldsymbol{\tau }%
+\mathbf{W}^{\prime }\mathbf{u}\right) \right] ^{\prime }\left( T^{-1}%
\mathbf{W}^{\prime }\mathbf{W}\right) ^{-1}\left[ T^{-1}\left( \mathbf{W}%
^{\prime }\mathbf{R}\boldsymbol{\tau }+\mathbf{W}^{\prime }\mathbf{u}\right) %
\right] - \\
& \qquad 2\left[ T^{-1}\left( \mathbf{W}^{\prime }\mathbf{R}\boldsymbol{\
\tau }+\mathbf{W}^{\prime }\mathbf{u}\right) \right] ^{\prime }\left( T^{-1}%
\mathbf{W}^{\prime }\mathbf{W}\right) ^{-1}\left[ T^{-1}\left( \mathbf{W}%
^{\prime }\mathbf{R}\boldsymbol{\tau }+\mathbf{W}^{\prime }\mathbf{u}\right) %
\right] \\
& =T^{-1}\left( \boldsymbol{\tau }^{\prime }\mathbf{R}^{\prime }\mathbf{R}%
\boldsymbol{\tau }+\mathbf{u}^{\prime }\mathbf{u}\right) +2T^{-1}\boldsymbol{%
\tau }^{\prime }\mathbf{R}^{\prime }\mathbf{u}- \\
& \qquad \left[ T^{-1}\left( \mathbf{W}^{\prime }\mathbf{R}\boldsymbol{\tau }%
+\mathbf{W}^{\prime }\mathbf{u}\right) \right] ^{\prime }\left( T^{-1}%
\mathbf{W}^{\prime }\mathbf{W}\right) ^{-1}\left[ T^{-1}\left( \mathbf{W}%
^{\prime }\mathbf{R}\boldsymbol{\tau }+\mathbf{W}^{\prime }\mathbf{u}
\right) \right] .
\end{align*}%
we can further write 
\begin{align*}
& T^{-1}\boldsymbol{\hat{\eta}}^{\prime }\boldsymbol{\hat{\eta}}=T^{-1}%
\mathbb{E}\left( \boldsymbol{\tau }^{\prime }\mathbf{R}^{\prime }\mathbf{R}%
\boldsymbol{\tau }+\mathbf{u}^{\prime }\mathbf{u}\right) +T^{-1}\left\{ %
\left[ \boldsymbol{\tau }^{\prime }\mathbf{R}^{\prime }\mathbf{R}\boldsymbol{%
\ \tau }-\mathbb{E}\left( \boldsymbol{\tau }^{\prime }\mathbf{R}^{\prime }%
\mathbf{R}\boldsymbol{\tau }\right) \right] +\left[ \mathbf{u}^{\prime }%
\mathbf{u}-\mathbb{E}\left( \mathbf{u}^{\prime }\mathbf{u}\right) \right]
\right\} + \\
& \quad 2T^{-1}\boldsymbol{\tau }^{\prime }\mathbf{R}^{\prime }\mathbf{u}-%
\left[ T^{-1}\left( \mathbf{W}^{\prime }\mathbf{R}\boldsymbol{\tau }+\mathbf{%
W}^{\prime }\mathbf{u}\right) \right] ^{\prime }\left[ \mathbb{E}\left(
T^{-1}\mathbf{W}^{\prime }\mathbf{W}\right) \right] ^{-1}\left[ T^{-1}\left( 
\mathbf{W}^{\prime }\mathbf{R}\boldsymbol{\tau }+\mathbf{W}^{\prime }\mathbf{%
u}\right) \right] - \\
& \quad \left[ T^{-1}\left( \mathbf{W}^{\prime }\mathbf{R}\boldsymbol{\tau }+%
\mathbf{W}^{\prime }\mathbf{u}\right) \right] ^{\prime }\left\{ \left( T^{-1}%
\mathbf{W}^{\prime }\mathbf{W}\right) ^{-1}-\left[ \mathbb{E}\left( T^{-1}%
\mathbf{W}^{\prime }\mathbf{W}\right) \right] ^{-1}\right\} \left[
T^{-1}\left( \mathbf{W}^{\prime }\mathbf{R}\boldsymbol{\tau }+\mathbf{W}%
^{\prime }\mathbf{u}\right) \right] .
\end{align*}%
Therefore, 
\begin{equation}
\begin{split}
& T^{-1}\boldsymbol{\hat{\eta}}^{\prime }\boldsymbol{\hat{\eta}}-T^{-1}%
\mathbb{E}\left( \boldsymbol{\tau }^{\prime }\mathbf{R}^{\prime }\mathbf{R}%
\boldsymbol{\tau }+\mathbf{u}^{\prime }\mathbf{u}\right) \leq \\
& \qquad T^{-1}\left[ \boldsymbol{\tau }^{\prime }\mathbf{R}^{\prime }%
\mathbf{\ \ R}\boldsymbol{\tau }-\mathbb{E}\left( \boldsymbol{\tau }^{\prime
}\mathbf{R}^{\prime }\mathbf{R}\boldsymbol{\tau }\right) \right] +T^{-1}%
\left[ \mathbf{u}^{\prime }\mathbf{u}-\mathbb{E}\left( \mathbf{u}^{\prime }%
\mathbf{u}\right) \right] + 2T^{-1}\boldsymbol{\tau }^{\prime }\mathbf{R}%
^{\prime }\mathbf{u} + \\
& \qquad \left\Vert T^{-1}\left[ \mathbf{W}^{\prime } \mathbf{R}\boldsymbol{%
\tau }+ \mathbf{W}^{\prime } \mathbf{u} - \mathbb{E}\left(\mathbf{W}%
^{\prime} \mathbf{u}\right) \right] \right\Vert^{2} \left\Vert \left[ 
\mathbb{E}\left( T^{-1}\mathbf{W}^{\prime }\mathbf{W}\right) \right]
^{-1}\right\Vert _{2} + \left\Vert T^{-1} \mathbb{E}\left(\mathbf{W}^{\prime
} \mathbf{u}\right) \right\Vert ^{2} \left\Vert \left[ \mathbb{E}\left(
T^{-1}\mathbf{W}^{\prime }\mathbf{W}\right) \right] ^{-1}\right\Vert _{2} +
\\
& \qquad \left\Vert T^{-1} \left[ \mathbf{W}^{\prime } \mathbf{R}\boldsymbol{%
\tau }+\mathbf{W}^{\prime }\mathbf{u} - \mathbb{E} \left(\mathbf{W}^{\prime
} \mathbf{u} \right)\right] \right\Vert ^{2}\left\Vert \left( T^{-1}\mathbf{W%
}^{\prime }\mathbf{W}\right) ^{-1}-\left[ \mathbb{E}\left( T^{-1}\mathbf{W}%
^{\prime }\mathbf{W}\right) \right] ^{-1}\right\Vert _{F} + \\
& \qquad \left\Vert T^{-1} \mathbb{E} \left(\mathbf{W}^{\prime } \mathbf{u}%
\right) \right\Vert ^{2}\left\Vert \left( T^{-1}\mathbf{W}^{\prime }\mathbf{W%
}\right) ^{-1}-\left[ \mathbb{E}\left( T^{-1}\mathbf{W}^{\prime }\mathbf{W}%
\right) \right] ^{-1}\right\Vert _{F}.
\end{split}
\label{model error bound}
\end{equation}%
First, consider $T^{-1}\left[ \boldsymbol{\tau }^{\prime }\mathbf{R}^{\prime
}\mathbf{R}\boldsymbol{\tau }-\mathbb{E}\left( \boldsymbol{\tau }^{\prime }%
\mathbf{R}^{\prime }\mathbf{R}\boldsymbol{\tau }\right) \right] $. Note that 
\begin{equation*}
\boldsymbol{\tau }^{\prime }\mathbf{R}^{\prime }\mathbf{R}\boldsymbol{\tau }=%
\boldsymbol{\tau }^{\prime }\left( \sum_{t=1}^{T}\mathbf{r}_{t}\mathbf{r}%
_{t}^{\prime }\right) \boldsymbol{\tau }=\sum_{t=1}^{T}\left( \boldsymbol{\
\tau }^{\prime }\mathbf{r}_{t}\right) \left( \mathbf{r}_{t}^{\prime }%
\boldsymbol{\tau }\right) =\sum_{t=1}^{T}\left( \sum_{i=1}^{k}r_{it}\right)
\left( \sum_{j=1}^{k}r_{jt}\right)
=\sum_{i=1}^{k}\sum_{j=1}^{k}\sum_{t=1}^{T}r_{it}r_{jt}.
\end{equation*}%
Recalling that $r_{it}=x_{it}(\beta _{it}-\beta _{i})$, and hence, 
\begin{equation*}
T^{-1}\left[ \boldsymbol{\tau }^{\prime }\mathbf{R}^{\prime }\mathbf{R}%
\boldsymbol{\tau }-\mathbb{E}\left( \boldsymbol{\tau }^{\prime }\mathbf{R}%
^{\prime }\mathbf{R}\boldsymbol{\tau }\right) \right] =\sum_{i=1}^{k}%
\sum_{j=1}^{k}\left( T^{-1}\sum_{t=1}^{T} \tilde{r}_{ij,t}\right) ,
\end{equation*}%
where 
\begin{equation*}
\tilde{r}_{ij,t}= r_{it}r_{jt} - \mathbb{E}(r_{it} r_{jt})
\end{equation*}%
Now consider $\mathbb{E}\left( T^{-1}\sum_{t=1}^{T}\tilde{r}_{ij,t}\right)
^{2}$ and note that 
\begin{align*}
\mathbb{E}\left( T^{-1}\sum_{t=1}^{T}\tilde{r}_{ij,t}\right) ^{2}& =
T^{-2}\sum_{t=1}^{T}\sum_{t^{\prime }=1}^{T}\mathbb{E}\left( \tilde{r}_{ij,t}%
\tilde{r}_{ij,t^{\prime }}\right).
\end{align*}%
By Assumption \ref{weak time dependence}, $T^{-2}\sum_{t=1}^{T}\sum_{t^{%
\prime }=1}^{T}\mathbb{E}\left( \tilde{r}_{ij,t}\tilde{r}_{ij,t^{\prime
}}\right)=O\left( T^{-1}\right) ,$ and hence, by Lemma \ref{Op sample cov
dev}, it follows that 
\begin{equation*}
\left\vert T^{-1}\sum_{t=1}^{T}\tilde{r}_{ij,t}\right\vert =O_{p}\left( 
\frac{1}{\sqrt{T}}\right) .
\end{equation*}%
Since by Assumption \ref{signal}, $k$ is a finite fixed integer, we can
further conclude that 
\begin{equation}
T^{-1}\left[ \boldsymbol{\tau }^{\prime }\mathbf{R}^{\prime }\mathbf{R}%
\boldsymbol{\tau }-\mathbb{E}\left( \boldsymbol{\tau }^{\prime }\mathbf{R}%
^{\prime }\mathbf{R}\boldsymbol{\tau }\right) \right] =\sum_{i=1}^{k}%
\sum_{j=1}^{k}\left( T^{-1}\sum_{t=1}^{T}\tilde{r}_{ij,t}\right)
=O_{p}\left( \frac{1}{\sqrt{T}}\right) .  \label{term one in model error}
\end{equation}%
Now, consider, $T^{-1}\boldsymbol{\tau }^{\prime }\mathbf{R}^{\prime }%
\mathbf{u}$. Note that 
\begin{equation*}
T^{-1}\boldsymbol{\tau }^{\prime }\mathbf{R}^{\prime }\mathbf{u}=T^{-1}%
\boldsymbol{\tau }^{\prime }\left( \sum_{t=1}^{T}\mathbf{r}_{t}u_{t}\right)
=T^{-1}\sum_{t=1}^{T}\boldsymbol{\tau }^{\prime }\mathbf{r}%
_{t}u_{t}=T^{-1}\sum_{t=1}^{T}\sum_{i=1}^{k}r_{it}u_{t}=\sum_{i=1}^{k}\left(
T^{-1}\sum_{t=1}^{T}r_{it}u_{t}\right) .
\end{equation*}%
We have 
\begin{equation*}
\mathbb{E}\left( T^{-1}\sum_{t=1}^{T}r_{it}u_{t}\right)
^{2}=T^{-2}\sum_{t=1}^{T}\mathbb{E}\left( r_{it}^{2}u_{t}^{2}\right)
+2T^{-2}\sum_{t=2}^{T}\sum_{t^{\prime }=1}^{t}\mathbb{E}\left(
r_{it}r_{it^{\prime }}u_{t}u_{t^{\prime }}\right) .
\end{equation*}%
Since $r_{it}=x_{it}(\beta _{it}-\beta _{i})$, and $\beta _{it}$ for $%
i=1,2,\cdots ,k$ are distributed independently of $x_{js}$, $j=1,2,\cdots ,N$
, and $u_{s}$ for all $t$ and $s$, we can further write for any $t^{\prime
}<t$ 
\begin{align*}
\mathbb{E}\left( r_{it}r_{it^{\prime }}u_{t}u_{t^{\prime }}\right) & =%
\mathbb{E}\left( x_{it}u_{t}x_{it^{\prime }}u_{t^{\prime }}\right) \mathbb{E}%
\left[ (\beta _{it}-\beta _{i})(\beta _{it^{\prime }}-\beta _{i})\right] \\
& =\mathbb{E}\left( x_{it}u_{t}x_{it^{\prime }}u_{t^{\prime }}\right) 
\mathbb{E}\left\{ (\beta _{it^{\prime }}-\beta _{i})\mathbb{E}\left[ (\beta
_{it}-\beta _{i})|\mathcal{F}_{t-1}\right] \right\} .
\end{align*}%
But, by Assumption \ref{md}, $\mathbb{E}\left[ (\beta _{it}-\beta _{i})|%
\mathcal{F}_{t-1}\right] =0$ and thus $\mathbb{E}\left( r_{it}r_{it^{\prime
}}u_{t}u_{t^{\prime }}\right) =0$ for any $t^{\prime }<t$. Therefore, 
\begin{equation*}
\mathbb{E}\left( T^{-1}\sum_{t=1}^{T}r_{it}u_{t}\right)
^{2}=T^{-2}\sum_{t=1}^{T}\mathbb{E}\left( r_{it}^{2}u_{t}^{2}\right)
=O\left( \frac{1}{T}\right) .
\end{equation*}%
Hence, by Lemma \ref{Op sample cov dev}, $\left\vert
T^{-1}\sum_{t=1}^{T}r_{it}u_{t}\right\vert =O_{p}\left( \frac{1}{\sqrt{T}}%
\right) $. Since, by Assumption \ref{signal}, $k$ is a finite fixed integer,
we conclude that 
\begin{equation}
T^{-1}\boldsymbol{\tau }^{\prime }\mathbf{R}^{\prime }\mathbf{u}%
=\sum_{i=1}^{k}\left( T^{-1}\sum_{t=1}^{T}r_{it}u_{t}\right) =O_{p}\left( 
\frac{1}{\sqrt{T}}\right) .  \label{term three in model error}
\end{equation}%
By substituting (\ref{term one in model error}) and (\ref{term three in
model error}) into (\ref{model error bound}), and noting that $\left \Vert
T^{-1}\mathbb{E}\left( \mathbf{W}^{\prime }\mathbf{u}\right) \right\Vert^2 =
O\left(T^{-(2\epsilon-d) }\right)$, for some $\epsilon \geq 1/2 $, 
\begin{equation*}
\left\Vert T^{-1}\left[ \mathbf{W}^{\prime }\mathbf{R}\boldsymbol{\tau }+ 
\mathbf{W}^{\prime }\mathbf{u} - \mathbb{E}\left(\mathbf{W}^{\prime }\mathbf{%
u} \right) \right] \right\Vert ^{2} = O_{p}(T^{-(1-d)}),
\end{equation*}
\begin{equation*}
\left\Vert \left( T^{-1}\mathbf{W}^{\prime }\mathbf{W}\right) ^{-1}-\left[ 
\mathbb{E}\left( T^{-1}\mathbf{W}^{\prime }\mathbf{W}\right) \right]
^{-1}\right\Vert _{F}=O_{p}(T^{-(1/2-d)}),
\end{equation*}
and 
\begin{equation*}
T^{-1}\left[ \mathbf{u}^{\prime }\mathbf{u}-\mathbb{E}\left(\mathbf{\ u}%
^{\prime }\mathbf{u}\right) \right] =O_{p}(1/\sqrt{T}),
\end{equation*}
we conclude that 
\begin{equation*}
T^{-1}\boldsymbol{\hat{\eta}}^{\prime }\boldsymbol{\hat{\eta}}%
=\sum_{i=1}^{k}\sum_{j=1}^{k}\left( T^{-1}\sum_{t=1}^{T}\sigma
_{ijt,x}\sigma _{ijt,\beta }\right) +\bar{\sigma}_{u,T}^{2}+O_{p}\left( 
\frac{1}{\sqrt{T}}\right) +O_{p}\left( T^{-(1-d)}\right) ,
\end{equation*}%
where $\sigma _{ijt,x}=\mathbb{E}\left( x_{it}x_{jt}\right) $, $\sigma
_{ijt,\beta }=\mathbb{E}\left[ (\beta _{it}-\beta _{i})(\beta _{jt}-\beta
_{j})\right] $, and $\bar{\sigma}_{u,T}^{2}=T^{-1}\mathbb{E}\left( \mathbf{u}%
^{\prime }\mathbf{u}\right) $. We further have 
\begin{equation*}
\bar{\Delta}_{\beta,T} = \sum_{i=1}^{k}\sum_{j=1}^{k}\left(
T^{-1}\sum_{t=1}^{T}\sigma_{ijt,x}\sigma _{ijt,\beta }\right) =
T^{-1}\sum_{t=1}^{T} \left( \sum_{i=1}^{k}\sum_{j=1}^{k}\sigma_{ijt,x}\sigma
_{ijt,\beta }\right) = \frac{1}{T} \sum_{t=1}^{T} \text{tr}\left(\boldsymbol{%
\Omega}_{\beta,t} \boldsymbol{\Sigma}_{\mathbf{x}_k,t}\right),
\end{equation*}
where $\boldsymbol{\Omega}_{\beta,t} \equiv \left(\sigma_{ijt,\beta}\right)$
and $\boldsymbol{\Sigma}_{\mathbf{x}_k,t} \equiv \left(\sigma_{ijt,x}\right)$
for $i,j = 1,2, \cdots, k$. By result 9(b) on page 44 of \cite%
{Lutkepohl1996handbook}, we can further write 
\begin{equation*}
\text{tr}\left(\boldsymbol{\Omega}_{\beta,t} \boldsymbol{\Sigma}_{\mathbf{x}%
_k,t}\right) \geq k \left[\text{det}\left(\boldsymbol{\Omega}_{\beta,t}
\right) \right]^{1/k}\left[\text{det}\left(\boldsymbol{\Sigma}_{\mathbf{x}%
_k,t}\right)\right]^{1/k}.
\end{equation*}
But $k $ is a finite fixed integer. Furthermore, $\text{det}\left(%
\boldsymbol{\Omega}_{\beta,t} \right) \geq 0 $ and $\text{det}\left(%
\boldsymbol{\Sigma}_{\mathbf{x}_k,t}\right) > 0 $, since $\boldsymbol{\Omega}%
_{\beta,t}$ and $\boldsymbol{\Sigma}_{\mathbf{x}_k,t} $ are positive
semi-definite and positive definite matrices, respectively. So, we can
conclude that $\bar{\Delta}_{\beta,T} \geq 0 $ as required.

In the second scenario, where $\mathbb{E}\left( \mathbf{w}_{t}\mathbf{w}%
_{t}^{\prime }\right) $ is time-invariant, we can write (\ref{dgp yt supp})
as 
\begin{equation*}
y_{t}=\sum_{i=1}^{k}x_{it}\bar{\beta}_{iT}+\sum_{i=1}^{k}x_{it}\left( \beta
_{it}-\bar{\beta}_{iT}\right) +u_{t}=\sum_{i=1}^{k}x_{it}\bar{\beta}%
_{iT}+\sum_{i=1}^{k}h_{it}+u_{t}=\mathbf{x}_{kt}^{\prime }\bar{\boldsymbol{%
\beta }}+\mathbf{h}_{t}^{\prime }\boldsymbol{\tau }+u_{t},
\end{equation*}%
where $h_{it}=x_{it}\left( \beta _{it}-\bar{\beta}_{iT}\right) $, and $%
\mathbf{h}_{t}=(h_{1t},h_{2t},\cdots ,h_{kt})^{\prime }$. We can further
write the DGP in (\ref{dgp yt supp}) in matrix format as 
\begin{equation*}
\mathbf{y}=\mathbf{X}_{k}\bar{\boldsymbol{\beta }}+\mathbf{H}\boldsymbol{%
\tau }+\mathbf{u},
\end{equation*}%
where $\mathbf{H}=(\mathbf{h}_{1},\mathbf{h}_{2},\cdots ,\mathbf{h}%
_{T})^{\prime }$. Now, by using the similar lines of arguments as in the
first scenario, we obtain 
\begin{equation*}
\hat{\boldsymbol{\gamma }}_{T}-\boldsymbol{\gamma }_{T}^{\diamond }=\left(
T^{-1}\mathbf{W}^{\prime }\mathbf{W}\right) ^{-1}\left( T^{-1}\mathbf{W}%
^{\prime }\mathbf{H}\boldsymbol{\tau }\right) +\left( T^{-1}\mathbf{W}%
^{\prime }\mathbf{W}\right) ^{-1}\left( T^{-1}\mathbf{W}^{\prime }\mathbf{u}%
\right) .
\end{equation*}%
%
%
%
%
%
%
We can further use the similar lines of arguments as in the first scenario
and write 
\begin{align*}
\left\Vert \hat{\boldsymbol{\gamma }}_{T}-\boldsymbol{\gamma }_{T}^{\diamond
}\right\Vert \leq & \left\Vert \left( T^{-1}\mathbf{W}^{\prime }\mathbf{W}%
\right) ^{-1}-\left[ \mathbb{E}\left( T^{-1}\mathbf{W}^{\prime }\mathbf{W}%
\right) \right] ^{-1}\right\Vert _{F}\left\Vert T^{-1}\mathbf{W}^{\prime }%
\mathbf{H}\boldsymbol{\tau }\right\Vert + \\
& \left\Vert \left[ \mathbb{E}\left( T^{-1}\mathbf{W}^{\prime }\mathbf{W}%
\right) \right] ^{-1}\right\Vert _{2}\left\Vert T^{-1}\mathbf{W}^{\prime }%
\mathbf{H}\boldsymbol{\tau }\right\Vert + \\
& \left\Vert \left( T^{-1}\mathbf{W}^{\prime }\mathbf{W}\right) ^{-1}-\left[ 
\mathbb{E}\left( T^{-1}\mathbf{W}^{\prime }\mathbf{W}\right) \right]
^{-1}\right\Vert _{F}\left\Vert T^{-1}\left[ \left( \mathbf{W}^{\prime }%
\mathbf{u}\right) -\mathbb{E}\left( \mathbf{W}^{\prime }\mathbf{u}\right) %
\right] \right\Vert _{}+ \\
& \left\Vert \left( T^{-1}\mathbf{W}^{\prime }\mathbf{W}\right) ^{-1}-\left[ 
\mathbb{E}\left( T^{-1}\mathbf{W}^{\prime }\mathbf{W}\right) \right]
^{-1}\right\Vert _{F}\left\Vert T^{-1}\mathbb{E}\left( \mathbf{W}^{\prime }%
\mathbf{u}\right) \right\Vert + \\
& \left\Vert \left[ \mathbb{E}\left( T^{-1}\mathbf{W}^{\prime }\mathbf{W}%
\right) \right] ^{-1}\right\Vert _{2}\left\Vert T^{-1}\left[ \left( \mathbf{W%
}^{\prime }\mathbf{u}\right) -\mathbb{E}\left( \mathbf{W}^{\prime }\mathbf{\
u}\right) \right] \right\Vert + \\
& \left\Vert \left[ \mathbb{E}\left( T^{-1}\mathbf{W}^{\prime }\mathbf{W}%
\right) \right] ^{-1}\right\Vert_{2} \left\Vert T^{-1}\mathbb{E}\left( 
\mathbf{W}^{\prime }\mathbf{\ u}\right) \right\Vert
\end{align*}%
We know that $\left\Vert T^{-1}\mathbb{E}\left( \mathbf{W}^{\prime }\mathbf{%
\ u}\right) \right\Vert = O\left(T^{-\frac{2 \epsilon - d}{2}}\right) $ for
some $\epsilon \geq 1/2 $. Also, 
\begin{equation*}
\left\Vert T^{-1}\left[ \left( \mathbf{W}^{\prime }\mathbf{u}\right) -%
\mathbb{E}\left( \mathbf{W}^{\prime }\mathbf{u}\right) \right] \right\Vert
=O_{p}\left( T^{-\frac{1 - d}{2}}\right) ,
\end{equation*}%
and 
\begin{equation*}
\left\Vert \left( T^{-1}\mathbf{W}^{\prime }\mathbf{W}\right) ^{-1}-\left[ 
\mathbb{E}\left( T^{-1}\mathbf{W}^{\prime }\mathbf{W}\right) \right]
^{-1}\right\Vert _{F}=O_{p}\left( T^{-(1/2 - d)}\right) .
\end{equation*}%
Now consider $\left\Vert T^{-1}\mathbf{W}^{\prime }\mathbf{H}\boldsymbol{%
\tau }\right\Vert $. By using the similar lines of arguments as in the first
scenario, we have 
\begin{equation*}
\left\Vert T^{-1}\mathbf{W}^{\prime }\mathbf{H}\boldsymbol{\tau }\right\Vert
^{2}\leq T^{-2}\sum_{i=1}^{k}\sum_{\ell
=1}^{k+l_{T}}\sum_{t=1}^{T}\sum_{t^{\prime }=1}^{T}w_{\ell t}w_{\ell
t^{\prime }}x_{it}x_{it^{\prime }}\left( \beta _{it}-\bar{\beta}_{i}\right)
\left( \beta _{it^{\prime }}-\bar{\beta}_{i}\right) .
\end{equation*}%
Since, by Assumption \ref{signal}, $\beta _{it}$ for $i=1,2,\cdots ,k$ are
distributed independently of $\mathbf{w}_{t}$ for $t=1,2,\cdots ,T$, we can
further write, 
\begin{align*}
\mathbb{E}\left\Vert T^{-1}\mathbf{W}^{\prime }\mathbf{H}\boldsymbol{\tau }%
\right\Vert ^{2}& \leq T^{-2}\sum_{i=1}^{k}\sum_{\ell
=1}^{k+l_{T}}\sum_{t=1}^{T}\sum_{t^{\prime }=1}^{T}\mathbb{E}\left( w_{\ell
t}w_{\ell t^{\prime }}x_{it}x_{it^{\prime }}\right) \mathbb{E}\left[ \left(
\beta _{it}-\bar{\beta}_{i}\right) \left( \beta _{it^{\prime }}-\bar{\beta}%
_{i}\right) \right] \\
& =T^{-2}\sum_{i=1}^{k}\sum_{\ell =1}^{k+l_{T}}\sum_{t=1}^{T}\mathbb{E}%
\left( w_{\ell t}^{2}x_{it}^{2}\right) \mathbb{E}\left[ \left( \beta _{it}-%
\bar{\beta}_{i}\right) ^{2}\right] + \\
& \quad T^{-2}\sum_{i=1}^{k}\sum_{\ell
=1}^{k+l_{T}}\sum_{t=1}^{T}\sum_{t^{\prime }\neq t}\mathbb{E}\left( w_{\ell
t}w_{\ell t^{\prime }}x_{it}x_{it^{\prime }}\right) \mathbb{E}\left[ \left(
\beta _{it}-\bar{\beta}_{i}\right) \left( \beta _{it^{\prime }}-\bar{\beta}%
_{i}\right) \right] .
\end{align*}%
Since, by Assumption \ref{md}, $\mathbb{E}\left[ w_{\ell t}w_{\ell ^{\prime
}t}-\mathbb{E}(w_{\ell t}w_{\ell ^{\prime }t})|\mathcal{F}_{t-1}\right] =0$
for all $\ell $, $\ell ^{\prime }$ and $t=1,2,\cdots ,T$, we have for any $%
t^{\prime }\neq t$ 
\begin{equation*}
\mathbb{E}\left( w_{\ell t}w_{\ell t^{\prime }}x_{it}x_{it^{\prime }}\right)
=\mathbb{E}\left( w_{\ell t}x_{it}\right) \mathbb{E}\left( w_{\ell t^{\prime
}}x_{it^{\prime }}\right) .
\end{equation*}%
Therefore, 
\begin{align*}
& \sum_{t=1}^{T}\sum_{t^{\prime }\neq t}\mathbb{E}\left( w_{\ell t}w_{\ell
t^{\prime }}x_{it}x_{it^{\prime }}\right) \mathbb{E}\left[ \left( \beta
_{it}-\bar{\beta}_{i}\right) \left( \beta _{it^{\prime }}-\bar{\beta}%
_{i}\right) \right] \\
& \qquad =\sum_{t=1}^{T}\sum_{t^{\prime }\neq t}\mathbb{E}\left( w_{\ell
t}x_{it}\right) \mathbb{E}\left( w_{\ell t^{\prime }}x_{it^{\prime }}\right) 
\mathbb{E}\left[ \left( \beta _{it}-\bar{\beta}_{i}\right) \left( \beta
_{it^{\prime }}-\bar{\beta}_{i}\right) \right] .
\end{align*}%
Since $\mathbb{E}\left( \mathbf{w}_{t}\mathbf{w}_{t}^{\prime }\right) $ is
time-invariant, we can further write 
\begin{align*}
& \sum_{t=1}^{T}\sum_{t^{\prime }\neq t}\mathbb{E}\left( w_{\ell t}w_{\ell
t^{\prime }}x_{it}x_{it^{\prime }}\right) \mathbb{E}\left[ \left( \beta
_{it}-\bar{\beta}_{i}\right) \left( \beta _{it^{\prime }}-\bar{\beta}%
_{i}\right) \right] \\
& \qquad =\mathbb{E}\left( w_{\ell t}x_{it}\right)
^{2}\sum_{t=1}^{T}\sum_{t^{\prime }\neq t}\mathbb{E}\left[ \left( \beta
_{it}-\bar{\beta}_{i}\right) \left( \beta _{it^{\prime }}-\bar{\beta}%
_{i}\right) \right] .
\end{align*}%
Note that, by Assumption \ref{md}, for any $t^{\prime }\neq t$, $\mathbb{E}%
\left[ \left( \beta _{it}-\bar{\beta}_{i}\right) \left( \beta _{it^{\prime
}}-\bar{\beta}_{i}\right) \right] =\left[ \mathbb{E}\left( \beta
_{it}\right) -\bar{\beta}_{i}\right] \left[ \mathbb{E}\left( \beta
_{it^{\prime }}\right) -\bar{\beta}_{i}\right] $. Therefore 
\begin{align*}
& \sum_{t=1}^{T}\sum_{t^{\prime }\neq t}\mathbb{E}\left( w_{\ell t}w_{\ell
t^{\prime }}x_{it}x_{it^{\prime }}\right) \mathbb{E}\left[ \left( \beta
_{it}-\bar{\beta}_{i}\right) \left( \beta _{it^{\prime }}-\bar{\beta}%
_{i}\right) \right] \\
& \qquad =\left[ \mathbb{E}\left( w_{\ell t}x_{it}\right) \right]
^{2}\sum_{t=1}^{T}\sum_{t^{\prime }\neq t}\left[ \mathbb{E}\left( \beta
_{it}\right) -\bar{\beta}_{i}\right] \left[ \mathbb{E}\left( \beta
_{it^{\prime }}\right) -\bar{\beta}_{i}\right] .
\end{align*}%
We can further write, 
\begin{align*}
& \sum_{t=1}^{T}\sum_{t^{\prime }\neq t}\mathbb{E}\left( w_{\ell t}w_{\ell
t^{\prime }}x_{it}x_{it^{\prime }}\right) \mathbb{E}\left[ \left( \beta
_{it}-\bar{\beta}_{i}\right) \left( \beta _{it^{\prime }}-\bar{\beta}%
_{i}\right) \right] \\
& \qquad =\left[ \mathbb{E}\left( w_{\ell t}x_{it}\right) \right]
^{2}\left\{ \sum_{t=1}^{T}\sum_{t^{\prime }=1}^{T}\left[ \mathbb{E}\left(
\beta _{it}\right) -\bar{\beta}_{i}\right] \left[ \mathbb{E}\left( \beta
_{it^{\prime }}\right) -\bar{\beta}_{i}\right] -\sum_{t=1}^{T}\left[ \mathbb{%
E}\left( \beta _{it}\right) -\bar{\beta}_{i}\right] ^{2}\right\} \\
& \qquad =\left[ \mathbb{E}\left( w_{\ell t}x_{it}\right) \right]
^{2}\left\{ \sum_{t=1}^{T}\left[ \mathbb{E}\left( \beta _{it}\right) -\bar{%
\beta}_{i}\right] \right\} \left\{ \sum_{t^{\prime }=1}^{T}\left[ \mathbb{E}%
\left( \beta _{it^{\prime }}\right) -\bar{\beta}_{i}\right] \right\} - \\
& \qquad \quad \left[ \mathbb{E}\left( w_{\ell t}x_{it}\right) \right]
^{2}\sum_{t=1}^{T}\left[ \mathbb{E}\left( \beta _{it}\right) -\bar{\beta}_{i}%
\right] ^{2}.
\end{align*}%
But, $\sum_{t=1}^{T}\left[ \mathbb{E}\left( \beta _{it}\right) -\bar{\beta}%
_{i}\right] =0$, and therefore, 
\begin{equation*}
\sum_{t=1}^{T}\sum_{t^{\prime }\neq t}\mathbb{E}\left( w_{\ell t}w_{\ell
t^{\prime }}x_{it}x_{it^{\prime }}\right) \mathbb{E}\left[ \left( \beta
_{it}-\bar{\beta}_{i}\right) \left( \beta _{it^{\prime }}-\bar{\beta}%
_{i}\right) \right] =-\left[ \mathbb{E}\left( w_{\ell t}x_{it}\right) \right]
^{2}\sum_{t=1}^{T}\left[ \mathbb{E}\left( \beta _{it}\right) -\bar{\beta}_{i}%
\right] ^{2}.
\end{equation*}%
So, 
\begin{align*}
& \mathbb{E}\left\Vert T^{-1}\mathbf{W}^{\prime }\mathbf{H}\boldsymbol{\tau }%
\right\Vert ^{2} \\
& \qquad \leq T^{-2}\sum_{i=1}^{p}\sum_{\ell
=1}^{p+l_{T}}\sum_{t=1}^{T}\left\{ \mathbb{E}\left( w_{\ell
t}^{2}x_{it}^{2}\right) \mathbb{E}\left[ \left( \beta _{it}-\bar{\beta}%
_{i}\right) ^{2}\right] -\left[ \mathbb{E}\left( w_{\ell t}x_{it}\right) %
\right] ^{2}\left[ \mathbb{E}\left( \beta _{it}\right) -\bar{\beta}_{i}%
\right] ^{2}\right\} \\
& \qquad =O\left( T^{-(1 - d)}\right) ,
\end{align*}%
and hence, by Lemma \ref{Op sample cov dev}, 
\begin{equation*}
\left\Vert T^{-1}\mathbf{W}^{\prime }\mathbf{H}\boldsymbol{\tau }\right\Vert
=O_{p}\left( T^{-\frac{1 - d}{2}}\right) .
\end{equation*}%
So, we conclude that 
\begin{equation*}
\left\Vert \hat{\boldsymbol{\gamma }}_{T}-\boldsymbol{\gamma }_{T}^{\diamond
}\right\Vert =O_{p}\left( T^{-\frac{1 - d}{2}}\right) .
\end{equation*}%
Lastly, consider the model mean squared errors for the second scenario.
Following the same lines of argument as in the first scenario, we can write, 
\begin{equation}
\begin{split}
& T^{-1}\boldsymbol{\hat{\eta}}^{\prime }\boldsymbol{\hat{\eta}}-T^{-1}%
\mathbb{E}\left( \boldsymbol{\tau }^{\prime }\mathbf{H}^{\prime }\mathbf{H}%
\boldsymbol{\tau }+\mathbf{u}^{\prime }\mathbf{u}\right) \leq \\
& \qquad T^{-1}\left[ \boldsymbol{\tau }^{\prime }\mathbf{H}^{\prime }%
\mathbf{\ \ H}\boldsymbol{\tau }-\mathbb{E}\left( \boldsymbol{\tau }^{\prime
}\mathbf{H}^{\prime }\mathbf{H}\boldsymbol{\tau }\right) \right] +T^{-1}%
\left[ \mathbf{u}^{\prime }\mathbf{u}-\mathbb{E}\left( \mathbf{u}^{\prime }%
\mathbf{u}\right) \right] + 2T^{-1}\boldsymbol{\tau }^{\prime }\mathbf{H}%
^{\prime }\mathbf{u} + \\
& \qquad \left\Vert T^{-1} \left[\mathbf{W}^{\prime}\mathbf{H}\boldsymbol{%
\tau }+ \mathbf{W}^{\prime} \mathbf{u} - \mathbb{E}\left(\mathbf{W}^{\prime} 
\mathbf{u}\right)\right] \right\Vert ^{2}\left\Vert \left[ \mathbb{E}\left(
T^{-1}\mathbf{W}^{\prime }\mathbf{W}\right) \right] ^{-1}\right\Vert _{2}+
\left\Vert T^{-1} \mathbb{E}\left(\mathbf{W}^{\prime} \mathbf{u}\right)
\right\Vert ^{2}\left\Vert \left[ \mathbb{E}\left( T^{-1}\mathbf{W}^{\prime }%
\mathbf{W}\right) \right] ^{-1}\right\Vert _{2} + \\
& \qquad \left\Vert T^{-1} \left[\mathbf{W}^{\prime}\mathbf{H}\boldsymbol{%
\tau }+ \mathbf{W}^{\prime} \mathbf{u} - \mathbb{E}\left(\mathbf{W}^{\prime} 
\mathbf{u}\right)\right] \right\Vert ^{2} \left\Vert \left( T^{-1}\mathbf{W}%
^{\prime }\mathbf{W}\right) ^{-1}-\left[ \mathbb{E}\left( T^{-1}\mathbf{W}%
^{\prime }\mathbf{W}\right) \right] ^{-1}\right\Vert _{F} + \\
& \qquad \left\Vert T^{-1} \mathbb{E}\left(\mathbf{W}^{\prime} \mathbf{u}%
\right) \right\Vert ^{2} \left\Vert \left( T^{-1}\mathbf{W}^{\prime }\mathbf{%
W}\right) ^{-1}-\left[ \mathbb{E}\left( T^{-1}\mathbf{W}^{\prime }\mathbf{W}%
\right) \right] ^{-1}\right\Vert _{F} \\
\end{split}
\label{model error bound 2}
\end{equation}%
First, consider $T^{-1}\left[ \boldsymbol{\tau }^{\prime }\mathbf{H}^{\prime
}\mathbf{H}\boldsymbol{\tau }-\mathbb{E}\left( \boldsymbol{\tau }^{\prime }%
\mathbf{H}^{\prime }\mathbf{H}\boldsymbol{\tau }\right) \right] $. Note that 
\begin{equation*}
\boldsymbol{\tau }^{\prime }\mathbf{H}^{\prime }\mathbf{H}\boldsymbol{\tau }=%
\boldsymbol{\tau }^{\prime }\left( \sum_{t=1}^{T}\mathbf{h}_{t}\mathbf{h}%
_{t}^{\prime }\right) \boldsymbol{\tau }=\sum_{t=1}^{T}\left( \boldsymbol{\
\tau }^{\prime }\mathbf{r}_{t}\right) \left( \mathbf{r}_{t}^{\prime }%
\boldsymbol{\tau }\right) =\sum_{t=1}^{T}\left( \sum_{i=1}^{k}h_{it}\right)
\left( \sum_{j=1}^{k}h_{jt}\right)
=\sum_{i=1}^{k}\sum_{j=1}^{k}\sum_{t=1}^{T}h_{it}h_{jt}.
\end{equation*}%
Recalling that $h_{it}=x_{it}(\beta _{it}-\bar{\beta}_{iT})$, and hence, 
\begin{equation*}
T^{-1}\left[ \boldsymbol{\tau }^{\prime }\mathbf{H}^{\prime }\mathbf{H}%
\boldsymbol{\tau }-\mathbb{E}\left( \boldsymbol{\tau }^{\prime }\mathbf{H}%
^{\prime }\mathbf{H}\boldsymbol{\tau }\right) \right] =\sum_{i=1}^{k}%
\sum_{j=1}^{k}\left( T^{-1}\sum_{t=1}^{T}\tilde{h}_{ij,t}\right) ,
\end{equation*}%
where 
\begin{equation*}
\tilde{h}_{ij,t}=h_{it}h_{jt}-\mathbb{E}(h_{it}h_{jt}).
\end{equation*}%
Now consider $\mathbb{E}\left( T^{-1}\sum_{t=1}^{T}\tilde{h}_{ij,t}\right)
^{2}$ and note that 
\begin{equation*}
\mathbb{E}\left( T^{-1}\sum_{t=1}^{T}\tilde{h}_{ij,t}\right)
^{2}=T^{-2}\sum_{t=1}^{T}\sum_{t^{\prime }=1}^{T}\mathbb{E}\left( \tilde{h}%
_{ij,t}\tilde{h}_{ij,t^{\prime }}\right) .
\end{equation*}%
By Assumption \ref{weak time dependence}, $T^{-2}\sum_{t=1}^{T}\sum_{t^{%
\prime }=1}^{T}\mathbb{E}\left( \tilde{h}_{ij,t}\tilde{h}_{ij,t^{\prime
}}\right) =O\left( T^{-1}\right) $, and hence, by Lemma \ref{Op sample cov
dev}, it follows that 
\begin{equation*}
\left\vert T^{-1}\sum_{t=1}^{T}\tilde{h}_{ij,t}\right\vert =O_{p}\left( 
\frac{1}{\sqrt{T}}\right) .
\end{equation*}%
Since by Assumption \ref{signal}, $k$ is a finite fixed integer, we can
further conclude that 
\begin{equation}
T^{-1}\left[ \boldsymbol{\tau }^{\prime }\mathbf{H}^{\prime }\mathbf{H}%
\boldsymbol{\tau }-\mathbb{E}\left( \boldsymbol{\tau }^{\prime }\mathbf{H}%
^{\prime }\mathbf{H}\boldsymbol{\tau }\right) \right] =\sum_{i=1}^{k}%
\sum_{j=1}^{k}\left( T^{-1}\sum_{t=1}^{T}\tilde{h}_{ij,t}\right)
=O_{p}\left( \frac{1}{\sqrt{T}}\right) .  \label{term one in model error2}
\end{equation}%
Now, consider, $T^{-1}\boldsymbol{\tau }^{\prime }\mathbf{H}^{\prime }%
\mathbf{u}$. Note that 
\begin{equation*}
T^{-1}\boldsymbol{\tau }^{\prime }\mathbf{H}^{\prime }\mathbf{u}=T^{-1}%
\boldsymbol{\tau }^{\prime }\left( \sum_{t=1}^{T}\mathbf{h}_{t}u_{t}\right)
=T^{-1}\sum_{t=1}^{T}\boldsymbol{\tau }^{\prime }\mathbf{h}%
_{t}u_{t}=T^{-1}\sum_{t=1}^{T}\sum_{i=1}^{k}h_{it}u_{t}=\sum_{i=1}^{k}\left(
T^{-1}\sum_{t=1}^{T}h_{it}u_{t}\right) .
\end{equation*}%
We have 
\begin{equation*}
\mathbb{E}\left( T^{-1}\sum_{t=1}^{T}h_{it}u_{t}\right)
^{2}=T^{-2}\sum_{t=1}^{T}\mathbb{E}\left[ \left( h_{it}u_{t}\right) ^{2}%
\right] +T^{-2}\sum_{t=1}^{T}\sum_{t^{\prime }\neq t}\mathbb{E}\left(
h_{it}h_{it^{\prime }}u_{t}u_{t^{\prime }}\right) .
\end{equation*}%
Since $h_{it}=x_{it}(\beta _{it}-\bar{\beta}_{iT})$, and $\beta _{it}$ for $%
i=1,2,\cdots ,k$ are distributed independently of $x_{js}$, $j=1,2,\cdots ,N$%
, and $u_{s}$ for all $t$ and $s$, we can further write for any $t^{\prime
}\neq t$ 
\begin{equation*}
\mathbb{E}\left( h_{it}h_{it^{\prime }}u_{t}u_{t^{\prime }}\right) =\mathbb{E%
}\left( x_{it}u_{t}x_{it^{\prime }}u_{t^{\prime }}\right) \mathbb{E}\left[
(\beta _{it}-\bar{\beta}_{iT})(\beta _{it^{\prime }}-\bar{\beta}_{iT})\right]
.
\end{equation*}%
But, by Assumption \ref{md}, $\mathbb{E}\left[ x_{it}u_{t}-\mathbb{E}%
(x_{it}u_{t})|\mathcal{F}_{t-1}\right] =0$ and we also have $\mathbb{E}%
(x_{it}u_{t})=0$ for $i=1,2,\cdots ,k$ and thus for any $t^{\prime }\neq t$
we have 
\begin{equation*}
\mathbb{E}\left( x_{it}u_{t}x_{it^{\prime }}u_{t^{\prime }}\right) =\mathbb{%
\ E}\left( x_{it}u_{t}\right) \mathbb{E}\left( x_{it^{\prime }}u_{t^{\prime
}}\right) =0.
\end{equation*}%
Therefore, 
\begin{equation*}
\mathbb{E}\left( T^{-1}\sum_{t=1}^{T}h_{it}u_{t}\right)
^{2}=T^{-2}\sum_{t=1}^{T}\mathbb{E}\left[ \left( h_{it}u_{t}\right) ^{2}%
\right] =O\left( \frac{1}{T}\right) .
\end{equation*}%
Hence, by Lemma \ref{Op sample cov dev}, $\left\vert
T^{-1}\sum_{t=1}^{T}h_{it}u_{t}\right\vert =O_{p}\left( \frac{1}{\sqrt{T}}%
\right) $. Since, by Assumption \ref{signal}, $k$ is a finite fixed integer,
we conclude that 
\begin{equation}
T^{-1}\boldsymbol{\tau }^{\prime }\mathbf{H}^{\prime }\mathbf{u}%
=\sum_{i=1}^{k}\left( T^{-1}\sum_{t=1}^{T}h_{it}u_{t}\right) =O_{p}\left( 
\frac{1}{\sqrt{T}}\right) .  \label{term three in model error2}
\end{equation}%
By substituting (\ref{term one in model error2}) and (\ref{term three in
model error2}) into (\ref{model error bound 2}), and noting that $%
\left
\Vert T^{-1}\mathbb{E}\left( \mathbf{W}^{\prime }\mathbf{u}\right)
\right\Vert^2 = O\left(T^{-(2\epsilon-d) }\right)$, for some $\epsilon \geq
1/2 $, 
\begin{equation*}
\left\Vert T^{-1}\left[ \mathbf{W}^{\prime }\mathbf{R}\boldsymbol{\tau }+ 
\mathbf{W}^{\prime }\mathbf{u} - \mathbb{E}\left(\mathbf{W}^{\prime }\mathbf{%
u} \right) \right] \right\Vert ^{2} = O_{p}(T^{-(1-d)}),
\end{equation*}
\begin{equation*}
\left\Vert \left( T^{-1}\mathbf{W}^{\prime }\mathbf{W}\right) ^{-1}-\left[ 
\mathbb{E}\left( T^{-1}\mathbf{W}^{\prime }\mathbf{W}\right) \right]
^{-1}\right\Vert _{F}=O_{p}(T^{-(1/2-d)}),
\end{equation*}
and 
\begin{equation*}
T^{-1}\left[ \mathbf{u}^{\prime }\mathbf{u}-\mathbb{E}\left(\mathbf{\ u}%
^{\prime }\mathbf{u}\right) \right] =O_{p}(1/\sqrt{T}),
\end{equation*}
we conclude that 
\begin{equation*}
T^{-1}\boldsymbol{\hat{\eta}}^{\prime }\boldsymbol{\hat{\eta}}%
=\sum_{i=1}^{k}\sum_{j=1}^{k}\left( T^{-1}\sum_{t=1}^{T}\sigma
_{ijt,x}\sigma _{ijt,\beta }^{\ast }\right) +\bar{\sigma}_{u,T}^{2}+O_{p}%
\left( \frac{1}{\sqrt{T}}\right) +O_{p}\left(T^{-(1-d)}\right),
\end{equation*}%
where $\sigma _{ijt,\beta }^{\ast }=\mathbb{E}\left[ (\beta _{it}-\bar{\beta}%
_{i,T})(\beta _{jt}-\bar{\beta}_{j,T})\right] $, $\bar{\beta}%
_{iT}=T^{-1}\sum_{t=1}^{T}\mathbb{E}(\beta _{it})$, and $\bar{\sigma}%
_{u,T}^{2}=T^{-1}\mathbb{E}\left( \mathbf{u}^{\prime }\mathbf{u}\right) $.
We further have 
\begin{equation*}
\bar{\Delta}_{\beta ,T}^{\ast }=\sum_{i=1}^{k}\sum_{j=1}^{k}\left(
T^{-1}\sum_{t=1}^{T}\sigma _{ijt,x}\sigma _{ijt,\beta }^{\ast }\right)
=T^{-1}\sum_{t=1}^{T}\left( \sum_{i=1}^{k}\sum_{j=1}^{k}\sigma
_{ijt,x}\sigma _{ijt,\beta }^{\ast }\right) =\frac{1}{T}\sum_{t=1}^{T}\text{%
tr}\left( \boldsymbol{\Omega }_{\beta ,t}^{\ast }\boldsymbol{\Sigma }_{%
\mathbf{x}_{k},t}\right) ,
\end{equation*}%
where $\boldsymbol{\Omega }_{\beta ,t}^{\ast }\equiv \left( \sigma
_{ijt,\beta }^{\ast }\right) $ and $\boldsymbol{\Sigma }_{\mathbf{x}%
_{k},t}\equiv \left( \sigma _{ijt,x}\right) $ for $i,j=1,2,\cdots ,k$. By
result 9(b) on page 44 of \cite{Lutkepohl1996handbook}, we can further write 
\begin{equation*}
\text{tr}\left( \boldsymbol{\Omega }_{\beta ,t}^{\ast }\boldsymbol{\Sigma }_{%
\mathbf{x}_{k},t}\right) \geq k\left[ \text{det}\left( \boldsymbol{\Omega }%
_{\beta ,t}^{\ast }\right) \right] ^{1/k}\left[ \text{det}\left( \boldsymbol{%
\Sigma }_{\mathbf{x}_{k},t}\right) \right] ^{1/k}.
\end{equation*}%
But $k$ is a finite fixed integer. Furthermore, $\text{det}\left( 
\boldsymbol{\Omega }_{\beta ,t}^{\ast }\right) \geq 0$ and $\text{det}\left( 
\boldsymbol{\Sigma }_{\mathbf{x}_{k},t}\right) >0$, since $\boldsymbol{%
\Omega }_{\beta ,t}^{\ast }$ and $\boldsymbol{\Sigma }_{\mathbf{x}_{k},t}$
are positive semi-definite and positive definite matrices, respectively. So,
we can conclude that $\bar{\Delta}_{\beta ,T}^{\ast }\geq 0$ as required.
\end{proof}


\begin{lemma}
\label{lem:weak time dependence} Let $y_{t} $ $t = 1, 2, \cdots, T $ be
generated by (\ref{dgp y_t}). Suppose Assumption \ref{signal} and \ref{md}
hold, and the cross products of coefficients of the signals in DGP for $y_t$
follow martingale difference processes such that 
\begin{equation*}
\mathbb{E}\left[ \beta _{it}\beta _{jt}-\mathbb{E}(\beta _{it}\beta _{jt})|%
\mathcal{F}_{t-1}\right] =0,\text{ for }i=1,2,\cdots ,k,\ j=1,2,\cdots ,k,%
\text{ and }t=1,2,\cdots ,T.
\end{equation*}%
Then, $\sum_{t=1}^{T}\sum_{t^{\prime}=1}^{T}\text{cov}(h_{ij,t},h_{ij,t^{%
\prime }})=O(T)$ where $h_{ij,t}=x_{it}x_{jt}(\beta _{it}-\bar{\beta}%
_{iT})(\beta _{jt}-\bar{\beta}_{jT})$.
\end{lemma}

\begin{proof}
To show this, let $\tilde{h}_{ij,t}=h_{ij,t}-\mathbb{E}\left(
h_{ij,t}\right) $. We have 
\begin{align*}
\sum_{t=1}^{T}\sum_{t^{\prime }=1}^{T}\text{cov}(h_{ij,t},h_{ij,t^{\prime
}})& =\sum_{t=1}^{T}\mathbb{E}\left( \tilde{h}_{ij,t}^{2}\right)
+2\sum_{t=2}^{T}\sum_{t^{\prime }=1}^{t}\mathbb{E}\left( \tilde{h}_{ij,t}%
\tilde{h}_{ij,t^{\prime }}\right) \\
& =\sum_{t=1}^{T}\mathbb{E}\left( \tilde{h}_{ij,t}^{2}\right)
+2\sum_{t=2}^{T}\sum_{t^{\prime }=1}^{t}\mathbb{E}\left[ \tilde{h}%
_{ij,t^{\prime }}\mathbb{E}\left( \tilde{h}_{ij,t}|\mathcal{F}_{t-1}\right) %
\right] .
\end{align*}%
But, $\mathbb{E}\left( \tilde{h}_{ij,t}|\mathcal{F}_{t-1}\right) =\mathbb{E}%
\left( h_{ij,t}|\mathcal{F}_{t-1}\right) -\mathbb{E}\left( h_{ij,t}\right) $
and under the conditions mentioned in this Lemma, 
\begin{align*}
\mathbb{E}\left( h_{ij,t}|\mathcal{F}_{t-1}\right) & =\mathbb{E}\left(
x_{it}x_{jt}|\mathcal{F}_{t-1}\right) \mathbb{E}\left[ (\beta _{it}-\bar{%
\beta}_{iT})(\beta _{jt}-\bar{\beta}_{jT})|\mathcal{F}_{t-1}\right] \\
& =\mathbb{E}\left( x_{it}x_{jt}\right) \left\{ \mathbb{E}(\beta _{it}\beta
_{jt}|\mathcal{F}_{t-1})-\bar{\beta}_{jT}\mathbb{E}(\beta _{it}|\mathcal{F}%
_{t-1})-\bar{\beta}_{iT}\mathbb{E}(\beta _{jt}|\mathcal{F}_{t-1})+\bar{\beta}%
_{iT}\bar{\beta}_{jT}\right\} \\
& =\mathbb{E}\left( x_{it}x_{jt}\right) \left\{ \mathbb{E}(\beta _{it}\beta
_{jt})-\bar{\beta}_{jT}\mathbb{E}(\beta _{it})-\bar{\beta}_{iT}\mathbb{E}%
(\beta _{jt})+\bar{\beta}_{iT}\bar{\beta}_{jT}\right\} \\
& =\mathbb{E}\left( x_{it}x_{jt}\right) \mathbb{E}\left[ (\beta _{it}-\bar{%
\beta}_{iT})(\beta _{jt}-\bar{\beta}_{jT})\right] =\mathbb{E}\left(
h_{ij,t}\right) .
\end{align*}%
Therefore, $\mathbb{E}\left( \tilde{h}_{ij,t}|\mathcal{F}_{t-1}\right) =0$.
Hence, $\sum_{t=1}^{T}\sum_{t^{\prime }=1}^{T}\text{cov}(h_{ij,t},h_{ij,t^{%
\prime }})=\sum_{t=1}^{T}\mathbb{E}\left( \tilde{h}_{ij,t}^{2}\right)
=O\left( T\right) $.
\end{proof}


\section{Supplementary lemmas}

\label{Complementary_lemmas}

\begin{lemma}
\label{mart_diff_proc_exp_tail} Let $z_{t}$ be a martingale difference
process with respect to $\mathcal{F}_{t-1}^{z}=\sigma (z_{t-1},\allowbreak
z_{t-2},\cdots )$, and suppose that there exist some finite positive
constants $C_{0}$ and $C_{1}$, and $s>0$ such that 
\begin{equation*}
\sup_{t}\Pr (\left\vert z_{t}\right\vert >\alpha )\leq C_{0}\exp
(-C_{1}\alpha ^{s}),\quad \text{for all}\ \alpha >0.
\end{equation*}
Let also $\sigma _{zt}^{2}=\mathbb{E}(z_{t}^{2}|\mathcal{F}_{t-1}^{z})$ and $%
\bar{\sigma}_{z,T}^{2}=T^{-1}\sum_{t=1}^{T}\sigma _{zt}^{2}$. Suppose that $%
\zeta _{T}=\ominus (T^{\lambda })$, for some $0<\lambda \leq (s+1)/(s+2)$.
Then for any $\pi $ in the range $0<\pi <1$, we have, 
\begin{equation*}
\textstyle\Pr \left( \lvert \sum_{t=1}^{T}z_{t}\rvert >\zeta _{T}\right)
\leq \exp \left[ \frac{-(1-\pi )^{2}\zeta _{T}^{2}}{2T\bar{\sigma}_{z,T}^{2}}
\right] .
\end{equation*}
If $\lambda >(s+1)/(s+2)$, then for some finite positive constant $C_{2}$, 
\begin{equation*}
\textstyle\Pr \left( \lvert \sum_{t=1}^{T}z_{t}\rvert >\zeta _{T}\right)
\leq \exp \left( -C_{2}\zeta _{T}^{s/(s+1)}\right) .
\end{equation*}
\end{lemma}

\begin{proof}
The results follow from Lemma A3 of Chudik et al. (2018) Online Theory
Supplement.
\end{proof}


\begin{lemma}
\label{cv_lemma} Let 
\begin{equation}
c_{p}(n,\delta )=\Phi ^{-1}\left( 1-\frac{p}{2f(n,\delta )}\right) ,
\label{cv_eq}
\end{equation}
where $\Phi ^{-1}(.)$ is the inverse of standard normal distribution
function, $p$ $(0<p<1)$ is the nominal size of a test, and $f(n,\delta
)=cn^{\delta }$ for some positive constants $\delta $ and $c$. Moreover, let 
$a>0$ and $0<b<1$. Then (I) $c_{p}(n,\delta )=O\left[ \sqrt{\delta \ln (n)} %
\right] $ and (II) $n^{a}\exp \left[ -bc_{p}^{2}(n,\delta )\right] =\ominus
(n^{a-2b\delta })$.
\end{lemma}

\begin{proof}
The results follow from Lemma 3 of Bailey et al. (2019) Supplementary
Appendix A.
\end{proof}


\begin{lemma}
\label{prob_sum} Let $x_{i}$, for $i=1,2,\cdots ,n$, be random variables.
Then for any constants $\pi _{i}$, for $i=1,2,\cdots ,n$, satisfying $0<\pi
_{i}<1$ and $\sum_{i=1}^{n}\pi _{i}=1$, we have 
\begin{equation*}
\textstyle\Pr (\sum_{i=1}^{n}\left\vert x_{i}\right\vert >C_{0})\leq
\sum_{i=1}^{n}\Pr (\left\vert x_{i}\right\vert >\pi _{i}C_{0}),
\end{equation*}
where $C_{0}$ is a finite positive constant.
\end{lemma}

\begin{proof}
The result follows from Lemma A11 of Chudik et al. (2018) Online Theory
Supplement.
\end{proof}


\begin{lemma}
\label{prob_product} Let $x$, $y$ and $z$ be random variables. Then for any
finite positive constants $C_{0}$, $C_{1}$, and $C_{2}$, we have 
\begin{equation*}
\Pr (|x|\times |y|>C_{0})\leq \Pr (|x|>C_{0}/C_{1})+\Pr (|y|>C_{1}),
\end{equation*}
and 
\begin{equation*}
\Pr (|x|\times |y|\times |z|>C_{0})\leq \Pr (|x|>C_{0}/(C_{1}C_{2}))+\Pr
(|y|>C_{1})+\Pr (|z|>C_{2}).
\end{equation*}
\end{lemma}

\begin{proof}
The results follow from Lemma A11 of Chudik et al. (2018) Online Theory
Supplement.
\end{proof}


\begin{lemma}
\label{prob_rand_sum_constnt} Let $x$ be a random variable. Then for some
finite constants $B$, and $C$, with $|B|\geq C>0$, we have 
\begin{equation*}
\Pr (|x+B|\leq C)\leq \Pr (|x|>|B|-C).
\end{equation*}
\end{lemma}

\begin{proof}
The results follow from Lemma A12 of Chudik et al. (2018) Online Theory
Supplement.
\end{proof}


\begin{lemma}
\label{prob_inv_sqrt_rand} Let $x_T $ to be a random variable. Then for a
deterministic sequence, $\alpha_T > 0 $, with $\alpha_T \rightarrow 0 $ as $%
T \rightarrow \infty $, there exists $T_0 > 0 $ such that for all $T > T_0 $
we have 
\begin{equation*}
\Pr \left( \left| \frac{1}{\sqrt{x_T}} - 1 \right| > \alpha_T \right) \leq
\Pr( |x_T - 1| > \alpha_T ).
\end{equation*}
\end{lemma}

\begin{proof}
The results follow from Lemma A3 of Chudik et al. (2018) Online Theory
Supplement.
\end{proof}


\begin{lemma}
\label{exp_tail_prod} Consider random variables $x_{t}$ and $z_{t}$ with the
exponentially bounded probability tail distributions such that 
\begin{equation*}
\begin{split}
& \sup_{t}\Pr (|x_{t}|>\alpha )\leq C_{0}\exp (-C_{1}\alpha ^{s_{x}}),\ 
\text{for all}\ \alpha >0, \\
& \sup_{t}\Pr (|z_{t}|>\alpha )\leq C_{0}\exp (-C_{1}\alpha ^{s_{z}}),\ 
\text{for all}\ \alpha >0,
\end{split}%
\end{equation*}
where $C_{0}$, and $C_{1}$ are some finite positive constants, $s_{x}>0$,
and $s_{z}>0$ . Then 
\begin{equation*}
\sup_{t}\Pr (|x_{t}z_{t}|>\alpha )\leq C_{0}\exp (-C_{1}\alpha ^{s/2}),\ 
\text{for all}\ \alpha >0,
\end{equation*}
where $s=\min \{s_{x},s_{z}\}$.
\end{lemma}

\begin{proof}
By using Lemma \ref{prob_product}, for all $\alpha > 0 $, 
\begin{equation*}
\Pr( |x_t z_t| > \alpha ) \leq \Pr(|x_t| > \alpha^{1/2}) + \Pr(|z_t| >
\alpha^{1/2})
\end{equation*}
So, 
\begin{equation*}
\begin{split}
& \sup_t \Pr( |x_t z_t| > \alpha ) \leq \sup_t \Pr(|x_t| > \alpha^{1/2}) +
\sup_t \Pr(|z_t| > \alpha^{1/2}) \\
& \qquad \leq C_0 \exp(-C_1 \alpha^{s_x/2}) + C_0 \exp(-C_1 \alpha^{s_z/2})
\\
& \qquad \leq C_0 \exp(-C_1 \alpha^{s/2})
\end{split}%
\end{equation*}
where $s = \min\{s_x, s_z \} $.
\end{proof}


\begin{lemma}
\label{prob_product_two} Let $x$, $y$ and $z$ be random variables. Then for
some finite positive constants $C_{0}$, and $C_{1}$, we have 
\begin{equation*}
\Pr (|x|\times |y|<C_{0})\leq \Pr (|x|<C_{0}/C_{1})+\Pr (|y|<C_{1}),
\end{equation*}
\end{lemma}

\begin{proof}
Define events $\mathfrak{A} = \{ \lvert x \rvert \times \lvert y \rvert <
C_0 \} $, $\mathfrak{B} = \{ \lvert x \rvert < C_0 / C_1 \} $ and $\mathfrak{%
\ \ C} = \{ \lvert y \rvert < C_1 \} $. Then $\mathfrak{A} \in \mathfrak{B}
\cup \mathfrak{C} $. Therefore, $\Pr(\mathfrak{A}) \leq \Pr( \mathfrak{B}
\cup \mathfrak{C})$. But $\Pr( \mathfrak{B} \cup \mathfrak{C}) \leq \Pr( 
\mathfrak{B} ) + \Pr(\mathfrak{C}) $ and hence $\Pr(\mathfrak{A}) \leq \Pr( 
\mathfrak{B} ) + \Pr(\mathfrak{C}) $.
\end{proof}


\begin{lemma}
\label{F_norm_submultiplicative} Let $\mathbf{A}$ and $\mathbf{B}$ be $%
n\times p$ and $p\times m$ matrices respectively, then 
\begin{equation*}
\Vert \mathbf{A}\mathbf{B}\Vert _{F}\leq \Vert \mathbf{A}\Vert _{F}\Vert 
\mathbf{B}\Vert _{2}\text{, and }\Vert \mathbf{A}\mathbf{B}\Vert _{F}\leq
\Vert \mathbf{A}\Vert _{2}\Vert \mathbf{B}\Vert _{F}.
\end{equation*}
\end{lemma}

\begin{proof}
$\Vert \mathbf{A}\mathbf{B}\Vert _{F}^{2}=\text{tr}(\mathbf{A}\mathbf{B} 
\mathbf{B}^{\prime }\mathbf{A}^{\prime })=\text{tr}[\mathbf{A}(\mathbf{B} 
\mathbf{B}^{\prime })\mathbf{A}^{\prime }]$, and by result (12) of L\"{u}%
tkepohl (1996, p.44), 
\begin{equation*}
\text{tr}\left[ \mathbf{A}(\mathbf{B}\mathbf{B}^{\prime })\mathbf{A}^{\prime
}\right] \leq \lambda _{\max }(\mathbf{B}\mathbf{B}^{\prime })\text{tr}( 
\mathbf{A}\mathbf{A}^{\prime })=\Vert \mathbf{A}\Vert _{F}^{2}\Vert \mathbf{%
\ B }\Vert _{2}^{2},
\end{equation*}
where $\lambda _{\max }(\mathbf{B}\mathbf{B}^{\prime })$ is the largest
eigenvalue of $\mathbf{B}\mathbf{B}^{\prime }$. Therefore, $\Vert \mathbf{A} 
\mathbf{B}\Vert _{F}\leq \Vert \mathbf{A}\Vert _{F}\Vert \mathbf{B}\Vert
_{2} $, as required. Similarly, 
\begin{equation*}
\Vert \mathbf{A}\mathbf{B}\Vert _{F}^{2}=\text{tr}(\mathbf{B}^{\prime } 
\mathbf{A}^{\prime }\mathbf{A}\mathbf{B})=\text{tr}[\mathbf{B}^{\prime }( 
\mathbf{A}^{\prime }\mathbf{A})\mathbf{B}]\leq \lambda _{\max }(\mathbf{A}
^{\prime}\mathbf{A})\text{tr}(\mathbf{B}^{\prime }\mathbf{B})=\Vert \mathbf{A%
}\Vert _{2}^{2}\Vert \mathbf{B}\Vert _{F}^{2},
\end{equation*}
and hence 
\begin{equation*}
\Vert \mathbf{A}\mathbf{B}\Vert _{F}\leq \Vert \mathbf{A}\Vert _{2}\Vert 
\mathbf{B}\Vert _{F}.
\end{equation*}
\end{proof}


\begin{lemma}
\label{specral_norm_bound} Let $\mathbf{A}=(a_{ij})_{n\times m}$ where $%
\sup_{ij}|a_{ij}|<C<\infty $, then 
\begin{equation*}
\left\Vert \mathbf{A}\right\Vert _{2}=O\left( \sqrt{nm}\right) .
\end{equation*}
\end{lemma}

\begin{proof}
This result follows, since $\left\| \mathbf{A} \right\|_{2} \leq \sqrt{
\left\| \mathbf{A} \right\|_{\infty} \left\| \mathbf{A} \right\|_{1}} $, $%
\left\| \mathbf{A} \right\|_{\infty} = O(m) $ and $\left\| \mathbf{A}
\right\|_{1} = O(n)$.
\end{proof}


\begin{lemma}
\label{inverse_matrices_diff_F_norm} Consider two $N \times N $ nonsingular
matrices $\mathbf{A} $ and $\mathbf{B} $ such that 
\begin{equation*}
\| \mathbf{B}^{-1} \|_2 \| \mathbf{A} - \mathbf{B} \|_F < 1.
\end{equation*}
Then 
\begin{equation*}
\| \mathbf{A}^{-1} - \mathbf{B}^{-1} \|_F \leq \frac{ \| \mathbf{B}^{-1}
\|_2^2 \| \mathbf{A} - \mathbf{B} \|_F }{1 - \| \mathbf{B}^{-1} \|_2 \| 
\mathbf{A} - \mathbf{B} \|_F }.
\end{equation*}
\end{lemma}

\begin{proof}
By Lemma \ref{F_norm_submultiplicative}, 
\begin{equation*}
\| \mathbf{A}^{-1} - \mathbf{B}^{-1} \|_F = \| \mathbf{A}^{-1} (\mathbf{B}- 
\mathbf{A}) \mathbf{B}^{-1} \|_F \leq \| \mathbf{A}^{-1} \|_2 \| \mathbf{B}
- \mathbf{A} \|_F \| \mathbf{B} ^{-1} \|_2
\end{equation*}
Note that 
\begin{equation*}
\begin{split}
& \| \mathbf{A}^{-1} \|_2 = \| \mathbf{A}^{-1} - \mathbf{B}^{-1} + \mathbf{B}
^{-1} \|_2 \leq \| \mathbf{A}^{-1} - \mathbf{B}^{-1} \|_2 + \| \mathbf{B}
^{-1} \|_2 \\
& \qquad \quad \ \leq \| \mathbf{A}^{-1} - \mathbf{B}^{-1} \|_F + \| \mathbf{%
\ B}^{-1} \|_2,
\end{split}%
\end{equation*}
and therefore, 
\begin{equation*}
\| \mathbf{A}^{-1} - \mathbf{B}^{-1} \|_F \leq (\| \mathbf{A}^{-1} - \mathbf{%
\ \ \ B}^{-1} \|_F + \| \mathbf{B}^{-1} \|_2) \| \mathbf{B} - \mathbf{A}
\|_F \| \mathbf{B}^{-1} \|_2.
\end{equation*}
Hence, 
\begin{equation*}
\| \mathbf{A}^{-1} - \mathbf{B}^{-1} \|_F (1 - \| \mathbf{B}^{-1} \|_2 \| 
\mathbf{B} - \mathbf{A} \|_F) \leq \| \mathbf{B}^{-1} \|_2^2 \| \mathbf{B} - 
\mathbf{A} \|_F.
\end{equation*}
Since $\| \mathbf{B}^{-1} \|_2 \| \mathbf{B} - \mathbf{A} \|_F < 1$, we can
further write, 
\begin{equation*}
\| \mathbf{A}^{-1} - \mathbf{B}^{-1} \|_F \leq \frac{ \| \mathbf{B}^{-1}
\|_2^2 \| \mathbf{A} - \mathbf{B} \|_F }{1 - \| \mathbf{B}^{-1} \|_2 \| 
\mathbf{A} - \mathbf{B} \|_F }.
\end{equation*}
\end{proof}


\begin{lemma}
\label{bound on expected sample cov dev} Let $\mathbf{X} $ and $\mathbf{Y} $
be $T \times N_x $ and $T \times N_y $ matrices of observations on random
variables $x_{it} $ and $y_{jt} $, for $i = 1,2,\cdots, N_x $, $j =
1,2,\cdots, N_y $ and $t = 1,2, \cdots, T $, respectively. Denote 
\begin{equation*}
w_{ij,t} = x_{it} y_{jt} - \mathbb{E} (x_{it} y_{jt}), \text{ for all } i, j 
\text{ and } t.
\end{equation*}
Suppose that

\begin{enumerate}
\item[(i)] $\sup_{i,t} \mathbb{E} \left| x_{it} \right|^{4} < C $, $%
\sup_{j,t} \mathbb{E} \left| y_{jt} \right|^{4} < C $, and

\item[(ii)] $\sup_{i,j} \left[ \sum_{t=1}^{T} \sum_{t^{\prime }=1}^{T} 
\mathbb{E}(w_{ij,t} w_{ij,t^{\prime }}) \right] = O(T). $
\end{enumerate}

Then, 
\begin{equation*}  \label{cov_matrix_F2_order}
\mathbb{E}\left\| T^{-1} \left[\mathbf{X}^{\prime }\mathbf{Y} - \mathbb{E} ( 
\mathbf{X}^{\prime }\mathbf{Y}) \right] \right\|_{F}^2 = O \left( \frac{N_x
N_{y}}{T} \right).
\end{equation*}
\end{lemma}

\begin{proof}
The results follow from Lemma A18 of Chudik et al. (2018) Online Theory
Supplement.
\end{proof}


\begin{lemma}
\label{Op sample cov dev} Let $\mathbf{X}=(x_{ij})_{T\times N_{x}}$ and $%
\mathbf{Y}=(y_{ij})_{T\times N_{y}}$ be matrices of random variables,
respectively. Suppose that, 
\begin{equation*}
\mathbb{E}\left\Vert T^{-1}\left[ \mathbf{X}^{\prime }\mathbf{Y}-\mathbb{E}(%
\mathbf{X}^{\prime }\mathbf{Y})\right] \right\Vert _{F}^{2}=O(a_{T}),
\end{equation*}%
where $a_{T}>0$. Then 
\begin{equation*}
\left\Vert T^{-1}\left[ \mathbf{X}^{\prime }\mathbf{Y}-\mathbb{E}(\mathbf{X}%
^{\prime }\mathbf{Y})\right] \right\Vert _{F}=O_{p}(\sqrt{a_{T}}).
\end{equation*}
\end{lemma}

\begin{proof}
For any $B > 0 $, by the Markov's inequality 
\begin{equation*}
\Pr \left( \left\| T^{-1} \left[\mathbf{X}^{\prime }\mathbf{Y} - \mathbb{E}
( \mathbf{X}^{\prime }\mathbf{Y}) \right] \right\|_{F} > B \sqrt{a_T}
\right) \leq \frac{\mathbb{E} \left\| T^{-1} \left[\mathbf{X}^{\prime }%
\mathbf{Y} - \mathbb{E} (\mathbf{X}^{\prime }\mathbf{Y}) \right]
\right\|_{F}^2}{a_T B^2}
\end{equation*}
Since $\mathbb{E} \left\| T^{-1} \left[\mathbf{X}^{\prime }\mathbf{Y} - 
\mathbb{E} (\mathbf{X}^{\prime }\mathbf{Y}) \right] \right\|_{F}^2 = O(a_T) $
, there exist $C $ and $T_0 $ such that for all $T > T_0 $ 
\begin{equation*}
\mathbb{E} \left\| T^{-1} \left[\mathbf{X}^{\prime }\mathbf{Y} - \mathbb{E}
( \mathbf{X}^{\prime }\mathbf{Y}) \right] \right\|_{F}^2 \leq C a_T.
\end{equation*}
Hence, for any $\varepsilon > 0 $, there exist $B_{\varepsilon} = \sqrt{ 
\frac{C}{\varepsilon}} $ and $T_{\varepsilon} = T_{0} $, such that for all $%
T > T_{\varepsilon} $ 
\begin{equation*}
\Pr \left( \left\| T^{-1} \left[\mathbf{X}^{\prime }\mathbf{Y} - \mathbb{E}
( \mathbf{X}^{\prime }\mathbf{Y}) \right] \right\|_{F} > B_{\varepsilon} 
\sqrt{ a_T} \right) \leq \varepsilon.
\end{equation*}
Therefore, 
\begin{equation*}
\left\| T^{-1} \left[\mathbf{X}^{\prime }\mathbf{Y} - \mathbb{E} (\mathbf{X}
^{\prime }\mathbf{Y}) \right] \right\|_{F} = O_p\left(\sqrt{a_T}\right).
\end{equation*}
\end{proof}


\begin{lemma}
\label{Op inv sample cov dev} Let $\boldsymbol{\Sigma }_{T}$ be a positive
definite matrix and $\hat{\boldsymbol{\Sigma }}_{T}$ be its corresponding
estimator. Suppose that $\lambda _{\min }\left( \boldsymbol{\Sigma }%
_{T}\right) >c>0$, and 
\begin{equation*}
\mathbb{E}\left\Vert \hat{\boldsymbol{\Sigma }}_{T}-\boldsymbol{\Sigma }%
_{T}\right\Vert _{F}^{2}=O(a_{T})
\end{equation*}%
where $a_{T}>0$, and $a_{T}=o(1)$. Then 
\begin{equation*}
\left\Vert \hat{\boldsymbol{\Sigma }}_{T}^{-1}-\boldsymbol{\Sigma }%
_{T}^{-1}\right\Vert _{F}=O_{p}(\sqrt{a_{T}})
\end{equation*}
\end{lemma}

\begin{proof}
Let $\mathcal{A}_{T} = \left\{ \left\| \boldsymbol{\Sigma}_{T}^{-1}
\right\|_2 \left\| \hat{\boldsymbol{\Sigma}}_{T} - \boldsymbol{\Sigma}_{T}
\right\|_F < 1 \right\} $, $\mathcal{B}_{T} = \left\{ \left\| \hat{ 
\boldsymbol{\Sigma}}_{T}^{-1} - \boldsymbol{\Sigma}_{T}^{-1} \right\|_F > B 
\sqrt{a_T} \right\} $ and $\mathcal{D}_{T} = \left\{\frac{\left\| 
\boldsymbol{\Sigma}_{T}^{-1} \right\|_2^2 \left\| \hat{\boldsymbol{\Sigma}}
_{T} - \boldsymbol{\Sigma}_{T} \right\|_F } {\left(1 - \left\| \boldsymbol{\
\Sigma}_{T}^{-1} \right\|_2 \left\| \hat{\boldsymbol{\Sigma}}_{T} - 
\boldsymbol{\Sigma}_{T} \right\|_F\right)} > B \sqrt{a_T} \right\} $ where $%
B > 0 $ is an arbitrary constant. If $\mathcal{A}_{T} $ holds, by Lemma \ref%
{inverse_matrices_diff_F_norm}, 
\begin{equation*}
\left\| \hat{\boldsymbol{\Sigma}}_{T}^{-1} - \boldsymbol{\Sigma}_{T}^{-1}
\right\|_F \leq \frac{\left\| \boldsymbol{\Sigma}_{T}^{-1} \right\|_2^2
\left\| \hat{\boldsymbol{\Sigma}}_{T} - \boldsymbol{\Sigma}_{T} \right\|_F 
} {1 - \left\| \boldsymbol{\Sigma}_{T}^{-1} \right\|_2 \left\| \hat{ 
\boldsymbol{\Sigma}}_{T} - \boldsymbol{\Sigma}_{T} \right\|_F}.
\end{equation*}
Hence $\mathcal{B}_{T} \cap \mathcal{A}_{T} \subseteq \mathcal{D}_T $.
Therefore 
\begin{equation*}
\begin{split}
& \Pr(\mathcal{B}_{T} \cap \mathcal{A}_{T} ) \leq \Pr \left( \frac{ \left\| 
\boldsymbol{\Sigma}_{T}^{-1} \right\|_2^2 \left\| \hat{\boldsymbol{\Sigma}}
_{T} - \boldsymbol{\Sigma}_{T} \right\|_F }{ \left(1 - \left\| \boldsymbol{\
\Sigma}_{T}^{-1} \right\|_2 \left\| \hat{\boldsymbol{\Sigma}}_{T} - 
\boldsymbol{\Sigma}_{T} \right\|_F \right)} > B \sqrt{a_T} \right) \\
& \qquad \qquad \quad \ = \Pr \left( \left\| \hat{\boldsymbol{\Sigma}}_{T} - 
\boldsymbol{\Sigma}_{T} \right\|_F > \frac{B \sqrt{a_T}}{\left\| \boldsymbol{%
\ \Sigma}_{T}^{-1} \right\|_2 \left( \left\| \boldsymbol{\Sigma}_{T}^{-1}
\right\|_2 + B \sqrt{a_T} \right) } \right)
\end{split}%
\end{equation*}
By the Markov's inequality, we can further conclude that 
\begin{align*}
\Pr(\mathcal{B}_{T} \cap \mathcal{A}_{T} ) \leq \frac{\mathbb{E}\left\| \hat{
\boldsymbol{\Sigma}}_{T} - \boldsymbol{\Sigma}_{T} \right\|_{F}^{2} }{a_T}
\times \frac{\left\| \boldsymbol{\Sigma}_{T}^{-1} \right\|_{2}^{2} \left(
\left\| \boldsymbol{\Sigma}_{T}^{-1} \right\|_2 + B \sqrt{a_T} \right)^{2}}{
B^2}.
\end{align*}
Since by assumption $\mathbb{E}\left\| \hat{\boldsymbol{\Sigma}}_{T} - 
\boldsymbol{\Sigma}_{T} \right\|_{F}^{2} = O(a_T)$, there exist C and $T_{0}
> 0 $ such that for all $T > T_{0} $, 
\begin{equation*}
\mathbb{E} \left\| \hat{ \boldsymbol{\Sigma}}_{T} - \boldsymbol{\Sigma}_{T}
\right\|_{F}^2 \leq C a_T.
\end{equation*}
Therefore, for all $T > T_{0} $, 
\begin{equation*}
\Pr(\mathcal{B}_{T} \cap \mathcal{A}_{T} ) \leq \frac{ C \left\| \boldsymbol{%
\ \Sigma}_{T}^{-1} \right\|_{2}^{2} \left( \left\| \boldsymbol{\Sigma}
_{T}^{-1} \right\|_2 + B \sqrt{a_T} \right)^{2}}{B^2}.
\end{equation*}
Moreover, 
\begin{equation*}
\Pr(\mathcal{A}_T^c) = \Pr \left( \left\| \boldsymbol{\Sigma}_{T}^{-1}
\right\|_2 \left\| \hat{\boldsymbol{\Sigma}}_{T} - \boldsymbol{\Sigma}_{T}
\right\|_F \geq 1 \right) = \Pr\left( \left\| \hat{\boldsymbol{\Sigma}}_{T}
- \boldsymbol{\Sigma}_{T} \right\|_F \geq \frac{1}{\left\| \boldsymbol{%
\Sigma }_{T}^{-1} \right\|_2} \right).
\end{equation*}
By the Markov's inequality, we can further write 
\begin{equation*}
\Pr(\mathcal{A}_T^c) \leq \left\| \boldsymbol{\Sigma}_{T}^{-1}
\right\|_{2}^2 \times \mathbb{E} \left\| \hat{\boldsymbol{\Sigma}}_{T} - 
\boldsymbol{\Sigma}_{T} \right\|_F^2,
\end{equation*}
and hence, for all $T > T_{0} $, 
\begin{equation*}
\Pr(\mathcal{A}_T^c) \leq C \left\| \boldsymbol{\Sigma}_{T}^{-1}
\right\|_{2}^2 a_T.
\end{equation*}
Note that 
\begin{equation*}
\Pr( \mathcal{B}_{T} ) = \Pr \left(\mathcal{B}_{T} \cap \mathcal{A}
_{T}\right) + \Pr( \mathcal{B}_{T} | \mathcal{A}_{T}^c ) \Pr ( \mathcal{A}
_{T}^c ),
\end{equation*}
and since $\Pr( \mathcal{B}_T \cap \mathcal{A}_{T} ) \leq \Pr(\mathcal{D}_T)$
and $\Pr( \mathcal{B}_{T} | \mathcal{A}_{T}^c ) \leq 1 $, we have 
\begin{equation*}
\Pr( \mathcal{B}_{T} ) \leq \Pr( \mathcal{B}_T \cap \mathcal{A}_{T} ) + \Pr
( \mathcal{A}_{T}^c ).
\end{equation*}
Therefore, for all $T > T_{0} $, 
\begin{equation*}
\Pr\left( \left\| \hat{\boldsymbol{\Sigma}}_{T}^{-1} - \boldsymbol{\Sigma}
_{T}^{-1} \right\|_F > B \sqrt{a_T} \right) \leq \frac{ C \left\| 
\boldsymbol{\Sigma}_{T}^{-1} \right\|_{2}^{2} \left( \left\| \boldsymbol{\
\Sigma}_{T}^{-1} \right\|_2 + B \sqrt{a_T} \right)^{2}}{B^2} + C \left\| 
\boldsymbol{\Sigma}_{T}^{-1} \right\|_{2}^2 a_T.
\end{equation*}
Now, for a given $\varepsilon > 0 $, we are interested to find $%
B_{\varepsilon} > 0 $ and $T_{\varepsilon} > 0 $ such that for all $T >
T_{\varepsilon} $, 
\begin{equation*}
\Pr\left( \left\| \hat{\boldsymbol{\Sigma}}_{T}^{-1} - \boldsymbol{\Sigma}
_{T}^{-1} \right\|_F > B_{\varepsilon} \sqrt{a_T} \right) \leq \varepsilon.
\end{equation*}
To do so, we first find a value of $B $ such that 
\begin{equation*}
\frac{ C \left\| \boldsymbol{\Sigma}_{T}^{-1} \right\|_{2}^{2} \left(
\left\| \boldsymbol{\Sigma}_{T}^{-1} \right\|_2 + B \sqrt{a_T} \right)^{2}}{
B^2} + C \left\| \boldsymbol{\Sigma}_{T}^{-1} \right\|_{2}^2 a_T =
\varepsilon.
\end{equation*}
By multiplying both sides of the above equality by $B^2 $ and bringing all
the equations to the left hand side we have 
\begin{equation*}
\left( \varepsilon - 2 C \left\| \boldsymbol{\Sigma}_{T}^{-1}
\right\|_{2}^{2} a_T \right)B^2 - 2 C \left\| \boldsymbol{\Sigma}_{T}^{-1}
\right\|_{2}^{3} \sqrt{a_T} B - C \left\| \boldsymbol{\Sigma}_{T}^{-1}
\right\|_{2}^{4} = 0.
\end{equation*}
By solving the above quadratic equation of $B $ we have 
\begin{align*}
B^{\ast} &= \frac{2 C \left\| \boldsymbol{\Sigma}_{T}^{-1} \right\|_{2}^{3} 
\sqrt{a_T} \pm \sqrt{4 C \left\| \boldsymbol{\Sigma}_{T}^{-1}
\right\|_{2}^{4} \varepsilon - 4 C^2 \left\| \boldsymbol{\Sigma}_{T}^{-1}
\right\|_{2}^{6} a_T }}{2 \left( \varepsilon - 2 C \left\| \boldsymbol{%
\Sigma }_{T}^{-1} \right\|_{2}^{2} a_T \right)} \\
& = \frac{ \left\| \boldsymbol{\Sigma}_{T}^{-1} \right\|_{2} \left(\sqrt{a_T}
\pm \sqrt{ \frac{\varepsilon}{C \left\| \boldsymbol{\Sigma}_{T}^{-1}
\right\|_{2}^{2}} - a_T }\right)}{ \frac{\varepsilon}{C \left\| \boldsymbol{%
\ \Sigma}_{T}^{-1} \right\|_{2}^{2}} - 2 a_T }.
\end{align*}
Notice that $a_T \rightarrow 0 $ as $T \rightarrow \infty $, therefore for
large enough $T^{\ast} $ we have both $\frac{\varepsilon}{C \left\| 
\boldsymbol{\Sigma}_{T}^{-1} \right\|_{2}^{2}} - 2 a_T $ and $\frac{
\varepsilon}{C \left\| \boldsymbol{\Sigma}_{T}^{-1} \right\|_{2}^{2}} - a_T $
being greater than zero for all $T > T^* $. Now, by setting $T_{\varepsilon}
= \max\{T^{\ast}, T_{0}\} $ and 
\begin{equation*}
B_{\varepsilon} = \frac{ \left\| \boldsymbol{\Sigma}_{T}^{-1} \right\|_{2}
\left(\sqrt{a_T} + \sqrt{ \frac{\varepsilon}{C \left\| \boldsymbol{\Sigma}
_{T}^{-1} \right\|_{2}^{2}} - a_T }\right)}{ \frac{\varepsilon}{C \left\| 
\boldsymbol{\Sigma}_{T}^{-1} \right\|_{2}^{2}} - 2 a_T } > 0,
\end{equation*}
we achieve our goal that for all $T > T_{\varepsilon} $, 
\begin{equation*}
\Pr\left( \left\| \hat{\boldsymbol{\Sigma}}_{T}^{-1} - \boldsymbol{\Sigma}
_{T}^{-1} \right\|_F > B_{\varepsilon} \sqrt{a_T} \right) \leq \varepsilon.
\end{equation*}
\end{proof}

By using Lemma \ref{inverse_matrices_diff_F_norm} we achieve the probability
convergence order for $\left\| \hat{\boldsymbol{\Sigma}}_{T}^{-1} - 
\boldsymbol{\Sigma}_{T}^{-1} \right\|_F $ that is sharper than the one shown
in the proof Lemma A21 of Chudik et al. (2018) (see equations (B.103) and
(B.105) of Chudik et al. (2018) Online Theory Supplement).


\begin{lemma}
\label{double sum prob bound} Let $z_{ij}$ be a random variable for $%
i=1,2,\cdots ,N$, and $j=1,2,\cdots ,N$. Then, for any $d_{T}>0$, 
\begin{equation*}
\textstyle\Pr (N^{-2}\sum_{i=1}^{N}\sum_{j=1}^{N}|z_{ij}|>d_{T})\leq
N^{2}\sup_{i,j}\Pr (|z_{ij}|>d_{T}).
\end{equation*}
\end{lemma}

\begin{proof}
We know that $N^{-2} \sum_{i=1}^{N} \sum_{j=1}^{N} |z_{ij}| \leq \sup_{i,j}
|z_{ij}| $. Therefore, 
\begin{equation*}
\begin{split}
& \textstyle \Pr(N^{-2} \sum_{i=1}^{N} \sum_{j=1}^{N} |z_{ij}| > d_T) \leq
\Pr( \sup_{i,j} |z_{ij}| > d_T) \\
& \textstyle \qquad \leq \Pr[ \cup_{i = 1}^{N} \cup_{j = 1}^{N} ( |z_{ij}| >
d_T) ] \leq \sum_{i=1}^{N} \sum_{j=1}^{N} \Pr( |z_{ij}| > d_T ) \\
& \textstyle \qquad \leq N^2 \sup_{i,j} \Pr( |z_{ij}| > d_T).
\end{split}%
\end{equation*}
\end{proof}


\begin{lemma}
\label{inverse_est_matrix_F_norm_bound_1} Let $\hat{\boldsymbol{\Sigma }}$
be an estimator of a $N\times N$ symmetric invertible matrix $\boldsymbol{\
\Sigma }$. Suppose that there exits a finite positive constant $C_{0}$, such
that 
\begin{equation*}
\sup_{i,j}\Pr (|\hat{\sigma}_{ij}-\sigma_{ij}| > d_{T}) \leq \exp
(-C_{0}Td_{T}^{2}),\ \text{for any}\ d_{T}>0,
\end{equation*}
where $\sigma_{ij}$ and $\hat{\sigma}_{ij}$ are the elements of $\boldsymbol{%
\ \Sigma } $ and $\hat{\boldsymbol{\Sigma }}$ respectively. Then, for any $%
b_{T}>0$, 
\begin{equation*}
\begin{split}
& \Pr (\Vert \hat{\boldsymbol{\Sigma }}^{-1}-\boldsymbol{\Sigma }^{-1}\Vert
_{F}>b_{T})\leq N^{2}\exp \left[ -C_{0}\frac{Tb_{T}^{2}}{N^{2}\Vert 
\boldsymbol{\Sigma }^{-1}\Vert _{2}^{2}(\Vert \boldsymbol{\Sigma }^{-1}\Vert
_{2}+b_{T})^{2}}\right] + \\
& \qquad \qquad \qquad \qquad \qquad \qquad N^{2}\exp \left( -C_{0}\frac{T}{
N^{2}\Vert \boldsymbol{\Sigma }^{-1}\Vert _{2}^{2}}\right) .
\end{split}%
\end{equation*}
\end{lemma}

\begin{proof}
Let $\mathcal{A}_{N} = \{ \| \boldsymbol{\Sigma}^{-1} \|_2 \| \boldsymbol{\ 
\hat{\Sigma}} - \boldsymbol{\Sigma} \|_F \leq 1 \} $ and $\mathcal{B}_{N} =
\{ \| \boldsymbol{\hat{\Sigma}}^{-1} - \boldsymbol{\Sigma}^{-1} \|_F > b_T
\} $, and note that by Lemma \ref{inverse_matrices_diff_F_norm} if $\mathcal{%
\ A}_{N} $ holds we have 
\begin{equation*}
\| \boldsymbol{\hat{\Sigma}}^{-1} - \boldsymbol{\Sigma}^{-1} \|_F \leq \frac{
\| \boldsymbol{\Sigma}^{-1} \|_2^2 \| \boldsymbol{\hat{\Sigma}} - 
\boldsymbol{\Sigma} \|_F }{1 - \| \boldsymbol{\Sigma}^{-1} \|_2 \| 
\boldsymbol{\hat{\Sigma}} - \boldsymbol{\Sigma} \|_F}.
\end{equation*}
Hence 
\begin{equation*}
\begin{split}
& \Pr(\mathcal{B}_{N} | \mathcal{A}_{N} ) \leq \Pr \left( \frac{\| 
\boldsymbol{\Sigma}^{-1} \|_2^2 \| \boldsymbol{\hat{\Sigma}} - \boldsymbol{\
\Sigma} \|_F }{1 - \| \boldsymbol{\Sigma}^{-1} \|_2 \| \boldsymbol{\hat{
\Sigma}} - \boldsymbol{\Sigma} \|_F} > b_T \right) \\
& \qquad \qquad \quad \ = \Pr \left[ \| \boldsymbol{\hat{\Sigma}} - 
\boldsymbol{\Sigma} \|_F > \frac{b_T}{ \| \boldsymbol{\Sigma}^{-1} \|_2 ( \| 
\boldsymbol{\Sigma}^{-1} \|_2 + b_T ) } \right].
\end{split}%
\end{equation*}
Note that $\| \boldsymbol{\hat{\Sigma}} - \boldsymbol{\Sigma} \|_F = \left(
\sum_{i=1}^{N} \sum_{j=1}^{N} (\hat{\sigma}_{ij} - \sigma_{ij})^2
\right)^{1/2}$. Therefore, 
\begin{equation*}
\begin{split}
& \Pr(\mathcal{B}_{N} | \mathcal{A}_{N} ) \leq \Pr \left[ \left(
\sum_{i=1}^{N} \sum_{j=1}^{N} (\hat{\sigma}_{ij} - \sigma_{ij})^2
\right)^{1/2} > \frac{b_T}{ \| \boldsymbol{\Sigma}^{-1} \|_2 ( \| 
\boldsymbol{\Sigma}^{-1} \|_2 + b_T ) } \right] \\
& \qquad \qquad \quad \ = \Pr \left[ \sum_{i=1}^{N} \sum_{j=1}^{N} (\hat{
\sigma}_{ij} - \sigma_{ij})^2 > \frac{b_T^2}{ \| \boldsymbol{\Sigma}^{-1}
\|_2^2 ( \| \boldsymbol{\Sigma}^{-1} \|_2 + b_T )^2 } \right].
\end{split}%
\end{equation*}
By Lemma \ref{double sum prob bound}, we can further write, 
\begin{equation*}
\begin{split}
& \Pr(\mathcal{B}_{N} | \mathcal{A}_{N} ) \leq N^2 \sup_{i,j} \Pr \left[ ( 
\hat{\sigma}_{ij} - \sigma_{ij})^2 > \frac{b_T^2}{ N^2 \| \boldsymbol{\Sigma}
^{-1} \|_2^2 ( \| \boldsymbol{\Sigma}^{-1} \|_2 + b_T )^2 } \right] \\
& \qquad \qquad \quad \ = N^2 \sup_{i,j} \Pr \left[ | \hat{\sigma}_{ij} -
\sigma_{ij} | > \frac{b_T}{ N \| \boldsymbol{\Sigma}^{-1} \|_2 ( \| 
\boldsymbol{\Sigma} ^{-1} \|_2 + b_T ) } \right] \\
& \qquad \qquad \quad \ \leq N^2 \exp \left[ -C_0 \frac{ T b_T^2}{ N^2 \| 
\boldsymbol{\Sigma}^{-1} \|_2^2 ( \| \boldsymbol{\Sigma}^{-1} \|_2 + b_T )^2 
}\right]
\end{split}%
\end{equation*}
Furthermore, 
\begin{equation*}
\begin{split}
& \Pr( \mathcal{A}_{N}^c ) = \Pr( \| \boldsymbol{\Sigma}^{-1} \|_2 \| 
\boldsymbol{\hat{\Sigma}} - \boldsymbol{\Sigma} \|_F > 1 ) \\
& \qquad \quad \ \ = \Pr( \| \boldsymbol{\hat{\Sigma}} - \boldsymbol{\Sigma}
\|_F > \| \boldsymbol{\Sigma}^{-1} \|_2^{-1} ) \\
& \qquad \quad \ \ = \Pr \left[ \left(\sum_{i=1}^{N} \sum_{j=1}^{N} (\hat{
\sigma} _{ij} - \sigma_{ij})^2 \right)^{1/2} > \| \boldsymbol{\Sigma}^{-1}
\|_2^{-1} \right] \\
& \qquad \quad \ \ = \Pr \left[ \sum_{i=1}^{N} \sum_{j=1}^{N} (\hat{\sigma}
_{ij} - \sigma_{ij})^2 > \| \boldsymbol{\Sigma}^{-1} \|_2^{-2} \right] \\
& \qquad \quad \ \ \leq N^2 \sup_{i,j} \Pr \left[ (\hat{\sigma}_{ij} -
\sigma_{ij})^2 > \frac{1}{N^2 \| \boldsymbol{\Sigma}^{-1} \|_2^{2}} \right]
\\
& \qquad \quad \ \ \leq N^2 \sup_{i,j} \Pr \left[ | \hat{\sigma}_{ij} -
\sigma_{ij} | > \frac{1}{N \| \boldsymbol{\Sigma}^{-1} \|_2} \right] \\
& \qquad \quad \ \ \leq N^2 \exp \left[ -C_0 \frac{T}{N^2 \| \boldsymbol{\
\Sigma}^{-1} \|_2^2 } \right].
\end{split}%
\end{equation*}
Note that 
\begin{equation*}
\Pr( \mathcal{B}_{N} ) = \Pr( \mathcal{B}_{N} | \mathcal{A}_{N} ) \Pr( 
\mathcal{A}_{N} ) + \Pr( \mathcal{B}_{N} | \mathcal{A}_{N}^c ) \Pr ( 
\mathcal{A}_{N}^c ),
\end{equation*}
and since $\Pr( \mathcal{A}_{N} ) $ and $\Pr( \mathcal{B}_{N} | \mathcal{A}
_{N}^c ) $ are less than equal to one, we have 
\begin{equation*}
\Pr( \mathcal{B}_{N} ) \leq \Pr( \mathcal{B}_{N} | \mathcal{A}_{N} ) + \Pr ( 
\mathcal{A}_{N}^c ).
\end{equation*}
Therefore, 
\begin{equation*}
\Pr( \mathcal{B}_{NT} ) \leq N^2 \exp \left[ -C_0 \frac{ T b_T^2}{ N^2 \| 
\boldsymbol{\Sigma}^{-1} \|_2^2 ( \| \boldsymbol{\Sigma}^{-1} \|_2 + b_T )^2 
}\right] + N^2 \exp \left[ -C_0 \frac{T}{N^2 \| \boldsymbol{\Sigma}^{-1}
\|_2^2 } \right].
\end{equation*}
\end{proof}


\begin{lemma}
\label{inverse_est_matrix_F_norm_bound_2} Let $\hat{\boldsymbol{\Sigma }}$
be an estimator of a $N\times N$ symmetric invertible matrix $\boldsymbol{\
\Sigma }$. Suppose that there exits a finite positive constant $C_{0}$, such
that 
\begin{equation*}
\sup_{i,j}\Pr (|\hat{\sigma}_{ij}- \sigma_{ij}|>d_{T})\leq \exp \left[
-C_{0}(Td_{T})^{s/s+2}\right] ,\ \text{for any}\ d_{T}>0,
\end{equation*}
where $\sigma_{ij}$ and $\hat{\sigma}_{ij}$ are the elements of $\boldsymbol{%
\ \Sigma } $ and $\hat{\boldsymbol{\Sigma }}$ respectively. Then, for any $%
b_{T}>0$, 
\begin{equation*}
\begin{split}
& \Pr (\Vert \hat{\boldsymbol{\Sigma }}^{-1}-\boldsymbol{\Sigma }^{-1}\Vert
_{F}>b_{T})\leq N^{2}\exp \left[ -C_{0}\frac{(Tb_{T})^{s/s+2}}{
N^{s/s+2}\Vert \boldsymbol{\Sigma }^{-1}\Vert _{2}^{s/s+2}(\Vert \boldsymbol{%
\ \ \Sigma }^{-1}\Vert _{2}+b_{T})^{s/s+2}}\right] + \\
& \qquad \qquad \qquad \qquad \qquad \qquad N^{2}\exp \left( -C_{0}\frac{
T^{s/s+2}}{N^{s/s+2}\Vert \boldsymbol{\Sigma }^{-1}\Vert _{2}^{s/s+2}}
\right) .
\end{split}%
\end{equation*}
\end{lemma}

\begin{proof}
The proof is similar to the proof of Lemma \ref%
{inverse_est_matrix_F_norm_bound_1}.
\end{proof}


\begin{lemma}
\label{project reg md} Let $\{x_{it}\}_{t=1}^{T}$ for $i=1,2,\cdots ,N$ and $%
\{z_{jt}\}_{t=1}^{T}$ for $j=1,2,\cdots ,m$ be time-series processes. Also
let $\mathcal{F}_{it}^{x}=\sigma (x_{it},x_{i,t-1},\cdots )$ for $%
i=1,2,\cdots ,N$, $\mathcal{F}_{jt}^{z}=\sigma (z_{jt},z_{j,t-1},\cdots )$
for $j=1,2,\cdots ,m$, $\mathcal{F}_{t}^{x}=\cup _{i=1}^{N}\mathcal{F}
_{it}^{x}$, $\mathcal{F}_{t}^{z}=\cup _{j=1}^{m}\mathcal{F}_{jt}^{z}$, and $%
\mathcal{F}_{t}=\mathcal{F}_{t}^{x}\cup \mathcal{F}_{t}^{z}$. Define the
projection regression of $x_{it}$ on $\mathbf{z}_{t}=(z_{1t},z_{2t},\cdots
,\allowbreak z_{m,t})^{\prime }$ as 
\begin{equation*}
x_{it}=\mathbf{z}_{t}^{\prime }\boldsymbol{\bar{\psi}}_{i,T}+\tilde{x}_{it}
\end{equation*}
where $\boldsymbol{\bar{\psi}}_{i,T}=(\psi_{1i,T},\psi_{2i,T},\cdots ,\psi
_{mi,T})^{\prime }$ is the $m\times 1$ vector of projection coefficients
which is equal to 
\begin{equation*}
\left[ T^{-1}\sum_{t=1}^{T}\mathbb{E}\left( \mathbf{z} _{t} \mathbf{z}%
_{t}^{\prime }\right) \right] ^{-1}[T^{-1}\sum_{t=1}^{T} \mathbb{E}( \mathbf{%
z}_{t}x_{it})].
\end{equation*}
Suppose, $\mathbb{E}[x_{it}x_{i^{\prime }t}-\mathbb{\ E}(x_{it}x_{i^{\prime
}t})|\mathcal{F}_{t-1}]=0$ for all $i,i^{\prime }=1,2,\cdots ,N$, $\mathbb{E}%
[z_{jt}z_{j^{\prime }t}-\mathbb{E} (z_{jt}z_{j^{\prime }t})|\mathcal{F}%
_{t-1}]=0$ for all $j,j^{\prime }=1,2,\cdots ,m$, and $\mathbb{E}%
[z_{jt}x_{it}-\mathbb{E}(z_{jt}x_{it})| \mathcal{F}_{t-1}]=0$ for all $%
j=1,2,\cdots ,m$ and for all $i=1,2,\cdots ,N$. Then 
\begin{equation*}
\mathbb{E}\left[ \tilde{x}_{it}\tilde{x}_{i^{\prime}t}-\mathbb{E}(\tilde{x}%
_{it}\tilde{x}_{i^{\prime}t})|\mathcal{F}_{t-1}\right] =0,
\end{equation*}
for all $j,j^{\prime }=1,2,\cdots ,N$, 
\begin{equation*}
\mathbb{E}\left[ \tilde{x}_{it}z_{jt}-\mathbb{E}(\tilde{x}_{it}z_{jt})|%
\mathcal{F} _{t-1}\right] =0,
\end{equation*}
for all $i=1,2,\cdots ,N$ and $j=1,2,\cdots ,m$, and 
\begin{equation*}
\textstyle T^{-1}\sum_{t=1}^{T}\mathbb{E}(\tilde{x}_{it}z_{jt})=0,
\end{equation*}
for all $i=1,2,\cdots ,N$ and $j=1,2,\cdots ,m$.
\end{lemma}

\begin{proof}
\begin{equation*}
\begin{split}
\mathbb{E}( \tilde{x}_{it} \tilde{x}_{i^{\prime }t} | \mathcal{F}_{t-1} ) & %
\textstyle = \mathbb{E}( x_{it} x_{i^{\prime }t} | \mathcal{F}_{t-1} ) - 
\mathbb{E}( x_{it} \mathbf{z}_t^{\prime }| \mathcal{F}_{t-1} ) \boldsymbol{%
\bar{\psi}}_{i^{\prime },T} - \\
& \quad \ \mathbb{E}( x_{i^{\prime }t} \mathbf{z}_t^{\prime }| \mathcal{F}
_{t-1} ) \boldsymbol{\bar{\psi}}_{i,T} + \boldsymbol{\bar{\psi}}%
_{i,T}^{\prime }\mathbb{\ E }( \mathbf{z}_t \mathbf{z}_t^{\prime }| \mathcal{%
F}_{t-1} ) \boldsymbol{\bar{\psi}}_{i^{\prime },T} \\
& \textstyle = \mathbb{E}( x_{it} x_{i^{\prime }t} ) - \mathbb{E}( x_{it} 
\mathbf{z}_t^{\prime }) \boldsymbol{\bar{\psi}}_{i^{\prime },T} - \mathbb{E}%
( x_{i^{\prime }t} \mathbf{z}_t^{\prime }) \boldsymbol{\bar{\psi}}_{i,T} + \\
& \quad \ \boldsymbol{\bar{\psi}}_{i,T}^{\prime }\mathbb{E}( \mathbf{z}_t 
\mathbf{\ z }_t^{\prime }) \boldsymbol{\bar{\psi}}_{i^{\prime },T} = \mathbb{%
E}( \tilde{x}_{it} \tilde{x}_{i^{\prime }t} ).
\end{split}%
\end{equation*}
\begin{equation*}
\begin{split}
\mathbb{E}( \tilde{x}_{it} z_{jt} | \mathcal{F}_{t-1} ) & \textstyle = 
\mathbb{E}( x_{it} z_{jt} | \mathcal{F}_{t-1} ) - \mathbb{E}( \mathbf{z}%
_t^{\prime }z_{jt} | \mathcal{F}_{t-1} ) \boldsymbol{\bar{\psi}}_{i,T} \\
& \textstyle = \mathbb{E}( x_{it} z_{jt}) - \mathbb{E}( \mathbf{z}_t^{\prime
}z_{jt}) \boldsymbol{\bar{\psi}}_{i,T} = \mathbb{E}( \tilde{x}_{it} z_{it} ).
\end{split}%
\end{equation*}
\begin{equation*}
\begin{split}
\textstyle T^{-1} \sum_{t=1}^{T} \mathbb{E} (\tilde{x}_{it} \mathbf{z}_t ) & %
\textstyle = T^{-1} \sum_{t=1}^{T} \mathbb{E}( x_{it} \mathbf{z}_t ) - 
\boldsymbol{\bar{\psi}}_{i,T}^{\prime -1} \sum_{t=1}^{T} \mathbb{E}( \mathbf{%
z}_t \mathbf{z}_t^{\prime })] \\
& \textstyle = T^{-1} \sum_{t=1}^{T} \mathbb{E}( x_{it} \mathbf{z}_t ) -
T^{-1} \sum_{t=1}^{T} \mathbb{E}( x_{it} \mathbf{z}_t ) = \mathbf{0}.
\end{split}%
\end{equation*}
\end{proof}


\begin{lemma}
\label{project reg subg} Let $\{x_{it}\}_{t=1}^{T}$ for $i=1,2,\cdots ,N$
and $\{z_{jt}\}_{t=1}^{T}$ for $j=1,2,\cdots ,m$ be time-series processes.
Define the projection regression of $x_{it}$ on $\mathbf{z}
_{t}=(z_{1t},z_{2t},\cdots ,\allowbreak z_{m,t})^{\prime }$ as 
\begin{equation*}
x_{it}=\mathbf{z}_{t}^{\prime }\boldsymbol{\bar{\psi}}_{i,T}+\tilde{x} _{it}
\end{equation*}
where $\boldsymbol{\bar{\psi}}_{i,T}=(\psi_{1i,T},\psi_{2i,T},\cdots
,\psi_{mi,T})^{\prime }$ is the $m\times 1$ vector of projection
coefficients which is equal to 
\begin{equation*}
\left[ T^{-1}\sum_{t=1}^{T}\mathbb{E}\left( \mathbf{z} _{t} \mathbf{z}%
_{t}^{\prime }\right) \right] ^{-1}[T^{-1}\sum_{t=1}^{T} \mathbb{E}( \mathbf{%
z}_{t}x_{it})].
\end{equation*}
Suppose that only a finite number of elements in $\boldsymbol{\bar{\psi}}%
_{i,T}$ is different from zero for all $i=1,2,\cdots ,N$ and there exist
sufficiently large positive constants $C_{0} $ and $C_{1}$, and $s>0$ such
that

\begin{enumerate}
\item[(i)] $\sup_{j,t}\Pr (\lvert z_{jt}\rvert >\alpha )\leq C_{0}\exp
(-C_{1}\alpha ^{s}),\ for\ all\ \alpha >0$, and

\item[(ii)] $\sup_{i,t}\Pr (\lvert x_{it}\rvert >\alpha )\leq C_{0}\exp
(-C_{1}\alpha ^{s}),\ for\ all\ \alpha >0.$
\end{enumerate}

Then, there exist sufficiently large positive constants $C_0 $ and $C_1 $,
and $s > 0 $ such that 
\begin{equation*}
\sup_{i,t} \Pr( \lvert \tilde{x}_{it} \rvert > \alpha) \leq C_0 \exp (-C_1
\alpha^{s}), \ for \ all \ \alpha > 0.
\end{equation*}
\end{lemma}

\begin{proof}
Without loss of generality assume that the first finite $\ell $ elements of $%
\psi_{i,T} $ are different from zero and write 
\begin{equation*}
\textstyle x_{it} = \sum_{j=1}^{\ell} \psi_{j i,T} z_{jt} + \tilde{x}_{it}.
\end{equation*}
Now, note that 
\begin{equation*}
\textstyle \Pr ( \lvert \tilde{x}_{it} \rvert > \alpha ) \leq \Pr \left(
\lvert x_{it} \rvert + \sum_{j=1}^{\ell} \lvert \psi_{j i,T} z_{jt} \rvert >
\alpha \right),
\end{equation*}
and hence by Lemma \ref{prob_sum}, for any $0 < \pi_j < 1 $, $j = 1, 2,
\cdots, \ell +1 $ we have, 
\begin{equation*}
\begin{split}
\Pr ( \lvert \tilde{x}_{it} \rvert > \alpha ) & \textstyle \leq
\sum_{j=1}^{\ell} \Pr \left( \lvert \psi_{j i,T} z_{jt} \rvert > \pi_j
\alpha \right) + \Pr \left( \lvert x_{it} \rvert > \pi_{\ell+1} \alpha
\right) \\
& \textstyle = \sum_{j=1}^{\ell} \Pr \left( \lvert z_{jt} \rvert > \lvert
\psi_{j i,T} \rvert^{-1} \pi_j \alpha \right) + \Pr \left( \lvert x_{it}
\rvert > \pi_{\ell+1} \alpha \right) \\
& \textstyle \leq \ell \sup_{j,t} \Pr \left( \lvert z_{jt} \rvert > \lvert
\psi_{T}^* \rvert^{-1} \pi^* \alpha \right) + \sup_{i,t} \Pr \left( \lvert
x_{it} \rvert > \pi^* \alpha \right),
\end{split}%
\end{equation*}
where $\psi_{T}^* = sup_{i,j} \{\psi_{ji,T}\}$ and $\pi^* = \inf_{j \in 1,
2, \cdots, \ell+1} \{\pi_j\} $. Also, there exist sufficiently large
positive constants $C_{0} $ and $C_{1}$, and $s>0$ such that for all $\alpha
>0 $, 
\begin{equation*}
\sup_{j,t}\Pr (\lvert z_{jt}\rvert >\alpha )\leq C_{0}\exp (-C_{1}\alpha
^{s}),
\end{equation*}
and 
\begin{equation*}
\sup_{i,t}\Pr (\lvert x_{it}\rvert >\alpha )\leq C_{0}\exp (-C_{1}\alpha
^{s}).
\end{equation*}
Therefore, 
\begin{equation*}
\begin{split}
\Pr ( \lvert \tilde{x}_{it} \rvert > \alpha ) \leq \ell C_0 \exp (-C_1
\alpha^{s}) + C_0 \exp (-C_1 \alpha^{s}),
\end{split}%
\end{equation*}
and hence there exist sufficiently large positive constants $C_0 $ and $C_1$%
, and $s > 0 $ such that for all $\alpha > 0$, 
\begin{equation*}
\begin{split}
\sup_{i,t} \Pr( \lvert \tilde{x}_{it} \rvert > \alpha) \leq C_0 \exp (-C_1
\alpha^{s}).
\end{split}%
\end{equation*}
\end{proof}


\begin{lemma}
\label{sum_martig_diff_prod_x} Let $\{x_{it}\}_{t=1}^{T}$ for $i=1,2,\cdots
,N$ and $\{z_{\ell t}\}_{t=1}^{T}$ for $\ell =1,2,\cdots ,m$ be time-series
processes and $m=\ominus (T^{d})$. Also let $\mathcal{F}_{it}^{x}=\sigma
(x_{it},x_{i,t-1},\cdots )$ for $i=1,2,\cdots ,N$, $\mathcal{F}_{\ell
t}^{z}=\sigma (z_{\ell t},z_{\ell ,t-1},\cdots )$ for $\ell =1,2,\cdots ,m$, 
$\mathcal{F}_{t}^{x}=\cup _{i=1}^{N}\mathcal{F}_{it}^{x}$, $\mathcal{F}
_{t}^{z}=\cup _{\ell =1}^{m}\mathcal{F}_{\ell t}^{z}$, and $\mathcal{F}_{t}= 
\mathcal{F}_{t}^{x}\cup \mathcal{F}_{t}^{z}$. Define the projection
regression of $x_{it}$ on $\mathbf{z}_{t}=(z_{1t},z_{2t},\cdots ,\allowbreak
z_{m,t})^{\prime }$ as 
\begin{equation*}
x_{it}=\mathbf{z}_{t}^{\prime }\boldsymbol{\bar{\psi}}_{i,T}+\tilde{x}_{it},
\end{equation*}
where $\boldsymbol{\bar{\psi}}_{i,T}=(\psi _{1i,T},\psi _{2i,T},\cdots ,\psi
_{mi,T})^{\prime }$ is the $m\times 1$ vector of projection coefficients
which is equal to 
\begin{equation*}
\left[ T^{-1}\sum_{t=1}^{T}\mathbb{E}\left( \mathbf{z} _{t} \mathbf{z}%
_{t}^{\prime }\right) \right] ^{-1}[T^{-1}\sum_{t=1}^{T} \mathbb{E}( \mathbf{%
z}_{t}x_{it})].
\end{equation*}
Suppose, $\mathbb{E}[x_{it}x_{jt}- \mathbb{E} (x_{it}x_{jt})|\mathcal{F}%
_{t-1}]=0$ for all $i,j=1,2,\cdots ,N$, $\mathbb{E} [z_{\ell t}z_{\ell
^{\prime }t}-\mathbb{E}(z_{\ell t}z_{\ell t})| \mathcal{F} _{t-1}]=0$ for
all $\ell ,\ell ^{\prime }=1,2,\cdots ,m$, and $\mathbb{E} [z_{\ell t}x_{it}-%
\mathbb{E}(z_{\ell t}x_{it})|\mathcal{F} _{t-1}]=0$ for all $\ell
=1,2,\cdots ,m$ and for all $i=1,2,\cdots ,N$. Additionally, assume that
only a finite number of elements in $\boldsymbol{\bar{\psi}}_{i,T}$ is
different from zero for all $i=1,2,\cdots ,N$ and there exist sufficiently
large positive constants $C_{0}$ and $C_{1}$, and $s>0$ such that

\begin{enumerate}
\item[(i)] $\sup_{j,t}\Pr (\lvert z_{\ell t}\rvert >\alpha )\leq C_{0}\exp
(-C_{1}\alpha ^{s}),\ for\ all\ \alpha >0,$ and

\item[(ii)] $\sup_{i,t}\Pr (\lvert x_{\ell t}\rvert >\alpha )\leq C_{0}\exp
(-C_{1}\alpha ^{s}),\ for\ all\ \alpha >0.$
\end{enumerate}

Then, there exist some finite positive constants $C_0 $, $C_1 $ and $C_2 $
such that if $d < \lambda \leq (s+2)/(s+4)$, 
\begin{equation*}
\Pr ( | \mathbf{x}_i^{\prime }\mathbf{M}_z \mathbf{x}_{j} - \mathbb{E}( 
\boldsymbol{\tilde{x}}_i^{\prime }\boldsymbol{\tilde{x}}_{j}) | > \zeta_T )
\leq \exp(-C_0 T^{-1} \zeta_T^2) + \exp(-C_1 T^{C_2} ),
\end{equation*}
and if $\lambda > (s+2)/(s+4) $ 
\begin{equation*}
\Pr ( | \mathbf{x}_i^{\prime }\mathbf{M}_z \mathbf{x}_{j} - \mathbb{E}( 
\boldsymbol{\tilde{x}}_i^{\prime }\boldsymbol{\tilde{x}}_{j}) | > \zeta_T )
\leq \exp(-C_0 \zeta_T^{s/(s+1)}) + \exp(-C_1 T^{C_2} ),
\end{equation*}
for all $i , j = 1,2,\cdots,N $, where $\boldsymbol{\tilde{x}}_i = (\tilde{x}%
_{i1}, \tilde{x}_{i2}, \cdots, \tilde{x}_{iT})^{\prime }$, $\mathbf{x}_i =
(x_{i1}, x_{i2}, \cdots, x_{iT})^{\prime }$, and $\mathbf{M}_z = \mathbf{I}
- T^{-1} \mathbf{\ Z } \hat{\boldsymbol{\Sigma}}_{zz}^{-1} \mathbf{Z}%
^{\prime }$ with $\mathbf{Z } = (\mathbf{z}_1, \mathbf{z}_2, \cdots, \mathbf{%
z}_T)^{\prime }$ and $\hat{ \boldsymbol{\Sigma}}_{zz} = T^{-1}
\sum_{t=1}^{T}(\mathbf{z}_t \mathbf{z} _t^{\prime }) $.
\end{lemma}

\begin{proof}
\begin{equation*}
\begin{split}
& \Pr \left[ |\mathbf{x}_{i}^{\prime }\mathbf{M}_{z}\mathbf{x}_{j}-\mathbb{E}
(\boldsymbol{\tilde{x} }_{i}^{\prime }\boldsymbol{\tilde{x} }_{j})|>\zeta
_{T}\right] =\Pr \left[ |\boldsymbol{\tilde{x} }_{i}^{\prime }\mathbf{M}_{z}%
\boldsymbol{\tilde{x}} _{j}-\mathbb{E}(\boldsymbol{\tilde{x}}_{i}^{\prime }%
\boldsymbol{\tilde{x}}_{j})|>\zeta _{T}\right] \\
& \quad =\Pr \left[ |\boldsymbol{\tilde{x}}_{i}^{\prime }\boldsymbol{\tilde{x%
}}_{j}- \mathbb{E}(\boldsymbol{\tilde{x}}_{i}^{\prime }\boldsymbol{\tilde{x}}%
_{j})-T^{-1} \boldsymbol{\tilde{x}}_{i}^{\prime }\mathbf{Z}\boldsymbol{%
\Sigma }_{zz}^{-1} \mathbf{Z}^{\prime }\boldsymbol{\tilde{x}}_{j}-T^{-1}%
\boldsymbol{\tilde{x}} _{i}^{\prime }\mathbf{Z}(\hat{\boldsymbol{\Sigma }}%
_{zz}^{-1}-\boldsymbol{\ \Sigma }_{zz}^{-1})\mathbf{Z}^{\prime }\boldsymbol{%
\tilde{x}}_{j}|>\zeta_{T} \right],
\end{split}%
\end{equation*}
where $\boldsymbol{\Sigma }_{zz}=\mathbb{E}[T^{-1}\sum_{t=1}^{T}(\mathbf{z}
_{t}\mathbf{z}_{t}^{\prime })]$. By Lemma \ref{prob_sum}, we can further
write 
\begin{equation*}
\begin{split}
& \Pr \left[ |\mathbf{x}_{i}^{\prime }\mathbf{M}_{z}\mathbf{x}_{j}-\mathbb{E}
(\boldsymbol{\tilde{x}}_{i}^{\prime }\boldsymbol{\tilde{x}}_{j})|>\zeta _{T}%
\right] \\
& \qquad \leq \Pr \left[ |\boldsymbol{\tilde{x}}_{i}^{\prime }\boldsymbol{%
\tilde{x}} _{j}-\mathbb{E}(\boldsymbol{\tilde{x}}_{i}^{\prime }\boldsymbol{%
\tilde{x}}_{j})|>\pi _{1}\zeta _{T}\right] +\Pr (|T^{-1}\boldsymbol{\tilde{x}%
}_{i}^{\prime }\mathbf{Z} \boldsymbol{\Sigma }_{zz}^{-1}\mathbf{Z}^{\prime }%
\boldsymbol{\tilde{x}}_{j}|>\pi _{2}\zeta _{T})+ \\
& \qquad \quad \Pr \left[ |T^{-1}\boldsymbol{\tilde{x}}_{i}^{\prime }\mathbf{%
Z}( \hat{\boldsymbol{\Sigma }}_{zz}^{-1}-\boldsymbol{\Sigma }_{zz}^{-1})%
\mathbf{\ Z }^{\prime }\boldsymbol{\tilde{x}}_{j}|)>\pi _{3}\zeta _{T}\right]%
,
\end{split}%
\end{equation*}
where $0<\pi _{i}<1$ and $\sum_{i=1}^{3}\pi _{i}=1$. By Lemma \ref%
{F_norm_submultiplicative}, 
\begin{equation*}
\Pr (|T^{-1}\boldsymbol{\tilde{x}}_{i}^{\prime }\mathbf{Z}\boldsymbol{\Sigma 
} _{zz}^{-1}\mathbf{Z}^{\prime }\boldsymbol{\tilde{x}}_{j}|>\pi _{2}\zeta
_{T})\leq \Pr (\Vert \boldsymbol{\tilde{x}}_{i}^{\prime }\mathbf{Z}\Vert
_{F}\Vert \boldsymbol{\Sigma }_{zz}^{-1}\Vert _{2}\Vert \mathbf{Z}^{\prime } 
\boldsymbol{\tilde{x}}_{j}\Vert _{F}>\pi _{2}\zeta _{T}T),
\end{equation*}
and again by Lemma \ref{prob_product}, we have 
\begin{equation*}
\begin{split}
& \Pr (|T^{-1}\boldsymbol{\tilde{x}}_{i}^{\prime }\mathbf{Z}\boldsymbol{%
\Sigma } _{zz}^{-1}\mathbf{Z}^{\prime }\boldsymbol{\tilde{x}}_{j}|>\pi
_{2}\zeta _{T}) \\
& \qquad \leq \Pr (\Vert \boldsymbol{\tilde{x}}_{i}^{\prime }\mathbf{Z}\Vert
_{F}>\Vert \boldsymbol{\Sigma }_{zz}^{-1}\Vert _{2}^{-1/2}\pi
_{2}^{1/2}\zeta _{T}^{1/2}T^{1/2})+\Pr (\Vert \mathbf{Z}^{\prime } 
\boldsymbol{\tilde{x}}_{j}\Vert _{F}>\Vert \boldsymbol{\Sigma }%
_{zz}^{-1}\Vert _{2}^{-1/2}\pi _{2}^{1/2}\zeta _{T}^{1/2}T^{1/2}).
\end{split}%
\end{equation*}
Similarly, we can show that 
\begin{equation*}
\begin{split}
& \Pr (|T^{-1}\boldsymbol{\tilde{x}}_{i}^{\prime }\mathbf{Z}(\hat{%
\boldsymbol{\ \Sigma }}_{zz}^{-1}-\boldsymbol{\Sigma }_{zz}^{-1})\mathbf{Z}%
^{\prime } \boldsymbol{\tilde{x}}_{j}|>\pi _{3}\zeta _{T}) \\
& \qquad \leq \Pr (\Vert \boldsymbol{\tilde{x}}_{i}^{\prime }\mathbf{Z}\Vert
_{F}\Vert \hat{\boldsymbol{\Sigma }}_{zz}^{-1}-\boldsymbol{\Sigma }
_{zz}^{-1}\Vert _{F}\Vert \mathbf{Z}^{\prime }\boldsymbol{\tilde{x}}%
_{j}\Vert _{F}>\pi _{3}\zeta _{T}T) \\
& \qquad \leq \Pr (\Vert \hat{\boldsymbol{\Sigma }}_{zz}^{-1}-\boldsymbol{\
\Sigma }_{zz}^{-1}\Vert _{F}>\delta _{T}^{-1}\zeta _{T})+\Pr (\Vert 
\boldsymbol{\tilde{x}}_{i}^{\prime }\mathbf{Z}\Vert _{F}>\pi
_{3}^{1/2}\delta _{T}^{1/2}T^{1/2}) \\
& \qquad \quad +\Pr (\Vert \mathbf{Z}^{\prime }\boldsymbol{\tilde{x}}%
_{j}\Vert _{F}>\pi _{3}^{1/2}\delta _{T}^{1/2}T^{1/2}),
\end{split}%
\end{equation*}
where $\delta _{T}=\ominus (T^{\alpha })$ with $0<\alpha <\lambda $.

Note that $\Pr (\Vert \mathbf{Z}^{\prime }\boldsymbol{\tilde{x}}_{i}\Vert
_{F}>c)=\Pr (\Vert \mathbf{Z}^{\prime }\boldsymbol{\tilde{x}}_{i}\Vert
_{F}^{2}>c^{2})=\Pr [\sum_{\ell =1}^{m}(\sum_{t=1}^{T}\tilde{x}_{it}z_{\ell
t})^{2}>c^{2}]$, where $c$ is a positive constant. So, by Lemma \ref%
{prob_sum}, we have 
\begin{equation*}
\textstyle\Pr (\Vert \mathbf{Z}^{\prime }\boldsymbol{\tilde{x}}_{i}\Vert
_{F}>c)\leq \sum_{\ell =1}^{m}\Pr [(\sum_{t=1}^{T}\tilde{x} _{it}z_{\ell
t})^{2}>m^{-1}c^{2}].
\end{equation*}
Hence, $\Pr (\Vert \mathbf{Z}^{\prime }\boldsymbol{\tilde{x}}_{i}\Vert
_{F}>c)\leq \sum_{\ell =1}^{m}\Pr (|\sum_{t=1}^{T}\tilde{x}_{it}z_{\ell
t}|>m^{-1/2}c)$. Also, by Lemma \ref{project reg md} we have $\sum_{t=1}^{T} 
\mathbb{E}(\tilde{x} _{it}z_{\ell t})=0$ and hence we can further write 
\begin{equation*}
\textstyle\Pr (\Vert \mathbf{Z}^{\prime }\boldsymbol{\tilde{x}}_{i}\Vert
_{F}>c)\leq \sum_{\ell =1}^{m}\Pr \{|\sum_{t=1}^{T}[\tilde{x} _{it}z_{\ell
t}- \mathbb{E}(\tilde{x}_{it}z_{\ell t})]|>m^{-1/2}c\}.
\end{equation*}
Note that $\Vert \boldsymbol{\Sigma }_{zz}^{-1}\Vert _{2}$ is equal to the
largest eigenvalue of $\boldsymbol{\Sigma }_{zz}^{-1}$ and it is a finite
positive constant. So, there exists a positive constant $C>0$ such that, 
\begin{equation*}
\begin{split}
& \Pr (|\mathbf{x}_{i}^{\prime }\mathbf{M}_{z}\mathbf{x}_{j}-\mathbb{E}( 
\boldsymbol{\tilde{x}}_{i}^{\prime }\boldsymbol{\tilde{x}}_{j})|>\zeta _{T})
\\
& \qquad \textstyle\leq \Pr \{|\sum_{t=1}^{T}[\tilde{x} _{it}\tilde{x} _{jt}-%
\mathbb{E} (\tilde{x} _{it}\tilde{x}_{jt})]|>CT^{\lambda }\}+ \\
& \qquad \quad \textstyle\sum_{\ell =1}^{m}\Pr \{|\sum_{t=1}^{T}[\tilde{x}%
_{it}z_{\ell t}-\mathbb{E}(\tilde{x} _{it}z_{\ell t}]|>CT^{1/2+\lambda
/2-d/2}\}+ \\
& \qquad \quad \textstyle\sum_{\ell =1}^{m}\Pr \{|\sum_{t=1}^{T}[\tilde{x}
_{jt}z_{\ell t}-\mathbb{E}(\tilde{x} _{jt}z_{\ell t}]|>CT^{1/2+\lambda
/2-d/2}\}+ \\
& \qquad \quad \textstyle\sum_{\ell =1}^{m}\Pr \{|\sum_{t=1}^{T}[\tilde{x}
_{it}z_{\ell t}-\mathbb{E}(\tilde{x} _{it}z_{\ell t}]|>CT^{1/2+\alpha
/2-d/2}\}+ \\
& \qquad \quad \textstyle\sum_{\ell =1}^{m}\Pr \{|\sum_{t=1}^{T}[\tilde{x}
_{jt}z_{\ell t}-\mathbb{E}(\tilde{x} _{jt}z_{\ell t}]|>CT^{1/2+\alpha
/2-d/2}\}+ \\
& \qquad \quad \textstyle\Pr (\Vert \hat{\boldsymbol{\Sigma }}_{zz}^{-1}- 
\boldsymbol{\Sigma }_{zz}^{-1}\Vert _{F}>\delta _{T}^{-1}\zeta _{T}).
\end{split}%
\end{equation*}
Let 
\begin{equation*}
\kappa _{T,i}\left( h,d\right) =\sum_{\ell =1}^{m}\Pr \{|\sum_{t=1}^{T}[%
\tilde{x} _{it}z_{\ell t}-\mathbb{E}(\tilde{x} _{it}z_{\ell t}]|>CT^{1/2+h
/2-d/2}\} \text{, for }h=\lambda ,\alpha \text{,}
\end{equation*}
and $i=1,2,...,N$. By Lemmas \ref{exp_tail_prod}, \ref{project reg md}, and %
\ref{project reg subg}, we have $\tilde{x} _{it}\tilde{x} _{jt}-\mathbb{E}(%
\tilde{x} _{it}\tilde{x} _{jt})$ and $\tilde{x} _{it}z_{\ell t}-\mathbb{E}(%
\tilde{x} _{it}z_{\ell t})$ are martingale difference processes with
exponentially bounded probability tail, $\frac{s}{2}$. So, depending on the
value of exponentially bounded probability tail parameter, from Lemma \ref%
{mart_diff_proc_exp_tail}, there exists a finite positive constant $C$ such
that 
\begin{equation*}
\kappa _{T,i}\left( h,d\right) \leq m\exp \left[ -C T^{h-d} \right],
\end{equation*}
or 
\begin{equation*}
\kappa _{T,i}\left( h,d\right) \leq m\exp \left[ - C
T^{s(1/2+h/2-d/2)/(s+2)} \right] \text{,}
\end{equation*}
for $h=\lambda ,\alpha $. Also, depending on the value of exponentially
bounded probability tail parameter, from Lemmas \ref%
{inverse_est_matrix_F_norm_bound_1} and \ref%
{inverse_est_matrix_F_norm_bound_2} we have, 
\begin{equation*}
\begin{split}
& \Pr (\Vert \hat{\boldsymbol{\Sigma }}_{zz}^{-1}-\boldsymbol{\Sigma }
_{zz}^{-1}\Vert _{F}>\delta _{T}^{-1}\zeta _{T})\leq m^{2}\exp \left[ -C_{0} 
\frac{T\delta _{T}^{-2}\zeta _{T}^{2}}{m^{2}\Vert \boldsymbol{\Sigma }
_{zz}^{-1}\Vert _{2}^{2}(\Vert \boldsymbol{\Sigma }_{zz}^{-1}\Vert
_{2}+\delta _{T}^{-1}\zeta _{T})^{2}}\right] + \\
& \qquad \qquad \qquad \qquad \qquad \qquad \qquad m^{2}\exp \left( -C_{0} 
\frac{T}{m^{2}\Vert \boldsymbol{\Sigma }_{zz}^{-1}\Vert _{2}^{2}}\right) ,
\end{split}%
\end{equation*}
or 
\begin{equation*}
\begin{split}
& \Pr (\Vert \hat{\boldsymbol{\Sigma }}_{zz}^{-1}-\boldsymbol{\Sigma }
_{zz}^{-1}\Vert _{F}>\delta _{T}^{-1}\zeta _{T})\leq m^{2}\exp \left[ -C_{0} 
\frac{(T\delta _{T}^{-1}\zeta _{T})^{s/s+2}}{m^{s/s+2}\Vert \boldsymbol{\
\Sigma }_{zz}^{-1}\Vert _{2}^{s/s+2}(\Vert \boldsymbol{\Sigma }
_{zz}^{-1}\Vert _{2}+\delta _{T}^{-1}\zeta _{T})^{s/s+2}}\right] + \\
& \qquad \qquad \qquad \qquad \qquad \qquad \qquad m^{2}\exp \left( -C_{0} 
\frac{T^{s/s+2}}{m^{s/s+2}\Vert \boldsymbol{\Sigma }_{zz}^{-1}\Vert
_{2}^{s/s+2}}\right) .
\end{split}%
\end{equation*}
Therefore, there exists a finite positive constant $C$ such that 
\begin{equation*}
\begin{split}
& \Pr (\Vert \hat{\boldsymbol{\Sigma }}_{zz}^{-1}-\boldsymbol{\Sigma }
_{zz}^{-1}\Vert _{F}>\delta _{T}^{-1}\zeta _{T}) \\
& \qquad \leq m^{2} \exp [-C T^{\max \{1-2d+2(\lambda -\alpha ),1-2d+\lambda
-\alpha ,1-2d\}}]+ \\
& \qquad \quad m^{2} \exp [-C T^{1-2d}],
\end{split}%
\end{equation*}
or, 
\begin{equation*}
\begin{split}
& \Pr (\Vert \hat{\boldsymbol{\Sigma }}_{zz}^{-1}-\boldsymbol{\Sigma }
_{zz}^{-1}\Vert _{F}>\delta _{T}^{-1}\zeta _{T}) \\
& \qquad \leq m^{2} \exp [-C T^{s(\max \{1-d+\lambda -\alpha ,1-d\})/(s+2)}]+
\\
& \qquad \quad m^{2} \exp [-C T^{s(1-d)/(s+2)}].
\end{split}%
\end{equation*}
Setting $d<1/2$, $\alpha =1/2$, and $\lambda >d$, we have all the terms
going to zero as $T\rightarrow \infty $ and there exist some finite positive
constants $C_{1}$ and $C_{2}$ such that 
\begin{equation*}
\kappa _{T,i}\left( \lambda ,d\right) \leq \exp \left(
-C_{1}T^{C_{2}}\right) \text{, }\kappa _{T,i}\left( \alpha ,d\right) \leq
\exp \left( -C_{1}T^{C_{2}}\right) \text{,}
\end{equation*}
and 
\begin{equation*}
\Pr (\Vert \hat{\boldsymbol{\Sigma }}_{zz}^{-1}-\boldsymbol{\Sigma }
_{zz}^{-1}\Vert _{F}>\delta _{T}^{-1}\zeta _{T})\leq \exp (-C_{1}T^{C_{2}}).
\end{equation*}
Hence, if $d<\lambda \leq (s+2)/(s+4)$, 
\begin{equation*}
\Pr (|\mathbf{x}_{i}^{\prime }\mathbf{M}_{z}\mathbf{x}_{j}-\mathbb{E}( 
\boldsymbol{\tilde{x} }_{i}^{\prime }\boldsymbol{\tilde{x} }_{j})|>\zeta
_{T})\leq \exp (-C_{0}T^{-1}\zeta _{T}^{2})+\exp (-C_{1}T^{C_{2}}),
\end{equation*}
and if $\lambda >(s+2)/(s+4)$, 
\begin{equation*}
\Pr (|\mathbf{x}_{i}^{\prime }\mathbf{M}_{z}\mathbf{x}_{j}-\mathbb{E}( 
\boldsymbol{\tilde{x} }_{i}^{\prime }\boldsymbol{\tilde{x} }_{j})|>\zeta
_{T})\leq \exp (-C_{0}\zeta _{T}^{s/(s+1)})+\exp (-C_{1}T^{C_{2}}),
\end{equation*}
where $C_{0}$, $C_{1}$ and $C_{2}$ are some finite positive constants.
\end{proof}

\begin{flushleft}
{\Large \ \textbf{References} }
\end{flushleft}

\normalfont

Bailey, N., Pesaran, M. H., and Smith, L. V. (2019). A multiple testing
approach to the regularisation of large sample correlation matrices. \textit{%
\ Journal of Econometrics}, 208(2): 507-534.
https://doi.org/10.1016/j.jeconom.2018.10.006 \bigskip

Chudik, A., Kapetanios, G., and Pesaran, M. H. (2018). A one covariate at a
time, multiple testing approach to variable selection in high-dimensional
linear regression models. \textit{Econometrica}, 86(4): 1479-1512.
https://doi.org/10.3982/ECTA14176\bigskip

Friedman, J., Hastie, T., and Tibshirani, R. (2010). Regularization paths
for generalized linear models via coordinate descent. \textit{Journal of
statistical software}, 33(1):1-22. \newline
https://doi.org/10.18637/jss.v033.i01 \bigskip

L\"{u}tkepohl, H. (1996). \textit{Handbook of Matrices}. John Wiley \& Sons,
West Sussex, UK. ISBN-10: 9780471970156

\newpage 

\setcounter{section}{0} \renewcommand{\thesection}{S-\arabic{section}}

\setcounter{table}{0} \renewcommand{\thetable}{S.\arabic{table}}

\setcounter{page}{1} \renewcommand{\thepage}{S.\arabic{page}}

\setcounter{equation}{0} \renewcommand{\theequation}{S.\arabic{equation}}

\vspace{0.05in}

\begin{center}
\textbf{\ {\large Online Monte Carlo Supplement to} }\\[0pt]

\textbf{{\large {``Variable Selection in High Dimensional Linear Regressions
with Parameter Instability''}}} \\[0pt]

Alexander Chudik

Federal Reserve Bank of Dallas\bigskip

M. Hashem Pesaran

University of Southern California, USA and Trinity College, Cambridge,
UK\bigskip

Mahrad Sharifvaghefi

University of Pittsburgh\\[0pt]
\bigskip \bigskip

\today\bigskip
\end{center}


\noindent This online Monte Carlo supplement has three sections. Section \ref%
{lasso, ad-lasso and cv} explains the algorithms used for implementing
Lasso, A-Lasso, boosting and cross-validation. We provide additional summary
tables of our Monte Carlo simulation findings in Section \ref%
{more_mc_summary_tabs}. The full set of Monte Carlo results for all the
baseline experiments are provided in Section \ref{detailed_mc_tabs}.

\section{Lasso, A-Lasso, boosting and cross-validation algorithms}

\label{lasso, ad-lasso and cv}

This section explains how Lasso, $K$-fold cross-validation, A-Lasso, and
boosting are implemented in this paper.\footnote{%
For the implementation of Lasso and A-Lasso, we use the Matlab Glmnet
package available at \url{https://hastie.su.domains/glmnet_matlab/}.} Let $%
\mathbf{y}=(y_{1},y_{2},\cdots ,y_{T})^{\prime }$ be a $T\times 1$ vector of
target variable, and let $\mathbf{Z}=(\mathbf{z}_{1},\mathbf{z}_{2},\cdots ,%
\mathbf{z}_{T})^{\prime }$ be a $T\times m$ matrix of conditioning
covariates where $\{\mathbf{z}_{t}:t=1,2,\cdots ,T\}$ are $m\times 1$
vectors and let $\mathbf{X}=(\mathbf{\ x}_{1},\mathbf{x}_{2},\cdots ,\mathbf{%
x}_{T})^{\prime }$ be a $T\times N$ matrix of covariates in the active set
where $\{\mathbf{x}_{t}:t=1,2,\cdots ,T\}$ are $N\times 1$ vectors.

\begin{flushleft}
\textbf{Lasso Procedure}
\end{flushleft}

\itshape

\begin{enumerate}
\item Construct the filtered variables $\mathbf{\tilde{y}}=\mathbf{M}_{z} 
\mathbf{y}$ and $\mathbf{\tilde{X}}=\mathbf{M}_{z}\mathbf{X=}\left( \mathbf{%
\ \tilde{x}}_{1\circ },\mathbf{\tilde{x}}_{2\circ },...,\mathbf{\tilde{x}}
_{N\circ }\right) $, where $\mathbf{M}_{z}=\mathbf{I}_{T}-\mathbf{Z}(\mathbf{%
\ Z}^{\prime }\mathbf{Z})^{-1}\mathbf{Z}^{\prime }$, and $\mathbf{\tilde{x}}
_{i\circ }=(\tilde{x}_{i1},\tilde{x}_{i2},\cdots ,\tilde{x}_{iT})^{\prime }$.

\item Normalize each covariate $\mathbf{\tilde{x}}_{i\circ }=(\tilde{x}%
_{i1}, \tilde{x}_{i2},\cdots ,\tilde{x}_{iT})^{\prime }$ by its $\ell _{2}$
norm, such that 
\begin{equation*}
\mathbf{\tilde{x}}_{i\circ }^{\ast }=\mathbf{\tilde{x}}_{i\circ }/\lVert 
\mathbf{\tilde{x}}_{i\circ }\rVert _{2},
\end{equation*}
where $\lVert .\rVert _{2}$ denotes the $\ell _{2}$ norm of a vector. The
corresponding matrix of normalized covariates in the active set is now
denoted by $\mathbf{\tilde{X}}^{\ast }$.

\item For a given value of $\varphi \geq 0$, find $\boldsymbol{\hat{\gamma}}
_{x}^{\ast }(\varphi )\equiv \lbrack \hat{\gamma}_{1x}^{\ast }(\varphi ), 
\hat{\gamma}_{2x}^{\ast }(\varphi ),\cdots ,\hat{\gamma}_{Nx}^{\ast
}(\varphi )]^{\prime }$ such that 
\begin{equation*}
\boldsymbol{\hat{\gamma}}_{x}^{\ast }(\varphi )=\arg \min_{\boldsymbol{\
\gamma }_{x}^{\ast }}\left\{ \lVert \mathbf{\tilde{y}}-\mathbf{\tilde{X}}
^{\ast }\boldsymbol{\gamma }_{x}^{\ast }\rVert _{2}^{2}+\varphi \lVert 
\boldsymbol{\gamma }_{x}^{\ast }\rVert _{1}\right\} ,
\end{equation*}
where $\lVert .\rVert _{1}$ denotes the $\ell _{1}$ norm of a vector.

\item Divide $\hat{\gamma}_{ix}^{\ast }(\varphi )$ for $i=1,2,\cdots ,N$ by $%
\ell _{2}$ norm of the $\mathbf{\tilde{x}}_{i\circ }$ to match the original
scale of $\mathbf{\tilde{x}}_{i\circ }$, namely set 
\begin{equation*}
\hat{\gamma}_{ix}(\varphi )=\hat{\gamma}_{ix}^{\ast }(\varphi )/\lVert 
\mathbf{\tilde{x}}_{i\circ }\rVert _{2},
\end{equation*}
where $\boldsymbol{\hat{\gamma}}_{x}(\varphi )\equiv \lbrack \hat{\gamma}
_{1x}(\varphi ),\hat{\gamma}_{2x}(\varphi ),\cdots ,\hat{\gamma}
_{Nx}(\varphi )]^{\prime }$ denotes the vector of scaled coefficients.

\item Compute $\boldsymbol{\hat{\gamma}}_{z}(\varphi )\equiv \lbrack \hat{
\gamma}_{1z}(\varphi ),\hat{\gamma}_{2z}(\varphi ),\cdots ,\hat{\gamma}
_{mz}(\varphi )]^{\prime }$ by $\boldsymbol{\hat{\gamma}}_{z}(\varphi )=( 
\mathbf{Z}^{\prime }\mathbf{Z})^{-1}\mathbf{Z}^{\prime }\mathbf{\hat{e}}
(\varphi )$ where $\mathbf{\hat{e}}(\varphi )=\mathbf{\tilde{y}}-\mathbf{\ 
\tilde{X}}\boldsymbol{\hat{\gamma}}_{x}(\varphi )$.
\end{enumerate}

\normalfont

For a given set of values of $\varphi $'s, say $\{\varphi_j: j = 1,2,
\cdots, h \} $, the optimal value of $\varphi $ is chosen by $K $-fold
cross-validation as described below.

\begin{flushleft}
\textbf{$K $-fold Cross-validation}
\end{flushleft}

\itshape

\begin{enumerate}
\item Create a $T \times 1 $ vector $\mathbf{m} = (1,2,\cdots, K, 1,2,
\cdots, K, \cdots)^{\prime }$ where $K $ is the number of folds.

\item Let $\mathbf{m}^{\ast }=(m_{1}^{\ast },m_{2}^{\ast },\cdots
,m_{T}^{\ast })^{\prime }$ be a $T\times 1$ vector generated by randomly
permuting the elements of $\mathbf{m}$.

\item Group observations into K folds such that 
\begin{equation*}
g_{\ell}=\{t:t\in \{1,2,\cdots ,T\}\text{ and }m_{t}^{\ast }=\ell\}\text{
for } \ell=1,2,\cdots ,K.
\end{equation*}

\item For a given value of $\varphi_j $ and each fold $\ell \in
\{1,2,\cdots, K\}$,

\begin{enumerate}
\item Remove the observations related to fold $\ell $ from the set of all
observations.

\item Given the value of $\varphi_j $, use the remaining observations to
estimate the coefficients of the model.

\item Use the estimated coefficients to compute predicted values of the
target variable for the observations in fold $\ell$, denoted by $\hat{y}%
_{t}^{f}(\varphi _{j})$. 
\end{enumerate}

\item Compute the mean squared forecast errors for a given value of $\varphi
_{j}$ by 
\begin{equation*}
\text{MSFE}(\varphi _{j})= \frac{1}{T}\sum_{t=1}^{T} \left[y_{t} - \hat{y}%
_{t}^{f}(\varphi _{j}) \right]^2.
\end{equation*}

\item Repeat steps 1 to 5 for all values of $\{\varphi_j: j = 1,2, \cdots, h
\} $.

\item Select $\varphi_j $ with the lowest corresponding mean squared
forecast errors as the optimal value of $\varphi $.
\end{enumerate}

\normalfont

In this study, following Friedman et al. (2010), we consider a sequence of
100 values of $\varphi $'s decreasing from $\varphi _{\max }$ to $\varphi
_{\min }$ on $\log $ scale where $\varphi _{\max }=\max_{i=1,2,\cdots
,N}\left\{ \lvert \sum_{t=1}^{T}\tilde{x}_{it}^{\ast }\tilde{y}_{t}\rvert
\right\} $ and $\varphi _{\min }=0.001\varphi _{\max }$. We use $10$-fold
cross-validation ($K=10$) to find the optimal value of $\varphi $.

Denote $\boldsymbol{\hat{\gamma}}_{x} \equiv \boldsymbol{\hat{\gamma}}%
_{x}(\varphi_{op})$ where $\varphi_{op} $ is the optimal value of $\varphi $
obtained by the $K $-fold cross-validation. Given $\boldsymbol{\hat{\gamma}}%
_{x} $, we implement A-Lasso as described below.

\begin{flushleft}
\textbf{A-Lasso}
\end{flushleft}

\itshape

\begin{enumerate}
\item Let $\mathcal{S}=\{i:i\in \{1,2,\cdots ,N\}\text{ and }\hat{\gamma}
_{ix}\neq 0\}$ and $\mathbf{X}_{\mathcal{S}}$ be the $T\times s$ set of
covariates in the active set with $\hat{\gamma}_{ix}\neq 0$ (from the Lasso
step) where $s=\lvert \mathcal{S}\rvert $. Additionally, denote the
corresponding $s\times 1$ vector of non-zero Lasso coefficients by $%
\boldsymbol{\hat{\gamma}}_{x,\mathcal{S}}=(\hat{\gamma}_{1x,\mathcal{S}}, 
\hat{\gamma}_{2x,\mathcal{S}},\cdots ,\hat{\gamma}_{sx,\mathcal{S}})^{\prime
}$.

\item For a given value of $\psi \geq 0$, find $\boldsymbol{\hat{\delta}}%
_{x, \mathcal{S}}^{\ast }(\psi )\equiv \lbrack \hat{\delta}_{1x,\mathcal{S}
}^{\ast }(\psi ),\hat{\delta}_{2x,\mathcal{S}}^{\ast }(\psi ),\cdots ,\hat{
\delta}_{sx,\mathcal{S}}^{\ast }(\psi )]^{\prime }$ such that 
\begin{equation*}
\boldsymbol{\hat{\delta}}_{x,\mathcal{S}}^{\ast }(\psi )=\arg \min_{ 
\boldsymbol{\delta }_{x,\mathcal{S}}^{\ast }}\left\{ \lVert \mathbf{\tilde{y}
}-\mathbf{\tilde{X}}_{\mathcal{S}}diag(\boldsymbol{\hat{\gamma}}_{x,\mathcal{%
\ S}})\boldsymbol{\delta }_{x,\mathcal{S}}^{\ast }\rVert _{2}^{2}+\psi
\lVert \boldsymbol{\delta }_{x,\mathcal{S}}^{\ast }\rVert _{1}\right\} ,
\end{equation*}
where $diag(\boldsymbol{\hat{\gamma}}_{x,\mathcal{S}})$ is an $s\times s$
diagonal matrix with its diagonal elements given by the corresponding
elements of $\boldsymbol{\hat{\gamma}}_{x,\mathcal{S}}$.

\item Post multiply $\boldsymbol{\hat{\delta}}_{x,\mathcal{S}}^{\ast }(\psi
) $ by $diag(\boldsymbol{\hat{\gamma}}_{x,\mathcal{S}})$ to match the
original scale of $\mathbf{\tilde{X}}_{\mathcal{S}}$, such that 
\begin{equation*}
\boldsymbol{\hat{\delta}}_{x,\mathcal{S}}(\psi )=diag(\boldsymbol{\hat{%
\gamma }}_{x,\mathcal{S}})\boldsymbol{\hat{\delta}}_{x,\mathcal{S}}^{\ast
}(\psi ).
\end{equation*}
The coefficients of the covariates in the active set that belong to $%
\mathcal{S}^{c}$ are set equal to zero. In other words, $\boldsymbol{\hat{
\delta}}_{x,\mathcal{S}^{c}}(\psi )=0$ for all $\psi \geq 0$.

\item Compute $\boldsymbol{\hat{\delta}}_z (\psi) \equiv [\hat{\delta}
_{1z}(\psi), \hat{\delta}_{2z}(\psi), \cdots, \hat{\delta}_{mz}(\psi)
]^{\prime }$ by $\boldsymbol{\hat{\delta}}_z(\psi) = (\mathbf{Z}^{\prime } 
\mathbf{Z})^{-1} \mathbf{Z}^{\prime }\mathbf{\hat{e}}(\psi) $ where $\mathbf{%
\ \hat{e}}(\psi) = \mathbf{\tilde{y}} - \mathbf{\tilde{X}}_{\mathcal{S}} 
\boldsymbol{\hat{\delta}}_{x,\mathcal{S}} (\psi) $.
\end{enumerate}

\normalfont

As in the Lasso step, the optimal value $\psi $ is set using $10$-fold
cross-validation as described before.

\begin{flushleft}
\textbf{Boosting}
\end{flushleft}

We implement boosting algorithm using a BIC stopping criterion. We have also
considered corrected AIC stopping criterion (\cite{buhlmann2006boosting})
and these results are available in the working paper version of this paper. 
\itshape

\begin{enumerate}
\item Consider the matrix of normalized filtered covariates $\mathbf{\tilde{X%
}}^{\ast }=\boldsymbol{(}\mathbf{\tilde{x}}_{1\circ }^{\ast },\mathbf{\tilde{%
x}}_{2\circ }^{\ast },...,\mathbf{\tilde{x}}_{n\circ }^{\ast })$, defined in
Step 2 of the Lasso procedure above. Let the row $t$ (for $t=1,2,...,T$) of $%
\mathbf{\tilde{X}}^{\ast }$ be denoted as $\mathbf{\tilde{x}}_{\circ
t}^{\ast \prime }=\left( \tilde{x}_{1t}^{\ast },\tilde{x}_{2t}^{\ast },...,%
\tilde{x}_{nt}^{\ast }\right) $. Given the normalized covariates matrix $%
\mathbf{\tilde{X}}^{\ast }$ and any vector $\mathbf{e}%
=(e_{1},e_{2},..,e_{T})^{\prime }$, define the least squares base procedure: 
\begin{equation*}
\hat{g}_{\mathbf{\tilde{X}}^{\ast },\mathbf{e}}(\mathbf{\tilde{x}}_{\circ
t}^{\ast })=\hat{\delta}_{\hat{s}}\tilde{x}_{\hat{s}t}^{\ast }\text{,\ }\hat{%
s}=\arg \min_{1\leq i\leq n}\left( \mathbf{e}-\hat{\delta}_{i}\mathbf{\tilde{%
x}}_{i}^{\ast }\right) ^{\prime }\left( \mathbf{e}-\hat{\delta}_{i}\mathbf{%
\tilde{x}}_{i}^{\ast }\right) \text{, \ }\hat{\delta}_{i}=\frac{\mathbf{e}%
^{\prime }\mathbf{\tilde{x}}_{i}^{\ast }}{\mathbf{\tilde{x}}_{i}^{\ast
\prime }\mathbf{\tilde{x}}_{i}^{\ast }}\text{, }
\end{equation*}

\item Given the normalized filtered covariates data $\mathbf{\tilde{X}}%
^{\ast }$ and the filtered target variable $\mathbf{\tilde{y}}=\mathbf{M}_{z}%
\mathbf{y}$, apply the base procedure to obtain $\hat{g}_{\mathbf{\tilde{X}}%
^{\ast }\mathbf{,\tilde{y}}}^{(1)}(\mathbf{\tilde{x}}_{\circ t}^{\ast }).$
Set $\hat{F}^{(1)}(\mathbf{\tilde{x}}_{\circ t}^{\ast })=v\hat{g}_{\mathbf{%
\tilde{X}}^{\ast }\mathbf{,\tilde{y}}}^{(1)}(\mathbf{\tilde{x}}_{\circ
t}^{\ast })$, for some $v>0$. Set $\hat{s}^{(1)}=\hat{s}$ and $m=1$.

\item Compute the residual vector $\mathbf{e}^{\left( m\right) }=\mathbf{%
\tilde{y}}-\hat{F}^{(m)}(\mathbf{\tilde{X}}^{\ast }),$ where $\hat{F}^{(m)}(%
\mathbf{\tilde{X}}^{\ast })=[\hat{F}^{(m)}(\mathbf{\tilde{x}}_{\circ
1}^{\ast }),\hat{F}^{(m)}(\mathbf{\tilde{x}}_{\circ 2}^{\ast })$, $...$, $%
\hat{F}^{(m)}(\mathbf{\tilde{x}}_{\circ T}^{\ast })]^{\prime },$ and fit the
base procedure to these residuals to obtain the fit values $\hat{g}_{\mathbf{%
\tilde{X}}^{\ast }\mathbf{,e}^{\left( m\right) }}^{(m+1)}(\mathbf{\tilde{x}}%
_{\circ t}^{\ast })$ and $\hat{s}^{(m)}$. Update%
\begin{equation*}
\hat{F}^{(m+1)}(\mathbf{\tilde{x}}_{\circ t}^{\ast })=\hat{F}^{(m)}(\mathbf{%
\tilde{x}}_{\circ t}^{\ast })+v\hat{g}_{\mathbf{\tilde{X}}^{\ast }\mathbf{,e}%
^{\left( m\right) }}^{(m+1)}(\mathbf{\tilde{x}}_{\circ t}^{\ast })\text{.}
\end{equation*}

\item Increase the iteration index $m$ by one and repeat Step 3 until the
stopping iteration $M$ is achieved. The stopping iteration is given by 
\begin{equation*}
M=\arg \min_{1\leq m\leq m_{\max }}BIC_{C}(m),
\end{equation*}%
for some predetermined large $m_{\max }$, where%
\begin{equation*}
BIC(m)=\log (\hat{\sigma}^{2})+1+tr\left( \mathcal{B}_{m}\right)ln(T) /T%
\text{,}
\end{equation*}%
\begin{equation*}
\hat{\sigma}^{2}=\frac{1}{T}\left( \mathbf{\tilde{y}}-\mathcal{B}_{m}\mathbf{%
\tilde{y}}\right) ^{\prime }\left( \mathbf{\tilde{y}}-\mathcal{B}_{m}\mathbf{%
\tilde{y}}\right) \text{,}
\end{equation*}%
\begin{equation*}
\mathcal{B}_{m}=I-\left( I-v\mathcal{H}^{(\hat{s}_{m})}\right) \left( I-v%
\mathcal{H}^{(\hat{s}_{m-1})}\right) ...\left( I-v\mathcal{H}^{(\hat{s}%
_{1})}\right) \text{,}
\end{equation*}%
\begin{equation*}
\mathcal{H}^{(j)}=\frac{\mathbf{\tilde{x}}_{j\circ }^{\ast }\mathbf{\tilde{x}%
}_{j\circ }^{\ast \prime }}{\mathbf{\tilde{x}}_{j\circ }^{\ast \prime }%
\mathbf{\tilde{x}}_{j\circ }^{\ast }}\text{.}
\end{equation*}
\end{enumerate}

\normalfont

We set $m_{\max }=500$ and $v=0.5$.

\section{Additional Monte Carlo summary tables}

\label{more_mc_summary_tabs}

\begin{table}[h]
	\caption{Comparison of the effects of down-weighting for TPR
performance in MC experiments with and without parameter instability. \bigskip} \label{mc_tab_s1}
	\renewcommand{\arraystretch}{1.05}%
	\centering
	\footnotesize%
\begin{tabular}{rrrrrrrrrrrr}
\hline\hline
& \multicolumn{11}{c}{\textbf{Average TPR}} \\ \cline{2-12}
Down-weighting: & No & Light & Heavy &  & No & Light & Heavy &  & No & Light
& Heavy \\ \cline{2-4}\cline{6-8}\cline{10-12}
$N\backslash T$ & \multicolumn{3}{c}{\textbf{100}} &  & \multicolumn{3}{c}{%
\textbf{150}} &  & \multicolumn{3}{c}{\textbf{200}} \\ \hline
\multicolumn{12}{l}{A. Without parameter instability} \\ \hline
\multicolumn{1}{l}{} & \multicolumn{11}{l}{OCMT (down-weighting at the
selection stage)} \\ \hline
\textbf{20} & \textbf{0.83} & 0.68 & 0.60 &  & \textbf{0.91} & 0.73 & 0.67 & 
& \textbf{0.96} & 0.77 & 0.74 \\ 
\textbf{40} & \textbf{0.80} & 0.64 & 0.57 &  & \textbf{0.91} & 0.71 & 0.66 & 
& \textbf{0.95} & 0.76 & 0.73 \\ 
\textbf{100} & \textbf{0.77} & 0.60 & 0.53 &  & \textbf{0.88} & 0.67 & 0.63
&  & \textbf{0.93} & 0.72 & 0.71 \\ \hline
\multicolumn{1}{l}{} & \multicolumn{10}{l}{Lasso} & \multicolumn{1}{l}{} \\ 
\hline
\textbf{20} & \textbf{0.84} & 0.78 & 0.73 &  & \textbf{0.89} & 0.80 & 0.74 & 
& \textbf{0.93} & 0.82 & 0.75 \\ 
\textbf{40} & \textbf{0.82} & 0.76 & 0.73 &  & \textbf{0.89} & 0.79 & 0.74 & 
& \textbf{0.92} & 0.81 & 0.75 \\ 
\textbf{100} & \textbf{0.79} & 0.74 & 0.70 &  & \textbf{0.87} & 0.78 & 0.75
&  & \textbf{0.90} & 0.79 & 0.76 \\ \hline
\multicolumn{1}{l}{} & \multicolumn{10}{l}{A-Lasso} & \multicolumn{1}{l}{}
\\ \hline
\textbf{20} & \textbf{0.73} & 0.69 & 0.64 &  & \textbf{0.80} & 0.73 & 0.66 & 
& \textbf{0.85} & 0.75 & 0.67 \\ 
\textbf{40} & \textbf{0.73} & 0.69 & 0.65 &  & \textbf{0.81} & 0.74 & 0.68 & 
& \textbf{0.86} & 0.76 & 0.69 \\ 
\textbf{100} & \textbf{0.72} & 0.68 & 0.64 &  & \textbf{0.81} & 0.73 & 0.69
&  & \textbf{0.86} & 0.75 & 0.70 \\ \hline
\multicolumn{1}{l}{} & \multicolumn{10}{l}{Boosting} & \multicolumn{1}{l}{}
\\ \hline
\textbf{20} & 0.77 & 0.76 & \textbf{0.78} &  & \textbf{0.83} & 0.81 & 0.83 & 
& \textbf{0.88} & 0.85 & 0.86 \\ 
\textbf{40} & 0.76 & 0.77 & \textbf{0.81} &  & 0.83 & 0.83 & \textbf{0.86} & 
& 0.87 & 0.86 & \textbf{0.89} \\ 
\textbf{100} & 0.75 & 0.79 & \textbf{0.81} &  & 0.82 & 0.85 & \textbf{0.85}
&  & 0.86 & \textbf{0.88} & 0.87 \\ \hline
\multicolumn{12}{l}{B. With parameter instability} \\ \hline
\multicolumn{1}{l}{} & \multicolumn{11}{l}{OCMT (down-weighting at the
selection stage)} \\ \hline
\textbf{20} & \textbf{0.73} & 0.59 & 0.57 &  & \textbf{0.85} & 0.69 & 0.68 & 
& \textbf{0.92} & 0.76 & 0.75 \\ 
\textbf{40} & \textbf{0.70} & 0.56 & 0.54 &  & \textbf{0.84} & 0.67 & 0.66 & 
& \textbf{0.91} & 0.75 & 0.75 \\ 
\textbf{100} & \textbf{0.66} & 0.51 & 0.50 &  & \textbf{0.81} & 0.63 & 0.63
&  & \textbf{0.88} & 0.70 & 0.72 \\ \hline
\multicolumn{1}{l}{} & \multicolumn{10}{l}{Lasso} & \multicolumn{1}{l}{} \\ 
\hline
\textbf{20} & \textbf{0.76} & 0.72 & 0.71 &  & \textbf{0.82} & 0.78 & 0.75 & 
& \textbf{0.87} & 0.81 & 0.77 \\ 
\textbf{40} & \textbf{0.74} & 0.70 & 0.70 &  & \textbf{0.82} & 0.76 & 0.75 & 
& \textbf{0.86} & 0.79 & 0.77 \\ 
\textbf{100} & \textbf{0.70} & 0.67 & 0.66 &  & \textbf{0.79} & 0.73 & 0.73
&  & \textbf{0.83} & 0.77 & 0.76 \\ 
\multicolumn{1}{l}{} & \multicolumn{10}{l}{A-Lasso} & \multicolumn{1}{l}{}
\\ 
\textbf{20} & \textbf{0.65} & 0.63 & 0.62 &  & \textbf{0.72} & 0.69 & 0.66 & 
& \textbf{0.78} & 0.74 & 0.69 \\ 
\textbf{40} & \textbf{0.64} & 0.62 & 0.62 &  & \textbf{0.73} & 0.70 & 0.68 & 
& \textbf{0.79} & 0.74 & 0.71 \\ 
\textbf{100} & \textbf{0.63} & 0.61 & 0.59 &  & \textbf{0.73} & 0.68 & 0.67
&  & \textbf{0.78} & 0.72 & 0.71 \\ \hline
\multicolumn{1}{l}{} & \multicolumn{10}{l}{Boosting} & \multicolumn{1}{l}{}
\\ \hline
\textbf{20} & 0.68 & 0.69 & \textbf{0.74} &  & 0.74 & 0.77 & \textbf{0.81} & 
& 0.79 & 0.82 & \textbf{0.85} \\ 
\textbf{40} & 0.68 & 0.70 & \textbf{0.77} &  & 0.75 & 0.78 & \textbf{0.84} & 
& 0.80 & 0.83 & \textbf{0.88} \\ 
\textbf{100} & 0.67 & 0.71 & \textbf{0.76} &  & 0.74 & 0.79 & \textbf{0.82}
&  & 0.78 & 0.83 & \textbf{0.85} \\ \hline\hline
\end{tabular}%
	\vspace{-.5cm}
	\begin{flushleft}
		\noindent 
		\scriptsize%
		\singlespacing%
		Notes: Down-weighting column label "No" stands for no down-weighting, "Light"
stands for light down-weighting given by values $\lambda
=0.975,0.98,0.985,0.99,0.995,1$, and "Heavy" stands for heavy down-weighting
given by values $\lambda =0.95,0.96,0.97,0.98,0.99,1$. For each of the two
sets of exponential down-weighting (light/heavy) we average TPR across the
choices for $\lambda $. Best results are highlighted by bold fonts. The
reported results are based on 4 experiments for models without parameter
instabilities (panel A) and 4 experiments with parameter instabilities
(panel B). See Section \ref{sec:MC-studies} for the description of the Monte
Carlo design. 
	\end{flushleft}
\end{table}

\begin{table}
	\caption{Comparison of the effects of down-weighting for FPR
performance in MC experiments with and without parameter instability. \bigskip} \label{mc_tab_s2}
	\renewcommand{\arraystretch}{1.05}%
	\centering
	\footnotesize%
\begin{tabular}{rrrrrrrrrrrr}
\hline\hline
& \multicolumn{11}{c}{\textbf{Average FPR}} \\ \cline{2-12}
Down-weighting: & No & Light & Heavy &  & No & Light & Heavy &  & No & Light
& Heavy \\ \cline{2-4}\cline{6-8}\cline{10-12}
$N\backslash T$ & \multicolumn{3}{c}{\textbf{100}} &  & \multicolumn{3}{c}{%
\textbf{150}} &  & \multicolumn{3}{c}{\textbf{200}} \\ \hline
\multicolumn{12}{l}{A. Without parameter instability} \\ \hline
\multicolumn{1}{l}{} & \multicolumn{11}{l}{OCMT (down-weighting at the
selection stage)} \\ \hline
\textbf{20} & 0.08 & \textbf{0.06} & 0.10 &  & 0.13 & \textbf{0.09} & 0.17 & 
& 0.17 & \textbf{0.13} & 0.23 \\ 
\textbf{40} & 0.04 & \textbf{0.03} & 0.09 &  & \textbf{0.06} & 0.06 & 0.16 & 
& \textbf{0.08} & 0.10 & 0.22 \\ 
\textbf{100} & \textbf{0.01} & 0.02 & 0.07 &  & \textbf{0.02} & 0.04 & 0.14
&  & \textbf{0.03} & 0.07 & 0.21 \\ \hline
\multicolumn{1}{l}{} & \multicolumn{10}{l}{Lasso} & \multicolumn{1}{l}{} \\ 
\hline
\textbf{20} & \textbf{0.17} & 0.20 & 0.23 &  & \textbf{0.17} & 0.20 & 0.23 & 
& \textbf{0.17} & 0.20 & 0.23 \\ 
\textbf{40} & \textbf{0.12} & 0.17 & 0.24 &  & \textbf{0.13} & 0.17 & 0.24 & 
& \textbf{0.13} & 0.17 & 0.24 \\ 
\textbf{100} & \textbf{0.08} & 0.14 & 0.19 &  & \textbf{0.08} & 0.14 & 0.23
&  & \textbf{0.07} & 0.14 & 0.23 \\ \hline
\multicolumn{1}{l}{} & \multicolumn{10}{l}{A-Lasso} & \multicolumn{1}{l}{}
\\ \hline
\textbf{20} & \textbf{0.11} & 0.14 & 0.17 &  & \textbf{0.11} & 0.14 & 0.16 & 
& \textbf{0.11} & 0.14 & 0.16 \\ 
\textbf{40} & \textbf{0.09} & 0.13 & 0.18 &  & \textbf{0.09} & 0.13 & 0.18 & 
& \textbf{0.09} & 0.13 & 0.18 \\ 
\textbf{100} & \textbf{0.06} & 0.11 & 0.15 &  & \textbf{0.06} & 0.11 & 0.18
&  & \textbf{0.05} & 0.12 & 0.18 \\ \hline
\multicolumn{1}{l}{} & \multicolumn{10}{l}{Boosting} & \multicolumn{1}{l}{}
\\ \hline
\textbf{20} & \textbf{0.08} & 0.16 & 0.28 &  & \textbf{0.07} & 0.19 & 0.33 & 
& \textbf{0.06} & 0.22 & 0.37 \\ 
\textbf{40} & \textbf{0.07} & 0.19 & 0.38 &  & \textbf{0.06} & 0.23 & 0.43 & 
& \textbf{0.05} & 0.27 & 0.47 \\ 
\textbf{100} & \textbf{0.08} & 0.25 & 0.40 &  & \textbf{0.06} & 0.30 & 0.44
&  & \textbf{0.05} & 0.34 & 0.46 \\ \hline
\multicolumn{12}{l}{B. With parameter instability} \\ \hline
\multicolumn{1}{l}{} & \multicolumn{11}{l}{OCMT (down-weighting at the
selection stage)} \\ \hline
\textbf{20} & 0.06 & \textbf{0.05} & 0.10 &  & \textbf{0.08} & 0.08 & 0.17 & 
& \textbf{0.11} & 0.12 & 0.23 \\ 
\textbf{40} & \textbf{0.02} & 0.03 & 0.09 &  & \textbf{0.04} & 0.07 & 0.16 & 
& \textbf{0.05} & 0.10 & 0.23 \\ 
\textbf{100} & \textbf{0.01} & 0.02 & 0.07 &  & \textbf{0.01} & 0.05 & 0.14
&  & \textbf{0.02} & 0.08 & 0.21 \\ \hline
\multicolumn{1}{l}{} & \multicolumn{10}{l}{Lasso} & \multicolumn{1}{l}{} \\ 
\hline
\textbf{20} & \textbf{0.21} & 0.23 & 0.26 &  & \textbf{0.22} & 0.23 & 0.26 & 
& \textbf{0.23} & 0.24 & 0.26 \\ 
\textbf{40} & \textbf{0.17} & 0.20 & 0.26 &  & \textbf{0.18} & 0.21 & 0.27 & 
& \textbf{0.19} & 0.21 & 0.27 \\ 
\textbf{100} & \textbf{0.11} & 0.16 & 0.21 &  & \textbf{0.12} & 0.17 & 0.24
&  & \textbf{0.12} & 0.17 & 0.25 \\ 
\multicolumn{1}{l}{} & \multicolumn{10}{l}{A-Lasso} & \multicolumn{1}{l}{}
\\ 
\textbf{20} & \textbf{0.15} & 0.17 & 0.19 &  & \textbf{0.15} & 0.16 & 0.18 & 
& \textbf{0.16} & 0.17 & 0.18 \\ 
\textbf{40} & \textbf{0.12} & 0.15 & 0.20 &  & \textbf{0.13} & 0.16 & 0.20 & 
& \textbf{0.14} & 0.16 & 0.20 \\ 
\textbf{100} & \textbf{0.08} & 0.12 & 0.16 &  & \textbf{0.09} & 0.13 & 0.19
&  & \textbf{0.09} & 0.14 & 0.19 \\ \hline
\multicolumn{1}{l}{} & \multicolumn{10}{l}{Boosting} & \multicolumn{1}{l}{}
\\ \hline
\textbf{20} & \textbf{0.09} & 0.17 & 0.28 &  & \textbf{0.08} & 0.20 & 0.34 & 
& \textbf{0.08} & 0.23 & 0.38 \\ 
\textbf{40} & \textbf{0.10} & 0.20 & 0.38 &  & \textbf{0.08} & 0.23 & 0.43 & 
& \textbf{0.08} & 0.27 & 0.48 \\ 
\textbf{100} & \textbf{0.10} & 0.25 & 0.40 &  & \textbf{0.08} & 0.30 & 0.44
&  & \textbf{0.07} & 0.34 & 0.46 \\ \hline\hline
\end{tabular}%
	\vspace{-.5cm}
	\begin{flushleft}
		\noindent 
		\scriptsize%
		\singlespacing%
		Notes: Down-weighting column label "No" stands for no down-weighting, "Light"
stands for light down-weighting given by values $\lambda
=0.975,0.98,0.985,0.99,0.995,1$, and "Heavy" stands for heavy down-weighting
given by values $\lambda =0.95,0.96,0.97,0.98,0.99,1$. For each of the two
sets of exponential down-weighting (light/heavy) we average FPR across the
choices for $\lambda $. Best results are highlighted by bold fonts. The
reported results are based on 4 experiments for models without parameter
instabilities (panel A) and 4 experiments with parameter instabilities
(panel B). See Section \ref{sec:MC-studies} for the description of the Monte
Carlo design. 
	\end{flushleft}
\end{table}

\begin{table}
	\caption{Comparison of the effects of down-weighting for the number of selected variables $\hat{k}$ in MC experiments with and without parameter instability. \bigskip}
	\renewcommand{\arraystretch}{1.05}%
	\centering
	\footnotesize%
\begin{tabular}{rrrrrrrrrrrr}
\hline\hline
& \multicolumn{11}{c}{\textbf{Average }$\hat{k}$} \\ \cline{2-12}
Down-weighting: & No & Light & Heavy &  & No & Light & Heavy &  & No & Light
& Heavy \\ \cline{2-4}\cline{6-8}\cline{10-12}
$N\backslash T$ & \multicolumn{3}{c}{\textbf{100}} &  & \multicolumn{3}{c}{%
\textbf{150}} &  & \multicolumn{3}{c}{\textbf{200}} \\ \hline
\multicolumn{12}{l}{A. Without parameter instability} \\ \hline
\multicolumn{1}{l}{} & \multicolumn{11}{l}{OCMT (down-weighting at the
selection stage)} \\ \hline
\textbf{20} & 5.03 & 3.82 & 4.44 &  & 6.17 & 4.70 & 6.05 &  & 7.22 & 5.63 & 
7.54 \\ 
\textbf{40} & 4.69 & 3.91 & 5.69 &  & 5.98 & 5.42 & 8.89 &  & 6.87 & 6.95 & 
11.89 \\ 
\textbf{100} & 4.31 & 4.31 & 8.94 &  & 5.52 & 7.01 & 16.13 &  & 6.35 & 10.11
& 23.92 \\ \hline
\multicolumn{1}{l}{} & \multicolumn{10}{l}{Lasso} & \multicolumn{1}{l}{} \\ 
\hline
\textbf{20} & 6.82 & 7.09 & 7.60 &  & 7.00 & 7.17 & 7.55 &  & 7.20 & 7.31 & 
7.62 \\ 
\textbf{40} & 8.26 & 9.78 & 12.49 &  & 8.57 & 10.01 & 12.58 &  & 8.74 & 10.10
& 12.68 \\ 
\textbf{100} & 10.76 & 16.51 & 22.19 &  & 11.00 & 17.58 & 26.09 &  & 10.51 & 
17.59 & 26.14 \\ \hline
\multicolumn{1}{l}{} & \multicolumn{10}{l}{A-Lasso} & \multicolumn{1}{l}{}
\\ \hline
\textbf{20} & 5.15 & 5.52 & 5.92 &  & 5.35 & 5.63 & 5.90 &  & 5.55 & 5.79 & 
5.97 \\ 
\textbf{40} & 6.39 & 7.83 & 9.98 &  & 6.78 & 8.08 & 10.03 &  & 6.96 & 8.21 & 
10.14 \\ 
\textbf{100} & 8.65 & 13.27 & 17.38 &  & 9.05 & 14.37 & 20.36 &  & 8.83 & 
14.51 & 20.52 \\ \hline
\multicolumn{1}{l}{} & \multicolumn{10}{l}{Boosting} & \multicolumn{1}{l}{}
\\ \hline
\textbf{20} & 4.59 & 6.20 & 8.66 &  & 4.63 & 7.00 & 9.92 &  & 4.70 & 7.76 & 
10.87 \\ 
\textbf{40} & 6.04 & 10.58 & 18.30 &  & 5.79 & 12.30 & 20.67 &  & 5.69 & 
14.16 & 22.54 \\ 
\textbf{100} & 11.36 & 28.60 & 43.23 &  & 9.27 & 33.19 & 47.00 &  & 8.43 & 
37.53 & 49.32 \\ \hline
\multicolumn{12}{l}{B. With parameter instability} \\ \hline
\multicolumn{1}{l}{} & \multicolumn{11}{l}{OCMT (down-weighting at the
selection stage)} \\ \hline
\textbf{20} & 4.04 & 3.39 & 4.32 &  & 5.07 & 4.45 & 6.03 &  & 5.96 & 5.45 & 
7.52 \\ 
\textbf{40} & 3.78 & 3.62 & 5.71 &  & 4.90 & 5.28 & 8.96 &  & 5.67 & 6.92 & 
12.00 \\ 
\textbf{100} & 3.54 & 4.37 & 9.28 &  & 4.62 & 7.25 & 16.44 &  & 5.26 & 10.38
& 24.22 \\ \hline
\multicolumn{1}{l}{} & \multicolumn{10}{l}{Lasso} & \multicolumn{1}{l}{} \\ 
\hline
\textbf{20} & 7.28 & 7.51 & 8.00 &  & 7.76 & 7.69 & 8.13 &  & 8.17 & 7.96 & 
8.27 \\ 
\textbf{40} & 9.80 & 10.90 & 13.32 &  & 10.60 & 11.31 & 13.64 &  & 11.13 & 
11.44 & 13.83 \\ 
\textbf{100} & 13.68 & 18.72 & 23.30 &  & 14.83 & 20.09 & 27.36 &  & 15.56 & 
20.23 & 27.93 \\ 
\multicolumn{1}{l}{} & \multicolumn{10}{l}{A-Lasso} & \multicolumn{1}{l}{}
\\ 
\textbf{20} & 5.49 & 5.83 & 6.22 &  & 5.95 & 6.04 & 6.35 &  & 6.30 & 6.27 & 
6.47 \\ 
\textbf{40} & 7.55 & 8.63 & 10.58 &  & 8.28 & 9.03 & 10.81 &  & 8.76 & 9.23
& 11.03 \\ 
\textbf{100} & 10.71 & 14.82 & 18.14 &  & 11.85 & 16.20 & 21.27 &  & 12.58 & 
16.50 & 21.86 \\ \hline
\multicolumn{1}{l}{} & \multicolumn{10}{l}{Boosting} & \multicolumn{1}{l}{}
\\ \hline
\textbf{20} & 4.59 & 6.16 & 8.59 &  & 4.66 & 7.00 & 9.95 &  & 4.75 & 7.79 & 
10.95 \\ 
\textbf{40} & 6.52 & 10.78 & 18.18 &  & 6.35 & 12.49 & 20.70 &  & 6.21 & 
14.24 & 22.61 \\ 
\textbf{100} & 12.70 & 28.14 & 42.94 &  & 10.73 & 33.22 & 47.05 &  & 10.03 & 
37.60 & 49.49 \\ \hline\hline
\end{tabular}%
	\vspace{-.5cm}
	\begin{flushleft}
		\noindent 
		\scriptsize%
		\singlespacing%
		Notes: Down-weighting column label "No" stands for no down-weighting, "Light"
stands for light down-weighting given by values $\lambda
=0.975,0.98,0.985,0.99,0.995,1$, and "Heavy" stands for heavy down-weighting
given by values $\lambda =0.95,0.96,0.97,0.98,0.99,1$. For each of the two
sets of exponential down-weighting (light/heavy) we average $\hat{k}$ across the
choices for $\lambda $. The
reported results are based on 4 experiments for models without parameter
instabilities (panel A) and 4 experiments with parameter instabilities
(panel B). See Section \ref{sec:MC-studies} for the description of the Monte
Carlo design. 
	\end{flushleft}
\end{table}

\begin{table}
	\caption{The number of selected variables  ($\hat{k}$), True Positive Rate (TRP), and False Positive Rate (FPR) averaged across MC experiments with  and without dynamics. \bigskip}\label{mc-tab-s4}
	\renewcommand{\arraystretch}{1}%
	\centering
	\footnotesize%
\begin{tabular}{rccccccccccc}
\hline\hline
& \multicolumn{3}{c}{$\hat{k}$} &  & \multicolumn{3}{c}{TPR} &  & 
\multicolumn{3}{c}{FPR} \\ \cline{2-4}\cline{6-8}\cline{10-12}
\multicolumn{1}{c}{$N\backslash T$} & \textbf{100} & \textbf{150} & \textbf{%
200} &  & \textbf{100} & \textbf{150} & \textbf{200} &  & \textbf{100} & 
\textbf{150} & \textbf{200} \\ \cline{1-4}\cline{6-8}\cline{10-12}
\multicolumn{12}{l}{A. Static} \\ \hline
& \multicolumn{11}{l}{OCMT} \\ \hline
\textbf{20} & 5.78 & 6.93 & 7.95 &  & 0.92 & 0.97 & 0.99 &  & 0.11 & 0.15 & 
0.20 \\ 
\textbf{40} & 5.51 & 6.78 & 7.63 &  & 0.90 & 0.97 & 0.99 &  & 0.05 & 0.07 & 
0.09 \\ 
\textbf{100} & 5.24 & 6.46 & 7.20 &  & 0.87 & 0.96 & 0.98 &  & 0.02 & 0.03 & 
0.03 \\ \hline
& \multicolumn{11}{l}{Lasso} \\ \hline
\textbf{20} & 7.51 & 7.87 & 8.08 &  & 0.87 & 0.92 & 0.95 &  & 0.20 & 0.21 & 
0.21 \\ 
\textbf{40} & 9.43 & 10.12 & 10.59 &  & 0.86 & 0.92 & 0.95 &  & 0.15 & 0.16
& 0.17 \\ 
\textbf{100} & 12.32 & 13.47 & 13.84 &  & 0.83 & 0.90 & 0.93 &  & 0.09 & 0.10
& 0.10 \\ \hline
& \multicolumn{11}{l}{A-Lasso} \\ \hline
\textbf{20} & 5.75 & 6.10 & 6.31 &  & 0.78 & 0.84 & 0.89 &  & 0.13 & 0.14 & 
0.14 \\ 
\textbf{40} & 7.32 & 7.98 & 8.40 &  & 0.78 & 0.86 & 0.90 &  & 0.11 & 0.11 & 
0.12 \\ 
\textbf{100} & 9.75 & 10.89 & 11.30 &  & 0.77 & 0.85 & 0.89 &  & 0.07 & 0.07
& 0.08 \\ \hline
& \multicolumn{11}{l}{Boosting} \\ \hline
\textbf{20} & 5.00 & 5.11 & 5.14 &  & 0.81 & 0.87 & 0.91 &  & 0.09 & 0.08 & 
0.08 \\ 
\textbf{40} & 6.67 & 6.56 & 6.45 &  & 0.81 & 0.87 & 0.91 &  & 0.09 & 0.08 & 
0.07 \\ 
\textbf{100} & 11.71 & 10.32 & 9.74 &  & 0.79 & 0.86 & 0.90 &  & 0.09 & 0.07
& 0.06 \\ \hline
\multicolumn{12}{l}{B. Dynamic} \\ \hline
& \multicolumn{11}{l}{OCMT} \\ \hline
\textbf{20} & 3.29 & 4.31 & 5.23 &  & 0.65 & 0.80 & 0.89 &  & 0.03 & 0.06 & 
0.08 \\ 
\textbf{40} & 2.96 & 4.10 & 4.91 &  & 0.60 & 0.78 & 0.87 &  & 0.01 & 0.02 & 
0.04 \\ 
\textbf{100} & 2.61 & 3.69 & 4.41 &  & 0.55 & 0.73 & 0.83 &  & 0.00 & 0.01 & 
0.01 \\ \hline
& \multicolumn{11}{l}{Lasso} \\ \hline
\textbf{20} & 6.59 & 6.88 & 7.28 &  & 0.72 & 0.80 & 0.85 &  & 0.18 & 0.19 & 
0.19 \\ 
\textbf{40} & 8.63 & 9.04 & 9.28 &  & 0.69 & 0.78 & 0.84 &  & 0.15 & 0.15 & 
0.15 \\ 
\textbf{100} & 12.12 & 12.35 & 12.23 &  & 0.66 & 0.75 & 0.80 &  & 0.09 & 0.09
& 0.09 \\ \hline
& \multicolumn{11}{l}{A-Lasso} \\ \hline
\textbf{20} & 4.89 & 5.20 & 5.54 &  & 0.60 & 0.68 & 0.74 &  & 0.12 & 0.12 & 
0.13 \\ 
\textbf{40} & 6.63 & 7.08 & 7.31 &  & 0.60 & 0.69 & 0.75 &  & 0.11 & 0.11 & 
0.11 \\ 
\textbf{100} & 9.61 & 10.01 & 10.11 &  & 0.58 & 0.69 & 0.75 &  & 0.07 & 0.07
& 0.07 \\ \hline
& \multicolumn{11}{l}{Boosting} \\ \hline
\textbf{20} & 4.17 & 4.18 & 4.30 &  & 0.64 & 0.70 & 0.76 &  & 0.08 & 0.07 & 
0.06 \\ 
\textbf{40} & 5.90 & 5.58 & 5.45 &  & 0.63 & 0.71 & 0.76 &  & 0.08 & 0.07 & 
0.06 \\ 
\textbf{100} & 12.35 & 9.67 & 8.72 &  & 0.63 & 0.70 & 0.75 &  & 0.10 & 0.07
& 0.06 \\ \hline\hline
\end{tabular}%
	\vspace{-.5cm}
	\begin{flushleft}
		\noindent 
		\scriptsize%
		\singlespacing%
		Notes: There are $k=4$ signal variables out of $N$ observed covariates. The top panel reports results averaged across 4 static experiments, which do not feature lagged dependent variable. The bottom panel reports results averaged across 4 dynamic experiments featuring lagged dependent variable. Each experiment is based on 2000 Monte Carlo simulations. OCMT, Lasso and A-Lasso methods in this table are based on original (not down-weighted) observations. See Section 5 of the paper for the detailed description of the Monte Carlo design.
	\end{flushleft}
\end{table}

\begin{table}[h] 
	\caption{\footnotesize Comparison of the effects of down-weighting on one-step-ahead MSFE of OCMT, Lasso, A-Lasso and boosting averaged across all the static MC experiments without and with parameter instabilities.}\label{mc_tab_s6}
	\centering
    \vspace{0.5cm}
	\renewcommand{\arraystretch}{1}%

	\footnotesize%
\begin{tabular}{cccccccccccc}
\hline\hline
\multicolumn{1}{l}{Down-weighting:} & No & Light & Heavy &  & No & Light & 
Heavy &  & No & Light & Heavy \\ \cline{2-12}
$N\backslash T$ & \multicolumn{3}{c}{\textbf{100}} &  & \multicolumn{3}{c}{%
\textbf{200}} &  & \multicolumn{3}{c}{\textbf{300}} \\ \hline
\multicolumn{11}{l}{A. Without parameter instabilities} &  \\ \hline
& \multicolumn{11}{l}{OCMT(Down-weighting only at the estimation stage)} \\ 
\hline
\textbf{20} & \textbf{17.03} & 17.50 & 18.43 &  & \textbf{15.59} & 16.07 & 
17.13 &  & \textbf{14.44} & 15.20 & 16.39 \\ 
\textbf{40} & \textbf{15.83} & 16.27 & 17.06 &  & \textbf{14.71} & 15.17 & 
16.18 &  & \textbf{16.87} & 18.43 & 20.01 \\ 
\textbf{100} & \textbf{16.00} & 16.28 & 16.96 &  & \textbf{15.03} & 15.83 & 
17.05 &  & \textbf{15.81} & 16.41 & 17.66 \\ \hline
& \multicolumn{11}{l}{OCMT(Down-weighting only at the variable selection and
estimation stages)} \\ \hline
\textbf{20} & \textbf{17.03} & 17.40 & 18.95 &  & \textbf{15.59} & 16.19 & 
17.83 &  & \textbf{14.44} & 14.99 & 17.16 \\ 
\textbf{40} & \textbf{15.83} & 16.30 & 18.07 &  & \textbf{14.71} & 14.96 & 
17.80 &  & \textbf{16.87} & 18.61 & 23.16 \\ 
\textbf{100} & \textbf{16.00} & 16.50 & 19.89 &  & \textbf{15.03} & 16.55 & 
22.85 &  & \textbf{15.81} & 17.40 & 25.81 \\ \hline
& \multicolumn{11}{l}{Lasso} \\ \hline
\textbf{20} & \textbf{17.23} & 17.77 & 18.89 &  & \textbf{15.61} & 16.15 & 
17.16 &  & \textbf{14.41} & 14.76 & 15.71 \\ 
\textbf{40} & \textbf{16.11} & 17.00 & 18.33 &  & \textbf{14.43} & 15.21 & 
17.06 &  & \textbf{16.56} & 17.97 & 19.91 \\ 
\textbf{100} & \textbf{16.44} & 17.99 & 20.58 &  & \textbf{15.38} & 16.76 & 
18.84 &  & \textbf{15.94} & 17.26 & 18.90 \\ \hline
& \multicolumn{11}{l}{A-Lasso} \\ \hline
\textbf{20} & \textbf{18.08} & 18.52 & 19.76 &  & \textbf{16.05} & 16.77 & 
17.87 &  & \textbf{14.75} & 15.10 & 16.32 \\ 
\textbf{40} & \textbf{17.25} & 18.23 & 19.71 &  & \textbf{15.25} & 16.30 & 
18.26 &  & \textbf{17.24} & 19.10 & 21.13 \\ 
\textbf{100} & \textbf{18.56} & 20.33 & 22.88 &  & \textbf{16.55} & 18.29 & 
20.87 &  & \textbf{16.93} & 18.87 & 20.57 \\ \hline
& \multicolumn{11}{l}{Boosting} \\ \hline
\textbf{20} & \textbf{17.56} & 18.43 & 21.17 &  & \textbf{15.94} & 16.94 & 
19.68 &  & \textbf{14.49} & 15.63 & 18.94 \\ 
\textbf{40} & \textbf{16.62} & 17.87 & 21.69 &  & \textbf{14.82} & 16.52 & 
21.52 &  & \textbf{17.02} & 20.22 & 26.76 \\ 
\textbf{100} & \textbf{17.45} & 21.81 & 25.46 &  & \textbf{15.93} & 19.98 & 
24.11 &  & \textbf{16.26} & 21.13 & 25.42 \\ \hline
\multicolumn{11}{l}{B. With parameter instabilities} &  \\ \hline
& \multicolumn{11}{l}{OCMT(Down-weighting only at the estimation stage)} \\ 
\hline
\textbf{20} & 20.09 & \textbf{19.32} & 19.47 &  & 17.80 & \textbf{17.09} & 
17.58 &  & 16.82 & \textbf{15.82} & 16.62 \\ 
\textbf{40} & 18.46 & \textbf{17.94} & 18.20 &  & 17.32 & \textbf{16.38} & 
16.84 &  & 19.43 & \textbf{19.18} & 20.34 \\ 
\textbf{100} & 19.01 & \textbf{18.55} & 18.67 &  & 17.97 & \textbf{17.19} & 
17.76 &  & 18.75 & \textbf{17.46} & 18.21 \\ \hline
& \multicolumn{11}{l}{OCMT(Down-weighting only at the variable selection and
estimation stages)} \\ \hline
\textbf{20} & 20.09 & \textbf{19.95} & 20.91 &  & 17.80 & \textbf{17.71} & 
19.32 &  & 16.82 & \textbf{15.92} & 17.97 \\ 
\textbf{40} & \textbf{18.46} & 18.65 & 20.43 &  & 17.32 & \textbf{17.01} & 
20.09 &  & \textbf{19.43} & 20.20 & 25.09 \\ 
\textbf{100} & \textbf{19.01} & 19.11 & 22.92 &  & \textbf{17.97} & 19.09 & 
25.28 &  & \textbf{18.75} & 19.66 & 28.60 \\ \hline
& \multicolumn{11}{l}{Lasso} \\ \hline
\textbf{20} & 20.93 & \textbf{20.63} & 21.04 &  & 18.29 & \textbf{17.76} & 
18.43 &  & 17.12 & \textbf{16.01} & 16.70 \\ 
\textbf{40} & \textbf{19.23} & 19.47 & 20.40 &  & 17.70 & \textbf{17.47} & 
19.09 &  & \textbf{19.68} & 19.68 & 21.40 \\ 
\textbf{100} & \textbf{19.95} & 20.64 & 22.69 &  & \textbf{18.87} & 19.35 & 
20.84 &  & 19.50 & \textbf{19.18} & 20.70 \\ \hline
& \multicolumn{11}{l}{A-Lasso} \\ \hline
\textbf{20} & 21.69 & \textbf{21.16} & 21.64 &  & 18.85 & \textbf{18.18} & 
18.93 &  & 17.41 & \textbf{16.12} & 17.26 \\ 
\textbf{40} & \textbf{20.32} & 20.56 & 21.84 &  & 18.74 & \textbf{18.52} & 
20.44 &  & \textbf{20.52} & 20.63 & 22.66 \\ 
\textbf{100} & \textbf{22.24} & 22.91 & 24.98 &  & \textbf{20.58} & 20.93 & 
23.00 &  & 21.21 & \textbf{21.00} & 22.66 \\ \hline
& \multicolumn{11}{l}{Boosting} \\ \hline
\textbf{20} & 21.02 & \textbf{20.61} & 22.66 &  & 18.36 & \textbf{18.34} & 
20.97 &  & 17.05 & \textbf{16.72} & 20.02 \\ 
\textbf{40} & \textbf{19.52} & 20.23 & 23.87 &  & \textbf{17.44} & 18.57 & 
23.42 &  & \textbf{19.74} & 21.43 & 27.85 \\ 
\textbf{100} & \textbf{19.98} & 23.60 & 27.42 &  & \textbf{18.69} & 21.53 & 
25.57 &  & \textbf{19.11} & 22.61 & 26.85 \\ \hline\hline
\end{tabular}%
	\vspace{-0.2in}
	\begin{flushleft}
		\noindent 
		
		\scriptsize%
		
		\singlespacing%
		
		Notes:  The reported results are averaged over two experiments (low fit and high fit) for models without and with parameter instabilities. See Section \ref{sec:MC-studies} for the description of the Monte Carlo
		design. Full set of results is presented in the online Monte Carlo supplement.
		
		$^{\dagger }$For each of the two sets of exponential down-weighting
		(light/heavy) forecasts of the target variable are computed as the simple
		average of the forecasts obtained using the down-weighting coefficient, $%
		\lambda $.
	\end{flushleft}
\end{table}

\begin{table}[h] 
	\caption{\footnotesize Comparison of the effects of down-weighting on one-step-ahead MSFE of OCMT, Lasso, A-Lasso and boosting averaged across all the dynamic MC experiments without and with parameter instabilities.}\label{mc_tab_s7}
	\centering
    \vspace{0.5cm}
	\renewcommand{\arraystretch}{1}%

	\footnotesize%
\begin{tabular}{cccccccccccc}
\hline\hline
\multicolumn{1}{l}{Down-weighting:} & No & Light & Heavy &  & No & Light & 
Heavy &  & No & Light & Heavy \\ \cline{2-12}
$N\backslash T$ & \multicolumn{3}{c}{\textbf{100}} &  & \multicolumn{3}{c}{%
\textbf{200}} &  & \multicolumn{3}{c}{\textbf{300}} \\ \hline
\multicolumn{11}{l}{A. Without parameter instabilities} &  \\ \hline
& \multicolumn{11}{l}{OCMT(Down-weighting only at the estimation stage)} \\ 
\hline
\textbf{20} & \textbf{46.49} & 47.83 & 49.98 &  & \textbf{41.46} & 42.60 & 
44.99 &  & \textbf{37.93} & 39.14 & 41.11 \\ 
\textbf{40} & \textbf{42.43} & 42.85 & 43.97 &  & \textbf{38.73} & 39.30 & 
40.86 &  & \textbf{47.24} & 49.56 & 52.57 \\ 
\textbf{100} & \textbf{42.50} & 42.85 & 44.01 &  & \textbf{40.84} & 42.11 & 
44.33 &  & \textbf{42.07} & 42.88 & 45.30 \\ \hline
& \multicolumn{11}{l}{OCMT(Down-weighting only at the variable selection and
estimation stages)} \\ \hline
\textbf{20} & \textbf{46.49} & 47.82 & 51.20 &  & \textbf{41.46} & 42.32 & 
46.54 &  & \textbf{37.93} & 39.72 & 44.19 \\ 
\textbf{40} & \textbf{42.43} & 42.62 & 45.84 &  & \textbf{38.73} & 39.45 & 
45.19 &  & \textbf{47.24} & 49.92 & 59.09 \\ 
\textbf{100} & \textbf{42.50} & 43.90 & 47.82 &  & \textbf{40.84} & 42.37 & 
50.59 &  & \textbf{42.07} & 44.99 & 54.47 \\ \hline
& \multicolumn{11}{l}{Lasso} \\ \hline
\textbf{20} & \textbf{46.41} & 48.93 & 52.09 &  & \textbf{41.58} & 42.83 & 
46.06 &  & \textbf{38.09} & 39.68 & 42.44 \\ 
\textbf{40} & \textbf{42.85} & 44.83 & 49.99 &  & \textbf{38.27} & 40.80 & 
47.56 &  & \textbf{47.00} & 49.53 & 55.69 \\ 
\textbf{100} & \textbf{44.82} & 48.60 & 53.52 &  & \textbf{41.28} & 45.04 & 
51.48 &  & \textbf{42.33} & 45.60 & 51.30 \\ \hline
& \multicolumn{11}{l}{A-Lasso} \\ \hline
\textbf{20} & \textbf{48.41} & 50.92 & 54.42 &  & \textbf{42.90} & 44.11 & 
47.83 &  & \textbf{39.27} & 40.54 & 43.94 \\ 
\textbf{40} & \textbf{46.08} & 47.51 & 52.88 &  & \textbf{40.71} & 43.86 & 
51.19 &  & \textbf{48.83} & 51.08 & 56.65 \\ 
\textbf{100} & \textbf{52.01} & 55.44 & 60.11 &  & \textbf{45.27} & 49.56 & 
56.53 &  & \textbf{46.12} & 49.87 & 55.69 \\ \hline
& \multicolumn{11}{l}{Boosting} \\ \hline
\textbf{20} & \textbf{47.82} & 52.59 & 61.33 &  & \textbf{43.09} & 47.03 & 
56.49 &  & \textbf{39.04} & 42.81 & 52.70 \\ 
\textbf{40} & \textbf{44.73} & 50.57 & 62.92 &  & \textbf{39.59} & 46.80 & 
62.00 &  & \textbf{48.77} & 60.10 & 76.55 \\ 
\textbf{100} & \textbf{49.91} & 62.19 & 71.42 &  & \textbf{42.63} & 57.36 & 
69.52 &  & \textbf{43.70} & 58.55 & 68.93 \\ \hline
\multicolumn{11}{l}{B. With parameter instabilities} &  \\ \hline
& \multicolumn{11}{l}{OCMT(Down-weighting only at the estimation stage)} \\ 
\hline
\textbf{20} & 51.65 & \textbf{50.55} & 51.44 &  & 44.55 & \textbf{43.26} & 
44.45 &  & 41.43 & \textbf{39.82} & 41.21 \\ 
\textbf{40} & 46.37 & \textbf{44.89} & 45.21 &  & 42.74 & \textbf{41.14} & 
41.98 &  & 51.12 & \textbf{50.38} & 52.58 \\ 
\textbf{100} & 47.60 & \textbf{46.54} & 46.98 &  & 45.35 & \textbf{43.91} & 
45.14 &  & 46.68 & \textbf{44.52} & 46.22 \\ \hline
& \multicolumn{11}{l}{OCMT(Down-weighting only at the variable selection and
estimation stages)} \\ \hline
\textbf{20} & 51.65 & \textbf{51.29} & 53.67 &  & 44.55 & \textbf{44.30} & 
48.16 &  & 41.43 & \textbf{41.00} & 45.22 \\ 
\textbf{40} & 46.37 & \textbf{45.53} & 48.40 &  & 42.74 & \textbf{42.01} & 
47.79 &  & 51.12 & \textbf{51.31} & 61.08 \\ 
\textbf{100} & \textbf{47.60} & 47.85 & 51.66 &  & 45.35 & \textbf{45.09} & 
52.79 &  & 46.68 & \textbf{47.05} & 59.42 \\ \hline
& \multicolumn{11}{l}{Lasso} \\ \hline
\textbf{20} & \textbf{52.76} & 53.46 & 55.50 &  & 45.12 & \textbf{44.79} & 
47.61 &  & 42.48 & \textbf{41.48} & 44.01 \\ 
\textbf{40} & \textbf{47.64} & 48.16 & 52.38 &  & \textbf{43.17} & 43.46 & 
49.71 &  & 51.54 & \textbf{51.12} & 56.91 \\ 
\textbf{100} & \textbf{49.94} & 52.33 & 56.53 &  & \textbf{46.40} & 49.09 & 
54.78 &  & \textbf{48.04} & 48.17 & 53.57 \\ \hline
& \multicolumn{11}{l}{A-Lasso} \\ \hline
\textbf{20} & \textbf{55.26} & 55.36 & 57.59 &  & 46.38 & \textbf{45.68} & 
49.11 &  & 43.39 & \textbf{42.12} & 45.31 \\ 
\textbf{40} & \textbf{50.97} & 51.14 & 55.62 &  & \textbf{46.07} & 46.25 & 
53.31 &  & 53.60 & \textbf{52.64} & 58.36 \\ 
\textbf{100} & \textbf{57.39} & 59.47 & 63.11 &  & \textbf{50.76} & 54.02 & 
60.09 &  & \textbf{52.35} & 52.63 & 58.42 \\ \hline
& \multicolumn{11}{l}{Boosting} \\ \hline
\textbf{20} & \textbf{52.12} & 55.54 & 63.86 &  & \textbf{45.19} & 48.95 & 
58.94 &  & \textbf{41.52} & 44.18 & 54.61 \\ 
\textbf{40} & \textbf{48.03} & 54.50 & 66.77 &  & \textbf{42.50} & 50.16 & 
65.76 &  & \textbf{51.12} & 61.14 & 77.47 \\ 
\textbf{100} & \textbf{52.19} & 65.79 & 75.81 &  & \textbf{44.87} & 59.92 & 
72.47 &  & \textbf{46.92} & 61.59 & 72.44 \\ \hline\hline
\end{tabular}%
	\vspace{-0.2in}
	\begin{flushleft}
		\noindent 
		
		\scriptsize%
		
		\singlespacing%
		
		Notes:  The reported results are averaged across two experiments (low fit and high fit) for models without and with parameter instabilities. See Section \ref{sec:MC-studies} for the description of the Monte Carlo
		design. Full set of results is presented in the online Monte Carlo supplement.
		
		$^{\dagger }$For each of the two sets of exponential down-weighting
		(light/heavy) forecasts of the target variable are computed as the simple
		average of the forecasts obtained using the down-weighting coefficient, $%
		\lambda $.
	\end{flushleft}
\end{table}

\clearpage

\section{Monte Carlo results for all the experiments}

\label{detailed_mc_tabs}

\subsection{MC findings for baseline experiments without parameter
instabilities}

\begin{table}[h]

	\vspace{0.1cm}
	\caption{\label{TBS}MC results for methods using no down-weighting in the baseline experiment with no dynamics ($%
	\rho _{y}=0$) and low fit.}
	\centering
 \vspace{0.2cm}
 \renewcommand{\arraystretch}{1.12}
	\scriptsize%
\begin{tabular}{|rrrrrrrrrrrrr|}
\hline\hline
& \multicolumn{3}{c}{\textbf{MSFE (}$\times 100$)} & \multicolumn{3}{|c}{$%
\hat{k}$} & \multicolumn{3}{|c}{\textbf{TPR}} & \multicolumn{3}{|c|}{\textbf{%
FPR}} \\ \hline
$N\backslash T$ & \textbf{100} & \textbf{150} & \textbf{200}
& \textbf{100} & \textbf{150} & \textbf{200} & \textbf{100} & \textbf{150} & 
\textbf{200} & \textbf{100} & \textbf{150} & \textbf{200} \\ \hline
\multicolumn{1}{|l}{} & \multicolumn{11}{l}{\textit{\textbf{Oracle}}} & 
\multicolumn{1}{l|}{} \\ \hline
\textbf{20} & 25.46 & 23.49 & 21.77 & 4.00 & 4.00 & 4.00 & 1.00 & 1.00 & 1.00
& 0.00 & 0.00 & 0.00 \\ 
\textbf{40} & 23.81 & 21.86 & 25.33 & 4.00 & 4.00 & 4.00 & 1.00 & 1.00 & 1.00
& 0.00 & 0.00 & 0.00 \\ 
\textbf{100} & 24.11 & 22.98 & 23.93 & 4.00 & 4.00 & 4.00 & 1.00 & 1.00 & 
1.00 & 0.00 & 0.00 & 0.00 \\ \hline
\multicolumn{1}{|l}{} & \multicolumn{11}{l}{\textit{\textbf{OCMT}}} &  \\ 
\hline
\textbf{20} & 26.27 & 24.00 & 22.26 & 5.22 & 6.37 & 7.38 & 0.90 & 0.96 & 0.99
& 0.08 & 0.13 & 0.17 \\ 
\textbf{40} & 24.39 & 22.60 & 25.97 & 4.87 & 6.18 & 7.03 & 0.88 & 0.97 & 0.99
& 0.03 & 0.06 & 0.08 \\ 
\textbf{100} & 24.73 & 23.14 & 24.40 & 4.49 & 5.72 & 6.48 & 0.84 & 0.95 & 
0.98 & 0.01 & 0.02 & 0.03 \\ \hline
\multicolumn{1}{|l}{} & \textit{\textbf{LASSO}} &  &  &  &  &  &  &  &  &  & 
&  \\ \hline
\textbf{20} & 26.50 & 24.14 & 22.28 & 6.71 & 7.03 & 7.22 & 0.85 & 0.91 & 0.95
& 0.17 & 0.17 & 0.17 \\ 
\textbf{40} & 24.76 & 22.28 & 25.55 & 7.91 & 8.41 & 8.73 & 0.82 & 0.90 & 0.94
& 0.12 & 0.12 & 0.12 \\ 
\textbf{100} & 25.34 & 23.77 & 24.66 & 9.86 & 10.56 & 10.23 & 0.79 & 0.88 & 
0.92 & 0.07 & 0.07 & 0.07 \\ \hline
\multicolumn{1}{|l}{} & \multicolumn{11}{l}{\textit{\textbf{LASSO for
variable selection only. LS for estimation/forecasting.}}} &  \\ \hline
\textbf{20} & 28.54 & 25.06 & 22.83 & 6.71 & 7.03 & 7.22 & 0.85 & 0.91 & 0.95
& 0.17 & 0.17 & 0.17 \\ 
\textbf{40} & 27.07 & 24.05 & 26.92 & 7.91 & 8.41 & 8.73 & 0.82 & 0.90 & 0.94
& 0.12 & 0.12 & 0.12 \\ 
\textbf{100} & 30.00 & 26.12 & 26.82 & 9.86 & 10.56 & 10.23 & 0.79 & 0.88 & 
0.92 & 0.07 & 0.07 & 0.07 \\ \hline
\multicolumn{1}{|l}{} & \multicolumn{11}{l}{\textit{\textbf{A-LASSO}}} &  \\ 
\hline
\textbf{20} & 27.85 & 24.90 & 22.84 & 5.01 & 5.31 & 5.50 & 0.72 & 0.80 & 0.86
& 0.11 & 0.11 & 0.10 \\ 
\textbf{40} & 26.46 & 23.55 & 26.66 & 6.04 & 6.58 & 6.91 & 0.72 & 0.82 & 0.87
& 0.08 & 0.08 & 0.09 \\ 
\textbf{100} & 28.71 & 25.63 & 26.23 & 7.86 & 8.65 & 8.54 & 0.71 & 0.82 & 
0.87 & 0.05 & 0.05 & 0.05 \\ \hline
& \multicolumn{11}{l}{\textit{\textbf{A-LASSO for variable selection only.
LS for estimation/forecasting.}}} &  \\ \hline
\textbf{20} & 28.54 & 25.26 & 22.96 & 5.01 & 5.31 & 5.50 & 0.72 & 0.80 & 0.86
& 0.11 & 0.11 & 0.10 \\ 
\textbf{40} & 27.10 & 24.08 & 27.06 & 6.04 & 6.58 & 6.91 & 0.72 & 0.82 & 0.87
& 0.08 & 0.08 & 0.09 \\ 
\textbf{100} & 29.79 & 26.08 & 26.70 & 7.86 & 8.65 & 8.54 & 0.71 & 0.82 & 
0.87 & 0.05 & 0.05 & 0.05 \\ \hline
\multicolumn{1}{|l}{} & \multicolumn{11}{l}{\textit{\textbf{Boosting}}} & 
\\ \hline
\textbf{20} & 26.71 & 24.37 & 22.23 & 4.50 & 4.66 & 4.77 & 0.76 & 0.84 & 0.90
& 0.07 & 0.06 & 0.06 \\ 
\textbf{40} & 25.22 & 22.67 & 26.03 & 5.82 & 5.79 & 5.73 & 0.76 & 0.84 & 0.89
& 0.07 & 0.06 & 0.05 \\ 
\textbf{100} & 26.55 & 24.36 & 25.00 & 10.31 & 8.92 & 8.31 & 0.74 & 0.83 & 
0.88 & 0.07 & 0.06 & 0.05 \\ \hline
\multicolumn{1}{|l}{} & \multicolumn{11}{l}{\textit{\textbf{Boosting for
variable selection only. LS for estimation/forecasting.}}} &  \\ \hline
\textbf{20} & 27.27 & 24.77 & 22.52 & 4.50 & 4.66 & 4.77 & 0.76 & 0.84 & 0.90
& 0.07 & 0.06 & 0.06 \\ 
\textbf{40} & 26.44 & 23.50 & 26.54 & 5.82 & 5.79 & 5.73 & 0.76 & 0.84 & 0.89
& 0.07 & 0.06 & 0.05 \\ 
\textbf{100} & 31.07 & 26.81 & 27.25 & 10.31 & 8.92 & 8.31 & 0.74 & 0.83 & 
0.88 & 0.07 & 0.06 & 0.05 \\ \hline\hline
\end{tabular}%
	\vspace{-0.2in}
	\begin{flushleft}
		\noindent 
		\scriptsize%
		\singlespacing%
		Notes: This table reports one-step-ahead Mean Square Forecast Error (MSFE, $%
		\times 100$) , average number of selected variables ($\hat{k}$), True
		Positive Rate (TPR), and False Positive Rate (FPR). The baseline model
		features no parameter instabilities in slopes and intercepts.  There are $%
		k=4$ signals variables out of $N$ observed variables. The DGP is given by $%
		y_{t}=d+\rho _{y}y_{t-1}+\sum_{j=1}^{4}\beta
		_{j}x_{jt}+\tau _{u}u_{t}$. Oracle model assumes
		the identity of signal variables is known. The reported results are based on 2000 Monte Carlo replications. See Section \ref{sec:MC-studies} of the paper for the
		detailed description of the Monte Carlo design.
	\end{flushleft}
\end{table}

\begin{table}

	\caption{MC results for methods using light down-weighting in the baseline experiment with no dynamics ($%
	\rho _{y}=0$), and low fit.}
	\centering
 \vspace{0.2cm}
 \renewcommand{\arraystretch}{1.12}
 \setlength{\tabcolsep}{8pt}
	\scriptsize%
	\begin{tabular}{|rrrrrrrrrrrrr|}
\hline\hline
& \multicolumn{3}{c}{\textbf{MSFE (}$\times 100$)} & \multicolumn{3}{|c}{$%
\hat{k}$} & \multicolumn{3}{|c}{\textbf{TPR}} & \multicolumn{3}{|c|}{\textbf{%
FPR}} \\ \hline
$N\backslash T$ & \textbf{100} & \textbf{150} & \textbf{200}
& \textbf{100} & \textbf{150} & \textbf{200} & \textbf{100} & \textbf{150} & 
\textbf{200} & \textbf{100} & \textbf{150} & \multicolumn{1}{r|}{\textbf{200}}
\\ \hline
\multicolumn{13}{|l|}{\textbf{A. Light down-weighting in the
estimation/forecasting stage only. }} \\ \hline
\multicolumn{13}{|l|}{Variable selection is based on original (not
down-weighted) data.} \\ \hline
\multicolumn{13}{|l|}{Forecasting stage is Least Squares on selected
down-weighted covariates for all methods} \\ \hline
& \multicolumn{11}{l}{\textit{\textbf{Oracle}}} & \multicolumn{1}{r|}{} \\ 
\hline
\textbf{20} & 26.03 & 23.94 & 22.22 & 4.00 & 4.00 & 4.00 & 1.00 & 1.00 & 1.00
& 0.00 & 0.00 & \multicolumn{1}{r|}{0.00} \\ 
\textbf{40} & 24.61 & 22.22 & 26.83 & 4.00 & 4.00 & 4.00 & 1.00 & 1.00 & 1.00
& 0.00 & 0.00 & \multicolumn{1}{r|}{0.00} \\ 
\textbf{100} & 24.50 & 23.74 & 24.32 & 4.00 & 4.00 & 4.00 & 1.00 & 1.00 & 
1.00 & 0.00 & 0.00 & \multicolumn{1}{r|}{0.00} \\ \hline
& \multicolumn{11}{l}{\textit{\textbf{OCMT}}} & \multicolumn{1}{r|}{} \\ 
\hline
\textbf{20} & 26.96 & 24.63 & 23.43 & 5.22 & 6.37 & 7.38 & 0.90 & 0.96 & 0.99
& 0.08 & 0.13 & \multicolumn{1}{r|}{0.17} \\ 
\textbf{40} & 24.98 & 23.25 & 28.17 & 4.87 & 6.18 & 7.03 & 0.88 & 0.97 & 0.99
& 0.03 & 0.06 & \multicolumn{1}{r|}{0.08} \\ 
\textbf{100} & 25.07 & 24.38 & 25.24 & 4.49 & 5.72 & 6.48 & 0.84 & 0.95 & 
0.98 & 0.01 & 0.02 & \multicolumn{1}{r|}{0.03} \\ \hline
& \multicolumn{11}{l}{\textit{\textbf{LASSO}}} & \multicolumn{1}{r|}{} \\ 
\hline
\textbf{20} & 29.17 & 26.43 & 23.52 & 6.71 & 7.03 & 7.22 & 0.85 & 0.91 & 0.95
& 0.17 & 0.17 & \multicolumn{1}{r|}{0.17} \\ 
\textbf{40} & 27.76 & 24.94 & 29.54 & 7.91 & 8.41 & 8.73 & 0.82 & 0.90 & 0.94
& 0.12 & 0.12 & \multicolumn{1}{r|}{0.12} \\ 
\textbf{100} & 30.72 & 28.34 & 27.99 & 9.86 & 10.56 & 10.23 & 0.79 & 0.88 & 
0.92 & 0.07 & 0.07 & \multicolumn{1}{r|}{0.07} \\ \hline
& \multicolumn{11}{l}{\textit{\textbf{A-LASSO}}} & \multicolumn{1}{r|}{} \\ 
\hline
\textbf{20} & 28.95 & 26.29 & 23.55 & 5.01 & 5.31 & 5.50 & 0.72 & 0.80 & 0.86
& 0.11 & 0.11 & \multicolumn{1}{r|}{0.10} \\ 
\textbf{40} & 27.67 & 24.46 & 29.55 & 6.04 & 6.58 & 6.91 & 0.72 & 0.82 & 0.87
& 0.08 & 0.08 & \multicolumn{1}{r|}{0.09} \\ 
\textbf{100} & 30.18 & 27.96 & 27.62 & 7.86 & 8.65 & 8.54 & 0.71 & 0.82 & 
0.87 & 0.05 & 0.05 & \multicolumn{1}{r|}{0.05} \\ \hline
& \multicolumn{11}{l}{\textit{\textbf{Boosting}}} & \multicolumn{1}{r|}{} \\ 
\hline
\textbf{20} & 27.77 & 25.76 & 23.25 & 4.50 & 4.66 & 4.77 & 0.76 & 0.84 & 0.90
& 0.07 & 0.06 & \multicolumn{1}{r|}{0.06} \\ 
\textbf{40} & 27.14 & 24.13 & 28.70 & 5.82 & 5.79 & 5.73 & 0.76 & 0.84 & 0.89
& 0.07 & 0.06 & \multicolumn{1}{r|}{0.05} \\ 
\textbf{100} & 31.86 & 28.51 & 28.00 & 10.31 & 8.92 & 8.31 & 0.74 & 0.83 & 
0.88 & 0.07 & 0.06 & \multicolumn{1}{r|}{0.05} \\ \hline
\multicolumn{13}{|l|}{\textbf{B. Light down-weighting in both the variable
selection and estimation/forecasting stages.}} \\ \hline
\multicolumn{13}{|l|}{OCMT uses down-weighted variables for selection as well
as for forecasting using Least Squares.} \\ \hline
\multicolumn{13}{|l|}{Remaining forecasts are based on Lasso, A-Lasso and
Boosting regressions applied to down-weighted data.} \\ \hline
& \multicolumn{11}{l}{\textit{\textbf{OCMT}}} & \multicolumn{1}{r|}{} \\ 
\hline
\textbf{20} & 26.77 & 24.91 & 23.15 & 4.02 & 5.01 & 6.09 & 0.72 & 0.77 & 0.81
& 0.06 & 0.10 & \multicolumn{1}{r|}{0.14} \\ 
\textbf{40} & 25.04 & 23.06 & 28.66 & 4.16 & 5.95 & 7.73 & 0.68 & 0.75 & 0.80
& 0.04 & 0.07 & \multicolumn{1}{r|}{0.11} \\ 
\textbf{100} & 25.38 & 25.31 & 26.52 & 4.72 & 8.01 & 11.67 & 0.63 & 0.72 & 
0.76 & 0.02 & 0.05 & \multicolumn{1}{r|}{0.09} \\ \hline
& \multicolumn{11}{l}{\textit{\textbf{LASSO}}} & \multicolumn{1}{r|}{} \\ 
\hline
\textbf{20} & 27.27 & 24.78 & 22.76 & 6.79 & 6.94 & 7.04 & 0.77 & 0.80 & 0.82
& 0.18 & 0.19 & \multicolumn{1}{r|}{0.19} \\ 
\textbf{40} & 26.08 & 23.41 & 27.64 & 8.88 & 9.14 & 9.32 & 0.75 & 0.79 & 0.80
& 0.15 & 0.15 & \multicolumn{1}{r|}{0.15} \\ 
\textbf{100} & 27.62 & 25.79 & 26.52 & 14.00 & 14.54 & 14.37 & 0.71 & 0.76 & 
0.77 & 0.11 & 0.12 & \multicolumn{1}{r|}{0.11} \\ \hline
& \multicolumn{11}{l}{\textit{\textbf{A-LASSO}}} & \multicolumn{1}{r|}{} \\ 
\hline
\textbf{20} & 28.46 & 25.79 & 23.33 & 5.21 & 5.39 & 5.54 & 0.66 & 0.71 & 0.74
& 0.13 & 0.13 & \multicolumn{1}{r|}{0.13} \\ 
\textbf{40} & 27.96 & 25.11 & 29.45 & 7.06 & 7.35 & 7.55 & 0.67 & 0.72 & 0.74
& 0.11 & 0.11 & \multicolumn{1}{r|}{0.11} \\ 
\textbf{100} & 31.28 & 28.11 & 28.98 & 11.26 & 11.91 & 11.91 & 0.65 & 0.71 & 
0.72 & 0.09 & 0.09 & \multicolumn{1}{r|}{0.09} \\ \hline
& \multicolumn{11}{l}{\textit{\textbf{Boosting}}} & \multicolumn{1}{r|}{} \\ 
\hline
\textbf{20} & 28.12 & 25.93 & 23.88 & 6.01 & 6.86 & 7.63 & 0.75 & 0.81 & 0.85
& 0.15 & 0.18 & \multicolumn{1}{r|}{0.21} \\ 
\textbf{40} & 27.19 & 25.31 & 31.04 & 10.08 & 11.92 & 13.77 & 0.76 & 0.82 & 
0.86 & 0.18 & 0.22 & \multicolumn{1}{r|}{0.26} \\ 
\textbf{100} & 33.31 & 30.66 & 32.46 & 27.16 & 32.12 & 36.55 & 0.77 & 0.84 & 
0.87 & 0.24 & 0.29 & \multicolumn{1}{r|}{0.33} \\ \hline\hline
\end{tabular}%
	\vspace{-0.2in}
	\begin{flushleft}
		\noindent 
		\scriptsize%
		\singlespacing%
		Notes: Light down-weighting is defined by by values $\lambda
=0.975,0.98,0.985,0.99,0.995,1$. For this set of exponential down-weighting
schemes we focus on simple average forecasts computed over the individual
forecasts obtained for each value of $\lambda $ in the set under
consideration. See notes to Table \ref{TBS}.
	\end{flushleft}
\end{table}

\begin{table}

	\caption{MC results for methods using heavy down-weighting in the baseline experiment with no dynamics ($%
	\rho _{y}=0$), and low fit.}
	\centering
 \vspace{0.2cm}
 \renewcommand{\arraystretch}{1.12}
 \setlength{\tabcolsep}{8pt}
	\scriptsize%
\begin{tabular}{|rrrrrrrrrrrrr|}
\hline\hline
& \multicolumn{3}{c}{\textbf{MSFE (}$\times 100$)} & \multicolumn{3}{|c}{$%
\hat{k}$} & \multicolumn{3}{|c}{\textbf{TPR}} & \multicolumn{3}{|c|}{\textbf{%
FPR}} \\ \hline
$N\backslash T$ & \textbf{100} & \textbf{150} & \textbf{200}
& \textbf{100} & \textbf{150} & \textbf{200} & \textbf{100} & \textbf{150} & 
\textbf{200} & \textbf{100} & \textbf{150} & \textbf{200} \\ \hline
\multicolumn{13}{|l|}{\textbf{A. Heavy down-weighting in the
estimation/forecasting stage only. }} \\ \hline
\multicolumn{13}{|l|}{Variable selection is based on original (not
down-weighted) data.} \\ \hline
\multicolumn{13}{|l|}{Forecasting stage is Least Squares on selected
down-weighted covariates for all methods} \\ \hline
& \multicolumn{11}{l}{\textit{\textbf{Oracle}}} &  \\ \hline
\textbf{20} & 27.19 & 25.06 & 23.13 & 4.00 & 4.00 & 4.00 & 1.00 & 1.00 & 1.00
& 0.00 & 0.00 & 0.00 \\ 
\textbf{40} & 25.87 & 23.33 & 28.02 & 4.00 & 4.00 & 4.00 & 1.00 & 1.00 & 1.00
& 0.00 & 0.00 & 0.00 \\ 
\textbf{100} & 25.46 & 24.97 & 25.48 & 4.00 & 4.00 & 4.00 & 1.00 & 1.00 & 
1.00 & 0.00 & 0.00 & 0.00 \\ \hline
& \multicolumn{11}{l}{\textit{\textbf{OCMT}}} &  \\ \hline
\textbf{20} & 28.34 & 26.13 & 25.21 & 5.22 & 6.37 & 7.38 & 0.90 & 0.96 & 0.99
& 0.08 & 0.13 & 0.17 \\ 
\textbf{40} & 26.08 & 24.62 & 30.46 & 4.87 & 6.18 & 7.03 & 0.88 & 0.97 & 0.99
& 0.03 & 0.06 & 0.08 \\ 
\textbf{100} & 25.92 & 26.21 & 27.04 & 4.49 & 5.72 & 6.48 & 0.84 & 0.95 & 
0.98 & 0.01 & 0.02 & 0.03 \\ \hline
& \multicolumn{11}{l}{\textit{\textbf{LASSO}}} &  \\ \hline
\textbf{20} & 30.43 & 28.22 & 24.94 & 6.71 & 7.03 & 7.22 & 0.85 & 0.91 & 0.95
& 0.17 & 0.17 & 0.17 \\ 
\textbf{40} & 29.11 & 26.49 & 31.76 & 7.91 & 8.41 & 8.73 & 0.82 & 0.90 & 0.94
& 0.12 & 0.12 & 0.12 \\ 
\textbf{100} & 32.33 & 31.25 & 30.51 & 9.86 & 10.56 & 10.23 & 0.79 & 0.88 & 
0.92 & 0.07 & 0.07 & 0.07 \\ \hline
& \multicolumn{11}{l}{\textit{\textbf{A-LASSO}}} &  \\ \hline
\textbf{20} & 29.95 & 27.62 & 24.58 & 5.01 & 5.31 & 5.50 & 0.72 & 0.80 & 0.86
& 0.11 & 0.11 & 0.10 \\ 
\textbf{40} & 28.69 & 25.49 & 31.33 & 6.04 & 6.58 & 6.91 & 0.72 & 0.82 & 0.87
& 0.08 & 0.08 & 0.09 \\ 
\textbf{100} & 31.01 & 30.40 & 29.63 & 7.86 & 8.65 & 8.54 & 0.71 & 0.82 & 
0.87 & 0.05 & 0.05 & 0.05 \\ \hline
& \multicolumn{11}{l}{\textit{\textbf{Boosting}}} &  \\ \hline
\textbf{20} & 28.76 & 26.95 & 24.44 & 4.50 & 4.66 & 4.77 & 0.76 & 0.84 & 0.90
& 0.07 & 0.06 & 0.06 \\ 
\textbf{40} & 28.33 & 25.36 & 30.09 & 5.82 & 5.79 & 5.73 & 0.76 & 0.84 & 0.89
& 0.07 & 0.06 & 0.05 \\ 
\textbf{100} & 33.44 & 30.59 & 29.83 & 10.31 & 8.92 & 8.31 & 0.74 & 0.83 & 
0.88 & 0.07 & 0.06 & 0.05 \\ \hline
\multicolumn{13}{|l|}{\textbf{B. Heavy down-weighting in both the variable
selection and estimation/forecasting stages.}} \\ \hline
\multicolumn{13}{|l|}{OCMT uses down-weighted variables for selection as well
as for forecasting using Least Squares.} \\ \hline
\multicolumn{13}{|l|}{Remaining forecasts are based on Lasso, A-Lasso and
Boosting regressions applied to down-weighted data.} \\ \hline
& \multicolumn{11}{l}{\textit{\textbf{OCMT}}} &  \\ \hline
\textbf{20} & 29.10 & 27.31 & 26.41 & 4.80 & 6.53 & 8.15 & 0.63 & 0.70 & 0.76
& 0.11 & 0.19 & 0.26 \\ 
\textbf{40} & 27.52 & 27.22 & 35.47 & 6.33 & 9.90 & 13.05 & 0.60 & 0.69 & 
0.76 & 0.10 & 0.18 & 0.25 \\ 
\textbf{100} & 29.94 & 34.28 & 38.83 & 10.34 & 18.38 & 26.56 & 0.56 & 0.66 & 
0.73 & 0.08 & 0.16 & 0.24 \\ \hline
& \multicolumn{11}{l}{\textit{\textbf{LASSO}}} &  \\ \hline
\textbf{20} & 28.91 & 26.21 & 24.14 & 6.99 & 7.03 & 6.94 & 0.70 & 0.71 & 0.72
& 0.21 & 0.21 & 0.20 \\ 
\textbf{40} & 27.96 & 26.15 & 30.41 & 10.97 & 10.81 & 10.96 & 0.69 & 0.70 & 
0.71 & 0.21 & 0.20 & 0.20 \\ 
\textbf{100} & 31.60 & 28.86 & 28.86 & 19.87 & 20.77 & 20.96 & 0.66 & 0.69 & 
0.70 & 0.17 & 0.18 & 0.18 \\ \hline
& \multicolumn{11}{l}{\textit{\textbf{A-LASSO}}} &  \\ \hline
\textbf{20} & 30.32 & 27.39 & 25.11 & 5.39 & 5.43 & 5.41 & 0.60 & 0.62 & 0.63
& 0.15 & 0.15 & 0.14 \\ 
\textbf{40} & 30.06 & 28.02 & 32.39 & 8.71 & 8.61 & 8.74 & 0.61 & 0.63 & 0.65
& 0.16 & 0.15 & 0.15 \\ 
\textbf{100} & 35.21 & 31.96 & 31.41 & 15.44 & 16.29 & 16.52 & 0.59 & 0.63 & 
0.65 & 0.13 & 0.14 & 0.14 \\ \hline
& \multicolumn{11}{l}{\textit{\textbf{Boosting}}} &  \\ \hline
\textbf{20} & 32.30 & 30.00 & 28.83 & 8.41 & 9.73 & 10.69 & 0.76 & 0.82 & 
0.85 & 0.27 & 0.32 & 0.36 \\ 
\textbf{40} & 33.00 & 32.93 & 41.08 & 17.77 & 20.32 & 22.23 & 0.80 & 0.85 & 
0.88 & 0.36 & 0.42 & 0.47 \\ 
\textbf{100} & 38.83 & 36.94 & 38.92 & 42.41 & 46.33 & 48.74 & 0.79 & 0.84 & 
0.86 & 0.39 & 0.43 & 0.45 \\ \hline\hline
\end{tabular}%
	\vspace{-0.2in}
	\begin{flushleft}
		\noindent 
		\scriptsize%
		\singlespacing%
		Notes: Heavy down-weighting is defined by by values $\lambda
=0.95,0.96,0.97,0.98,0.99,1$. For this set of exponential down-weighting
schemes we focus on simple average forecasts computed over the individual
forecasts obtained for each value of $\lambda $ in the set under
consideration. See notes to Table \ref{TBS}.
	\end{flushleft}
\end{table}

\begin{table}

	\vspace{0.1cm}
	\caption{MC results for methods using no down-weighting in the baseline experiment with no dynamics ($%
	\rho _{y}=0$) and high fit.}
	\centering
 \vspace{0.2cm}
 \renewcommand{\arraystretch}{1.12}
	\scriptsize%
\begin{tabular}{|rrrrrrrrrrrrr|}
\hline\hline
& \multicolumn{3}{c}{\textbf{MSFE (}$\times 100$)} & \multicolumn{3}{|c}{$%
\hat{k}$} & \multicolumn{3}{|c}{\textbf{TPR}} & \multicolumn{3}{|c|}{\textbf{%
FPR}} \\ \hline
$N\backslash T$ & \textbf{100} & \textbf{150} & \textbf{200}
& \textbf{100} & \textbf{150} & \textbf{200} & \textbf{100} & \textbf{150} & 
\textbf{200} & \textbf{100} & \textbf{150} & \textbf{200} \\ \hline
\multicolumn{1}{|l}{} & \multicolumn{11}{l}{\textit{\textbf{Oracle}}} & 
\multicolumn{1}{l|}{} \\ \hline
\textbf{20} & 7.42 & 6.84 & 6.34 & 4.00 & 4.00 & 4.00 & 1.00 & 1.00 & 1.00 & 
0.00 & 0.00 & 0.00 \\ 
\textbf{40} & 6.94 & 6.37 & 7.38 & 4.00 & 4.00 & 4.00 & 1.00 & 1.00 & 1.00 & 
0.00 & 0.00 & 0.00 \\ 
\textbf{100} & 7.03 & 6.69 & 6.97 & 4.00 & 4.00 & 4.00 & 1.00 & 1.00 & 1.00
& 0.00 & 0.00 & 0.00 \\ \hline
\multicolumn{1}{|l}{} & \multicolumn{11}{l}{\textit{\textbf{OCMT}}} &  \\ 
\hline
\textbf{20} & 7.79 & 7.19 & 6.62 & 7.34 & 8.61 & 9.82 & 1.00 & 1.00 & 1.00 & 
0.17 & 0.23 & 0.29 \\ 
\textbf{40} & 7.26 & 6.82 & 7.76 & 7.07 & 8.47 & 9.48 & 0.99 & 1.00 & 1.00 & 
0.08 & 0.11 & 0.14 \\ 
\textbf{100} & 7.26 & 6.91 & 7.22 & 6.72 & 8.05 & 8.98 & 0.99 & 1.00 & 1.00
& 0.03 & 0.04 & 0.05 \\ \hline
\multicolumn{1}{|l}{} & \textit{\textbf{LASSO}} &  &  &  &  &  &  &  &  &  & 
&  \\ \hline
\textbf{20} & 7.96 & 7.07 & 6.55 & 7.52 & 7.56 & 7.53 & 0.98 & 0.99 & 1.00 & 
0.18 & 0.18 & 0.18 \\ 
\textbf{40} & 7.46 & 6.59 & 7.56 & 8.88 & 8.97 & 9.10 & 0.97 & 0.99 & 1.00 & 
0.13 & 0.13 & 0.13 \\ 
\textbf{100} & 7.54 & 7.00 & 7.21 & 11.05 & 11.28 & 10.78 & 0.96 & 0.99 & 
1.00 & 0.07 & 0.07 & 0.07 \\ \hline
\multicolumn{1}{|l}{} & \multicolumn{11}{l}{\textit{\textbf{LASSO for
variable selection only. LS for estimation/forecasting.}}} &  \\ \hline
\textbf{20} & 8.59 & 7.32 & 6.72 & 7.52 & 7.56 & 7.53 & 0.98 & 0.99 & 1.00 & 
0.18 & 0.18 & 0.18 \\ 
\textbf{40} & 8.09 & 7.17 & 8.01 & 8.88 & 8.97 & 9.10 & 0.97 & 0.99 & 1.00 & 
0.13 & 0.13 & 0.13 \\ 
\textbf{100} & 8.72 & 7.69 & 7.80 & 11.05 & 11.28 & 10.78 & 0.96 & 0.99 & 
1.00 & 0.07 & 0.07 & 0.07 \\ \hline
\multicolumn{1}{|l}{} & \multicolumn{11}{l}{\textit{\textbf{A-LASSO}}} &  \\ 
\hline
\textbf{20} & 8.30 & 7.19 & 6.67 & 5.93 & 5.99 & 6.02 & 0.93 & 0.96 & 0.98 & 
0.11 & 0.11 & 0.10 \\ 
\textbf{40} & 8.03 & 6.94 & 7.82 & 7.09 & 7.30 & 7.38 & 0.93 & 0.97 & 0.99 & 
0.08 & 0.09 & 0.09 \\ 
\textbf{100} & 8.41 & 7.46 & 7.63 & 9.04 & 9.46 & 9.14 & 0.92 & 0.97 & 0.99
& 0.05 & 0.06 & 0.05 \\ \hline
& \multicolumn{11}{l}{\textit{\textbf{A-LASSO for variable selection only.
LS for estimation/forecasting.}}} &  \\ \hline
\textbf{20} & 8.47 & 7.26 & 6.74 & 5.93 & 5.99 & 6.02 & 0.93 & 0.96 & 0.98 & 
0.11 & 0.11 & 0.10 \\ 
\textbf{40} & 8.16 & 7.11 & 7.93 & 7.09 & 7.30 & 7.38 & 0.93 & 0.97 & 0.99 & 
0.08 & 0.09 & 0.09 \\ 
\textbf{100} & 8.66 & 7.57 & 7.74 & 9.04 & 9.46 & 9.14 & 0.92 & 0.97 & 0.99
& 0.05 & 0.06 & 0.05 \\ \hline
\multicolumn{1}{|l}{} & \multicolumn{11}{l}{\textit{\textbf{Boosting}}} & 
\\ \hline
\textbf{20} & 8.42 & 7.50 & 6.75 & 5.41 & 5.31 & 5.23 & 0.96 & 0.98 & 0.99 & 
0.08 & 0.07 & 0.06 \\ 
\textbf{40} & 8.01 & 6.97 & 8.01 & 6.79 & 6.45 & 6.23 & 0.95 & 0.98 & 0.99 & 
0.08 & 0.06 & 0.06 \\ 
\textbf{100} & 8.36 & 7.50 & 7.52 & 11.46 & 9.78 & 8.93 & 0.94 & 0.98 & 0.99
& 0.08 & 0.06 & 0.05 \\ \hline
\multicolumn{1}{|l}{} & \multicolumn{11}{l}{\textit{\textbf{Boosting for
variable selection only. LS for estimation/forecasting.}}} &  \\ \hline
\textbf{20} & 8.12 & 7.14 & 6.57 & 5.41 & 5.31 & 5.23 & 0.96 & 0.98 & 0.99 & 
0.08 & 0.07 & 0.06 \\ 
\textbf{40} & 7.92 & 6.93 & 7.89 & 6.79 & 6.45 & 6.23 & 0.95 & 0.98 & 0.99 & 
0.08 & 0.06 & 0.06 \\ 
\textbf{100} & 9.41 & 7.92 & 7.80 & 11.46 & 9.78 & 8.93 & 0.94 & 0.98 & 0.99
& 0.08 & 0.06 & 0.05 \\ \hline\hline
\end{tabular}%
	\vspace{-0.2in}
	\begin{flushleft}
		\noindent 
		\scriptsize%
		\singlespacing%
		Notes: See notes to Table \ref{TBS}.
	\end{flushleft}
\end{table}

\begin{table}

	\caption{MC results for methods using light down-weighting in the baseline experiment with no dynamics ($%
	\rho _{y}=0$), and high fit.}
	\centering
 \vspace{0.2cm}
 \renewcommand{\arraystretch}{1.12}
 \setlength{\tabcolsep}{8pt}
	\scriptsize%
\begin{tabular}{|rrrrrrrrrrrrr|}
\hline\hline
& \multicolumn{3}{c}{\textbf{MSFE (}$\times 100$)} & \multicolumn{3}{|c}{$%
\hat{k}$} & \multicolumn{3}{|c}{\textbf{TPR}} & \multicolumn{3}{|c|}{\textbf{%
FPR}} \\ \hline
$N\backslash T$ & \textbf{100} & \textbf{150} & \textbf{200}
& \textbf{100} & \textbf{150} & \textbf{200} & \textbf{100} & \textbf{150} & 
\textbf{200} & \textbf{100} & \textbf{150} & \textbf{200} \\ \hline
\multicolumn{13}{|l|}{\textbf{A. Light down-weighting in the
estimation/forecasting stage only. }} \\ \hline
\multicolumn{13}{|l|}{Variable selection is based on original (not
down-weighted) data.} \\ \hline
\multicolumn{13}{|l|}{Forecasting stage is Least Squares on selected
down-weighted covariates for all methods} \\ \hline
& \multicolumn{11}{l}{\textit{\textbf{Oracle}}} &  \\ \hline
\textbf{20} & 7.58 & 6.98 & 6.47 & 4.00 & 4.00 & 4.00 & 1.00 & 1.00 & 1.00 & 
0.00 & 0.00 & 0.00 \\ 
\textbf{40} & 7.17 & 6.47 & 7.82 & 4.00 & 4.00 & 4.00 & 1.00 & 1.00 & 1.00 & 
0.00 & 0.00 & 0.00 \\ 
\textbf{100} & 7.14 & 6.92 & 7.08 & 4.00 & 4.00 & 4.00 & 1.00 & 1.00 & 1.00
& 0.00 & 0.00 & 0.00 \\ \hline
& \multicolumn{11}{l}{\textit{\textbf{OCMT}}} &  \\ \hline
\textbf{20} & 8.05 & 7.50 & 6.96 & 7.34 & 8.61 & 9.82 & 1.00 & 1.00 & 1.00 & 
0.17 & 0.23 & 0.29 \\ 
\textbf{40} & 7.55 & 7.10 & 8.69 & 7.07 & 8.47 & 9.48 & 0.99 & 1.00 & 1.00 & 
0.08 & 0.11 & 0.14 \\ 
\textbf{100} & 7.49 & 7.28 & 7.57 & 6.72 & 8.05 & 8.98 & 0.99 & 1.00 & 1.00
& 0.03 & 0.04 & 0.05 \\ \hline
& \multicolumn{11}{l}{\textit{\textbf{LASSO}}} &  \\ \hline
\textbf{20} & 8.84 & 7.67 & 6.95 & 7.52 & 7.56 & 7.53 & 0.98 & 0.99 & 1.00 & 
0.18 & 0.18 & 0.18 \\ 
\textbf{40} & 8.33 & 7.41 & 8.84 & 8.88 & 8.97 & 9.10 & 0.97 & 0.99 & 1.00 & 
0.13 & 0.13 & 0.13 \\ 
\textbf{100} & 9.02 & 8.36 & 8.20 & 11.05 & 11.28 & 10.78 & 0.96 & 0.99 & 
1.00 & 0.07 & 0.07 & 0.07 \\ \hline
& \multicolumn{11}{l}{\textit{\textbf{A-LASSO}}} &  \\ \hline
\textbf{20} & 8.63 & 7.55 & 7.00 & 5.93 & 5.99 & 6.02 & 0.93 & 0.96 & 0.98 & 
0.11 & 0.11 & 0.10 \\ 
\textbf{40} & 8.34 & 7.22 & 8.67 & 7.09 & 7.30 & 7.38 & 0.93 & 0.97 & 0.99 & 
0.08 & 0.09 & 0.09 \\ 
\textbf{100} & 8.93 & 8.13 & 8.06 & 9.04 & 9.46 & 9.14 & 0.92 & 0.97 & 0.99
& 0.05 & 0.06 & 0.05 \\ \hline
& \multicolumn{11}{l}{\textit{\textbf{Boosting}}} &  \\ \hline
\textbf{20} & 8.28 & 7.42 & 6.82 & 5.41 & 5.31 & 5.23 & 0.96 & 0.98 & 0.99 & 
0.08 & 0.07 & 0.06 \\ 
\textbf{40} & 8.15 & 7.16 & 8.66 & 6.79 & 6.45 & 6.23 & 0.95 & 0.98 & 0.99 & 
0.08 & 0.06 & 0.06 \\ 
\textbf{100} & 9.77 & 8.48 & 8.10 & 11.46 & 9.78 & 8.93 & 0.94 & 0.98 & 0.99
& 0.08 & 0.06 & 0.05 \\ \hline
\multicolumn{13}{|l|}{\textbf{B. Light down-weighting in both the variable
selection and estimation/forecasting stages.}} \\ \hline
\multicolumn{13}{|l|}{OCMT uses down-weighted variables for selection as well
as for forecasting using Least Squares.} \\ \hline
\multicolumn{13}{|l|}{Remaining forecasts are based on Lasso, A-Lasso and
Boosting regressions applied to down-weighted data.} \\ \hline
& \multicolumn{11}{l}{\textit{\textbf{OCMT}}} &  \\ \hline
\textbf{20} & 8.04 & 7.47 & 6.83 & 5.77 & 6.78 & 7.86 & 0.90 & 0.90 & 0.91 & 
0.11 & 0.16 & 0.21 \\ 
\textbf{40} & 7.55 & 6.87 & 8.56 & 6.25 & 8.23 & 10.17 & 0.89 & 0.90 & 0.91
& 0.07 & 0.12 & 0.16 \\ 
\textbf{100} & 7.62 & 7.80 & 8.28 & 7.48 & 11.43 & 15.72 & 0.87 & 0.88 & 0.90
& 0.04 & 0.08 & 0.12 \\ \hline
& \multicolumn{11}{l}{\textit{\textbf{LASSO}}} &  \\ \hline
\textbf{20} & 8.27 & 7.53 & 6.75 & 8.00 & 8.10 & 8.14 & 0.95 & 0.96 & 0.97 & 
0.21 & 0.21 & 0.21 \\ 
\textbf{40} & 7.92 & 7.00 & 8.29 & 10.47 & 10.65 & 10.87 & 0.94 & 0.96 & 0.96
& 0.17 & 0.17 & 0.18 \\ 
\textbf{100} & 8.35 & 7.72 & 8.00 & 16.25 & 17.03 & 16.87 & 0.93 & 0.95 & 
0.95 & 0.13 & 0.13 & 0.13 \\ \hline
& \multicolumn{11}{l}{\textit{\textbf{A-LASSO}}} &  \\ \hline
\textbf{20} & 8.58 & 7.75 & 6.88 & 6.39 & 6.50 & 6.58 & 0.90 & 0.92 & 0.94 & 
0.14 & 0.14 & 0.14 \\ 
\textbf{40} & 8.50 & 7.48 & 8.75 & 8.53 & 8.74 & 8.92 & 0.90 & 0.93 & 0.94 & 
0.12 & 0.13 & 0.13 \\ 
\textbf{100} & 9.38 & 8.46 & 8.75 & 13.22 & 14.07 & 14.05 & 0.89 & 0.93 & 
0.94 & 0.10 & 0.10 & 0.10 \\ \hline
& \multicolumn{11}{l}{\textit{\textbf{Boosting}}} &  \\ \hline
\textbf{20} & 8.74 & 7.95 & 7.37 & 6.97 & 7.67 & 8.33 & 0.94 & 0.96 & 0.97 & 
0.16 & 0.19 & 0.22 \\ 
\textbf{40} & 8.54 & 7.74 & 9.39 & 11.13 & 12.79 & 14.57 & 0.94 & 0.97 & 0.98
& 0.18 & 0.22 & 0.27 \\ 
\textbf{100} & 10.30 & 9.31 & 9.80 & 28.66 & 33.21 & 37.40 & 0.94 & 0.97 & 
0.98 & 0.25 & 0.29 & 0.33 \\ \hline\hline
\end{tabular}%
	\vspace{-0.2in}
	\begin{flushleft}
		\noindent 
		\scriptsize%
		\singlespacing%
		Notes: Light down-weighting is defined by by values $\lambda
=0.975,0.98,0.985,0.99,0.995,1$. For this set of exponential down-weighting
schemes we focus on simple average forecasts computed over the individual
forecasts obtained for each value of $\lambda $ in the set under
consideration. See notes to Table \ref{TBS}.
	\end{flushleft}
\end{table}

\begin{table}

	\caption{MC results for methods using heavy down-weighting in the baseline experiment with no dynamics ($%
	\rho _{y}=0$), and high fit.}
	\centering
 \vspace{0.2cm}
 \renewcommand{\arraystretch}{1.12}
 \setlength{\tabcolsep}{8pt}
	\scriptsize%
\begin{tabular}{|rrrrrrrrrrrrr|}
\hline\hline
& \multicolumn{3}{c}{\textbf{MSFE (}$\times 100$)} & \multicolumn{3}{|c}{$%
\hat{k}$} & \multicolumn{3}{|c}{\textbf{TPR}} & \multicolumn{3}{|c|}{\textbf{%
FPR}} \\ \hline
$N\backslash T$ & \textbf{100} & \textbf{150} & \textbf{200}
& \textbf{100} & \textbf{150} & \textbf{200} & \textbf{100} & \textbf{150} & 
\textbf{200} & \textbf{100} & \textbf{150} & \textbf{200} \\ \hline
\multicolumn{13}{|l|}{\textbf{A. Heavy down-weighting in the
estimation/forecasting stage only. }} \\ \hline
\multicolumn{13}{|l|}{Variable selection is based on original (not
down-weighted) data.} \\ \hline
\multicolumn{13}{|l|}{Forecasting stage is Least Squares on selected
down-weighted covariates for all methods} \\ \hline
& \multicolumn{11}{l}{\textit{\textbf{Oracle}}} &  \\ \hline
\textbf{20} & 7.92 & 7.30 & 6.74 & 4.00 & 4.00 & 4.00 & 1.00 & 1.00 & 1.00 & 
0.00 & 0.00 & 0.00 \\ 
\textbf{40} & 7.54 & 6.80 & 8.16 & 4.00 & 4.00 & 4.00 & 1.00 & 1.00 & 1.00 & 
0.00 & 0.00 & 0.00 \\ 
\textbf{100} & 7.42 & 7.27 & 7.42 & 4.00 & 4.00 & 4.00 & 1.00 & 1.00 & 1.00
& 0.00 & 0.00 & 0.00 \\ \hline
& \multicolumn{11}{l}{\textit{\textbf{OCMT}}} &  \\ \hline
\textbf{20} & 8.51 & 8.13 & 7.58 & 7.34 & 8.61 & 9.82 & 1.00 & 1.00 & 1.00 & 
0.17 & 0.23 & 0.29 \\ 
\textbf{40} & 8.04 & 7.73 & 9.56 & 7.07 & 8.47 & 9.48 & 0.99 & 1.00 & 1.00 & 
0.08 & 0.11 & 0.14 \\ 
\textbf{100} & 8.01 & 7.90 & 8.28 & 6.72 & 8.05 & 8.98 & 0.99 & 1.00 & 1.00
& 0.03 & 0.04 & 0.05 \\ \hline
& \multicolumn{11}{l}{\textit{\textbf{LASSO}}} &  \\ \hline
\textbf{20} & 9.24 & 8.21 & 7.38 & 7.52 & 7.56 & 7.53 & 0.98 & 0.99 & 1.00 & 
0.18 & 0.18 & 0.18 \\ 
\textbf{40} & 8.74 & 7.93 & 9.59 & 8.88 & 8.97 & 9.10 & 0.97 & 0.99 & 1.00 & 
0.13 & 0.13 & 0.13 \\ 
\textbf{100} & 9.56 & 9.26 & 8.97 & 11.05 & 11.28 & 10.78 & 0.96 & 0.99 & 
1.00 & 0.07 & 0.07 & 0.07 \\ \hline
& \multicolumn{11}{l}{\textit{\textbf{A-LASSO}}} &  \\ \hline
\textbf{20} & 8.97 & 7.98 & 7.37 & 5.93 & 5.99 & 6.02 & 0.93 & 0.96 & 0.98 & 
0.11 & 0.11 & 0.10 \\ 
\textbf{40} & 8.70 & 7.57 & 9.28 & 7.09 & 7.30 & 7.38 & 0.93 & 0.97 & 0.99 & 
0.08 & 0.09 & 0.09 \\ 
\textbf{100} & 9.36 & 8.91 & 8.67 & 9.04 & 9.46 & 9.14 & 0.92 & 0.97 & 0.99
& 0.05 & 0.06 & 0.05 \\ \hline
& \multicolumn{11}{l}{\textit{\textbf{Boosting}}} &  \\ \hline
\textbf{20} & 8.63 & 7.82 & 7.21 & 5.41 & 5.31 & 5.23 & 0.96 & 0.98 & 0.99 & 
0.08 & 0.07 & 0.06 \\ 
\textbf{40} & 8.53 & 7.54 & 9.22 & 6.79 & 6.45 & 6.23 & 0.95 & 0.98 & 0.99 & 
0.08 & 0.06 & 0.06 \\ 
\textbf{100} & 10.41 & 9.20 & 8.76 & 11.46 & 9.78 & 8.93 & 0.94 & 0.98 & 0.99
& 0.08 & 0.06 & 0.05 \\ \hline
\multicolumn{13}{|l|}{\textbf{B. Heavy down-weighting in both the variable
selection and estimation/forecasting stages.}} \\ \hline
\multicolumn{13}{|l|}{OCMT uses down-weighted variables for selection as well
as for forecasting using Least Squares.} \\ \hline
\multicolumn{13}{|l|}{Remaining forecasts are based on Lasso, A-Lasso and
Boosting regressions applied to down-weighted data.} \\ \hline
& \multicolumn{11}{l}{\textit{\textbf{OCMT}}} &  \\ \hline
\textbf{20} & 8.81 & 8.35 & 7.91 & 6.52 & 8.25 & 9.82 & 0.80 & 0.83 & 0.87 & 
0.17 & 0.25 & 0.32 \\ 
\textbf{40} & 8.63 & 8.38 & 10.85 & 8.72 & 12.41 & 15.59 & 0.79 & 0.83 & 0.87
& 0.14 & 0.23 & 0.30 \\ 
\textbf{100} & 9.83 & 11.42 & 12.78 & 14.24 & 22.78 & 31.22 & 0.76 & 0.82 & 
0.86 & 0.11 & 0.19 & 0.28 \\ \hline
& \multicolumn{11}{l}{\textit{\textbf{LASSO}}} &  \\ \hline
\textbf{20} & 8.88 & 8.11 & 7.27 & 8.56 & 8.56 & 8.53 & 0.90 & 0.91 & 0.91 & 
0.25 & 0.25 & 0.24 \\ 
\textbf{40} & 8.70 & 7.96 & 9.40 & 13.08 & 12.99 & 13.18 & 0.90 & 0.91 & 0.91
& 0.24 & 0.23 & 0.24 \\ 
\textbf{100} & 9.57 & 8.82 & 8.95 & 22.24 & 23.59 & 23.96 & 0.89 & 0.90 & 
0.91 & 0.19 & 0.20 & 0.20 \\ \hline
& \multicolumn{11}{l}{\textit{\textbf{A-LASSO}}} &  \\ \hline
\textbf{20} & 9.20 & 8.35 & 7.54 & 6.78 & 6.78 & 6.78 & 0.84 & 0.86 & 0.87 & 
0.17 & 0.17 & 0.17 \\ 
\textbf{40} & 9.37 & 8.49 & 9.88 & 10.53 & 10.43 & 10.59 & 0.85 & 0.87 & 0.87
& 0.18 & 0.17 & 0.18 \\ 
\textbf{100} & 10.54 & 9.78 & 9.74 & 17.29 & 18.52 & 18.88 & 0.85 & 0.87 & 
0.88 & 0.14 & 0.15 & 0.15 \\ \hline
& \multicolumn{11}{l}{\textit{\textbf{Boosting}}} &  \\ \hline
\textbf{20} & 10.05 & 9.37 & 9.05 & 9.31 & 10.45 & 11.31 & 0.93 & 0.95 & 0.96
& 0.28 & 0.33 & 0.37 \\ 
\textbf{40} & 10.39 & 10.10 & 12.44 & 18.75 & 21.01 & 22.80 & 0.94 & 0.96 & 
0.97 & 0.37 & 0.43 & 0.47 \\ 
\textbf{100} & 12.09 & 11.28 & 11.92 & 43.33 & 47.00 & 49.27 & 0.94 & 0.96 & 
0.96 & 0.40 & 0.43 & 0.45 \\ \hline\hline
\end{tabular}%
	\vspace{-0.2in}
	\begin{flushleft}
		\noindent 
		\scriptsize%
		\singlespacing%
		Notes: Heavy down-weighting is defined by by values $\lambda
=0.95,0.96,0.97,0.98,0.99,1$. For this set of exponential down-weighting
schemes we focus on simple average forecasts computed over the individual
forecasts obtained for each value of $\lambda $ in the set under
consideration. See notes to Table \ref{TBS}.
	\end{flushleft}
\end{table}

\begin{table}

	\vspace{0.1cm}
	\caption{MC results for methods using no down-weighting in the baseline experiment with dynamics ($%
	\rho _{y} \neq 0$) and low fit.}
	\centering
 \vspace{0.2cm}
 \renewcommand{\arraystretch}{1.12}
	\scriptsize%
\begin{tabular}{|rrrrrrrrrrrrr|}
\hline\hline
& \multicolumn{3}{c}{\textbf{MSFE (}$\times 100$)} & \multicolumn{3}{|c}{$%
\hat{k}$} & \multicolumn{3}{|c}{\textbf{TPR}} & \multicolumn{3}{|c|}{\textbf{%
FPR}} \\ \hline
$N\backslash T$ & \textbf{100} & \textbf{150} & \textbf{200}
& \textbf{100} & \textbf{150} & \textbf{200} & \textbf{100} & \textbf{150} & 
\textbf{200} & \textbf{100} & \textbf{150} & \textbf{200} \\ \hline
\multicolumn{1}{|l}{} & \multicolumn{11}{l}{\textit{\textbf{Oracle}}} & 
\multicolumn{1}{l|}{} \\ \hline
\textbf{20} & 69.43 & 62.57 & 58.00 & 4.00 & 4.00 & 4.00 & 1.00 & 1.00 & 1.00
& 0.00 & 0.00 & 0.00 \\ 
\textbf{40} & 63.86 & 58.55 & 71.68 & 4.00 & 4.00 & 4.00 & 1.00 & 1.00 & 1.00
& 0.00 & 0.00 & 0.00 \\ 
\textbf{100} & 65.21 & 61.36 & 63.98 & 4.00 & 4.00 & 4.00 & 1.00 & 1.00 & 
1.00 & 0.00 & 0.00 & 0.00 \\ \hline
\multicolumn{1}{|l}{} & \multicolumn{11}{l}{\textit{\textbf{OCMT}}} &  \\ 
\hline
\textbf{20} & 71.51 & 64.00 & 58.39 & 2.49 & 3.58 & 4.57 & 0.52 & 0.71 & 0.86
& 0.02 & 0.04 & 0.06 \\ 
\textbf{40} & 65.69 & 59.71 & 72.73 & 2.11 & 3.37 & 4.26 & 0.45 & 0.69 & 0.83
& 0.01 & 0.02 & 0.02 \\ 
\textbf{100} & 65.53 & 63.15 & 64.85 & 1.74 & 2.89 & 3.72 & 0.38 & 0.61 & 
0.77 & 0.00 & 0.00 & 0.01 \\ \hline
\multicolumn{1}{|l}{} & \textit{\textbf{LASSO}} &  &  &  &  &  &  &  &  &  & 
&  \\ \hline
\textbf{20} & 71.15 & 63.81 & 58.62 & 5.86 & 6.17 & 6.55 & 0.65 & 0.74 & 0.80
& 0.16 & 0.16 & 0.17 \\ 
\textbf{40} & 65.64 & 58.78 & 72.31 & 7.37 & 7.82 & 8.05 & 0.61 & 0.72 & 0.79
& 0.12 & 0.12 & 0.12 \\ 
\textbf{100} & 68.64 & 63.49 & 65.05 & 10.23 & 10.36 & 9.88 & 0.57 & 0.68 & 
0.75 & 0.08 & 0.08 & 0.07 \\ \hline
\multicolumn{1}{|l}{} & \multicolumn{11}{l}{\textit{\textbf{LASSO for
variable selection only. LS for estimation/forecasting.}}} &  \\ \hline
\textbf{20} & 76.37 & 66.19 & 60.83 & 5.86 & 6.17 & 6.55 & 0.65 & 0.74 & 0.80
& 0.16 & 0.16 & 0.17 \\ 
\textbf{40} & 71.62 & 64.23 & 76.14 & 7.37 & 7.82 & 8.05 & 0.61 & 0.72 & 0.79
& 0.12 & 0.12 & 0.12 \\ 
\textbf{100} & 82.05 & 71.49 & 71.90 & 10.23 & 10.36 & 9.88 & 0.57 & 0.68 & 
0.75 & 0.08 & 0.08 & 0.07 \\ \hline
\multicolumn{1}{|l}{} & \multicolumn{11}{l}{\textit{\textbf{A-LASSO}}} &  \\ 
\hline
\textbf{20} & 74.27 & 65.89 & 60.46 & 4.26 & 4.51 & 4.82 & 0.52 & 0.60 & 0.67
& 0.11 & 0.11 & 0.11 \\ 
\textbf{40} & 70.61 & 62.33 & 75.17 & 5.59 & 6.00 & 6.19 & 0.51 & 0.61 & 0.68
& 0.09 & 0.09 & 0.09 \\ 
\textbf{100} & 79.63 & 69.76 & 70.80 & 8.11 & 8.37 & 8.17 & 0.49 & 0.60 & 
0.67 & 0.06 & 0.06 & 0.05 \\ \hline
& \multicolumn{11}{l}{\textit{\textbf{A-LASSO for variable selection only.
LS for estimation/forecasting.}}} &  \\ \hline
\textbf{20} & 75.96 & 66.61 & 61.08 & 4.26 & 4.51 & 4.82 & 0.52 & 0.60 & 0.67
& 0.11 & 0.11 & 0.11 \\ 
\textbf{40} & 72.24 & 63.79 & 76.04 & 5.59 & 6.00 & 6.19 & 0.51 & 0.61 & 0.68
& 0.09 & 0.09 & 0.09 \\ 
\textbf{100} & 81.59 & 71.57 & 71.96 & 8.11 & 8.37 & 8.17 & 0.49 & 0.60 & 
0.67 & 0.06 & 0.06 & 0.05 \\ \hline
\multicolumn{1}{|l}{} & \multicolumn{11}{l}{\textit{\textbf{Boosting}}} & 
\\ \hline
\textbf{20} & 72.22 & 65.02 & 59.44 & 3.65 & 3.71 & 3.86 & 0.55 & 0.62 & 0.69
& 0.07 & 0.06 & 0.05 \\ 
\textbf{40} & 67.34 & 60.07 & 74.17 & 5.19 & 4.89 & 4.84 & 0.54 & 0.62 & 0.69
& 0.08 & 0.06 & 0.05 \\ 
\textbf{100} & 75.45 & 64.80 & 66.42 & 11.45 & 8.63 & 7.69 & 0.54 & 0.61 & 
0.67 & 0.09 & 0.06 & 0.05 \\ \hline
\multicolumn{1}{|l}{} & \multicolumn{11}{l}{\textit{\textbf{Boosting for
variable selection only. LS for estimation/forecasting.}}} &  \\ \hline
\textbf{20} & 73.51 & 64.83 & 59.76 & 3.65 & 3.71 & 3.86 & 0.55 & 0.62 & 0.69
& 0.07 & 0.06 & 0.05 \\ 
\textbf{40} & 70.55 & 62.88 & 74.42 & 5.19 & 4.89 & 4.84 & 0.54 & 0.62 & 0.69
& 0.08 & 0.06 & 0.05 \\ 
\textbf{100} & 89.64 & 73.08 & 71.84 & 11.45 & 8.63 & 7.69 & 0.54 & 0.61 & 
0.67 & 0.09 & 0.06 & 0.05 \\ \hline\hline
\end{tabular}%
	\vspace{-0.2in}
	\begin{flushleft}
		\noindent 
		\scriptsize%
		\singlespacing%
		Notes: See notes to Table \ref{TBS}.
	\end{flushleft}
\end{table}

\begin{table}

	\caption{MC results for methods using light down-weighting in the baseline experiment with dynamics ($%
	\rho _{y} \neq 0$), and low fit.}
	\centering
 \vspace{0.2cm}
 \renewcommand{\arraystretch}{1.12}
 \setlength{\tabcolsep}{8pt}
	\scriptsize%
\begin{tabular}{|rrrrrrrrrrrrr|}
\hline\hline
& \multicolumn{3}{c}{\textbf{MSFE (}$\times 100$)} & \multicolumn{3}{|c}{$%
\hat{k}$} & \multicolumn{3}{|c}{\textbf{TPR}} & \multicolumn{3}{|c|}{\textbf{%
FPR}} \\ \hline
$N\backslash T$ & \textbf{100} & \textbf{150} & \textbf{200}
& \textbf{100} & \textbf{150} & \textbf{200} & \textbf{100} & \textbf{150} & 
\textbf{200} & \textbf{100} & \textbf{150} & \textbf{200} \\ \hline
\multicolumn{13}{|l|}{\textbf{A. Light down-weighting in the
estimation/forecasting stage only. }} \\ \hline
\multicolumn{13}{|l|}{Variable selection is based on original (not
down-weighted) data.} \\ \hline
\multicolumn{13}{|l|}{Forecasting stage is Least Squares on selected
down-weighted covariates for all methods} \\ \hline
& \multicolumn{11}{l}{\textit{\textbf{Oracle}}} &  \\ \hline
\textbf{20} & 72.07 & 63.71 & 59.72 & 4.00 & 4.00 & 4.00 & 1.00 & 1.00 & 1.00
& 0.00 & 0.00 & 0.00 \\ 
\textbf{40} & 65.77 & 59.95 & 74.18 & 4.00 & 4.00 & 4.00 & 1.00 & 1.00 & 1.00
& 0.00 & 0.00 & 0.00 \\ 
\textbf{100} & 67.21 & 63.44 & 65.55 & 4.00 & 4.00 & 4.00 & 1.00 & 1.00 & 
1.00 & 0.00 & 0.00 & 0.00 \\ \hline
& \multicolumn{11}{l}{\textit{\textbf{OCMT}}} &  \\ \hline
\textbf{20} & 73.41 & 65.83 & 59.88 & 2.49 & 3.58 & 4.57 & 0.52 & 0.71 & 0.86
& 0.02 & 0.04 & 0.06 \\ 
\textbf{40} & 66.08 & 60.20 & 75.93 & 2.11 & 3.37 & 4.26 & 0.45 & 0.69 & 0.83
& 0.01 & 0.02 & 0.02 \\ 
\textbf{100} & 65.96 & 64.78 & 65.76 & 1.74 & 2.89 & 3.72 & 0.38 & 0.61 & 
0.77 & 0.00 & 0.00 & 0.01 \\ \hline
& \multicolumn{11}{l}{\textit{\textbf{LASSO}}} &  \\ \hline
\textbf{20} & 79.09 & 68.02 & 63.50 & 5.86 & 6.17 & 6.55 & 0.65 & 0.74 & 0.80
& 0.16 & 0.16 & 0.17 \\ 
\textbf{40} & 72.99 & 67.27 & 81.07 & 7.37 & 7.82 & 8.05 & 0.61 & 0.72 & 0.79
& 0.12 & 0.12 & 0.12 \\ 
\textbf{100} & 84.04 & 76.52 & 74.54 & 10.23 & 10.36 & 9.88 & 0.57 & 0.68 & 
0.75 & 0.08 & 0.08 & 0.07 \\ \hline
& \multicolumn{11}{l}{\textit{\textbf{A-LASSO}}} &  \\ \hline
\textbf{20} & 77.97 & 68.21 & 63.69 & 4.26 & 4.51 & 4.82 & 0.52 & 0.60 & 0.67
& 0.11 & 0.11 & 0.11 \\ 
\textbf{40} & 73.29 & 64.85 & 80.55 & 5.59 & 6.00 & 6.19 & 0.51 & 0.61 & 0.68
& 0.09 & 0.09 & 0.09 \\ 
\textbf{100} & 83.30 & 75.81 & 73.69 & 8.11 & 8.37 & 8.17 & 0.49 & 0.60 & 
0.67 & 0.06 & 0.06 & 0.05 \\ \hline
& \multicolumn{11}{l}{\textit{\textbf{Boosting}}} &  \\ \hline
\textbf{20} & 75.68 & 66.22 & 61.82 & 3.65 & 3.71 & 3.86 & 0.55 & 0.62 & 0.69
& 0.07 & 0.06 & 0.05 \\ 
\textbf{40} & 72.16 & 64.47 & 78.63 & 5.19 & 4.89 & 4.84 & 0.54 & 0.62 & 0.69
& 0.08 & 0.06 & 0.05 \\ 
\textbf{100} & 91.72 & 77.53 & 73.52 & 11.45 & 8.63 & 7.69 & 0.54 & 0.61 & 
0.67 & 0.09 & 0.06 & 0.05 \\ \hline
\multicolumn{13}{|l|}{\textbf{B. Light down-weighting in both the variable
selection and estimation/forecasting stages.}} \\ \hline
\multicolumn{13}{|l|}{OCMT uses down-weighted variables for selection as well
as for forecasting using Least Squares.} \\ \hline
\multicolumn{13}{|l|}{Remaining forecasts are based on Lasso, A-Lasso and
Boosting regressions applied to down-weighted data.} \\ \hline
& \multicolumn{11}{l}{\textit{\textbf{OCMT}}} &  \\ \hline
\textbf{20} & 73.46 & 65.02 & 61.18 & 1.84 & 2.69 & 3.51 & 0.36 & 0.48 & 0.56
& 0.02 & 0.04 & 0.06 \\ 
\textbf{40} & 65.33 & 60.75 & 76.85 & 1.71 & 2.89 & 4.15 & 0.32 & 0.45 & 0.54
& 0.01 & 0.03 & 0.05 \\ 
\textbf{100} & 67.44 & 65.23 & 69.24 & 1.66 & 3.45 & 5.70 & 0.26 & 0.38 & 
0.48 & 0.01 & 0.02 & 0.04 \\ \hline
& \multicolumn{11}{l}{\textit{\textbf{LASSO}}} &  \\ \hline
\textbf{20} & 74.83 & 65.53 & 60.95 & 5.97 & 5.97 & 6.14 & 0.58 & 0.61 & 0.63
& 0.18 & 0.18 & 0.18 \\ 
\textbf{40} & 68.47 & 62.43 & 75.81 & 8.88 & 9.07 & 9.05 & 0.55 & 0.59 & 0.61
& 0.17 & 0.17 & 0.17 \\ 
\textbf{100} & 74.39 & 68.78 & 69.63 & 16.97 & 18.11 & 18.08 & 0.52 & 0.57 & 
0.59 & 0.15 & 0.16 & 0.16 \\ \hline
& \multicolumn{11}{l}{\textit{\textbf{A-LASSO}}} &  \\ \hline
\textbf{20} & 77.98 & 67.62 & 62.31 & 4.58 & 4.61 & 4.75 & 0.48 & 0.51 & 0.53
& 0.13 & 0.13 & 0.13 \\ 
\textbf{40} & 72.52 & 67.09 & 78.02 & 7.02 & 7.22 & 7.24 & 0.47 & 0.51 & 0.54
& 0.13 & 0.13 & 0.13 \\ 
\textbf{100} & 84.64 & 75.49 & 76.04 & 13.48 & 14.65 & 14.77 & 0.45 & 0.51 & 
0.54 & 0.12 & 0.13 & 0.13 \\ \hline
& \multicolumn{11}{l}{\textit{\textbf{Boosting}}} &  \\ \hline
\textbf{20} & 79.52 & 70.90 & 64.85 & 5.38 & 6.23 & 7.07 & 0.56 & 0.63 & 0.68
& 0.16 & 0.19 & 0.22 \\ 
\textbf{40} & 76.15 & 70.73 & 91.25 & 10.04 & 11.76 & 13.69 & 0.58 & 0.65 & 
0.70 & 0.19 & 0.23 & 0.27 \\ 
\textbf{100} & 94.31 & 86.72 & 88.62 & 28.88 & 33.31 & 37.76 & 0.61 & 0.69 & 
0.74 & 0.26 & 0.31 & 0.35 \\ \hline\hline
\end{tabular}%
	\vspace{-0.2in}
	\begin{flushleft}
		\noindent 
		\scriptsize%
		\singlespacing%
		Notes: Light down-weighting is defined by by values $\lambda
=0.975,0.98,0.985,0.99,0.995,1$. For this set of exponential down-weighting
schemes we focus on simple average forecasts computed over the individual
forecasts obtained for each value of $\lambda $ in the set under
consideration. See notes to Table \ref{TBS}.
	\end{flushleft}
\end{table}

\begin{table}

	\caption{MC results for methods using heavy down-weighting in the baseline experiment with dynamics ($%
	\rho _{y} \neq 0$), and low fit.}
	\centering
 \vspace{0.2cm}
 \renewcommand{\arraystretch}{1.12}
 \setlength{\tabcolsep}{8pt}
	\scriptsize%
\begin{tabular}{|rrrrrrrrrrrrr|}
\hline\hline
& \multicolumn{3}{c}{\textbf{MSFE (}$\times 100$)} & \multicolumn{3}{|c}{$%
\hat{k}$} & \multicolumn{3}{|c}{\textbf{TPR}} & \multicolumn{3}{|c|}{\textbf{%
FPR}} \\ \hline
$N\backslash T$ & \textbf{100} & \textbf{150} & \textbf{200}
& \textbf{100} & \textbf{150} & \textbf{200} & \textbf{100} & \textbf{150} & 
\textbf{200} & \textbf{100} & \textbf{150} & \textbf{200} \\ \hline
\multicolumn{13}{|l|}{\textbf{A. Heavy down-weighting in the
estimation/forecasting stage only. }} \\ \hline
\multicolumn{13}{|l|}{Variable selection is based on original (not
down-weighted) data.} \\ \hline
\multicolumn{13}{|l|}{Forecasting stage is Least Squares on selected
down-weighted covariates for all methods} \\ \hline
& \multicolumn{11}{l}{\textit{\textbf{Oracle}}} &  \\ \hline
\textbf{20} & 75.92 & 67.24 & 62.37 & 4.00 & 4.00 & 4.00 & 1.00 & 1.00 & 1.00
& 0.00 & 0.00 & 0.00 \\ 
\textbf{40} & 69.18 & 63.33 & 77.13 & 4.00 & 4.00 & 4.00 & 1.00 & 1.00 & 1.00
& 0.00 & 0.00 & 0.00 \\ 
\textbf{100} & 70.89 & 67.40 & 69.56 & 4.00 & 4.00 & 4.00 & 1.00 & 1.00 & 
1.00 & 0.00 & 0.00 & 0.00 \\ \hline
& \multicolumn{11}{l}{\textit{\textbf{OCMT}}} &  \\ \hline
\textbf{20} & 76.38 & 69.40 & 62.44 & 2.49 & 3.58 & 4.57 & 0.52 & 0.71 & 0.86
& 0.02 & 0.04 & 0.06 \\ 
\textbf{40} & 67.43 & 62.11 & 79.85 & 2.11 & 3.37 & 4.26 & 0.45 & 0.69 & 0.83
& 0.01 & 0.02 & 0.02 \\ 
\textbf{100} & 67.53 & 67.62 & 68.91 & 1.74 & 2.89 & 3.72 & 0.38 & 0.61 & 
0.77 & 0.00 & 0.00 & 0.01 \\ \hline
& \multicolumn{11}{l}{\textit{\textbf{LASSO}}} &  \\ \hline
\textbf{20} & 82.79 & 72.47 & 67.13 & 5.86 & 6.17 & 6.55 & 0.65 & 0.74 & 0.80
& 0.16 & 0.16 & 0.17 \\ 
\textbf{40} & 76.17 & 71.85 & 86.14 & 7.37 & 7.82 & 8.05 & 0.61 & 0.72 & 0.79
& 0.12 & 0.12 & 0.12 \\ 
\textbf{100} & 88.80 & 82.86 & 80.86 & 10.23 & 10.36 & 9.88 & 0.57 & 0.68 & 
0.75 & 0.08 & 0.08 & 0.07 \\ \hline
& \multicolumn{11}{l}{\textit{\textbf{A-LASSO}}} &  \\ \hline
\textbf{20} & 80.90 & 71.95 & 66.56 & 4.26 & 4.51 & 4.82 & 0.52 & 0.60 & 0.67
& 0.11 & 0.11 & 0.11 \\ 
\textbf{40} & 75.93 & 67.65 & 85.09 & 5.59 & 6.00 & 6.19 & 0.51 & 0.61 & 0.68
& 0.09 & 0.09 & 0.09 \\ 
\textbf{100} & 87.08 & 80.45 & 78.65 & 8.11 & 8.37 & 8.17 & 0.49 & 0.60 & 
0.67 & 0.06 & 0.06 & 0.05 \\ \hline
& \multicolumn{11}{l}{\textit{\textbf{Boosting}}} &  \\ \hline
\textbf{20} & 78.85 & 69.68 & 63.98 & 3.65 & 3.71 & 3.86 & 0.55 & 0.62 & 0.69
& 0.07 & 0.06 & 0.05 \\ 
\textbf{40} & 75.00 & 66.92 & 82.17 & 5.19 & 4.89 & 4.84 & 0.54 & 0.62 & 0.69
& 0.08 & 0.06 & 0.05 \\ 
\textbf{100} & 96.04 & 83.14 & 78.06 & 11.45 & 8.63 & 7.69 & 0.54 & 0.61 & 
0.67 & 0.09 & 0.06 & 0.05 \\ \hline
\multicolumn{13}{|l|}{\textbf{B. Heavy down-weighting in both the variable
selection and estimation/forecasting stages.}} \\ \hline
\multicolumn{13}{|l|}{OCMT uses down-weighted variables for selection as well
as for forecasting using Least Squares.} \\ \hline
\multicolumn{13}{|l|}{Remaining forecasts are based on Lasso, A-Lasso and
Boosting regressions applied to down-weighted data.} \\ \hline
& \multicolumn{11}{l}{\textit{\textbf{OCMT}}} &  \\ \hline
\textbf{20} & 78.41 & 71.21 & 68.06 & 2.51 & 4.05 & 5.46 & 0.35 & 0.47 & 0.56
& 0.06 & 0.11 & 0.16 \\ 
\textbf{40} & 69.92 & 69.21 & 90.54 & 3.15 & 5.93 & 8.79 & 0.31 & 0.45 & 0.55
& 0.05 & 0.10 & 0.16 \\ 
\textbf{100} & 72.73 & 77.44 & 83.26 & 4.95 & 10.98 & 18.27 & 0.26 & 0.40 & 
0.51 & 0.04 & 0.09 & 0.16 \\ \hline
& \multicolumn{11}{l}{\textit{\textbf{LASSO}}} &  \\ \hline
\textbf{20} & 79.47 & 70.28 & 64.99 & 6.59 & 6.42 & 6.62 & 0.54 & 0.55 & 0.56
& 0.22 & 0.21 & 0.22 \\ 
\textbf{40} & 76.33 & 72.52 & 84.94 & 11.99 & 12.23 & 12.29 & 0.54 & 0.57 & 
0.58 & 0.25 & 0.25 & 0.25 \\ 
\textbf{100} & 81.70 & 78.72 & 78.47 & 22.61 & 28.99 & 28.66 & 0.51 & 0.59 & 
0.60 & 0.21 & 0.27 & 0.26 \\ \hline
& \multicolumn{11}{l}{\textit{\textbf{A-LASSO}}} &  \\ \hline
\textbf{20} & 83.16 & 73.12 & 67.33 & 5.10 & 4.99 & 5.12 & 0.45 & 0.46 & 0.47
& 0.17 & 0.16 & 0.16 \\ 
\textbf{40} & 80.63 & 78.23 & 86.17 & 9.51 & 9.67 & 9.77 & 0.46 & 0.49 & 0.51
& 0.19 & 0.19 & 0.19 \\ 
\textbf{100} & 91.52 & 86.31 & 85.21 & 17.78 & 22.51 & 22.38 & 0.44 & 0.51 & 
0.52 & 0.16 & 0.20 & 0.20 \\ \hline
& \multicolumn{11}{l}{\textit{\textbf{Boosting}}} &  \\ \hline
\textbf{20} & 92.64 & 85.03 & 79.72 & 8.04 & 9.36 & 10.41 & 0.62 & 0.69 & 
0.73 & 0.28 & 0.33 & 0.37 \\ 
\textbf{40} & 95.06 & 93.72 & 116.17 & 17.98 & 20.35 & 22.26 & 0.67 & 0.73 & 
0.78 & 0.38 & 0.44 & 0.48 \\ 
\textbf{100} & 108.30 & 105.00 & 104.27 & 43.29 & 47.07 & 49.39 & 0.67 & 0.73
& 0.75 & 0.41 & 0.44 & 0.46 \\ \hline\hline
\end{tabular}%
	\vspace{-0.2in}
	\begin{flushleft}
		\noindent 
		\scriptsize%
		\singlespacing%
		Notes: Heavy down-weighting is defined by by values $\lambda
=0.95,0.96,0.97,0.98,0.99,1$. For this set of exponential down-weighting
schemes we focus on simple average forecasts computed over the individual
forecasts obtained for each value of $\lambda $ in the set under
consideration. See notes to Table \ref{TBS}.
	\end{flushleft}
\end{table}

\begin{table}

	\vspace{0.1cm}
	\caption{MC results for methods using no down-weighting in the baseline experiment with dynamics ($%
	\rho _{y} \neq 0$) and high fit.}
	\centering
 \vspace{0.2cm}
 \renewcommand{\arraystretch}{1.12}
	\scriptsize%
\begin{tabular}{|rrrrrrrrrrrrr|}
\hline\hline
& \multicolumn{3}{c}{\textbf{MSFE (}$\times 100$)} & \multicolumn{3}{|c}{$%
\hat{k}$} & \multicolumn{3}{|c}{\textbf{TPR}} & \multicolumn{3}{|c|}{\textbf{%
FPR}} \\ \hline
$N\backslash T$ & \textbf{100} & \textbf{150} & \textbf{200}
& \textbf{100} & \textbf{150} & \textbf{200} & \textbf{100} & \textbf{150} & 
\textbf{200} & \textbf{100} & \textbf{150} & \textbf{200} \\ \hline
\multicolumn{1}{|l}{} & \multicolumn{11}{l}{\textit{\textbf{Oracle}}} & 
\multicolumn{1}{l|}{} \\ \hline
\textbf{20} & 20.51 & 18.61 & 17.18 & 4.00 & 4.00 & 4.00 & 1.00 & 1.00 & 1.00
& 0.00 & 0.00 & 0.00 \\ 
\textbf{40} & 18.89 & 17.33 & 21.16 & 4.00 & 4.00 & 4.00 & 1.00 & 1.00 & 1.00
& 0.00 & 0.00 & 0.00 \\ 
\textbf{100} & 19.24 & 18.17 & 18.95 & 4.00 & 4.00 & 4.00 & 1.00 & 1.00 & 
1.00 & 0.00 & 0.00 & 0.00 \\ \hline
\multicolumn{1}{|l}{} & \multicolumn{11}{l}{\textit{\textbf{OCMT}}} &  \\ 
\hline
\textbf{20} & 21.46 & 18.93 & 17.47 & 5.07 & 6.11 & 7.12 & 0.92 & 0.97 & 0.99
& 0.07 & 0.11 & 0.16 \\ 
\textbf{40} & 19.16 & 17.75 & 21.75 & 4.73 & 5.90 & 6.73 & 0.90 & 0.98 & 0.99
& 0.03 & 0.05 & 0.07 \\ 
\textbf{100} & 19.48 & 18.53 & 19.28 & 4.31 & 5.43 & 6.21 & 0.86 & 0.96 & 
0.99 & 0.01 & 0.02 & 0.02 \\ \hline
\multicolumn{1}{|l}{} & \textit{\textbf{LASSO}} &  &  &  &  &  &  &  &  &  & 
&  \\ \hline
\textbf{20} & 21.67 & 19.35 & 17.57 & 7.18 & 7.24 & 7.48 & 0.88 & 0.93 & 0.96
& 0.18 & 0.18 & 0.18 \\ 
\textbf{40} & 20.05 & 17.76 & 21.70 & 8.89 & 9.09 & 9.09 & 0.86 & 0.93 & 0.96
& 0.14 & 0.13 & 0.13 \\ 
\textbf{100} & 21.00 & 19.07 & 19.61 & 11.89 & 11.79 & 11.16 & 0.83 & 0.92 & 
0.95 & 0.09 & 0.08 & 0.07 \\ \hline
\multicolumn{1}{|l}{} & \multicolumn{11}{l}{\textit{\textbf{LASSO for
variable selection only. LS for estimation/forecasting.}}} &  \\ \hline
\textbf{20} & 23.15 & 19.92 & 18.22 & 7.18 & 7.24 & 7.48 & 0.88 & 0.93 & 0.96
& 0.18 & 0.18 & 0.18 \\ 
\textbf{40} & 22.09 & 19.43 & 22.84 & 8.89 & 9.09 & 9.09 & 0.86 & 0.93 & 0.96
& 0.14 & 0.13 & 0.13 \\ 
\textbf{100} & 25.04 & 21.43 & 21.70 & 11.89 & 11.79 & 11.16 & 0.83 & 0.92 & 
0.95 & 0.09 & 0.08 & 0.07 \\ \hline
\multicolumn{1}{|l}{} & \multicolumn{11}{l}{\textit{\textbf{A-LASSO}}} &  \\ 
\hline
\textbf{20} & 22.55 & 19.91 & 18.09 & 5.39 & 5.60 & 5.87 & 0.77 & 0.84 & 0.90
& 0.12 & 0.11 & 0.11 \\ 
\textbf{40} & 21.55 & 19.09 & 22.49 & 6.86 & 7.25 & 7.36 & 0.77 & 0.86 & 0.91
& 0.09 & 0.10 & 0.09 \\ 
\textbf{100} & 24.39 & 20.79 & 21.43 & 9.59 & 9.73 & 9.46 & 0.76 & 0.86 & 
0.91 & 0.07 & 0.06 & 0.06 \\ \hline
& \multicolumn{11}{l}{\textit{\textbf{A-LASSO for variable selection only.
LS for estimation/forecasting.}}} &  \\ \hline
\textbf{20} & 23.09 & 20.05 & 18.25 & 5.39 & 5.60 & 5.87 & 0.77 & 0.84 & 0.90
& 0.12 & 0.11 & 0.11 \\ 
\textbf{40} & 22.24 & 19.45 & 22.81 & 6.86 & 7.25 & 7.36 & 0.77 & 0.86 & 0.91
& 0.09 & 0.10 & 0.09 \\ 
\textbf{100} & 25.35 & 21.27 & 21.84 & 9.59 & 9.73 & 9.46 & 0.76 & 0.86 & 
0.91 & 0.07 & 0.06 & 0.06 \\ \hline
\multicolumn{1}{|l}{} & \multicolumn{11}{l}{\textit{\textbf{Boosting}}} & 
\\ \hline
\textbf{20} & 23.41 & 21.17 & 18.64 & 4.79 & 4.85 & 4.93 & 0.81 & 0.88 & 0.93
& 0.08 & 0.07 & 0.06 \\ 
\textbf{40} & 22.11 & 19.11 & 23.38 & 6.37 & 6.03 & 5.94 & 0.81 & 0.88 & 0.93
& 0.08 & 0.06 & 0.06 \\ 
\textbf{100} & 24.38 & 20.46 & 20.99 & 12.21 & 9.73 & 8.78 & 0.79 & 0.87 & 
0.91 & 0.09 & 0.06 & 0.05 \\ \hline
\multicolumn{1}{|l}{} & \multicolumn{11}{l}{\textit{\textbf{Boosting for
variable selection only. LS for estimation/forecasting.}}} &  \\ \hline
\textbf{20} & 22.29 & 19.45 & 17.77 & 4.79 & 4.85 & 4.93 & 0.81 & 0.88 & 0.93
& 0.08 & 0.07 & 0.06 \\ 
\textbf{40} & 21.55 & 18.78 & 22.26 & 6.37 & 6.03 & 5.94 & 0.81 & 0.88 & 0.93
& 0.08 & 0.06 & 0.06 \\ 
\textbf{100} & 26.87 & 21.52 & 21.66 & 12.21 & 9.73 & 8.78 & 0.79 & 0.87 & 
0.91 & 0.09 & 0.06 & 0.05 \\ \hline\hline
\end{tabular}%
	\vspace{-0.2in}
	\begin{flushleft}
		\noindent 
		\scriptsize%
		\singlespacing%
		Notes: See notes to Table \ref{TBS}.
	\end{flushleft}
\end{table}

\begin{table}

	\caption{MC results for methods using light down-weighting in the baseline experiment with dynamics ($%
	\rho _{y} \neq 0$), and high fit.}
	\centering
 \vspace{0.2cm}
 \renewcommand{\arraystretch}{1.12}
 \setlength{\tabcolsep}{8pt}
	\scriptsize%
\begin{tabular}{|rrrrrrrrrrrrr|}
\hline\hline
& \multicolumn{3}{c}{\textbf{MSFE (}$\times 100$)} & \multicolumn{3}{|c}{$%
\hat{k}$} & \multicolumn{3}{|c}{\textbf{TPR}} & \multicolumn{3}{|c|}{\textbf{%
FPR}} \\ \hline
$N\backslash T$ & \textbf{100} & \textbf{150} & \textbf{200}
& \textbf{100} & \textbf{150} & \textbf{200} & \textbf{100} & \textbf{150} & 
\textbf{200} & \textbf{100} & \textbf{150} & \textbf{200} \\ \hline
\multicolumn{13}{|l|}{\textbf{A. Light down-weighting in the
estimation/forecasting stage only. }} \\ \hline
\multicolumn{13}{|l|}{Variable selection is based on original (not
down-weighted) data.} \\ \hline
\multicolumn{13}{|l|}{Forecasting stage is Least Squares on selected
down-weighted covariates for all methods} \\ \hline
& \multicolumn{11}{l}{\textit{\textbf{Oracle}}} &  \\ \hline
\textbf{20} & 21.30 & 18.94 & 17.68 & 4.00 & 4.00 & 4.00 & 1.00 & 1.00 & 1.00
& 0.00 & 0.00 & 0.00 \\ 
\textbf{40} & 19.46 & 17.78 & 21.97 & 4.00 & 4.00 & 4.00 & 1.00 & 1.00 & 1.00
& 0.00 & 0.00 & 0.00 \\ 
\textbf{100} & 19.82 & 18.78 & 19.37 & 4.00 & 4.00 & 4.00 & 1.00 & 1.00 & 
1.00 & 0.00 & 0.00 & 0.00 \\ \hline
& \multicolumn{11}{l}{\textit{\textbf{OCMT}}} &  \\ \hline
\textbf{20} & 22.24 & 19.36 & 18.40 & 5.07 & 6.11 & 7.12 & 0.92 & 0.97 & 0.99
& 0.07 & 0.11 & 0.16 \\ 
\textbf{40} & 19.61 & 18.40 & 23.20 & 4.73 & 5.90 & 6.73 & 0.90 & 0.98 & 0.99
& 0.03 & 0.05 & 0.07 \\ 
\textbf{100} & 19.73 & 19.43 & 20.00 & 4.31 & 5.43 & 6.21 & 0.86 & 0.96 & 
0.99 & 0.01 & 0.02 & 0.02 \\ \hline
& \multicolumn{11}{l}{\textit{\textbf{LASSO}}} &  \\ \hline
\textbf{20} & 24.16 & 20.75 & 19.11 & 7.18 & 7.24 & 7.48 & 0.88 & 0.93 & 0.96
& 0.18 & 0.18 & 0.18 \\ 
\textbf{40} & 22.69 & 20.21 & 24.37 & 8.89 & 9.09 & 9.09 & 0.86 & 0.93 & 0.96
& 0.14 & 0.13 & 0.13 \\ 
\textbf{100} & 25.80 & 23.17 & 22.61 & 11.89 & 11.79 & 11.16 & 0.83 & 0.92 & 
0.95 & 0.09 & 0.08 & 0.07 \\ \hline
& \multicolumn{11}{l}{\textit{\textbf{A-LASSO}}} &  \\ \hline
\textbf{20} & 23.77 & 20.66 & 19.03 & 5.39 & 5.60 & 5.87 & 0.77 & 0.84 & 0.90
& 0.12 & 0.11 & 0.11 \\ 
\textbf{40} & 22.71 & 20.05 & 24.26 & 6.86 & 7.25 & 7.36 & 0.77 & 0.86 & 0.91
& 0.09 & 0.10 & 0.09 \\ 
\textbf{100} & 25.87 & 22.66 & 22.72 & 9.59 & 9.73 & 9.46 & 0.76 & 0.86 & 
0.91 & 0.07 & 0.06 & 0.06 \\ \hline
& \multicolumn{11}{l}{\textit{\textbf{Boosting}}} &  \\ \hline
\textbf{20} & 23.13 & 20.07 & 18.57 & 4.79 & 4.85 & 4.93 & 0.81 & 0.88 & 0.93
& 0.08 & 0.07 & 0.06 \\ 
\textbf{40} & 22.01 & 19.35 & 23.72 & 6.37 & 6.03 & 5.94 & 0.81 & 0.88 & 0.93
& 0.08 & 0.06 & 0.06 \\ 
\textbf{100} & 27.93 & 22.95 & 22.21 & 12.21 & 9.73 & 8.78 & 0.79 & 0.87 & 
0.91 & 0.09 & 0.06 & 0.05 \\ \hline
\multicolumn{13}{|l|}{\textbf{B. Light down-weighting in both the variable
selection and estimation/forecasting stages.}} \\ \hline
\multicolumn{13}{|l|}{OCMT uses down-weighted variables for selection as well
as for forecasting using Least Squares.} \\ \hline
\multicolumn{13}{|l|}{Remaining forecasts are based on Lasso, A-Lasso and
Boosting regressions applied to down-weighted data.} \\ \hline
& \multicolumn{11}{l}{\textit{\textbf{OCMT}}} &  \\ \hline
\textbf{20} & 22.19 & 19.61 & 18.26 & 3.65 & 4.33 & 5.08 & 0.72 & 0.77 & 0.81
& 0.04 & 0.06 & 0.09 \\ 
\textbf{40} & 19.91 & 18.15 & 22.99 & 3.52 & 4.59 & 5.77 & 0.69 & 0.75 & 0.80
& 0.02 & 0.04 & 0.06 \\ 
\textbf{100} & 20.36 & 19.51 & 20.74 & 3.40 & 5.17 & 7.35 & 0.63 & 0.71 & 
0.76 & 0.01 & 0.02 & 0.04 \\ \hline
& \multicolumn{11}{l}{\textit{\textbf{LASSO}}} &  \\ \hline
\textbf{20} & 23.03 & 20.14 & 18.42 & 7.61 & 7.67 & 7.92 & 0.83 & 0.85 & 0.87
& 0.22 & 0.21 & 0.22 \\ 
\textbf{40} & 21.19 & 19.17 & 23.25 & 10.90 & 11.18 & 11.17 & 0.81 & 0.85 & 
0.86 & 0.19 & 0.19 & 0.19 \\ 
\textbf{100} & 22.81 & 21.30 & 21.56 & 18.84 & 20.65 & 21.04 & 0.78 & 0.84 & 
0.85 & 0.16 & 0.17 & 0.18 \\ \hline
& \multicolumn{11}{l}{\textit{\textbf{A-LASSO}}} &  \\ \hline
\textbf{20} & 23.86 & 20.60 & 18.77 & 5.90 & 6.04 & 6.29 & 0.73 & 0.77 & 0.80
& 0.15 & 0.15 & 0.15 \\ 
\textbf{40} & 22.51 & 20.63 & 24.14 & 8.72 & 9.02 & 9.11 & 0.73 & 0.78 & 0.81
& 0.14 & 0.15 & 0.15 \\ 
\textbf{100} & 26.24 & 23.62 & 23.70 & 15.13 & 16.87 & 17.33 & 0.72 & 0.79 & 
0.81 & 0.12 & 0.14 & 0.14 \\ \hline
& \multicolumn{11}{l}{\textit{\textbf{Boosting}}} &  \\ \hline
\textbf{20} & 25.67 & 23.16 & 20.77 & 6.43 & 7.25 & 8.02 & 0.80 & 0.85 & 0.89
& 0.16 & 0.19 & 0.22 \\ 
\textbf{40} & 24.98 & 22.87 & 28.94 & 11.07 & 12.74 & 14.61 & 0.81 & 0.87 & 
0.90 & 0.20 & 0.23 & 0.28 \\ 
\textbf{100} & 30.06 & 28.00 & 28.48 & 29.72 & 34.12 & 38.43 & 0.83 & 0.89 & 
0.91 & 0.26 & 0.31 & 0.35 \\ \hline\hline
\end{tabular}%
	\vspace{-0.2in}
	\begin{flushleft}
		\noindent 
		\scriptsize%
		\singlespacing%
		Notes: Light down-weighting is defined by by values $\lambda
=0.975,0.98,0.985,0.99,0.995,1$. For this set of exponential down-weighting
schemes we focus on simple average forecasts computed over the individual
forecasts obtained for each value of $\lambda $ in the set under
consideration. See notes to Table \ref{TBS}.
	\end{flushleft}
\end{table}

\begin{table}

	\caption{MC results for methods using heavy down-weighting in the baseline experiment with dynamics ($%
	\rho _{y} \neq 0$), and high fit.}
	\centering
 \vspace{0.2cm}
 \renewcommand{\arraystretch}{1.12}
 \setlength{\tabcolsep}{8pt}
	\scriptsize%
\begin{tabular}{|rrrrrrrrrrrrr|}
\hline\hline
& \multicolumn{3}{c}{\textbf{MSFE (}$\times 100$)} & \multicolumn{3}{|c}{$%
\hat{k}$} & \multicolumn{3}{|c}{\textbf{TPR}} & \multicolumn{3}{|c|}{\textbf{%
FPR}} \\ \hline
$N\backslash T$ & \textbf{100} & \textbf{150} & \textbf{200}
& \textbf{100} & \textbf{150} & \textbf{200} & \textbf{100} & \textbf{150} & 
\textbf{200} & \textbf{100} & \textbf{150} & \textbf{200} \\ \hline
\multicolumn{13}{|l|}{\textbf{A. Heavy down-weighting in the
estimation/forecasting stage only. }} \\ \hline
\multicolumn{13}{|l|}{Variable selection is based on original (not
down-weighted) data.} \\ \hline
\multicolumn{13}{|l|}{Forecasting stage is Least Squares on selected
down-weighted covariates for all methods} \\ \hline
& \multicolumn{11}{l}{\textit{\textbf{Oracle}}} &  \\ \hline
\textbf{20} & 22.47 & 19.96 & 18.48 & 4.00 & 4.00 & 4.00 & 1.00 & 1.00 & 1.00
& 0.00 & 0.00 & 0.00 \\ 
\textbf{40} & 20.47 & 18.82 & 22.89 & 4.00 & 4.00 & 4.00 & 1.00 & 1.00 & 1.00
& 0.00 & 0.00 & 0.00 \\ 
\textbf{100} & 20.90 & 19.92 & 20.53 & 4.00 & 4.00 & 4.00 & 1.00 & 1.00 & 
1.00 & 0.00 & 0.00 & 0.00 \\ \hline
& \multicolumn{11}{l}{\textit{\textbf{OCMT}}} &  \\ \hline
\textbf{20} & 23.57 & 20.58 & 19.77 & 5.07 & 6.11 & 7.12 & 0.92 & 0.97 & 0.99
& 0.07 & 0.11 & 0.16 \\ 
\textbf{40} & 20.51 & 19.61 & 25.29 & 4.73 & 5.90 & 6.73 & 0.90 & 0.98 & 0.99
& 0.03 & 0.05 & 0.07 \\ 
\textbf{100} & 20.49 & 21.04 & 21.68 & 4.31 & 5.43 & 6.21 & 0.86 & 0.96 & 
0.99 & 0.01 & 0.02 & 0.02 \\ \hline
& \multicolumn{11}{l}{\textit{\textbf{LASSO}}} &  \\ \hline
\textbf{20} & 25.55 & 22.37 & 20.44 & 7.18 & 7.24 & 7.48 & 0.88 & 0.93 & 0.96
& 0.18 & 0.18 & 0.18 \\ 
\textbf{40} & 24.11 & 21.71 & 26.07 & 8.89 & 9.09 & 9.09 & 0.86 & 0.93 & 0.96
& 0.14 & 0.13 & 0.13 \\ 
\textbf{100} & 27.58 & 25.23 & 24.80 & 11.89 & 11.79 & 11.16 & 0.83 & 0.92 & 
0.95 & 0.09 & 0.08 & 0.07 \\ \hline
& \multicolumn{11}{l}{\textit{\textbf{A-LASSO}}} &  \\ \hline
\textbf{20} & 24.86 & 22.00 & 20.11 & 5.39 & 5.60 & 5.87 & 0.77 & 0.84 & 0.90
& 0.12 & 0.11 & 0.11 \\ 
\textbf{40} & 23.78 & 21.23 & 25.86 & 6.86 & 7.25 & 7.36 & 0.77 & 0.86 & 0.91
& 0.09 & 0.10 & 0.09 \\ 
\textbf{100} & 27.08 & 24.37 & 24.49 & 9.59 & 9.73 & 9.46 & 0.76 & 0.86 & 
0.91 & 0.07 & 0.06 & 0.06 \\ \hline
& \multicolumn{11}{l}{\textit{\textbf{Boosting}}} &  \\ \hline
\textbf{20} & 24.35 & 21.32 & 19.57 & 4.79 & 4.85 & 4.93 & 0.81 & 0.88 & 0.93
& 0.08 & 0.07 & 0.06 \\ 
\textbf{40} & 23.01 & 20.22 & 25.14 & 6.37 & 6.03 & 5.94 & 0.81 & 0.88 & 0.93
& 0.08 & 0.06 & 0.06 \\ 
\textbf{100} & 29.65 & 24.69 & 23.65 & 12.21 & 9.73 & 8.78 & 0.79 & 0.87 & 
0.91 & 0.09 & 0.06 & 0.05 \\ \hline
\multicolumn{13}{|l|}{\textbf{B. Heavy down-weighting in both the variable
selection and estimation/forecasting stages.}} \\ \hline
\multicolumn{13}{|l|}{OCMT uses down-weighted variables for selection as well
as for forecasting using Least Squares.} \\ \hline
\multicolumn{13}{|l|}{Remaining forecasts are based on Lasso, A-Lasso and
Boosting regressions applied to down-weighted data.} \\ \hline
& \multicolumn{11}{l}{\textit{\textbf{OCMT}}} &  \\ \hline
\textbf{20} & 24.00 & 21.87 & 20.31 & 3.92 & 5.37 & 6.72 & 0.62 & 0.69 & 0.75
& 0.07 & 0.13 & 0.19 \\ 
\textbf{40} & 21.75 & 21.18 & 27.64 & 4.58 & 7.31 & 10.14 & 0.59 & 0.68 & 
0.75 & 0.06 & 0.11 & 0.18 \\ 
\textbf{100} & 22.91 & 23.74 & 25.67 & 6.25 & 12.37 & 19.62 & 0.54 & 0.65 & 
0.73 & 0.04 & 0.10 & 0.17 \\ \hline
& \multicolumn{11}{l}{\textit{\textbf{LASSO}}} &  \\ \hline
\textbf{20} & 24.70 & 21.84 & 19.90 & 8.25 & 8.17 & 8.38 & 0.77 & 0.78 & 0.79
& 0.26 & 0.25 & 0.26 \\ 
\textbf{40} & 23.65 & 22.59 & 26.43 & 13.93 & 14.29 & 14.30 & 0.78 & 0.80 & 
0.80 & 0.27 & 0.28 & 0.28 \\ 
\textbf{100} & 25.34 & 24.24 & 24.13 & 24.04 & 31.00 & 30.99 & 0.75 & 0.80 & 
0.82 & 0.21 & 0.28 & 0.28 \\ \hline
& \multicolumn{11}{l}{\textit{\textbf{A-LASSO}}} &  \\ \hline
\textbf{20} & 25.68 & 22.53 & 20.54 & 6.40 & 6.39 & 6.58 & 0.67 & 0.70 & 0.72
& 0.18 & 0.18 & 0.19 \\ 
\textbf{40} & 25.13 & 24.15 & 27.13 & 11.15 & 11.41 & 11.48 & 0.69 & 0.72 & 
0.74 & 0.21 & 0.21 & 0.21 \\ 
\textbf{100} & 28.70 & 26.75 & 26.17 & 19.00 & 24.14 & 24.29 & 0.68 & 0.74 & 
0.76 & 0.16 & 0.21 & 0.21 \\ \hline
& \multicolumn{11}{l}{\textit{\textbf{Boosting}}} &  \\ \hline
\textbf{20} & 30.02 & 27.95 & 25.67 & 8.87 & 10.13 & 11.08 & 0.80 & 0.86 & 
0.89 & 0.28 & 0.34 & 0.38 \\ 
\textbf{40} & 30.78 & 30.28 & 36.93 & 18.71 & 21.02 & 22.86 & 0.84 & 0.89 & 
0.92 & 0.38 & 0.44 & 0.48 \\ 
\textbf{100} & 34.53 & 34.04 & 33.58 & 43.88 & 47.60 & 49.86 & 0.84 & 0.88 & 
0.90 & 0.41 & 0.44 & 0.46 \\ \hline\hline
\end{tabular}%
	\vspace{-0.2in}
	\begin{flushleft}
		\noindent 
		\scriptsize%
		\singlespacing%
		Notes: Heavy down-weighting is defined by by values $\lambda
=0.95,0.96,0.97,0.98,0.99,1$. For this set of exponential down-weighting
schemes we focus on simple average forecasts computed over the individual
forecasts obtained for each value of $\lambda $ in the set under
consideration. See notes to Table \ref{TBS}.
	\end{flushleft}
\end{table}

\clearpage

\subsection{MC Findings for experiments with parameter instabilities}

\begin{table}[h]

	\vspace{0.1cm}
	\caption{\label{TNBS}MC results for methods using no down-weighting in the experiment with parameter instabilities, no dynamics ($%
	\rho _{y}=0$) and low fit.}
	\centering
 \vspace{0.2cm}
 \renewcommand{\arraystretch}{1.12}
	\scriptsize%
\begin{tabular}{|rrrrrrrrrrrrr|}
\hline\hline
& \multicolumn{3}{c}{\textbf{MSFE (}$\times 100$)} & \multicolumn{3}{|c}{$%
\hat{k}$} & \multicolumn{3}{|c}{\textbf{TPR}} & \multicolumn{3}{|c|}{\textbf{%
FPR}} \\ \hline
$N\backslash T$ & \textbf{100} & \textbf{150} & \textbf{200}
& \textbf{100} & \textbf{150} & \textbf{200} & \textbf{100} & \textbf{150} & 
\textbf{200} & \textbf{100} & \textbf{150} & \textbf{200} \\ \hline
\multicolumn{1}{|l}{} & \multicolumn{11}{l}{\textit{\textbf{Oracle}}} & 
\multicolumn{1}{l|}{} \\ \hline
\textbf{20} & 28.87 & 25.69 & 24.13 & 4.00 & 4.00 & 4.00 & 1.00 & 1.00 & 1.00
& 0.00 & 0.00 & 0.00 \\ 
\textbf{40} & 26.29 & 24.34 & 28.05 & 4.00 & 4.00 & 4.00 & 1.00 & 1.00 & 1.00
& 0.00 & 0.00 & 0.00 \\ 
\textbf{100} & 27.06 & 25.75 & 26.63 & 4.00 & 4.00 & 4.00 & 1.00 & 1.00 & 
1.00 & 0.00 & 0.00 & 0.00 \\ \hline
\multicolumn{1}{|l}{} & \multicolumn{11}{l}{\textit{\textbf{OCMT}}} &  \\ 
\hline
\textbf{20} & 29.25 & 26.04 & 24.51 & 4.42 & 5.51 & 6.39 & 0.81 & 0.93 & 0.97
& 0.06 & 0.09 & 0.12 \\ 
\textbf{40} & 27.08 & 25.16 & 28.35 & 4.13 & 5.34 & 6.11 & 0.78 & 0.92 & 0.97
& 0.03 & 0.04 & 0.06 \\ 
\textbf{100} & 27.69 & 26.10 & 27.31 & 3.89 & 5.07 & 5.68 & 0.73 & 0.90 & 
0.95 & 0.01 & 0.01 & 0.02 \\ \hline
\multicolumn{1}{|l}{} & \textit{\textbf{LASSO}} &  &  &  &  &  &  &  &  &  & 
&  \\ \hline
\textbf{20} & 30.31 & 26.73 & 24.93 & 7.04 & 7.47 & 7.91 & 0.78 & 0.85 & 0.89
& 0.20 & 0.20 & 0.22 \\ 
\textbf{40} & 27.91 & 25.57 & 28.69 & 8.99 & 9.90 & 10.38 & 0.75 & 0.84 & 
0.89 & 0.15 & 0.16 & 0.17 \\ 
\textbf{100} & 29.03 & 27.52 & 28.46 & 11.90 & 13.27 & 13.63 & 0.71 & 0.80 & 
0.86 & 0.09 & 0.10 & 0.10 \\ \hline
\multicolumn{1}{|l}{} & \multicolumn{11}{l}{\textit{\textbf{LASSO for
variable selection only. LS for estimation/forecasting.}}} &  \\ \hline
\textbf{20} & 32.21 & 28.09 & 25.63 & 7.04 & 7.47 & 7.91 & 0.78 & 0.85 & 0.89
& 0.20 & 0.20 & 0.22 \\ 
\textbf{40} & 30.04 & 27.62 & 30.41 & 8.99 & 9.90 & 10.38 & 0.75 & 0.84 & 
0.89 & 0.15 & 0.16 & 0.17 \\ 
\textbf{100} & 33.67 & 30.47 & 31.39 & 11.90 & 13.27 & 13.63 & 0.71 & 0.80 & 
0.86 & 0.09 & 0.10 & 0.10 \\ \hline
\multicolumn{1}{|l}{} & \multicolumn{11}{l}{\textit{\textbf{A-LASSO}}} &  \\ 
\hline
\textbf{20} & 31.61 & 27.63 & 25.38 & 5.28 & 5.67 & 6.07 & 0.65 & 0.73 & 0.80
& 0.13 & 0.14 & 0.14 \\ 
\textbf{40} & 29.30 & 27.00 & 29.88 & 6.89 & 7.67 & 8.11 & 0.65 & 0.75 & 0.81
& 0.11 & 0.12 & 0.12 \\ 
\textbf{100} & 32.45 & 29.75 & 30.64 & 9.27 & 10.63 & 11.02 & 0.63 & 0.74 & 
0.80 & 0.07 & 0.08 & 0.08 \\ \hline
& \multicolumn{11}{l}{\textit{\textbf{A-LASSO for variable selection only.
LS for estimation/forecasting.}}} &  \\ \hline
\textbf{20} & 32.13 & 28.02 & 25.75 & 5.28 & 5.67 & 6.07 & 0.65 & 0.73 & 0.80
& 0.13 & 0.14 & 0.14 \\ 
\textbf{40} & 30.08 & 27.66 & 30.48 & 6.89 & 7.67 & 8.11 & 0.65 & 0.75 & 0.81
& 0.11 & 0.12 & 0.12 \\ 
\textbf{100} & 33.64 & 30.53 & 31.26 & 9.27 & 10.63 & 11.02 & 0.63 & 0.74 & 
0.80 & 0.07 & 0.08 & 0.08 \\ \hline
\multicolumn{1}{|l}{} & \multicolumn{11}{l}{\textit{\textbf{Boosting}}} & 
\\ \hline
\textbf{20} & 30.49 & 26.72 & 24.84 & 4.46 & 4.67 & 4.76 & 0.69 & 0.76 & 0.81
& 0.08 & 0.08 & 0.08 \\ 
\textbf{40} & 28.15 & 25.28 & 28.83 & 6.20 & 6.15 & 6.07 & 0.69 & 0.77 & 0.82
& 0.09 & 0.08 & 0.07 \\ 
\textbf{100} & 28.99 & 27.22 & 27.94 & 11.25 & 10.05 & 9.50 & 0.67 & 0.75 & 
0.80 & 0.09 & 0.07 & 0.06 \\ \hline
\multicolumn{1}{|l}{} & \multicolumn{11}{l}{\textit{\textbf{Boosting for
variable selection only. LS for estimation/forecasting.}}} &  \\ \hline
\textbf{20} & 31.45 & 27.05 & 24.96 & 4.46 & 4.67 & 4.76 & 0.69 & 0.76 & 0.81
& 0.08 & 0.08 & 0.08 \\ 
\textbf{40} & 30.04 & 26.83 & 29.50 & 6.20 & 6.15 & 6.07 & 0.69 & 0.77 & 0.82
& 0.09 & 0.08 & 0.07 \\ 
\textbf{100} & 34.66 & 30.62 & 30.65 & 11.25 & 10.05 & 9.50 & 0.67 & 0.75 & 
0.80 & 0.09 & 0.07 & 0.06 \\ \hline\hline
\end{tabular}%
	\vspace{-0.2in}
	\begin{flushleft}
		\noindent 
		\scriptsize%
		\singlespacing%
		Notes: This table reports one-step-ahead Mean Square Forecast Error (MSFE, $%
		\times 100$), average number of selected variables ($\hat{k}$), True
		Positive Rate (TPR), and False Positive Rate (FPR). There are $k=4$ signals
		variables out of $N$ observed variables. The DGP is given by $%
		y_{t}=d_{t}+\rho
		_{y,t}y_{t-1}$ + $\sum_{j=1}^{4}\beta _{jt}x_{jt}$ + $\tau _{u}u_{t}$, where slopes $%
		\beta _{jt}=b_{jt}+\tau _{\eta _{j}}\eta _{jt}$ feature stochastic AR(1) component $\eta _{jt} $ and parameter instabilities in mean $b_{jt}$ given by (\ref{bm1})-(\ref{bm3}), intercepts
		are given by $d_{t}=\sum_{j=1}^{k}\beta _{jt}\mu _{jt}$ where parameter instabilities in $\mu _{jt}$ is given by (\ref{mu1})-(\ref{mu3}), and $\rho _{y,t}$ is zero in experiments without dynamics, and
		given by (\ref{rhoyt}) in experiments with dynamics. $u_{t}$ is given by a GARCH(1,1). See Section \ref{sec:MC-studies} of the paper for the detailed description of the Monte Carlo design. The reported results are based on 2000 simulations. Oracle model assumes the
		identity of signal variables is known.
	\end{flushleft}
\end{table}

\begin{table}

	\caption{MC results for methods using light down-weighting in the experiment with parameter instabilities, no dynamics ($%
	\rho _{y}=0$), and low fit.}
	\centering
 \vspace{0.2cm}
 \renewcommand{\arraystretch}{1.12}
 \setlength{\tabcolsep}{8pt}
	\scriptsize%
\begin{tabular}{|rrrrrrrrrrrrr|}
\hline\hline
& \multicolumn{3}{c}{\textbf{MSFE (}$\times 100$)} & \multicolumn{3}{|c}{$%
\hat{k}$} & \multicolumn{3}{|c}{\textbf{TPR}} & \multicolumn{3}{|c|}{\textbf{%
FPR}} \\ \hline
$N\backslash T$ & \textbf{100} & \textbf{150} & \textbf{200}
& \textbf{100} & \textbf{150} & \textbf{200} & \textbf{100} & \textbf{150} & 
\textbf{200} & \textbf{100} & \textbf{150} & \textbf{200} \\ \hline
\multicolumn{13}{|l|}{\textbf{A. Light down-weighting in the
estimation/forecasting stage only. }} \\ \hline
\multicolumn{13}{|l|}{Variable selection is based on original (not
down-weighted) data.} \\ \hline
\multicolumn{13}{|l|}{Forecasting stage is Least Squares on selected
down-weighted covariates for all methods} \\ \hline
& \multicolumn{11}{l}{\textit{\textbf{Oracle}}} &  \\ \hline
\textbf{20} & 28.11 & 24.94 & 22.95 & 4.00 & 4.00 & 4.00 & 1.00 & 1.00 & 1.00
& 0.00 & 0.00 & 0.00 \\ 
\textbf{40} & 25.99 & 23.34 & 27.74 & 4.00 & 4.00 & 4.00 & 1.00 & 1.00 & 1.00
& 0.00 & 0.00 & 0.00 \\ 
\textbf{100} & 26.55 & 24.96 & 25.34 & 4.00 & 4.00 & 4.00 & 1.00 & 1.00 & 
1.00 & 0.00 & 0.00 & 0.00 \\ \hline
& \multicolumn{11}{l}{\textit{\textbf{OCMT}}} &  \\ \hline
\textbf{20} & 28.67 & 25.57 & 23.76 & 4.42 & 5.51 & 6.39 & 0.81 & 0.93 & 0.97
& 0.06 & 0.09 & 0.12 \\ 
\textbf{40} & 26.72 & 24.36 & 28.71 & 4.13 & 5.34 & 6.11 & 0.78 & 0.92 & 0.97
& 0.03 & 0.04 & 0.06 \\ 
\textbf{100} & 27.33 & 25.69 & 26.22 & 3.89 & 5.07 & 5.68 & 0.73 & 0.90 & 
0.95 & 0.01 & 0.01 & 0.02 \\ \hline
& \multicolumn{11}{l}{\textit{\textbf{LASSO}}} &  \\ \hline
\textbf{20} & 32.00 & 28.08 & 25.13 & 7.04 & 7.47 & 7.91 & 0.78 & 0.85 & 0.89
& 0.20 & 0.20 & 0.22 \\ 
\textbf{40} & 30.08 & 27.28 & 31.95 & 8.99 & 9.90 & 10.38 & 0.75 & 0.84 & 
0.89 & 0.15 & 0.16 & 0.17 \\ 
\textbf{100} & 33.99 & 31.95 & 31.58 & 11.90 & 13.27 & 13.63 & 0.71 & 0.80 & 
0.86 & 0.09 & 0.10 & 0.10 \\ \hline
& \multicolumn{11}{l}{\textit{\textbf{A-LASSO}}} &  \\ \hline
\textbf{20} & 31.89 & 27.80 & 25.16 & 5.28 & 5.67 & 6.07 & 0.65 & 0.73 & 0.80
& 0.13 & 0.14 & 0.14 \\ 
\textbf{40} & 30.03 & 27.05 & 31.40 & 6.89 & 7.67 & 8.11 & 0.65 & 0.75 & 0.81
& 0.11 & 0.12 & 0.12 \\ 
\textbf{100} & 34.01 & 31.34 & 30.79 & 9.27 & 10.63 & 11.02 & 0.63 & 0.74 & 
0.80 & 0.07 & 0.08 & 0.08 \\ \hline
& \multicolumn{11}{l}{\textit{\textbf{Boosting}}} &  \\ \hline
\textbf{20} & 30.93 & 26.95 & 24.37 & 4.46 & 4.67 & 4.76 & 0.69 & 0.76 & 0.81
& 0.08 & 0.08 & 0.08 \\ 
\textbf{40} & 30.18 & 26.66 & 29.99 & 6.20 & 6.15 & 6.07 & 0.69 & 0.77 & 0.82
& 0.09 & 0.08 & 0.07 \\ 
\textbf{100} & 35.62 & 31.41 & 30.43 & 11.25 & 10.05 & 9.50 & 0.67 & 0.75 & 
0.80 & 0.09 & 0.07 & 0.06 \\ \hline
\multicolumn{13}{|l|}{\textbf{B. Light down-weighting in both the variable
selection and estimation/forecasting stages.}} \\ \hline
\multicolumn{13}{|l|}{OCMT uses down-weighted variables for selection as well
as for forecasting using Least Squares.} \\ \hline
\multicolumn{13}{|l|}{Remaining forecasts are based on Lasso, A-Lasso and
Boosting regressions applied to down-weighted data.} \\ \hline
& \multicolumn{11}{l}{\textit{\textbf{OCMT}}} &  \\ \hline
\textbf{20} & 29.58 & 26.42 & 23.95 & 3.78 & 4.98 & 6.13 & 0.65 & 0.74 & 0.81
& 0.06 & 0.10 & 0.15 \\ 
\textbf{40} & 27.52 & 25.17 & 30.26 & 4.11 & 6.15 & 8.05 & 0.61 & 0.72 & 0.80
& 0.04 & 0.08 & 0.12 \\ 
\textbf{100} & 28.03 & 28.08 & 29.18 & 5.18 & 8.87 & 12.67 & 0.56 & 0.68 & 
0.75 & 0.03 & 0.06 & 0.10 \\ \hline
& \multicolumn{11}{l}{\textit{\textbf{LASSO}}} &  \\ \hline
\textbf{20} & 30.41 & 26.42 & 24.12 & 7.13 & 7.42 & 7.63 & 0.73 & 0.78 & 0.82
& 0.21 & 0.21 & 0.22 \\ 
\textbf{40} & 28.59 & 25.80 & 29.48 & 9.64 & 10.28 & 10.44 & 0.70 & 0.76 & 
0.80 & 0.17 & 0.18 & 0.18 \\ 
\textbf{100} & 30.30 & 28.65 & 28.61 & 15.48 & 16.59 & 16.50 & 0.65 & 0.72 & 
0.76 & 0.13 & 0.14 & 0.13 \\ \hline
& \multicolumn{11}{l}{\textit{\textbf{A-LASSO}}} &  \\ \hline
\textbf{20} & 31.33 & 27.15 & 24.38 & 5.50 & 5.80 & 5.98 & 0.63 & 0.69 & 0.74
& 0.15 & 0.15 & 0.15 \\ 
\textbf{40} & 30.14 & 27.38 & 31.00 & 7.62 & 8.18 & 8.42 & 0.61 & 0.69 & 0.74
& 0.13 & 0.14 & 0.14 \\ 
\textbf{100} & 33.71 & 30.96 & 31.28 & 12.30 & 13.41 & 13.50 & 0.59 & 0.66 & 
0.71 & 0.10 & 0.11 & 0.11 \\ \hline
& \multicolumn{11}{l}{\textit{\textbf{Boosting}}} &  \\ \hline
\textbf{20} & 30.42 & 27.30 & 25.01 & 5.99 & 6.92 & 7.72 & 0.69 & 0.77 & 0.83
& 0.16 & 0.19 & 0.22 \\ 
\textbf{40} & 29.67 & 27.54 & 32.25 & 10.24 & 12.10 & 13.92 & 0.70 & 0.78 & 
0.83 & 0.19 & 0.22 & 0.26 \\ 
\textbf{100} & 35.02 & 32.18 & 33.88 & 26.74 & 32.14 & 36.59 & 0.70 & 0.78 & 
0.83 & 0.24 & 0.29 & 0.33 \\ \hline\hline
\end{tabular}%
	\vspace{-0.2in}
	\begin{flushleft}
		\noindent 
		\scriptsize%
		\singlespacing%
		Notes: Light down-weighting is defined by by values $\lambda
=0.975,0.98,0.985,0.99,0.995,1$. For this set of exponential down-weighting
schemes we focus on simple average forecasts computed over the individual
forecasts obtained for each value of $\lambda $ in the set under
consideration. See notes to Table \ref{TNBS}.
	\end{flushleft}
\end{table}

\begin{table}

	\caption{MC results for methods using heavy down-weighting in the experiment with parameter instabilities, no dynamics ($%
	\rho _{y}=0$), and low fit.}
	\centering
 \vspace{0.2cm}
 \renewcommand{\arraystretch}{1.12}
 \setlength{\tabcolsep}{8pt}
	\scriptsize%
\begin{tabular}{|rrrrrrrrrrrrr|}
\hline\hline
& \multicolumn{3}{c}{\textbf{MSFE (}$\times 100$)} & \multicolumn{3}{|c}{$%
\hat{k}$} & \multicolumn{3}{|c}{\textbf{TPR}} & \multicolumn{3}{|c|}{\textbf{%
FPR}} \\ \hline
$N\backslash T$ & \textbf{100} & \textbf{150} & \textbf{200}
& \textbf{100} & \textbf{150} & \textbf{200} & \textbf{100} & \textbf{150} & 
\textbf{200} & \textbf{100} & \textbf{150} & \textbf{200} \\ \hline
\multicolumn{13}{|l|}{\textbf{A. Heavy down-weighting in the
estimation/forecasting stage only. }} \\ \hline
\multicolumn{13}{|l|}{Variable selection is based on original (not
down-weighted) data.} \\ \hline
\multicolumn{13}{|l|}{Forecasting stage is Least Squares on selected
down-weighted covariates for all methods} \\ \hline
& \multicolumn{11}{l}{\textit{\textbf{Oracle}}} &  \\ \hline
\textbf{20} & 28.48 & 25.64 & 23.58 & 4.00 & 4.00 & 4.00 & 1.00 & 1.00 & 1.00
& 0.00 & 0.00 & 0.00 \\ 
\textbf{40} & 26.72 & 24.10 & 28.65 & 4.00 & 4.00 & 4.00 & 1.00 & 1.00 & 1.00
& 0.00 & 0.00 & 0.00 \\ 
\textbf{100} & 26.88 & 25.69 & 26.26 & 4.00 & 4.00 & 4.00 & 1.00 & 1.00 & 
1.00 & 0.00 & 0.00 & 0.00 \\ \hline
& \multicolumn{11}{l}{\textit{\textbf{OCMT}}} &  \\ \hline
\textbf{20} & 29.24 & 26.36 & 24.98 & 4.42 & 5.51 & 6.39 & 0.81 & 0.93 & 0.97
& 0.06 & 0.09 & 0.12 \\ 
\textbf{40} & 27.33 & 25.16 & 30.63 & 4.13 & 5.34 & 6.11 & 0.78 & 0.92 & 0.97
& 0.03 & 0.04 & 0.06 \\ 
\textbf{100} & 27.70 & 26.79 & 27.47 & 3.89 & 5.07 & 5.68 & 0.73 & 0.90 & 
0.95 & 0.01 & 0.01 & 0.02 \\ \hline
& \multicolumn{11}{l}{\textit{\textbf{LASSO}}} &  \\ \hline
\textbf{20} & 32.57 & 29.32 & 26.67 & 7.04 & 7.47 & 7.91 & 0.78 & 0.85 & 0.89
& 0.20 & 0.20 & 0.22 \\ 
\textbf{40} & 31.02 & 28.49 & 34.06 & 8.99 & 9.90 & 10.38 & 0.75 & 0.84 & 
0.89 & 0.15 & 0.16 & 0.17 \\ 
\textbf{100} & 35.75 & 35.32 & 34.44 & 11.90 & 13.27 & 13.63 & 0.71 & 0.80 & 
0.86 & 0.09 & 0.10 & 0.10 \\ \hline
& \multicolumn{11}{l}{\textit{\textbf{A-LASSO}}} &  \\ \hline
\textbf{20} & 32.43 & 28.56 & 26.21 & 5.28 & 5.67 & 6.07 & 0.65 & 0.73 & 0.80
& 0.13 & 0.14 & 0.14 \\ 
\textbf{40} & 30.72 & 27.76 & 32.81 & 6.89 & 7.67 & 8.11 & 0.65 & 0.75 & 0.81
& 0.11 & 0.12 & 0.12 \\ 
\textbf{100} & 35.32 & 33.70 & 32.74 & 9.27 & 10.63 & 11.02 & 0.63 & 0.74 & 
0.80 & 0.07 & 0.08 & 0.08 \\ \hline
& \multicolumn{11}{l}{\textit{\textbf{Boosting}}} &  \\ \hline
\textbf{20} & 31.16 & 27.72 & 25.59 & 4.46 & 4.67 & 4.76 & 0.69 & 0.76 & 0.81
& 0.08 & 0.08 & 0.08 \\ 
\textbf{40} & 30.88 & 27.55 & 31.20 & 6.20 & 6.15 & 6.07 & 0.69 & 0.77 & 0.82
& 0.09 & 0.08 & 0.07 \\ 
\textbf{100} & 37.56 & 33.47 & 32.40 & 11.25 & 10.05 & 9.50 & 0.67 & 0.75 & 
0.80 & 0.09 & 0.07 & 0.06 \\ \hline
\multicolumn{13}{|l|}{\textbf{B. Heavy down-weighting in both the variable
selection and estimation/forecasting stages.}} \\ \hline
\multicolumn{13}{|l|}{OCMT uses down-weighted variables for selection as well
as for forecasting using Least Squares.} \\ \hline
\multicolumn{13}{|l|}{Remaining forecasts are based on Lasso, A-Lasso and
Boosting regressions applied to down-weighted data.} \\ \hline
& \multicolumn{11}{l}{\textit{\textbf{OCMT}}} &  \\ \hline
\textbf{20} & 31.14 & 29.06 & 27.11 & 4.87 & 6.71 & 8.31 & 0.62 & 0.71 & 0.78
& 0.12 & 0.19 & 0.26 \\ 
\textbf{40} & 30.05 & 29.87 & 37.71 & 6.64 & 10.28 & 13.46 & 0.58 & 0.70 & 
0.78 & 0.11 & 0.19 & 0.26 \\ 
\textbf{100} & 33.39 & 36.66 & 42.02 & 11.31 & 19.36 & 27.60 & 0.53 & 0.67 & 
0.75 & 0.09 & 0.17 & 0.25 \\ \hline
& \multicolumn{11}{l}{\textit{\textbf{LASSO}}} &  \\ \hline
\textbf{20} & 31.31 & 27.55 & 25.21 & 7.42 & 7.61 & 7.63 & 0.69 & 0.74 & 0.75
& 0.23 & 0.23 & 0.23 \\ 
\textbf{40} & 30.17 & 28.45 & 32.23 & 11.69 & 11.92 & 12.11 & 0.67 & 0.72 & 
0.75 & 0.22 & 0.23 & 0.23 \\ 
\textbf{100} & 33.59 & 31.12 & 30.97 & 20.62 & 22.23 & 22.78 & 0.63 & 0.69 & 
0.73 & 0.18 & 0.19 & 0.20 \\ \hline
& \multicolumn{11}{l}{\textit{\textbf{A-LASSO}}} &  \\ \hline
\textbf{20} & 32.42 & 28.40 & 26.15 & 5.74 & 5.91 & 5.95 & 0.60 & 0.64 & 0.67
& 0.17 & 0.17 & 0.16 \\ 
\textbf{40} & 32.30 & 30.51 & 34.28 & 9.26 & 9.42 & 9.65 & 0.59 & 0.65 & 0.68
& 0.17 & 0.17 & 0.17 \\ 
\textbf{100} & 37.03 & 34.44 & 33.85 & 15.96 & 17.35 & 17.89 & 0.56 & 0.63 & 
0.67 & 0.14 & 0.15 & 0.15 \\ \hline
& \multicolumn{11}{l}{\textit{\textbf{Boosting}}} &  \\ \hline
\textbf{20} & 33.77 & 31.31 & 30.02 & 8.38 & 9.81 & 10.83 & 0.73 & 0.81 & 
0.85 & 0.27 & 0.33 & 0.37 \\ 
\textbf{40} & 35.26 & 35.00 & 42.12 & 17.72 & 20.39 & 22.32 & 0.76 & 0.84 & 
0.88 & 0.37 & 0.43 & 0.47 \\ 
\textbf{100} & 40.68 & 38.33 & 40.39 & 42.16 & 46.42 & 48.87 & 0.75 & 0.82 & 
0.85 & 0.39 & 0.43 & 0.45 \\ \hline\hline
\end{tabular}%
	\vspace{-0.2in}
	\begin{flushleft}
		\noindent 
		\scriptsize%
		\singlespacing%
		Notes: Heavy down-weighting is defined by by values $\lambda
=0.95,0.96,0.97,0.98,0.99,1$. For this set of exponential down-weighting
schemes we focus on simple average forecasts computed over the individual
forecasts obtained for each value of $\lambda $ in the set under
consideration. See notes to Table \ref{TNBS}.
	\end{flushleft}
\end{table}

\begin{table}

	\vspace{0.1cm}
	\caption{MC results for methods using no down-weighting in the experiment with parameter instabilities, no dynamics ($%
	\rho _{y}=0$) and high fit.}
	\centering
 \vspace{0.2cm}
 \renewcommand{\arraystretch}{1.12}
	\scriptsize%
\begin{tabular}{|rrrrrrrrrrrrr|}
\hline\hline
& \multicolumn{3}{c}{\textbf{MSFE (}$\times 100$)} & \multicolumn{3}{|c}{$%
\hat{k}$} & \multicolumn{3}{|c}{\textbf{TPR}} & \multicolumn{3}{|c|}{\textbf{%
FPR}} \\ \hline
$N\backslash T$ & \textbf{100} & \textbf{150} & \textbf{200}
& \textbf{100} & \textbf{150} & \textbf{200} & \textbf{100} & \textbf{150} & 
\textbf{200} & \textbf{100} & \textbf{150} & \textbf{200} \\ \hline
\multicolumn{1}{|l}{} & \multicolumn{11}{l}{\textit{\textbf{Oracle}}} & 
\multicolumn{1}{l|}{} \\ \hline
\textbf{20} & 10.59 & 9.37 & 8.83 & 4.00 & 4.00 & 4.00 & 1.00 & 1.00 & 1.00
& 0.00 & 0.00 & 0.00 \\ 
\textbf{40} & 9.61 & 9.03 & 10.18 & 4.00 & 4.00 & 4.00 & 1.00 & 1.00 & 1.00
& 0.00 & 0.00 & 0.00 \\ 
\textbf{100} & 9.91 & 9.48 & 9.70 & 4.00 & 4.00 & 4.00 & 1.00 & 1.00 & 1.00
& 0.00 & 0.00 & 0.00 \\ \hline
\multicolumn{1}{|l}{} & \multicolumn{11}{l}{\textit{\textbf{OCMT}}} &  \\ 
\hline
\textbf{20} & 10.93 & 9.56 & 9.13 & 6.13 & 7.22 & 8.21 & 0.96 & 0.99 & 1.00
& 0.11 & 0.16 & 0.21 \\ 
\textbf{40} & 9.84 & 9.48 & 10.52 & 5.97 & 7.12 & 7.91 & 0.96 & 0.99 & 1.00
& 0.05 & 0.08 & 0.10 \\ 
\textbf{100} & 10.33 & 9.85 & 10.19 & 5.87 & 7.00 & 7.65 & 0.94 & 0.99 & 1.00
& 0.02 & 0.03 & 0.04 \\ \hline
\multicolumn{1}{|l}{} & \textit{\textbf{LASSO}} &  &  &  &  &  &  &  &  &  & 
&  \\ \hline
\textbf{20} & 11.54 & 9.85 & 9.31 & 8.78 & 9.42 & 9.67 & 0.90 & 0.94 & 0.96
& 0.26 & 0.28 & 0.29 \\ 
\textbf{40} & 10.54 & 9.83 & 10.66 & 11.95 & 13.22 & 14.14 & 0.89 & 0.94 & 
0.96 & 0.21 & 0.24 & 0.26 \\ 
\textbf{100} & 10.88 & 10.23 & 10.55 & 16.46 & 18.79 & 20.73 & 0.87 & 0.92 & 
0.95 & 0.13 & 0.15 & 0.17 \\ \hline
\multicolumn{1}{|l}{} & \multicolumn{11}{l}{\textit{\textbf{LASSO for
variable selection only. LS for estimation/forecasting.}}} &  \\ \hline
\textbf{20} & 12.10 & 10.35 & 9.57 & 8.78 & 9.42 & 9.67 & 0.90 & 0.94 & 0.96
& 0.26 & 0.28 & 0.29 \\ 
\textbf{40} & 11.72 & 10.93 & 11.44 & 11.95 & 13.22 & 14.14 & 0.89 & 0.94 & 
0.96 & 0.21 & 0.24 & 0.26 \\ 
\textbf{100} & 12.95 & 11.99 & 12.51 & 16.46 & 18.79 & 20.73 & 0.87 & 0.92 & 
0.95 & 0.13 & 0.15 & 0.17 \\ \hline
\multicolumn{1}{|l}{} & \multicolumn{11}{l}{\textit{\textbf{A-LASSO}}} &  \\ 
\hline
\textbf{20} & 11.77 & 10.07 & 9.43 & 6.79 & 7.41 & 7.62 & 0.81 & 0.87 & 0.91
& 0.18 & 0.20 & 0.20 \\ 
\textbf{40} & 11.34 & 10.49 & 11.16 & 9.25 & 10.39 & 11.20 & 0.81 & 0.89 & 
0.92 & 0.15 & 0.17 & 0.19 \\ 
\textbf{100} & 12.02 & 11.42 & 11.79 & 12.84 & 14.83 & 16.51 & 0.81 & 0.88 & 
0.92 & 0.10 & 0.11 & 0.13 \\ \hline
& \multicolumn{11}{l}{\textit{\textbf{A-LASSO for variable selection only.
LS for estimation/forecasting.}}} &  \\ \hline
\textbf{20} & 12.01 & 10.24 & 9.54 & 6.79 & 7.41 & 7.62 & 0.81 & 0.87 & 0.91
& 0.18 & 0.20 & 0.20 \\ 
\textbf{40} & 11.70 & 10.94 & 11.45 & 9.25 & 10.39 & 11.20 & 0.81 & 0.89 & 
0.92 & 0.15 & 0.17 & 0.19 \\ 
\textbf{100} & 12.68 & 11.91 & 12.30 & 12.84 & 14.83 & 16.51 & 0.81 & 0.88 & 
0.92 & 0.10 & 0.11 & 0.13 \\ \hline
\multicolumn{1}{|l}{} & \multicolumn{11}{l}{\textit{\textbf{Boosting}}} & 
\\ \hline
\textbf{20} & 11.56 & 10.00 & 9.26 & 5.66 & 5.79 & 5.82 & 0.84 & 0.90 & 0.92
& 0.12 & 0.11 & 0.11 \\ 
\textbf{40} & 10.89 & 9.60 & 10.66 & 7.88 & 7.84 & 7.74 & 0.84 & 0.90 & 0.93
& 0.11 & 0.11 & 0.10 \\ 
\textbf{100} & 10.97 & 10.16 & 10.28 & 13.81 & 12.55 & 12.22 & 0.83 & 0.89 & 
0.92 & 0.10 & 0.09 & 0.09 \\ \hline
\multicolumn{1}{|l}{} & \multicolumn{11}{l}{\textit{\textbf{Boosting for
variable selection only. LS for estimation/forecasting.}}} &  \\ \hline
\textbf{20} & 11.66 & 9.86 & 9.13 & 5.66 & 5.79 & 5.82 & 0.84 & 0.90 & 0.92
& 0.12 & 0.11 & 0.11 \\ 
\textbf{40} & 11.27 & 10.07 & 10.90 & 7.88 & 7.84 & 7.74 & 0.84 & 0.90 & 0.93
& 0.11 & 0.11 & 0.10 \\ 
\textbf{100} & 12.98 & 11.26 & 11.25 & 13.81 & 12.55 & 12.22 & 0.83 & 0.89 & 
0.92 & 0.10 & 0.09 & 0.09 \\ \hline\hline
\end{tabular}%
	\vspace{-0.2in}
	\begin{flushleft}
		\noindent 
		\scriptsize%
		\singlespacing%
Notes: See notes to Table \ref{TNBS}.
	\end{flushleft}
\end{table}

\begin{table}

	\caption{MC results for methods using light down-weighting in the experiment with parameter instabilities, no dynamics ($%
	\rho _{y}=0$), and high fit.}
	\centering
 \vspace{0.2cm}
 \renewcommand{\arraystretch}{1.12}
 \setlength{\tabcolsep}{8pt}
	\scriptsize%
\begin{tabular}{|rrrrrrrrrrrrr|}
\hline\hline
& \multicolumn{3}{c}{\textbf{MSFE (}$\times 100$)} & \multicolumn{3}{|c}{$%
\hat{k}$} & \multicolumn{3}{|c}{\textbf{TPR}} & \multicolumn{3}{|c|}{\textbf{%
FPR}} \\ \hline
$N\backslash T$ & \textbf{100} & \textbf{150} & \textbf{200}
& \textbf{100} & \textbf{150} & \textbf{200} & \textbf{100} & \textbf{150} & 
\textbf{200} & \textbf{100} & \textbf{150} & \textbf{200} \\ \hline
\multicolumn{13}{|l|}{\textbf{A. Light down-weighting in the
estimation/forecasting stage only. }} \\ \hline
\multicolumn{13}{|l|}{Variable selection is based on original (not
down-weighted) data.} \\ \hline
\multicolumn{13}{|l|}{Forecasting stage is Least Squares on selected
down-weighted covariates for all methods} \\ \hline
& \multicolumn{11}{l}{\textit{\textbf{Oracle}}} &  \\ \hline
\textbf{20} & 9.53 & 8.16 & 7.32 & 4.00 & 4.00 & 4.00 & 1.00 & 1.00 & 1.00 & 
0.00 & 0.00 & 0.00 \\ 
\textbf{40} & 8.75 & 7.71 & 8.82 & 4.00 & 4.00 & 4.00 & 1.00 & 1.00 & 1.00 & 
0.00 & 0.00 & 0.00 \\ \hline
\textbf{100} & 9.06 & 8.19 & 8.11 & 4.00 & 4.00 & 4.00 & 1.00 & 1.00 & 1.00
& 0.00 & 0.00 & 0.00 \\ \hline
& \multicolumn{11}{l}{\textit{\textbf{OCMT}}} &  \\ 
\textbf{20} & 9.97 & 8.60 & 7.89 & 6.13 & 7.22 & 8.21 & 0.96 & 0.99 & 1.00 & 
0.11 & 0.16 & 0.21 \\ 
\textbf{40} & 9.17 & 8.41 & 9.64 & 5.97 & 7.12 & 7.91 & 0.96 & 0.99 & 1.00 & 
0.05 & 0.08 & 0.10 \\ \hline
\textbf{100} & 9.77 & 8.69 & 8.70 & 5.87 & 7.00 & 7.65 & 0.94 & 0.99 & 1.00
& 0.02 & 0.03 & 0.04 \\ \hline
& \multicolumn{11}{l}{\textit{\textbf{LASSO}}} &  \\ 
\textbf{20} & 11.36 & 9.51 & 8.39 & 8.78 & 9.42 & 9.67 & 0.90 & 0.94 & 0.96
& 0.26 & 0.28 & 0.29 \\ 
\textbf{40} & 11.19 & 10.09 & 10.75 & 11.95 & 13.22 & 14.14 & 0.89 & 0.94 & 
0.96 & 0.21 & 0.24 & 0.26 \\ \hline
\textbf{100} & 12.83 & 11.73 & 11.82 & 16.46 & 18.79 & 20.73 & 0.87 & 0.92 & 
0.95 & 0.13 & 0.15 & 0.17 \\ \hline
& \multicolumn{11}{l}{\textit{\textbf{A-LASSO}}} &  \\ 
\textbf{20} & 11.30 & 9.38 & 8.30 & 6.79 & 7.41 & 7.62 & 0.81 & 0.87 & 0.91
& 0.18 & 0.20 & 0.20 \\ 
\textbf{40} & 11.11 & 9.94 & 10.62 & 9.25 & 10.39 & 11.20 & 0.81 & 0.89 & 
0.92 & 0.15 & 0.17 & 0.19 \\ \hline
\textbf{100} & 12.44 & 11.52 & 11.57 & 12.84 & 14.83 & 16.51 & 0.81 & 0.88 & 
0.92 & 0.10 & 0.11 & 0.13 \\ \hline
& \multicolumn{11}{l}{\textit{\textbf{Boosting}}} &  \\ 
\textbf{20} & 10.83 & 9.06 & 7.85 & 5.66 & 5.79 & 5.82 & 0.84 & 0.90 & 0.92
& 0.12 & 0.11 & 0.11 \\ 
\textbf{40} & 10.86 & 9.16 & 10.11 & 7.88 & 7.84 & 7.74 & 0.84 & 0.90 & 0.93
& 0.11 & 0.11 & 0.10 \\ \hline
\textbf{100} & 12.79 & 10.90 & 10.55 & 13.81 & 12.55 & 12.22 & 0.83 & 0.89 & 
0.92 & 0.10 & 0.09 & 0.09 \\ \hline
\multicolumn{13}{|l|}{\textbf{B. Light down-weighting in both the variable
selection and estimation/forecasting stages.}} \\ \hline
\multicolumn{13}{|l|}{OCMT uses down-weighted variables for selection as well
as for forecasting using Least Squares.} \\ \hline
\multicolumn{13}{|l|}{Remaining forecasts are based on Lasso, A-Lasso and
Boosting regressions applied to down-weighted data.} \\ \hline
& \multicolumn{11}{l}{\textit{\textbf{OCMT}}} &  \\ \hline
\textbf{20} & 10.32 & 9.00 & 7.89 & 5.30 & 6.49 & 7.63 & 0.82 & 0.87 & 0.90
& 0.10 & 0.15 & 0.20 \\ 
\textbf{40} & 9.78 & 8.86 & 10.14 & 6.04 & 8.16 & 10.19 & 0.80 & 0.86 & 0.90
& 0.07 & 0.12 & 0.16 \\ \hline
\textbf{100} & 10.18 & 10.10 & 10.14 & 7.95 & 12.09 & 16.33 & 0.77 & 0.84 & 
0.88 & 0.05 & 0.09 & 0.13 \\ \hline
& \multicolumn{11}{l}{\textit{\textbf{LASSO}}} &  \\ 
\textbf{20} & 10.84 & 9.10 & 7.91 & 8.95 & 9.22 & 9.37 & 0.87 & 0.92 & 0.95
& 0.27 & 0.28 & 0.28 \\ 
\textbf{40} & 10.34 & 9.13 & 9.88 & 12.68 & 13.26 & 13.37 & 0.86 & 0.91 & 
0.93 & 0.23 & 0.24 & 0.24 \\ \hline
\textbf{100} & 10.98 & 10.05 & 9.75 & 19.68 & 21.18 & 21.59 & 0.83 & 0.89 & 
0.92 & 0.16 & 0.18 & 0.18 \\ \hline
& \multicolumn{11}{l}{\textit{\textbf{A-LASSO}}} &  \\ 
\textbf{20} & 11.00 & 9.21 & 7.87 & 7.00 & 7.32 & 7.48 & 0.79 & 0.86 & 0.90
& 0.19 & 0.19 & 0.19 \\ 
\textbf{40} & 10.98 & 9.67 & 10.25 & 10.08 & 10.59 & 10.80 & 0.80 & 0.87 & 
0.90 & 0.17 & 0.18 & 0.18 \\ \hline
\textbf{100} & 12.11 & 10.90 & 10.71 & 15.50 & 17.01 & 17.51 & 0.78 & 0.85 & 
0.89 & 0.12 & 0.14 & 0.14 \\ \hline
& \multicolumn{11}{l}{\textit{\textbf{Boosting}}} &  \\ 
\textbf{20} & 10.81 & 9.39 & 8.43 & 7.02 & 7.82 & 8.53 & 0.84 & 0.90 & 0.94
& 0.18 & 0.21 & 0.24 \\ 
\textbf{40} & 10.79 & 9.61 & 10.60 & 11.52 & 13.25 & 14.85 & 0.85 & 0.91 & 
0.94 & 0.20 & 0.24 & 0.28 \\ 
\textbf{100} & 12.18 & 10.88 & 11.35 & 27.20 & 32.73 & 37.25 & 0.84 & 0.90 & 
0.93 & 0.24 & 0.29 & 0.34 \\ \hline\hline
\end{tabular}%
	\vspace{-0.2in}
	\begin{flushleft}
		\noindent 
		\scriptsize%
		\singlespacing%
		Notes: Light down-weighting is defined by by values $\lambda
=0.975,0.98,0.985,0.99,0.995,1$. For this set of exponential down-weighting
schemes we focus on simple average forecasts computed over the individual
forecasts obtained for each value of $\lambda $ in the set under
consideration. See notes to Table \ref{TNBS}.
	\end{flushleft}
\end{table}

\begin{table}

	\caption{MC results for methods using heavy down-weighting in the experiment with parameter instabilities, no dynamics ($%
	\rho _{y}=0$), and high fit.}
	\centering
 \vspace{0.2cm}
 \renewcommand{\arraystretch}{1.12}
 \setlength{\tabcolsep}{8pt}
	\scriptsize%
\begin{tabular}{|rrrrrrrrrrrrr|}
\hline\hline
& \multicolumn{3}{c}{\textbf{MSFE (}$\times 100$)} & \multicolumn{3}{|c}{$%
\hat{k}$} & \multicolumn{3}{|c}{\textbf{TPR}} & \multicolumn{3}{|c|}{\textbf{%
FPR}} \\ \hline
$N\backslash T$ & \textbf{100} & \textbf{150} & \textbf{200}
& \textbf{100} & \textbf{150} & \textbf{200} & \textbf{100} & \textbf{150} & 
\textbf{200} & \textbf{100} & \textbf{150} & \textbf{200} \\ \hline
\multicolumn{13}{|l|}{\textbf{A. Heavy down-weighting in the
estimation/forecasting stage only. }} \\ \hline
\multicolumn{13}{|l|}{Variable selection is based on original (not
down-weighted) data.} \\ \hline
\multicolumn{13}{|l|}{Forecasting stage is Least Squares on selected
down-weighted covariates for all methods} \\ \hline
& \multicolumn{11}{l}{\textit{\textbf{Oracle}}} &  \\ \hline
\textbf{20} & 9.13 & 8.04 & 7.30 & 4.00 & 4.00 & 4.00 & 1.00 & 1.00 & 1.00 & 
0.00 & 0.00 & 0.00 \\ 
\textbf{40} & 8.52 & 7.61 & 8.86 & 4.00 & 4.00 & 4.00 & 1.00 & 1.00 & 1.00 & 
0.00 & 0.00 & 0.00 \\ 
\textbf{100} & 8.71 & 8.06 & 8.18 & 4.00 & 4.00 & 4.00 & 1.00 & 1.00 & 1.00
& 0.00 & 0.00 & 0.00 \\ \hline
& \multicolumn{11}{l}{\textit{\textbf{OCMT}}} &  \\ \hline
\textbf{20} & 9.69 & 8.79 & 8.25 & 6.13 & 7.22 & 8.21 & 0.96 & 0.99 & 1.00 & 
0.11 & 0.16 & 0.21 \\ 
\textbf{40} & 9.07 & 8.52 & 10.06 & 5.97 & 7.12 & 7.91 & 0.96 & 0.99 & 1.00
& 0.05 & 0.08 & 0.10 \\ 
\textbf{100} & 9.65 & 8.72 & 8.95 & 5.87 & 7.00 & 7.65 & 0.94 & 0.99 & 1.00
& 0.02 & 0.03 & 0.04 \\ \hline
& \multicolumn{11}{l}{\textit{\textbf{LASSO}}} &  \\ \hline
\textbf{20} & 11.11 & 9.68 & 8.81 & 8.78 & 9.42 & 9.67 & 0.90 & 0.94 & 0.96
& 0.26 & 0.28 & 0.29 \\ 
\textbf{40} & 11.23 & 10.32 & 11.21 & 11.95 & 13.22 & 14.14 & 0.89 & 0.94 & 
0.96 & 0.21 & 0.24 & 0.26 \\ 
\textbf{100} & 13.44 & 12.66 & 12.90 & 16.46 & 18.79 & 20.73 & 0.87 & 0.92 & 
0.95 & 0.13 & 0.15 & 0.17 \\ \hline
& \multicolumn{11}{l}{\textit{\textbf{A-LASSO}}} &  \\ \hline
\textbf{20} & 11.02 & 9.41 & 8.52 & 6.79 & 7.41 & 7.62 & 0.81 & 0.87 & 0.91
& 0.18 & 0.20 & 0.20 \\ 
\textbf{40} & 10.96 & 10.04 & 10.97 & 9.25 & 10.39 & 11.20 & 0.81 & 0.89 & 
0.92 & 0.15 & 0.17 & 0.19 \\ 
\textbf{100} & 12.64 & 12.23 & 12.32 & 12.84 & 14.83 & 16.51 & 0.81 & 0.88 & 
0.92 & 0.10 & 0.11 & 0.13 \\ \hline
& \multicolumn{11}{l}{\textit{\textbf{Boosting}}} &  \\ \hline
\textbf{20} & 10.44 & 9.04 & 8.01 & 5.66 & 5.79 & 5.82 & 0.84 & 0.90 & 0.92
& 0.12 & 0.11 & 0.11 \\ 
\textbf{40} & 10.76 & 9.21 & 10.37 & 7.88 & 7.84 & 7.74 & 0.84 & 0.90 & 0.93
& 0.11 & 0.11 & 0.10 \\ 
\textbf{100} & 13.12 & 11.48 & 11.11 & 13.81 & 12.55 & 12.22 & 0.83 & 0.89 & 
0.92 & 0.10 & 0.09 & 0.09 \\ \hline
\multicolumn{13}{|l|}{\textbf{B. Heavy down-weighting in both the variable
selection and estimation/forecasting stages.}} \\ \hline
\multicolumn{13}{|l|}{OCMT uses down-weighted variables for selection as well
as for forecasting using Least Squares.} \\ \hline
\multicolumn{13}{|l|}{Remaining forecasts are based on Lasso, A-Lasso and
Boosting regressions applied to down-weighted data.} \\ \hline
& \multicolumn{11}{l}{\textit{\textbf{OCMT}}} &  \\ \hline
\textbf{20} & 10.68 & 9.57 & 8.82 & 6.35 & 8.16 & 9.74 & 0.77 & 0.83 & 0.87
& 0.16 & 0.24 & 0.31 \\ 
\textbf{40} & 10.80 & 10.32 & 12.48 & 8.78 & 12.47 & 15.65 & 0.75 & 0.82 & 
0.88 & 0.14 & 0.23 & 0.30 \\ 
\textbf{100} & 12.45 & 13.91 & 15.18 & 14.85 & 23.23 & 31.55 & 0.72 & 0.81 & 
0.86 & 0.12 & 0.20 & 0.28 \\ \hline
& \multicolumn{11}{l}{\textit{\textbf{LASSO}}} &  \\ \hline
\textbf{20} & 10.76 & 9.31 & 8.19 & 9.23 & 9.46 & 9.48 & 0.86 & 0.90 & 0.92
& 0.29 & 0.29 & 0.29 \\ 
\textbf{40} & 10.62 & 9.73 & 10.56 & 14.34 & 14.72 & 15.02 & 0.85 & 0.90 & 
0.91 & 0.27 & 0.28 & 0.28 \\ 
\textbf{100} & 11.80 & 10.56 & 10.43 & 23.37 & 25.38 & 26.74 & 0.81 & 0.88 & 
0.90 & 0.20 & 0.22 & 0.23 \\ \hline
& \multicolumn{11}{l}{\textit{\textbf{A-LASSO}}} &  \\ \hline
\textbf{20} & 10.87 & 9.46 & 8.37 & 7.21 & 7.44 & 7.51 & 0.78 & 0.84 & 0.87
& 0.20 & 0.20 & 0.20 \\ 
\textbf{40} & 11.39 & 10.36 & 11.04 & 11.41 & 11.68 & 12.03 & 0.79 & 0.85 & 
0.87 & 0.21 & 0.21 & 0.21 \\ 
\textbf{100} & 12.93 & 11.56 & 11.46 & 17.98 & 19.74 & 20.84 & 0.76 & 0.83 & 
0.87 & 0.15 & 0.16 & 0.17 \\ \hline
& \multicolumn{11}{l}{\textit{\textbf{Boosting}}} &  \\ \hline
\textbf{20} & 11.55 & 10.64 & 10.03 & 9.21 & 10.53 & 11.47 & 0.87 & 0.92 & 
0.95 & 0.29 & 0.34 & 0.38 \\ 
\textbf{40} & 12.49 & 11.84 & 13.59 & 18.46 & 21.10 & 23.01 & 0.88 & 0.93 & 
0.95 & 0.37 & 0.43 & 0.48 \\ 
\textbf{100} & 14.16 & 12.82 & 13.31 & 42.35 & 46.88 & 49.42 & 0.87 & 0.92 & 
0.94 & 0.39 & 0.43 & 0.46 \\ \hline\hline
\end{tabular}%
	\vspace{-0.2in}
	\begin{flushleft}
		\noindent 
		\scriptsize%
		\singlespacing%
		Notes: Heavy down-weighting is defined by by values $\lambda
=0.95,0.96,0.97,0.98,0.99,1$. For this set of exponential down-weighting
schemes we focus on simple average forecasts computed over the individual
forecasts obtained for each value of $\lambda $ in the set under
consideration. See notes to Table \ref{TNBS}.
	\end{flushleft}
\end{table}

\begin{table}

	\vspace{0.1cm}
	\caption{MC results for methods using no down-weighting in the experiment with parameter instabilities, dynamics ($%
	\rho _{y} \neq 0$) and low fit.}
	\centering
 \vspace{0.2cm}
 \renewcommand{\arraystretch}{1.12}
	\scriptsize%
\begin{tabular}{|rrrrrrrrrrrrr|}
\hline\hline
& \multicolumn{3}{c}{\textbf{MSFE (}$\times 100$)} & \multicolumn{3}{|c}{$%
\hat{k}$} & \multicolumn{3}{|c}{\textbf{TPR}} & \multicolumn{3}{|c|}{\textbf{%
FPR}} \\ \hline
$N\backslash T$ & \textbf{100} & \textbf{150} & \textbf{200}
& \textbf{100} & \textbf{150} & \textbf{200} & \textbf{100} & \textbf{150} & 
\textbf{200} & \textbf{100} & \textbf{150} & \textbf{200} \\ \hline
\multicolumn{1}{|l}{} & \multicolumn{11}{l}{\textit{\textbf{Oracle}}} & 
\multicolumn{1}{l|}{} \\ \hline
\textbf{20} & 75.71 & 66.08 & 61.78 & 4.00 & 4.00 & 4.00 & 1.00 & 1.00 & 1.00
& 0.00 & 0.00 & 0.00 \\ 
\textbf{40} & 67.94 & 62.43 & 75.57 & 4.00 & 4.00 & 4.00 & 1.00 & 1.00 & 1.00
& 0.00 & 0.00 & 0.00 \\ 
\textbf{100} & 70.50 & 66.03 & 68.84 & 4.00 & 4.00 & 4.00 & 1.00 & 1.00 & 
1.00 & 0.00 & 0.00 & 0.00 \\ \hline
& \multicolumn{11}{l}{\textit{\textbf{OCMT}}} &  \\ \hline
\textbf{20} & 77.31 & 66.80 & 61.94 & 1.78 & 2.74 & 3.63 & 0.38 & 0.58 & 0.74
& 0.01 & 0.02 & 0.03 \\ 
\textbf{40} & 69.50 & 64.16 & 77.07 & 1.49 & 2.52 & 3.34 & 0.33 & 0.55 & 0.71
& 0.00 & 0.01 & 0.01 \\ 
\textbf{100} & 71.43 & 68.54 & 70.27 & 1.23 & 2.13 & 2.82 & 0.27 & 0.47 & 
0.62 & 0.00 & 0.00 & 0.00 \\ \hline
& \textit{\textbf{LASSO}} &  &  &  &  &  &  &  &  &  &  &  \\ \hline
\textbf{20} & 78.26 & 67.40 & 63.50 & 5.86 & 6.27 & 6.82 & 0.59 & 0.67 & 0.74
& 0.17 & 0.18 & 0.19 \\ 
\textbf{40} & 71.00 & 64.28 & 77.18 & 7.94 & 8.49 & 8.73 & 0.55 & 0.65 & 0.72
& 0.14 & 0.15 & 0.15 \\ 
\textbf{100} & 74.43 & 69.24 & 71.63 & 11.54 & 11.95 & 11.84 & 0.50 & 0.61 & 
0.67 & 0.10 & 0.10 & 0.09 \\ \hline
& \multicolumn{11}{l}{\textit{\textbf{LASSO for variable selection only. LS
for estimation/forecasting.}}} &  \\ \hline
\textbf{20} & 83.48 & 70.28 & 65.14 & 5.86 & 6.27 & 6.82 & 0.59 & 0.67 & 0.74
& 0.17 & 0.18 & 0.19 \\ 
\textbf{40} & 78.58 & 70.50 & 81.00 & 7.94 & 8.49 & 8.73 & 0.55 & 0.65 & 0.72
& 0.14 & 0.15 & 0.15 \\ 
\textbf{100} & 89.38 & 78.67 & 79.53 & 11.54 & 11.95 & 11.84 & 0.50 & 0.61 & 
0.67 & 0.10 & 0.10 & 0.09 \\ \hline
& \multicolumn{11}{l}{\textit{\textbf{A-LASSO}}} &  \\ \hline
\textbf{20} & 82.14 & 69.42 & 65.04 & 4.28 & 4.71 & 5.09 & 0.47 & 0.54 & 0.61
& 0.12 & 0.13 & 0.13 \\ 
\textbf{40} & 76.05 & 68.80 & 80.34 & 6.12 & 6.56 & 6.86 & 0.45 & 0.55 & 0.62
& 0.11 & 0.11 & 0.11 \\ 
\textbf{100} & 85.41 & 75.67 & 77.67 & 9.11 & 9.64 & 9.77 & 0.43 & 0.53 & 
0.60 & 0.07 & 0.08 & 0.07 \\ \hline
& \multicolumn{11}{l}{\textit{\textbf{A-LASSO for variable selection only.
LS for estimation/forecasting.}}} &  \\ \hline
\textbf{20} & 83.73 & 70.57 & 65.77 & 4.28 & 4.71 & 5.09 & 0.47 & 0.54 & 0.61
& 0.12 & 0.13 & 0.13 \\ 
\textbf{40} & 78.42 & 70.50 & 81.82 & 6.12 & 6.56 & 6.86 & 0.45 & 0.55 & 0.62
& 0.11 & 0.11 & 0.11 \\ 
\textbf{100} & 89.02 & 77.88 & 79.20 & 9.11 & 9.64 & 9.77 & 0.43 & 0.53 & 
0.60 & 0.07 & 0.08 & 0.07 \\ \hline
& \multicolumn{11}{l}{\textit{\textbf{Boosting}}} &  \\ \hline
\textbf{20} & 77.31 & 67.12 & 62.13 & 3.59 & 3.57 & 3.69 & 0.50 & 0.56 & 0.62
& 0.08 & 0.07 & 0.06 \\ 
\textbf{40} & 70.98 & 63.07 & 76.55 & 5.34 & 5.04 & 4.86 & 0.49 & 0.56 & 0.62
& 0.08 & 0.07 & 0.06 \\ 
\textbf{100} & 77.35 & 67.07 & 70.40 & 12.09 & 9.35 & 8.33 & 0.49 & 0.56 & 
0.60 & 0.10 & 0.07 & 0.06 \\ \hline
& \multicolumn{11}{l}{\textit{\textbf{Boosting for variable selection only.
LS for estimation/forecasting.}}} &  \\ \hline
\textbf{20} & 80.27 & 67.66 & 63.58 & 3.59 & 3.57 & 3.69 & 0.50 & 0.56 & 0.62
& 0.08 & 0.07 & 0.06 \\ 
\textbf{40} & 75.81 & 67.93 & 78.29 & 5.34 & 5.04 & 4.86 & 0.49 & 0.56 & 0.62
& 0.08 & 0.07 & 0.06 \\ 
\textbf{100} & 94.18 & 77.52 & 79.66 & 12.09 & 9.35 & 8.33 & 0.49 & 0.56 & 
0.60 & 0.10 & 0.07 & 0.06 \\ \hline\hline
\end{tabular}%
	\vspace{-0.2in}
	\begin{flushleft}
		\noindent 
		\scriptsize%
		\singlespacing%
Notes: See notes to Table \ref{TNBS}.
	\end{flushleft}
\end{table}

\begin{table}

	\caption{MC results for methods using light down-weighting in the experiment with parameter instabilities, dynamics ($%
	\rho _{y} \neq 0$), and low fit.}
	\centering
 \vspace{0.2cm}
 \renewcommand{\arraystretch}{1.12}
 \setlength{\tabcolsep}{8pt}
	\scriptsize%
\begin{tabular}{|rrrrrrrrrrrrr|}
\hline\hline
& \multicolumn{3}{c}{\textbf{MSFE (}$\times 100$)} & \multicolumn{3}{|c}{$%
\hat{k}$} & \multicolumn{3}{|c}{\textbf{TPR}} & \multicolumn{3}{|c|}{\textbf{%
FPR}} \\ \hline
$N\backslash T$ & \textbf{100} & \textbf{150} & \textbf{200}
& \textbf{100} & \textbf{150} & \textbf{200} & \textbf{100} & \textbf{150} & 
\textbf{200} & \textbf{100} & \textbf{150} & \textbf{200} \\ \hline
\multicolumn{13}{|l|}{\textbf{A. Light down-weighting in the
estimation/forecasting stage only. }} \\ \hline
\multicolumn{13}{|l|}{Variable selection is based on original (not
down-weighted) data.} \\ \hline
\multicolumn{13}{|l|}{Forecasting stage is Least Squares on selected
down-weighted covariates for all methods} \\ \hline
& \multicolumn{11}{l}{\textit{\textbf{Oracle}}} &  \\ \hline
\textbf{20} & 75.04 & 64.94 & 60.47 & 4.00 & 4.00 & 4.00 & 1.00 & 1.00 & 1.00
& 0.00 & 0.00 & 0.00 \\ 
\textbf{40} & 67.36 & 61.25 & 74.44 & 4.00 & 4.00 & 4.00 & 1.00 & 1.00 & 1.00
& 0.00 & 0.00 & 0.00 \\ 
\textbf{100} & 70.12 & 65.02 & 66.76 & 4.00 & 4.00 & 4.00 & 1.00 & 1.00 & 
1.00 & 0.00 & 0.00 & 0.00 \\ \hline
& \multicolumn{11}{l}{\textit{\textbf{OCMT}}} &  \\ \hline
\textbf{20} & 76.46 & 65.65 & 60.52 & 1.78 & 2.74 & 3.63 & 0.38 & 0.58 & 0.74
& 0.01 & 0.02 & 0.03 \\ 
\textbf{40} & 67.80 & 62.56 & 77.03 & 1.49 & 2.52 & 3.34 & 0.33 & 0.55 & 0.71
& 0.00 & 0.01 & 0.01 \\ 
\textbf{100} & 70.41 & 67.10 & 67.92 & 1.23 & 2.13 & 2.82 & 0.27 & 0.47 & 
0.62 & 0.00 & 0.00 & 0.00 \\ \hline
& \multicolumn{11}{l}{\textit{\textbf{LASSO}}} &  \\ \hline
\textbf{20} & 83.16 & 69.84 & 65.40 & 5.86 & 6.27 & 6.82 & 0.59 & 0.67 & 0.74
& 0.17 & 0.18 & 0.19 \\ 
\textbf{40} & 79.01 & 70.18 & 83.74 & 7.94 & 8.49 & 8.73 & 0.55 & 0.65 & 0.72
& 0.14 & 0.15 & 0.15 \\ 
\textbf{100} & 89.07 & 80.66 & 79.92 & 11.54 & 11.95 & 11.84 & 0.50 & 0.61 & 
0.67 & 0.10 & 0.10 & 0.09 \\ \hline
& \multicolumn{11}{l}{\textit{\textbf{A-LASSO}}} &  \\ \hline
\textbf{20} & 83.11 & 69.88 & 65.95 & 4.28 & 4.71 & 5.09 & 0.47 & 0.54 & 0.61
& 0.12 & 0.13 & 0.13 \\ 
\textbf{40} & 78.09 & 68.87 & 83.49 & 6.12 & 6.56 & 6.86 & 0.45 & 0.55 & 0.62
& 0.11 & 0.11 & 0.11 \\ 
\textbf{100} & 88.15 & 78.50 & 79.01 & 9.11 & 9.64 & 9.77 & 0.43 & 0.53 & 
0.60 & 0.07 & 0.08 & 0.07 \\ \hline
& \multicolumn{11}{l}{\textit{\textbf{Boosting}}} &  \\ \hline
\textbf{20} & 80.01 & 67.10 & 63.77 & 3.59 & 3.57 & 3.69 & 0.50 & 0.56 & 0.62
& 0.08 & 0.07 & 0.06 \\ 
\textbf{40} & 76.17 & 67.13 & 79.47 & 5.34 & 5.04 & 4.86 & 0.49 & 0.56 & 0.62
& 0.08 & 0.07 & 0.06 \\ 
\textbf{100} & 95.21 & 78.89 & 78.69 & 12.09 & 9.35 & 8.33 & 0.49 & 0.56 & 
0.60 & 0.10 & 0.07 & 0.06 \\ \hline
\multicolumn{13}{|l|}{\textbf{B. Light down-weighting in both the variable
selection and estimation/forecasting stages.}} \\ \hline
\multicolumn{13}{|l|}{OCMT uses down-weighted variables for selection as well
as for forecasting using Least Squares.} \\ \hline
\multicolumn{13}{|l|}{Remaining forecasts are based on Lasso, A-Lasso and
Boosting regressions applied to down-weighted data.} \\ \hline
& \multicolumn{11}{l}{\textit{\textbf{OCMT}}} &  \\ \hline
\textbf{20} & 77.30 & 67.14 & 62.56 & 1.54 & 2.46 & 3.36 & 0.30 & 0.44 & 0.55
& 0.02 & 0.03 & 0.06 \\ 
\textbf{40} & 68.17 & 63.51 & 78.15 & 1.48 & 2.68 & 4.02 & 0.25 & 0.40 & 0.52
& 0.01 & 0.03 & 0.05 \\ 
\textbf{100} & 71.87 & 68.17 & 71.51 & 1.53 & 3.31 & 5.57 & 0.21 & 0.34 & 
0.45 & 0.01 & 0.02 & 0.04 \\ \hline
& \multicolumn{11}{l}{\textit{\textbf{LASSO}}} &  \\ \hline
\textbf{20} & 80.00 & 67.40 & 62.93 & 6.20 & 6.21 & 6.52 & 0.55 & 0.60 & 0.64
& 0.20 & 0.19 & 0.20 \\ 
\textbf{40} & 72.11 & 65.23 & 77.42 & 9.52 & 9.74 & 9.83 & 0.52 & 0.58 & 0.62
& 0.19 & 0.19 & 0.18 \\ 
\textbf{100} & 78.28 & 73.66 & 72.59 & 18.48 & 19.90 & 19.93 & 0.49 & 0.55 & 
0.59 & 0.17 & 0.18 & 0.18 \\ \hline
& \multicolumn{11}{l}{\textit{\textbf{A-LASSO}}} &  \\ \hline
\textbf{20} & 83.11 & 68.96 & 64.09 & 4.79 & 4.82 & 5.09 & 0.46 & 0.50 & 0.55
& 0.15 & 0.14 & 0.15 \\ 
\textbf{40} & 76.73 & 69.71 & 79.78 & 7.53 & 7.76 & 7.89 & 0.44 & 0.50 & 0.55
& 0.14 & 0.14 & 0.14 \\ 
\textbf{100} & 88.96 & 81.22 & 79.27 & 14.68 & 16.02 & 16.24 & 0.43 & 0.49 & 
0.53 & 0.13 & 0.14 & 0.14 \\ \hline
& \multicolumn{11}{l}{\textit{\textbf{Boosting}}} &  \\ \hline
\textbf{20} & 82.60 & 72.71 & 66.02 & 5.41 & 6.23 & 7.08 & 0.53 & 0.61 & 0.67
& 0.16 & 0.19 & 0.22 \\ 
\textbf{40} & 80.49 & 74.27 & 92.06 & 10.25 & 11.94 & 13.79 & 0.54 & 0.63 & 
0.69 & 0.20 & 0.24 & 0.28 \\ 
\textbf{100} & 97.54 & 89.07 & 91.91 & 29.16 & 33.91 & 38.20 & 0.58 & 0.66 & 
0.71 & 0.27 & 0.31 & 0.35 \\ \hline\hline
\end{tabular}%
	\vspace{-0.2in}
	\begin{flushleft}
		\noindent 
		\scriptsize%
		\singlespacing%
		Notes: Light down-weighting is defined by by values $\lambda
=0.975,0.98,0.985,0.99,0.995,1$. For this set of exponential down-weighting
schemes we focus on simple average forecasts computed over the individual
forecasts obtained for each value of $\lambda $ in the set under
consideration. See notes to Table \ref{TNBS}.
	\end{flushleft}
\end{table}

\begin{table}

	\caption{MC results for methods using heavy down-weighting in the experiment with parameter instabilities, dynamics ($%
	\rho _{y} \neq 0$), and low fit.}
	\centering
 \vspace{0.2cm}
 \renewcommand{\arraystretch}{1.12}
 \setlength{\tabcolsep}{8pt}
	\scriptsize%
\begin{tabular}{|rrrrrrrrrrrrr|}
\hline\hline
& \multicolumn{3}{c}{\textbf{MSFE (}$\times 100$)} & \multicolumn{3}{|c}{$%
\hat{k}$} & \multicolumn{3}{|c}{\textbf{TPR}} & \multicolumn{3}{|c|}{\textbf{%
FPR}} \\ \hline
$N\backslash T$ & \textbf{100} & \textbf{150} & \textbf{200}
& \textbf{100} & \textbf{150} & \textbf{200} & \textbf{100} & \textbf{150} & 
\textbf{200} & \textbf{100} & \textbf{150} & \textbf{200} \\ \hline
\multicolumn{13}{|l|}{\textbf{A. Heavy down-weighting in the
estimation/forecasting stage only. }} \\ \hline
\multicolumn{13}{|l|}{Variable selection is based on original (not
down-weighted) data.} \\ \hline
\multicolumn{13}{|l|}{Forecasting stage is Least Squares on selected
down-weighted covariates for all methods} \\ \hline
& \multicolumn{11}{l}{\textit{\textbf{Oracle}}} &  \\ \hline
\textbf{20} & 77.37 & 67.65 & 62.80 & 4.00 & 4.00 & 4.00 & 1.00 & 1.00 & 1.00
& 0.00 & 0.00 & 0.00 \\ 
\textbf{40} & 69.77 & 64.11 & 76.89 & 4.00 & 4.00 & 4.00 & 1.00 & 1.00 & 1.00
& 0.00 & 0.00 & 0.00 \\ 
\textbf{100} & 72.54 & 68.20 & 70.34 & 4.00 & 4.00 & 4.00 & 1.00 & 1.00 & 
1.00 & 0.00 & 0.00 & 0.00 \\ \hline
& \multicolumn{11}{l}{\textit{\textbf{OCMT}}} &  \\ \hline
\textbf{20} & 78.00 & 67.66 & 62.51 & 1.78 & 2.74 & 3.63 & 0.38 & 0.58 & 0.74
& 0.01 & 0.02 & 0.03 \\ 
\textbf{40} & 68.28 & 63.77 & 80.08 & 1.49 & 2.52 & 3.34 & 0.33 & 0.55 & 0.71
& 0.00 & 0.01 & 0.01 \\ 
\textbf{100} & 71.14 & 68.87 & 70.25 & 1.23 & 2.13 & 2.82 & 0.27 & 0.47 & 
0.62 & 0.00 & 0.00 & 0.00 \\ \hline
& \multicolumn{11}{l}{\textit{\textbf{LASSO}}} &  \\ \hline
\textbf{20} & 85.84 & 73.54 & 69.28 & 5.86 & 6.27 & 6.82 & 0.59 & 0.67 & 0.74
& 0.17 & 0.18 & 0.19 \\ 
\textbf{40} & 82.61 & 73.75 & 89.00 & 7.94 & 8.49 & 8.73 & 0.55 & 0.65 & 0.72
& 0.14 & 0.15 & 0.15 \\ 
\textbf{100} & 92.80 & 87.68 & 86.51 & 11.54 & 11.95 & 11.84 & 0.50 & 0.61 & 
0.67 & 0.10 & 0.10 & 0.09 \\ \hline
& \multicolumn{11}{l}{\textit{\textbf{A-LASSO}}} &  \\ \hline
\textbf{20} & 85.01 & 72.33 & 68.99 & 4.28 & 4.71 & 5.09 & 0.47 & 0.54 & 0.61
& 0.12 & 0.13 & 0.13 \\ 
\textbf{40} & 80.01 & 71.15 & 87.53 & 6.12 & 6.56 & 6.86 & 0.45 & 0.55 & 0.62
& 0.11 & 0.11 & 0.11 \\ 
\textbf{100} & 89.74 & 83.20 & 84.42 & 9.11 & 9.64 & 9.77 & 0.43 & 0.53 & 
0.60 & 0.07 & 0.08 & 0.07 \\ \hline
& \multicolumn{11}{l}{\textit{\textbf{Boosting}}} &  \\ \hline
\textbf{20} & 81.73 & 69.53 & 66.63 & 3.59 & 3.57 & 3.69 & 0.50 & 0.56 & 0.62
& 0.08 & 0.07 & 0.06 \\ 
\textbf{40} & 78.95 & 69.26 & 82.99 & 5.34 & 5.04 & 4.86 & 0.49 & 0.56 & 0.62
& 0.08 & 0.07 & 0.06 \\ 
\textbf{100} & 100.44 & 83.45 & 82.31 & 12.09 & 9.35 & 8.33 & 0.49 & 0.56 & 
0.60 & 0.10 & 0.07 & 0.06 \\ \hline
\multicolumn{13}{|l|}{\textbf{B. Heavy down-weighting in both the variable
selection and estimation/forecasting stages.}} \\ \hline
\multicolumn{13}{|l|}{OCMT uses down-weighted variables for selection as well
as for forecasting using Least Squares.} \\ \hline
\multicolumn{13}{|l|}{Remaining forecasts are based on Lasso, A-Lasso and
Boosting regressions applied to down-weighted data.} \\ \hline
& \multicolumn{11}{l}{\textit{\textbf{OCMT}}} &  \\ \hline
\textbf{20} & 81.17 & 73.13 & 68.83 & 2.42 & 4.00 & 5.42 & 0.33 & 0.47 & 0.57
& 0.05 & 0.11 & 0.16 \\ 
\textbf{40} & 72.34 & 72.03 & 94.38 & 3.10 & 5.89 & 8.79 & 0.29 & 0.44 & 0.56
& 0.05 & 0.10 & 0.16 \\ 
\textbf{100} & 76.96 & 79.30 & 90.15 & 4.94 & 10.95 & 18.17 & 0.25 & 0.40 & 
0.52 & 0.04 & 0.09 & 0.16 \\ \hline
& \multicolumn{11}{l}{\textit{\textbf{LASSO}}} &  \\ \hline
\textbf{20} & 83.35 & 71.86 & 66.78 & 6.89 & 6.89 & 7.12 & 0.55 & 0.58 & 0.60
& 0.23 & 0.23 & 0.24 \\ 
\textbf{40} & 78.82 & 74.92 & 86.24 & 12.68 & 12.95 & 13.04 & 0.54 & 0.59 & 
0.61 & 0.26 & 0.27 & 0.26 \\ 
\textbf{100} & 84.81 & 82.65 & 81.18 & 23.76 & 30.22 & 30.01 & 0.50 & 0.58 & 
0.61 & 0.22 & 0.28 & 0.28 \\ \hline
& \multicolumn{11}{l}{\textit{\textbf{A-LASSO}}} &  \\ \hline
\textbf{20} & 86.88 & 74.40 & 68.83 & 5.35 & 5.35 & 5.51 & 0.46 & 0.48 & 0.51
& 0.18 & 0.17 & 0.17 \\ 
\textbf{40} & 83.60 & 80.66 & 88.42 & 10.07 & 10.21 & 10.35 & 0.46 & 0.50 & 
0.53 & 0.21 & 0.21 & 0.21 \\ 
\textbf{100} & 94.40 & 90.83 & 88.68 & 18.65 & 23.45 & 23.48 & 0.43 & 0.51 & 
0.54 & 0.17 & 0.21 & 0.21 \\ \hline
& \multicolumn{11}{l}{\textit{\textbf{Boosting}}} &  \\ \hline
\textbf{20} & 95.04 & 87.48 & 81.60 & 8.09 & 9.41 & 10.46 & 0.61 & 0.68 & 
0.74 & 0.28 & 0.33 & 0.38 \\ 
\textbf{40} & 98.94 & 97.85 & 116.43 & 18.06 & 20.44 & 22.35 & 0.65 & 0.73 & 
0.78 & 0.39 & 0.44 & 0.48 \\ 
\textbf{100} & 112.52 & 108.00 & 108.11 & 43.50 & 47.37 & 49.77 & 0.65 & 0.72
& 0.75 & 0.41 & 0.45 & 0.47 \\ \hline\hline
\end{tabular}%
	\vspace{-0.2in}
	\begin{flushleft}
		\noindent 
		\scriptsize%
		\singlespacing%
		Notes: Heavy down-weighting is defined by by values $\lambda
=0.95,0.96,0.97,0.98,0.99,1$. For this set of exponential down-weighting
schemes we focus on simple average forecasts computed over the individual
forecasts obtained for each value of $\lambda $ in the set under
consideration. See notes to Table \ref{TNBS}.
	\end{flushleft}
\end{table}

\begin{table}

	\vspace{0.1cm}
	\caption{MC results for methods using no down-weighting in the experiment with parameter instabilities, dynamics ($%
	\rho _{y} \neq 0$) and high fit.}
	\centering
 \vspace{0.2cm}
 \renewcommand{\arraystretch}{1.12}
	\scriptsize%
\begin{tabular}{|rrrrrrrrrrrrr|}
\hline\hline
& \multicolumn{3}{c}{\textbf{MSFE (}$\times 100$)} & \multicolumn{3}{|c}{$%
\hat{k}$} & \multicolumn{3}{|c}{\textbf{TPR}} & \multicolumn{3}{|c|}{\textbf{%
FPR}} \\ \hline
$N\backslash T$ & \textbf{100} & \textbf{150} & \textbf{200}
& \textbf{100} & \textbf{150} & \textbf{200} & \textbf{100} & \textbf{150} & 
\textbf{200} & \textbf{100} & \textbf{150} & \textbf{200} \\ \hline
\multicolumn{1}{|l}{} & \multicolumn{11}{l}{\textit{\textbf{Oracle}}} & 
\multicolumn{1}{l|}{} \\ \hline
\textbf{20} & 25.58 & 22.02 & 20.60 & 4.00 & 4.00 & 4.00 & 1.00 & 1.00 & 1.00
& 0.00 & 0.00 & 0.00 \\ 
\textbf{40} & 22.54 & 20.90 & 24.85 & 4.00 & 4.00 & 4.00 & 1.00 & 1.00 & 1.00
& 0.00 & 0.00 & 0.00 \\ 
\textbf{100} & 23.50 & 21.92 & 22.73 & 4.00 & 4.00 & 4.00 & 1.00 & 1.00 & 
1.00 & 0.00 & 0.00 & 0.00 \\ \hline
\multicolumn{1}{|l}{} & \multicolumn{11}{l}{\textit{\textbf{OCMT}}} &  \\ 
\hline
\textbf{20} & 25.99 & 22.31 & 20.92 & 3.81 & 4.81 & 5.59 & 0.78 & 0.91 & 0.97
& 0.04 & 0.06 & 0.09 \\ 
\textbf{40} & 23.23 & 21.32 & 25.18 & 3.53 & 4.62 & 5.33 & 0.74 & 0.90 & 0.96
& 0.01 & 0.03 & 0.04 \\ 
\textbf{100} & 23.77 & 22.17 & 23.10 & 3.16 & 4.29 & 4.88 & 0.68 & 0.87 & 
0.94 & 0.00 & 0.01 & 0.01 \\ \hline
\multicolumn{1}{|l}{} & \textit{\textbf{LASSO}} &  &  &  &  &  &  &  &  &  & 
&  \\ \hline
\textbf{20} & 27.25 & 22.84 & 21.46 & 7.44 & 7.86 & 8.27 & 0.78 & 0.84 & 0.89
& 0.22 & 0.22 & 0.24 \\ 
\textbf{40} & 24.27 & 22.07 & 25.90 & 10.31 & 10.78 & 11.25 & 0.76 & 0.83 & 
0.88 & 0.18 & 0.19 & 0.19 \\ 
\textbf{100} & 25.44 & 23.55 & 24.44 & 14.82 & 15.29 & 16.04 & 0.72 & 0.81 & 
0.85 & 0.12 & 0.12 & 0.13 \\ \hline
\multicolumn{1}{|l}{} & \multicolumn{11}{l}{\textit{\textbf{LASSO for
variable selection only. LS for estimation/forecasting.}}} &  \\ \hline
\textbf{20} & 29.07 & 23.73 & 21.98 & 7.44 & 7.86 & 8.27 & 0.78 & 0.84 & 0.89
& 0.22 & 0.22 & 0.24 \\ 
\textbf{40} & 26.88 & 24.34 & 27.30 & 10.31 & 10.78 & 11.25 & 0.76 & 0.83 & 
0.88 & 0.18 & 0.19 & 0.19 \\ 
\textbf{100} & 31.67 & 27.37 & 27.90 & 14.82 & 15.29 & 16.04 & 0.72 & 0.81 & 
0.85 & 0.12 & 0.12 & 0.13 \\ \hline
\multicolumn{1}{|l}{} & \multicolumn{11}{l}{\textit{\textbf{A-LASSO}}} &  \\ 
\hline
\textbf{20} & 28.39 & 23.34 & 21.74 & 5.62 & 6.00 & 6.40 & 0.66 & 0.73 & 0.79
& 0.15 & 0.15 & 0.16 \\ 
\textbf{40} & 25.89 & 23.34 & 26.86 & 7.95 & 8.49 & 8.85 & 0.66 & 0.75 & 0.81
& 0.13 & 0.14 & 0.14 \\ 
\textbf{100} & 29.38 & 25.86 & 27.03 & 11.63 & 12.29 & 13.03 & 0.65 & 0.75 & 
0.80 & 0.09 & 0.09 & 0.10 \\ \hline
& \multicolumn{11}{l}{\textit{\textbf{A-LASSO for variable selection only.
LS for estimation/forecasting.}}} &  \\ \hline
\textbf{20} & 29.04 & 23.73 & 21.98 & 5.62 & 6.00 & 6.40 & 0.66 & 0.73 & 0.79
& 0.15 & 0.15 & 0.16 \\ 
\textbf{40} & 26.92 & 24.16 & 27.53 & 7.95 & 8.49 & 8.85 & 0.66 & 0.75 & 0.81
& 0.13 & 0.14 & 0.14 \\ 
\textbf{100} & 31.03 & 26.76 & 28.10 & 11.63 & 12.29 & 13.03 & 0.65 & 0.75 & 
0.80 & 0.09 & 0.09 & 0.10 \\ \hline
\multicolumn{1}{|l}{} & \multicolumn{11}{l}{\textit{\textbf{Boosting}}} & 
\\ \hline
\textbf{20} & 26.93 & 23.25 & 20.92 & 4.64 & 4.60 & 4.73 & 0.69 & 0.75 & 0.81
& 0.09 & 0.08 & 0.07 \\ 
\textbf{40} & 25.08 & 21.93 & 25.69 & 6.68 & 6.36 & 6.15 & 0.69 & 0.77 & 0.82
& 0.10 & 0.08 & 0.07 \\ 
\textbf{100} & 27.04 & 22.67 & 23.43 & 13.67 & 10.97 & 10.07 & 0.69 & 0.76 & 
0.80 & 0.11 & 0.08 & 0.07 \\ \hline
\multicolumn{1}{|l}{} & \multicolumn{11}{l}{\textit{\textbf{Boosting for
variable selection only. LS for estimation/forecasting.}}} &  \\ \hline
\textbf{20} & 27.66 & 22.92 & 21.24 & 4.64 & 4.60 & 4.73 & 0.69 & 0.75 & 0.81
& 0.09 & 0.08 & 0.07 \\ 
\textbf{40} & 25.99 & 23.44 & 26.04 & 6.68 & 6.36 & 6.15 & 0.69 & 0.77 & 0.82
& 0.10 & 0.08 & 0.07 \\ 
\textbf{100} & 32.44 & 26.25 & 26.55 & 13.67 & 10.97 & 10.07 & 0.69 & 0.76 & 
0.80 & 0.11 & 0.08 & 0.07 \\ \hline\hline
\end{tabular}%
	\vspace{-0.2in}
	\begin{flushleft}
		\noindent 
		\scriptsize%
		\singlespacing%
Notes: See notes to Table \ref{TNBS}.
	\end{flushleft}
\end{table}

\begin{table}

	\caption{MC results for methods using light down-weighting in the experiment with parameter instabilities, dynamics ($%
	\rho _{y} \neq 0$), and high fit.}
	\centering
 \vspace{0.2cm}
 \renewcommand{\arraystretch}{1.12}
 \setlength{\tabcolsep}{8pt}
	\scriptsize%
\begin{tabular}{|rrrrrrrrrrrrr|}
\hline\hline
& \multicolumn{3}{c}{\textbf{MSFE (}$\times 100$)} & \multicolumn{3}{|c}{$%
\hat{k}$} & \multicolumn{3}{|c}{\textbf{TPR}} & \multicolumn{3}{|c|}{\textbf{%
FPR}} \\ \hline
$N\backslash T$ & \textbf{100} & \textbf{150} & \textbf{200}
& \textbf{100} & \textbf{150} & \textbf{200} & \textbf{100} & \textbf{150} & 
\textbf{200} & \textbf{100} & \textbf{150} & \textbf{200} \\ \hline
\multicolumn{13}{|l|}{\textbf{A. Light down-weighting in the
estimation/forecasting stage only. }} \\ \hline
\multicolumn{13}{|l|}{Variable selection is based on original (not
down-weighted) data.} \\ \hline
\multicolumn{13}{|l|}{Forecasting stage is Least Squares on selected
down-weighted covariates for all methods} \\ \hline
& \multicolumn{11}{l}{\textit{\textbf{Oracle}}} &  \\ \hline
\textbf{20} & 24.08 & 20.43 & 18.65 & 4.00 & 4.00 & 4.00 & 1.00 & 1.00 & 1.00
& 0.00 & 0.00 & 0.00 \\ 
\textbf{40} & 21.27 & 19.26 & 22.93 & 4.00 & 4.00 & 4.00 & 1.00 & 1.00 & 1.00
& 0.00 & 0.00 & 0.00 \\ 
\textbf{100} & 22.39 & 20.27 & 20.52 & 4.00 & 4.00 & 4.00 & 1.00 & 1.00 & 
1.00 & 0.00 & 0.00 & 0.00 \\ \hline
& \multicolumn{11}{l}{\textit{\textbf{OCMT}}} &  \\ \hline
\textbf{20} & 24.65 & 20.87 & 19.12 & 3.81 & 4.81 & 5.59 & 0.78 & 0.91 & 0.97
& 0.04 & 0.06 & 0.09 \\ 
\textbf{40} & 21.99 & 19.72 & 23.74 & 3.53 & 4.62 & 5.33 & 0.74 & 0.90 & 0.96
& 0.01 & 0.03 & 0.04 \\ 
\textbf{100} & 22.68 & 20.72 & 21.13 & 3.16 & 4.29 & 4.88 & 0.68 & 0.87 & 
0.94 & 0.00 & 0.01 & 0.01 \\ \hline
& \multicolumn{11}{l}{\textit{\textbf{LASSO}}} &  \\ \hline
\textbf{20} & 28.32 & 22.59 & 20.37 & 7.44 & 7.86 & 8.27 & 0.78 & 0.84 & 0.89
& 0.22 & 0.22 & 0.24 \\ 
\textbf{40} & 26.12 & 23.33 & 27.08 & 10.31 & 10.78 & 11.25 & 0.76 & 0.83 & 
0.88 & 0.18 & 0.19 & 0.19 \\ 
\textbf{100} & 31.41 & 27.03 & 26.92 & 14.82 & 15.29 & 16.04 & 0.72 & 0.81 & 
0.85 & 0.12 & 0.12 & 0.13 \\ \hline
& \multicolumn{11}{l}{\textit{\textbf{A-LASSO}}} &  \\ \hline
\textbf{20} & 28.11 & 22.74 & 20.49 & 5.62 & 6.00 & 6.40 & 0.66 & 0.73 & 0.79
& 0.15 & 0.15 & 0.16 \\ 
\textbf{40} & 25.90 & 22.88 & 27.02 & 7.95 & 8.49 & 8.85 & 0.66 & 0.75 & 0.81
& 0.13 & 0.14 & 0.14 \\ 
\textbf{100} & 30.62 & 26.33 & 27.09 & 11.63 & 12.29 & 13.03 & 0.65 & 0.75 & 
0.80 & 0.09 & 0.09 & 0.10 \\ \hline
& \multicolumn{11}{l}{\textit{\textbf{Boosting}}} &  \\ \hline
\textbf{20} & 26.24 & 22.00 & 19.70 & 4.64 & 4.60 & 4.73 & 0.69 & 0.75 & 0.81
& 0.09 & 0.08 & 0.07 \\ 
\textbf{40} & 25.47 & 22.50 & 25.02 & 6.68 & 6.36 & 6.15 & 0.69 & 0.77 & 0.82
& 0.10 & 0.08 & 0.07 \\ 
\textbf{100} & 32.79 & 25.84 & 25.02 & 13.67 & 10.97 & 10.07 & 0.69 & 0.76 & 
0.80 & 0.11 & 0.08 & 0.07 \\ \hline
\multicolumn{13}{|l|}{\textbf{B. Light down-weighting in both the variable
selection and estimation/forecasting stages.}} \\ \hline
\multicolumn{13}{|l|}{OCMT uses down-weighted variables for selection as well
as for forecasting using Least Squares.} \\ \hline
\multicolumn{13}{|l|}{Remaining forecasts are based on Lasso, A-Lasso and
Boosting regressions applied to down-weighted data.} \\ \hline
& \multicolumn{11}{l}{\textit{\textbf{OCMT}}} &  \\ \hline
\textbf{20} & 25.28 & 21.47 & 19.45 & 2.94 & 3.87 & 4.69 & 0.60 & 0.72 & 0.79
& 0.03 & 0.05 & 0.08 \\ 
\textbf{40} & 22.89 & 20.51 & 24.46 & 2.85 & 4.11 & 5.42 & 0.56 & 0.69 & 0.78
& 0.02 & 0.03 & 0.06 \\ 
\textbf{100} & 23.84 & 22.00 & 22.60 & 2.80 & 4.74 & 6.96 & 0.50 & 0.64 & 
0.73 & 0.01 & 0.02 & 0.04 \\ \hline
& \multicolumn{11}{l}{\textit{\textbf{LASSO}}} &  \\ \hline
\textbf{20} & 26.92 & 22.17 & 20.03 & 7.75 & 7.92 & 8.29 & 0.74 & 0.81 & 0.85
& 0.24 & 0.23 & 0.24 \\ 
\textbf{40} & 24.21 & 21.70 & 24.81 & 11.75 & 11.97 & 12.12 & 0.73 & 0.80 & 
0.83 & 0.22 & 0.22 & 0.22 \\ 
\textbf{100} & 26.38 & 24.51 & 23.75 & 21.23 & 22.70 & 22.92 & 0.70 & 0.76 & 
0.81 & 0.18 & 0.20 & 0.20 \\ \hline
& \multicolumn{11}{l}{\textit{\textbf{A-LASSO}}} &  \\ \hline
\textbf{20} & 27.60 & 22.40 & 20.14 & 6.04 & 6.22 & 6.54 & 0.65 & 0.72 & 0.77
& 0.17 & 0.17 & 0.17 \\ 
\textbf{40} & 25.56 & 22.79 & 25.50 & 9.30 & 9.61 & 9.79 & 0.64 & 0.73 & 0.77
& 0.17 & 0.17 & 0.17 \\ 
\textbf{100} & 29.97 & 26.82 & 26.00 & 16.80 & 18.34 & 18.73 & 0.63 & 0.71 & 
0.76 & 0.14 & 0.16 & 0.16 \\ \hline
& \multicolumn{11}{l}{\textit{\textbf{Boosting}}} &  \\ \hline
\textbf{20} & 28.49 & 25.18 & 22.34 & 6.21 & 7.05 & 7.85 & 0.71 & 0.79 & 0.84
& 0.17 & 0.19 & 0.22 \\ 
\textbf{40} & 28.51 & 26.06 & 30.23 & 11.12 & 12.68 & 14.39 & 0.72 & 0.80 & 
0.85 & 0.21 & 0.24 & 0.27 \\ 
\textbf{100} & 34.04 & 30.77 & 31.27 & 29.46 & 34.10 & 38.36 & 0.73 & 0.81 & 
0.85 & 0.27 & 0.31 & 0.35 \\ \hline\hline
\end{tabular}%
	\vspace{-0.2in}
	\begin{flushleft}
		\noindent 
		\scriptsize%
		\singlespacing%
		Notes: Light down-weighting is defined by by values $\lambda
=0.975,0.98,0.985,0.99,0.995,1$. For this set of exponential down-weighting
schemes we focus on simple average forecasts computed over the individual
forecasts obtained for each value of $\lambda $ in the set under
consideration. See notes to Table \ref{TNBS}.
	\end{flushleft}
\end{table}

\begin{table}

	\caption{MC results for methods using heavy down-weighting in the experiment with parameter instabilities, dynamics ($%
	\rho _{y} \neq 0$), and high fit.}
	\centering
 \vspace{0.2cm}
 \renewcommand{\arraystretch}{1.12}
 \setlength{\tabcolsep}{8pt}
	\scriptsize%
\begin{tabular}{|rrrrrrrrrrrrr|}
\hline\hline
& \multicolumn{3}{c}{\textbf{MSFE (}$\times 100$)} & \multicolumn{3}{|c}{$%
\hat{k}$} & \multicolumn{3}{|c}{\textbf{TPR}} & \multicolumn{3}{|c|}{\textbf{%
FPR}} \\ \hline
$N\backslash T$ & \textbf{100} & \textbf{150} & \textbf{200}
& \textbf{100} & \textbf{150} & \textbf{200} & \textbf{100} & \textbf{150} & 
\textbf{200} & \textbf{100} & \textbf{150} & \textbf{200} \\ \hline
\multicolumn{13}{|l|}{\textbf{A. Heavy down-weighting in the
estimation/forecasting stage only. }} \\ \hline
\multicolumn{13}{|l|}{Variable selection is based on original (not
down-weighted) data.} \\ \hline
\multicolumn{13}{|l|}{Forecasting stage is Least Squares on selected
down-weighted covariates for all methods} \\ \hline
& \multicolumn{11}{l}{\textit{\textbf{Oracle}}} &  \\ \hline
\textbf{20} & 24.13 & 20.79 & 19.13 & 4.00 & 4.00 & 4.00 & 1.00 & 1.00 & 1.00
& 0.00 & 0.00 & 0.00 \\ 
\textbf{40} & 21.47 & 19.80 & 23.45 & 4.00 & 4.00 & 4.00 & 1.00 & 1.00 & 1.00
& 0.00 & 0.00 & 0.00 \\ 
\textbf{100} & 22.57 & 20.80 & 21.32 & 4.00 & 4.00 & 4.00 & 1.00 & 1.00 & 
1.00 & 0.00 & 0.00 & 0.00 \\ \hline
& \multicolumn{11}{l}{\textit{\textbf{OCMT}}} &  \\ \hline
\textbf{20} & 24.87 & 21.25 & 19.91 & 3.81 & 4.81 & 5.59 & 0.78 & 0.91 & 0.97
& 0.04 & 0.06 & 0.09 \\ 
\textbf{40} & 22.13 & 20.19 & 25.07 & 3.53 & 4.62 & 5.33 & 0.74 & 0.90 & 0.96
& 0.01 & 0.03 & 0.04 \\ 
\textbf{100} & 22.82 & 21.40 & 22.20 & 3.16 & 4.29 & 4.88 & 0.68 & 0.87 & 
0.94 & 0.00 & 0.01 & 0.01 \\ \hline
& \multicolumn{11}{l}{\textit{\textbf{LASSO}}} &  \\ \hline
\textbf{20} & 28.82 & 23.33 & 21.38 & 7.44 & 7.86 & 8.27 & 0.78 & 0.84 & 0.89
& 0.22 & 0.22 & 0.24 \\ 
\textbf{40} & 26.91 & 24.49 & 29.04 & 10.31 & 10.78 & 11.25 & 0.76 & 0.83 & 
0.88 & 0.18 & 0.19 & 0.19 \\ 
\textbf{100} & 32.69 & 29.08 & 29.69 & 14.82 & 15.29 & 16.04 & 0.72 & 0.81 & 
0.85 & 0.12 & 0.12 & 0.13 \\ \hline
& \multicolumn{11}{l}{\textit{\textbf{A-LASSO}}} &  \\ \hline
\textbf{20} & 28.36 & 23.36 & 21.30 & 5.62 & 6.00 & 6.40 & 0.66 & 0.73 & 0.79
& 0.15 & 0.15 & 0.16 \\ 
\textbf{40} & 26.14 & 23.70 & 28.68 & 7.95 & 8.49 & 8.85 & 0.66 & 0.75 & 0.81
& 0.13 & 0.14 & 0.14 \\ 
\textbf{100} & 31.35 & 27.64 & 29.03 & 11.63 & 12.29 & 13.03 & 0.65 & 0.75 & 
0.80 & 0.09 & 0.09 & 0.10 \\ \hline
& \multicolumn{11}{l}{\textit{\textbf{Boosting}}} &  \\ \hline
\textbf{20} & 26.09 & 22.48 & 20.53 & 4.64 & 4.60 & 4.73 & 0.69 & 0.75 & 0.81
& 0.09 & 0.08 & 0.07 \\ 
\textbf{40} & 26.12 & 23.19 & 26.15 & 6.68 & 6.36 & 6.15 & 0.69 & 0.77 & 0.82
& 0.10 & 0.08 & 0.07 \\ 
\textbf{100} & 34.44 & 27.12 & 26.22 & 13.67 & 10.97 & 10.07 & 0.69 & 0.76 & 
0.80 & 0.11 & 0.08 & 0.07 \\ \hline
\multicolumn{13}{|l|}{\textbf{B. Heavy down-weighting in both the variable
selection and estimation/forecasting stages.}} \\ \hline
\multicolumn{13}{|l|}{OCMT uses down-weighted variables for selection as well
as for forecasting using Least Squares.} \\ \hline
\multicolumn{13}{|l|}{Remaining forecasts are based on Lasso, A-Lasso and
Boosting regressions applied to down-weighted data.} \\ \hline
& \multicolumn{11}{l}{\textit{\textbf{OCMT}}} &  \\ \hline
\textbf{20} & 26.17 & 23.19 & 21.60 & 3.63 & 5.24 & 6.62 & 0.58 & 0.70 & 0.78
& 0.07 & 0.12 & 0.18 \\ 
\textbf{40} & 24.46 & 23.54 & 27.79 & 4.32 & 7.20 & 10.10 & 0.54 & 0.68 & 
0.77 & 0.05 & 0.11 & 0.18 \\ 
\textbf{100} & 26.35 & 26.29 & 28.69 & 6.02 & 12.21 & 19.55 & 0.49 & 0.64 & 
0.74 & 0.04 & 0.10 & 0.17 \\ \hline
& \multicolumn{11}{l}{\textit{\textbf{LASSO}}} &  \\ \hline
\textbf{20} & 27.66 & 23.35 & 21.23 & 8.45 & 8.56 & 8.84 & 0.74 & 0.78 & 0.81
& 0.28 & 0.27 & 0.28 \\ 
\textbf{40} & 25.94 & 24.49 & 27.57 & 14.58 & 14.99 & 15.15 & 0.73 & 0.79 & 
0.81 & 0.29 & 0.30 & 0.30 \\ 
\textbf{100} & 28.26 & 26.92 & 25.96 & 25.44 & 31.60 & 32.20 & 0.69 & 0.77 & 
0.81 & 0.23 & 0.29 & 0.29 \\ \hline
& \multicolumn{11}{l}{\textit{\textbf{A-LASSO}}} &  \\ \hline
\textbf{20} & 28.30 & 23.82 & 21.80 & 6.59 & 6.70 & 6.91 & 0.64 & 0.69 & 0.73
& 0.20 & 0.20 & 0.20 \\ 
\textbf{40} & 27.65 & 25.96 & 28.30 & 11.58 & 11.91 & 12.10 & 0.65 & 0.72 & 
0.75 & 0.22 & 0.23 & 0.23 \\ 
\textbf{100} & 31.82 & 29.35 & 28.16 & 19.98 & 24.54 & 25.22 & 0.62 & 0.70 & 
0.75 & 0.18 & 0.22 & 0.22 \\ \hline
& \multicolumn{11}{l}{\textit{\textbf{Boosting}}} &  \\ \hline
\textbf{20} & 32.68 & 30.40 & 27.62 & 8.68 & 10.06 & 11.04 & 0.75 & 0.83 & 
0.87 & 0.28 & 0.34 & 0.38 \\ 
\textbf{40} & 34.60 & 33.67 & 38.51 & 18.48 & 20.88 & 22.76 & 0.78 & 0.86 & 
0.89 & 0.38 & 0.44 & 0.48 \\ 
\textbf{100} & 39.09 & 36.94 & 36.77 & 43.75 & 47.53 & 49.90 & 0.78 & 0.84 & 
0.87 & 0.41 & 0.44 & 0.46 \\ \hline\hline
\end{tabular}%
	\vspace{-0.2in}
	\begin{flushleft}
		\noindent 
		\scriptsize%
		\singlespacing%
		Notes: Heavy down-weighting is defined by by values $\lambda
=0.95,0.96,0.97,0.98,0.99,1$. For this set of exponential down-weighting
schemes we focus on simple average forecasts computed over the individual
forecasts obtained for each value of $\lambda $ in the set under
consideration. See notes to Table \ref{TNBS}.
	\end{flushleft}
\end{table}

\clearpage\newpage

\setcounter{section}{0} \renewcommand{\thesection}{S-\arabic{section}}

\setcounter{table}{0} \renewcommand{\thetable}{S.\arabic{table}}

\setcounter{page}{1} \renewcommand{\thepage}{S.\arabic{page}}

\setcounter{equation}{0} \renewcommand{\theequation}{S.\arabic{equation}}

\vspace{0.05in}

\begin{center}
\textbf{\ {\large Online Empirical Supplement to} }\\[0pt]

\textbf{{\large {``Variable Selection in High Dimensional Linear Regressions
with Parameter Instability''}}} \\[0pt]

Alexander Chudik

Federal Reserve Bank of Dallas\bigskip

M. Hashem Pesaran

University of Southern California, USA and Trinity College, Cambridge,
UK\bigskip

Mahrad Sharifvaghefi

University of Pittsburgh\\[0pt]
\bigskip \bigskip

\today\bigskip
\end{center}


\noindent This online empirical supplement has three sections. Section \ref%
{Appendix B} provides the full list and description of technical indicators
considered in the stock market application. Section \ref{GVAR_AS} provides
the list of variables in the conditioning and active sets in the application
on forecasting output growth rates across 33 countries. Last section focuses
on the third application, forecasting euro area quarterly output growth
using the European Central Bank (ECB) survey of professional forecasters.
The section starts with description of the data and then discusses the
results.


\section{Technical and financial indicators}

\label{Appendix B}

\noindent Our choice of the technical trading indicators is based on the
extensive literature on system trading, reviewed by \cite{wilder1978new} and 
\cite{kaufman2020trading}. Most of the technical indicators are based on
historical daily high, low and adjusted close prices, which we denote by $%
H_{it}(\tau )$, $L_{it}(\tau )$, and $P_{it}(\tau )$, respectively. These
prices refer to stock $i$ in month $t$, for day $\tau .$ Moreover, let $%
D_{t}^{i}$ be the number of trading days, and denote by $D_{l_{t}}^{i}$ the
last trading day of stock $i$ in month $t$. For each stock $i$, monthly
high, low and close prices are set to the last trading day of the month,
namely $H_{it}(D_{l_{t}}^{i})$, $L_{it}(D_{l_{t}}^{i})$ and $%
P_{it}(D_{l_{t}}^{i})$, or $H_{it}$, $L_{it}$, and $P_{it},$ for simplicity.
The logarithms of these are denoted by $h_{it}$, $l_{it}$, and $p_{it}$,
respectively.

The $28$ stocks considered in our study are allocated to 19 sectoral groups
according to Industry Classification Benchmark.\footnote{%
The 19 groups are as follows: Oil \& Gas, Chemicals, Basic Resources,
Construction \& Materials, Industrial Goods \& Services, Automobiles \&
Parts, Food \& Beverage, Personal \& Household Goods, Health Care, Retail,
Media, Travel \& Leisure, Telecommunications, Utilities, Banks, Insurance,
Real Estate, Financial Services, and Technology.} The group membership of
stock $i$ is denoted by the set $\mathfrak{g}_{i}$, which includes all S\&P
500 stocks in stock $i^{th}$ group, and $|\mathfrak{g}_{i}|$ is the number
of stocks in the group.

The technical and financial indicators considered are:

\begin{enumerate}
\item Return of Stock $i\,\ $($r_{it}$): $r_{it}=100(p_{it}-p_{i,t-1}). $

\item The Group Average Return of Stock $i$ ($\bar{r}_{it}^{g}$): $\bar{r}
_{it}^{g}=|\mathfrak{g}_{i}|^{-1}\sum_{j\in \mathfrak{g}_{i}}r_{jt}. $

\item Moving Average Stock Return of order $s$ ($mar_{it}(s)$): This
indicator, which is also known as s-day momentum (see, for example, %
\citealp{kaufman2020trading}), is defined as 
\begin{equation*}
mar_{it}(s)=\text{MA}(r_{it},s),
\end{equation*}
where $\text{MA}(x_{it},s)$ is Moving Average of a time-series process $%
x_{it}$ with degree of smoothness $s$ which can be written as 
\begin{equation*}
\text{MA}(x_{it},s)=s^{-1}\sum_{\ell =1}^{s}x_{i,t-\ell }.
\end{equation*}

\item Return Gap ($gr_{it}(s)$): This indicator represents a belief in mean
reversion that prices will eventually return to their means (for further
details see \citealp{kaufman2020trading}). 
\begin{equation*}
gr_{it}(s)=r_{it}-\text{MA}(r_{it},s).
\end{equation*}

\item Price Gap ($gp_{it}(s)$): $gp_{it}(s)=100\left[ p_{it}-\text{MA}
(p_{it},s)\right] . $

\item Realized Volatility ($RV_{it}$): $RV_{it}=\sqrt{\sum_{\tau
=1}^{D_{t}^{i}}\left( R_{it}(\tau )-\bar{R}_{it}\right) ^{2}}$, where 
\begin{equation*}
R_{it}(\tau )=100\left[ P_{it}(\tau )/P_{it}(\tau -1) -1 \right], \text{ and 
} \bar{R}_{it}=\sum_{\tau =1}^{D_{t}^{i}}R_{it}(\tau )/D_{t}^{i}.
\end{equation*}

\item Group Realized Volatility ($RV_{it}^{g}$): $RV_{it}^{g}=\sqrt{| 
\mathfrak{g}|^{-1}\sum_{i\in \mathfrak{g}}RV_{it}^{2}}. $

\item Moving Average Realized Volatility ($mav_{it}(s)$): \textquotedblleft
Signals are generated when a price change is accompanied by an unusually
large move relative to average volatility\textquotedblright\ %
\citep{kaufman2020trading}. The following two indicators are constructed to
capture such signals 
\begin{equation*}
mav_{it}(s)=\text{MA}(RV_{it},s)
\end{equation*}

\item Realized Volatility Gap ($RVG_{it}(s)$): $RVG_{it}(s)=RV_{it}-\text{MA}
(RV_{it},s) $

\item Percent Price Oscillator($PPO_{it}(s_{1},s_{2})$): 
\begin{equation*}
PPO_{it}(s_{1},s_{2})=100\left( \frac{\text{MA}(P_{it},s_{1})-\text{MA}
(P_{it},s_{2})}{\text{MA}(P_{it},s_{2})}\right), \text{ where } s_{1}<s_{2}.
\end{equation*}

\item Relative Strength Indicator ($RSI_{it}^{s}$): This is a price momentum
indicator developed by \cite{wilder1978new} to capture overbought and
oversold conditions. Let 
\begin{equation*}
\Delta P_{it}^{+}=\Delta P_{it}\text{I}_{\Delta P_{it}>0}(\Delta P_{it}),\ 
\text{and}\ \Delta P_{it}^{-}=\Delta P_{it}\text{I}_{\Delta P_{it}\leq
0}(\Delta P_{it}),
\end{equation*}
where $\Delta P_{it} = P_{it} - P_{i,t-1}$ and $\text{I}_{A}(x_{it})$ is an
indicator function that take a value of one if $x_{it}\in A$ and zero
otherwise. Then 
\begin{equation*}
RS_{it}^{s}=-\frac{\text{MA}(\Delta P_{it}^{+},s)}{\text{MA}(\Delta
P_{it}^{-},s)}, \text{ and } RSI_{it}^{s}=100\left( 1-\frac{1}{1+RS_{it}^{s}}
\right) .
\end{equation*}

\item Williams R ($WILLR_{it}(s)$): This indicator proposed by \cite%
{williams1979how} to measure buying and selling pressure. 
\begin{equation*}
WILLR_{it}(s)=-100\left( \frac{\underset{j\in \{1,\cdots ,s\}}{\max }
(h_{i,t-s+j})-p_{it}}{\underset{j\in \{1,\cdots ,s\}}{\max }(h_{i,t-s+j})- 
\underset{j\in \{1,\cdots ,s\}}{\min }(l_{i,t-s+i})}\right) .
\end{equation*}

\item Average Directional Movement Index ($ADX_{it}(s)$): This is a filtered
momentum indicator by \cite{wilder1978new}. To compute $ADX_{it}(s)$, we
first calculate up-ward directional movement ($DM_{it}^{+}$), down-ward
directional movement ($DM_{it}^{-}$), and true range ($TR_{it}$) as: 
\begin{equation*}
DM_{it}^{+}= 
\begin{cases}
h_{it}-h_{i,t-1}, & \text{if}\ h_{it}-h_{i,t-1}>0\ \text{and}\
h_{it}-h_{i,t-1}>l_{i,t-1}-l_{it}, \\ 
0, & \text{otherwise}.%
\end{cases}%
\end{equation*}
\begin{equation*}
DM_{it}^{-}= 
\begin{cases}
l_{i,t-1}-l_{it}, & if\ l_{i,t-1}-l_{it}>0\ and\
l_{i,t-1}-l_{it}>h_{it}-h_{i,t-1}, \\ 
0, & otherwise.%
\end{cases}%
\end{equation*}
\begin{equation*}
TR_{it}=\max \{h_{it}-l_{it},|h_{it}-p_{i,t-1}|,|p_{i,t-1}-l_{it}|\}.
\end{equation*}
Then, positive and negative directional indexes denoted by $ID_{it}^{+}(s)$
and $ID_{it}^{-}(s)$ respectively, are computed by 
\begin{equation*}
ID_{it}^{+}(s)=100\left( \frac{\text{MA}(DM_{it}^{+},s)}{\text{MA}%
(TR_{it},s) }\right) , \text{ and } ID_{it}^{-}(s)=100\left( \frac{\text{MA}
(DM_{it}^{-},s)}{\text{MA}(TR_{it},s)}\right) ,
\end{equation*}
Finally, directional index $DX_{it}(s)$ and $ADX_{it}(s)$ are computed as 
\begin{equation*}
DX_{it}(s)=100\left( \frac{|ID_{it}^{+}(s)-ID_{it}^{-}(s)|}{
ID_{it}^{+}(s)+ID_{it}^{-}(s)}\right) , \text{ and } ADX_{it}(s)=\text{MA}
(DX_{it}(s),s).
\end{equation*}

\item Percentage Change in Kaufman's Adaptive Moving Average ($\Delta
KAMA_{it}(s_{1},s_{2},m)$): Kaufman's Adaptive Moving Average accounts for
market noise or volatility. To compute $\Delta KAMA_{it}(s_{1},s_{2},m)$, we
first need to calculate the Efficiency Ratio ($ER_{it}$) defined by 
\begin{equation*}
ER_{it}=100\left( \frac{|p_{it}-p_{i,t-m}|}{\sum_{j=1}^{m}|\Delta
P_{i,t-m+j}|}\right) ,
\end{equation*}
where $\Delta P_{it} = P_{it} - P_{i,t-1}$, and then calculate the Smoothing
Constant ($SC_{it}$) which is 
\begin{equation*}
SC_{it}=\left[ ER_{it}\left(\frac{2}{s_{1}+1}-\frac{2}{s_{2}+1}\right)+\frac{
2}{s_{2}+1}\right] ^{2},
\end{equation*}
where $s_{1}<m<s_{2}$. Then, Kaufman's Adaptive Moving Average is computed
as 
\begin{equation*}
\text{KAMA}(P_{it},s_{1},s_{2},m)=SC_{it}P_{it}+(1-SC_{it})\text{KAMA}
(P_{i,t-1},s_{1},s_{2},m)
\end{equation*}
where 
\begin{equation*}
\text{KAMA}(P_{is_{2}},s_{1},s_{2},m)=s_{2}^{-1}\sum_{\kappa
=1}^{s_{2}}P_{i\kappa }.
\end{equation*}
The Percentage Change in Kaufman's Adaptive Moving Average is then computed
as 
\begin{equation*}
\Delta KAMA_{it}(s_{1},s_{2},m)=100\left( \frac{\text{KAMA}
(P_{it},s_{1},s_{2},m)-\text{KAMA}(P_{i,t-1},s_{1},s_{2},m)}{\text{KAMA}
(P_{i,t-1},s_{1},s_{2},m)}\right) .
\end{equation*}
For further details see \cite{kaufman2020trading}.
\end{enumerate}

\subsection*{Other financial indicators}

In addition to the above technical indicators, we also make use of daily
prices of Brent Crude Oil, S\&P 500 index, monthly series on Fama and French
market factors, and annualized percentage yield on 3-month, 2-year and
10-year US government bonds. Based on this data, we have constructed the
following variables. These series are denoted by $PO_{t}$ and $P_{sp,t}$
respectively, and their logs by $po_{t}$ and $p_{sp,t}$. The list of
additional variables are:

\begin{enumerate}
\item Return of S\&P 500 index ($r_{sp,t}$): $%
r_{sp,t}=100(p_{sp,t}-p_{sp,t-1})$, where $p_{sp,t}$ is the log of $S\&P$ $%
500$ index at the end of month $t$.

\item Realized Volatility of S\&P 500 index ($RV_{sp,t}$): 
\begin{equation*}
RV_{sp,t}=\sqrt{\sum_{\tau =1}^{D_{t}^{sp}}\left( R_{sp,t}(\tau )-\bar{R}
_{sp,t}\right) ^{2}},
\end{equation*}
where $\bar{R}_{sp,t}=\sum_{\tau =1}^{D_{t}^{sp}}R_{it}(\tau )/D_{t}^{sp}$, $%
\ R_{sp,t}(\tau )=100(\left[ P_{sp,t}(\tau )/P_{sp,t}(\tau -1)-1\right] $, $%
P_{sp,t}(\tau )$ is the $S\&P$ $500$ price index at close of day $\tau $ of
month $t$, and $D_{t}^{sp}$ is the number of days in month $t$.

\item Percent Rate of Change in Oil Prices ($\Delta po_{t}$): $\Delta
po_{t}=100(po_{t}-po_{t-1}), $ where $po_{t}$ is the log of oil princes at
the close of month $t$.

\item Long Term Interest Rate Spread ($LIRS_{t}$): The difference between
annualized percentage yield on 10-year and 3-month US government bonds.

\item Medium Term Interest Rate Spread ($MIRS_{t}$): The difference between
annualized percentage yield on 10-year and 2-year US government bonds.

\item Short Term Interest Rate Spread ($SIRS_{t}$): The difference between
annualized percentage yield on 2-year and 3-month US government bonds.

\item Small Minus Big Factor ($SMB_{t}$): Fama and French Small Minus Big
market factor.

\item High Minus Low Factor ($HML_{t}$): Fama and French High Minus Low
market factor.
\end{enumerate}

A summary of the covariates in the active set used for prediction of monthly
stock returns is given in Table \ref{active_set_tab}.

\begin{table}[tbph]
\caption{Active set for percentage change in equity price forecasting}
\label{active_set_tab}
{\small \ \centering
\begin{tabular}{ll}
\hline\hline
Target Variable: & $r_{it+1} $ (one-month ahead percentage change in equity
price of stock $i $) \\ \hline
A. Financial Variables: & $r_{it} $, $\bar{r}_{it}^g $, $r_{sp,t} $, $%
RV_{it} $, $RV_{it}^g $, $RV_{sp,t} $, $SMB_t $, $HML_t $. \\ 
&  \\ 
B. Economic Variables: & $\Delta po_t $, $LIRS_t - LIRS_{t-1}$, $MIRS_t -
MIRS_{t-1}$, $SIRS_t - SIRS_{t-1}$. \\ 
&  \\ 
C. Technical Indicators: & $mar_{it}^s $ for $s = \{3,6,12\} $, $mav_{it}^s $
for $s = \{3,6,12\} $, $gr_{it}^s $ for $s = \{3,6,12\} $, \\ 
& $gp_{it}^s $ for $s = \{3,6,12\} $, $RVG_{it}^s $ for $s = \{3,6,12\} $, $%
RSI_{it}^s $ for $s = \{3,6,12\} $, \\ 
& $ADX_{it}^s $ for $s = \{3,6,12\} $, $WILLR_{it}^s $ for $s = \{3,6,12\} $,
\\ 
& $PPO_{it}(s_1,s_2) $ for $(s_1,s_2) = \{ (3,6), (6,12), (3, 12) \} $, \\ 
& $\Delta KAMA_{it}(s_1,s_2,m) $ for $(s_1, s_2, m) = (2,12,6) $. \\ \hline
\end{tabular}
}
\end{table}

\section{List of variables used for forecasting output growth}

\label{GVAR_AS}

Variables in the conditioning and active sets for forecasting output growth
across 33 countries are listed in Table \ref{countries gdp growth active set}
below.

\begin{table}[t]
\caption{{\protect\small List of variables in the conditioning and active
sets for forecasting quarterly output growth across 33 countries}}
\label{countries gdp growth active set}
\vspace{0.2cm} 
\par
\centering\renewcommand{\arraystretch}{1.4}
\par
{\footnotesize \ 
\begin{tabular}{ll}
\hline
\textbf{Conditioning set} &  \\ \hline
$c$, $\Delta _{1}y_{it}$ &  \\ \hline
\textbf{Active Set} &  \\ \hline
(a) Domestic variables, $\ell =0,1$. & (b) Foreign counterparts, $\ell =0,1$.
\\ \hline
$\Delta _{1}y_{i,t-1}$ & $\Delta _{1}y_{i,t-\ell }^{\ast }$ \\ 
$\Delta _{1}r_{i,t-\ell }-\Delta _{1}\pi _{i,t-\ell }$ & $\Delta
_{1}r_{i,t-\ell }^{\ast }-\Delta _{1}\pi _{i,t-\ell }^{\ast }$ \\ 
$\Delta _{1}r_{i,t-\ell }^{L}-\Delta _{1}r_{i,t-\ell }$ & $\Delta
_{1}r_{i,t-\ell }^{L\ast }-\Delta _{1}r_{i,t-\ell }^{\ast }$ \\ 
$\Delta _{1}q_{i,t-\ell }-\Delta _{1}\pi _{i,t-\ell }$ & $\Delta
_{1}q_{i,t-\ell }^{\ast }-\Delta _{1}\pi _{i,t-\ell }^{\ast }$ \\ \hline
\multicolumn{2}{l}{Total number of variables in the active set $\mathbf{x}
_{t}$: $n=15$ (max)} \\ \hline\hline
\end{tabular}
}
\end{table} \clearpage

\section{Forecasting euro area output growth using ECB surveys of
professional forecasters}

\label{ECB survey of growth}

This application considers forecasting one-year ahead euro area real output
growth using the ECB survey of professional forecasters, recently analyzed
by \cite{diebold2018egalitarian}. The dataset consists of quarterly
predictions of 25 professional forecasters over the period 1999Q3 to 2014Q1.%
\footnote{%
We are grateful to Frank Diebold for providing us with the data set.} 
The predictions of these forecasters are highly correlated suggesting the
presence of a common factor across these forecasts. To deal with this issue
at the variable selection stage following \cite{sharifvaghefi2023variable}
we also include the simple average of the 25 forecasts in the conditioning
set, $\mathbf{z}_{t}$, as a proxy for the common factor in addition to the
intercept. We consider 39 quarterly forecasts (from 2004Q3 and 2014Q1) for
forecast evaluation, using expanding samples (weighted and unweighted) from
1999Q3. We also consider two simple baseline forecasts: a simple cross
sectional (CS) average of the professional forecasts, and forecasts computed
using a regression of output growths on an intercept and the CS average of
the professional forecasts.

Table \ref{msfe euro area ocmt} compares the forecast performance of OCMT
with and without down-weighting at the selection and forecasting stages, in
terms of MSFE. The results suggest that down-weighting at the selection
stage leaves us with larger forecasting errors. The MSFE goes from 3.765
(3.995) to 3.874 (4.672) in case of light (heavy) down-weighting. However,
the panel DM tests indicate that the MSFE among different scenarios are not
statistically significant, possibly due to the short samples being
considered. In Table \ref{ocmt vs lasso euro area}, we compare OCMT (with no
down-weighting at the selection stage) with Lasso, A-Lasso and boosting. The
results indicate that the OCMT procedure outperforms Lasso, A-Lasso and
boosting in terms of MSFE when using no down-weighting, light
down-weighting, and heavy down-weighting at the forecasting stage. It is
worth mentioning that OCMT selects 3 forecasters (Forecaster \#21 for
2004Q4-2005Q1, Forecaster \#7 for 2007Q2-2008Q3, and Forecaster \#18 for
2011Q2-2011Q3). This means that over the full evaluating sample, only 0.3
variables are selected by OCMT from the active set on average. In contrast,
Lasso selects 12.6 forecasters on average. Each individual
forecaster is selected for at least part of the evaluation period. As to be
expected, A-Lasso selects a fewer number of forecasters (9.8 on average) as
compared to Lasso (12.6 on average), and performs slightly worse. Boosting selects 11.6 forecasters on average. 

To summarize, we find that down-weighting at the selection stage of OCMT
leads to forecast deterioration (in terms of MSFE). OCMT outperforms Lasso,
A-Lasso and boosting, but the panel DM tests are not statistically
significant. Moreover, none of the considered big data methods can beat the
simple baseline models.

\begin{table}[tbph]
	\caption{{\protect\footnotesize Mean square forecast error (MSFE) and panel
			DM test of OCMT of one-year ahead euro area real output growth
			forecasts between 2004Q3 and 2014Q1 (39 forecasts)}}%
	\label{msfe euro area ocmt}\vspace{-0.2cm}
	
	\begin{center}
		\renewcommand{\arraystretch}{1.2}{\footnotesize 
  
			\begin{tabular}{cccccc}
				\hline\hline
				& \multicolumn{2}{c}{Down-weighting at$^{\dagger }$} &  &  &  \\ \cline{2-3}
				& Selection stage & Forecasting stage & \multicolumn{3}{c}{MSFE} \\ \hline
				(M$1$) & no & no & \multicolumn{3}{c}{3.507} \\ \hline
				\multicolumn{6}{c}{Light down-weighting, $\lambda =\left\{
					0.975,0.98,0.985,0.99,0.995,1\right\} $} \\ \hline
				(M$2$) & no & yes & \multicolumn{3}{c}{3.765} \\ 
				(M$3$) & yes & yes & \multicolumn{3}{c}{3.874} \\ \hline
				\multicolumn{6}{c}{Heavy down-weighting, $\lambda =\left\{
					0.95,0.96,0.97,0.98,0.99,1\right\} $} \\ \hline
				(M$4$) & no & yes & \multicolumn{3}{c}{3.995} \\ 
				(M$5$) & yes & yes & \multicolumn{3}{c}{4.672} \\ \hline
				\multicolumn{6}{c}{Pair-wise panel DM tests} \\ \cline{2-6}
				& \multicolumn{2}{c}{Light down-weighting} &  & \multicolumn{2}{c}{Heavy
					down-weighting} \\ \cline{2-3}\cline{5-6}
				& (M$2$) & (M$3$) &  & (M$4$) & (M$5$) \\ \cline{2-3}\cline{5-6}
				(M$1$) & -0.737 & -0.474 & (M$1$) & -0.656 & -0.741 \\ 
				(M$2$) & - & -0.187 & (M$5$) & - & -0.645 \\ \hline\hline
		\end{tabular}

}
	\end{center}
	
	\begin{flushleft}
		
		\footnotesize%
		Notes: The active set consists of 25 individual forecasts. The conditioning
		set consists of an intercept and the cross sectional average of 25 forecasts.
		
		$^{\dagger }$For each of the two sets of exponential down-weighting
		(light/heavy) forecasts of the target variable are computed as the simple
		average of the forecasts obtained using the down-weighting coefficient, $%
		\lambda $, in the \textquotedblleft light\textquotedblright\ or the
		\textquotedblleft heavy\textquotedblright\ down-weighting set under
		consideration.%
		
	\end{flushleft}

	\caption{{\protect\footnotesize Mean square forecast error (MSFE) and panel DM test of OCMT versus Lasso, A-Lasso and boosting of one-year ahead euro area real output growth forecasts between 2004Q3 and 2014Q1 (39 forecasts)}
	}%
	\label{ocmt vs lasso euro area}\vspace{-0.2cm}
	
	\begin{center}
		\renewcommand{\arraystretch}{1.2}{\footnotesize 

 \begin{tabular}{rccccccccc}
\hline\hline
& \multicolumn{9}{c}{MSFE under different down-weighting scenarios} \\ 
\cline{2-10}
& \multicolumn{3}{c}{No down-weighting} & \multicolumn{3}{c}{Light
down-weighting$^{\dag }$} & \multicolumn{3}{c}{Heavy down-weighting$^{\ddag }
$} \\ \cline{2-10}
OCMT &  & 3.507 &  &  & 3.765 &  &  & 3.995 &  \\ 
Lasso &  & 5.242 &  &  & 5.116 &  &  & 5.385 &  \\ 
A-Lasso &  & 7.559 &  &  & 6.475 &  &  & 6.539 &  \\ 
Boosting &  & 4.830 &  &  & 5.071 &  &  & 5.439 &  \\ \hline
& \multicolumn{9}{c}{Pair-wise Panel DM tests (All countries)} \\ 
\cline{2-10}
& \multicolumn{3}{c}{No down-weighting} & \multicolumn{3}{c}{Light
down-weighting} & \multicolumn{3}{c}{Heavy down-weighting} \\ 
\cline{2-4}\cline{5-7}\cline{8-10}
& Lasso & A-Lasso & Boosting & Lasso & A-Lasso & Boosting & Lasso & A-Lasso
& Boosting \\ \cline{2-4}\cline{5-7}\cline{8-10}
OCMT & -1.413 & -1.544 & -0.934 & -0.990 & -1.265 & -0.938 & -1.070 & -1.267
& -1.155 \\ 
Lasso & - & -1.484 & 0.819 & - & -1.589 & 0.144 & - & -1.527 & -0.417 \\ 
A-Lasso & - & - & 2.005 & - & - & 1.707 & - & - & 1.402 \\ \hline\hline
\end{tabular}

			\vspace{-0.2cm}}
	\end{center}
	
	\begin{flushleft}
		\footnotesize%
		Notes: The active set consists of forecasts by 25 individual forecasters.
		The conditioning set contains an intercept and the cross sectional average
		of the 25 forecasts.
		
		$^{\dagger }$ Light down-weighted forecasts are computed as simple averages
		of forecasts obtained using the down-weighting coefficient, $\lambda
		=\{0.975,0.98,0.985,0.99,0.995,1\}$.
		
		$^{\ddagger }$ Heavy down-weighted forecasts are computed as simple averages
		of forecasts obtained using the down-weighting coefficient, $\lambda
		=\{0.95,0.96,0.97,0.98,0.99,1\}$.%
	\end{flushleft}
	
\end{table}%

\end{document}